\documentclass[a4paper,reqno,11pt]{article}
\usepackage[a4paper,margin=3cm, marginparwidth=2cm]{geometry}
\usepackage[T1]{fontenc}
\usepackage{amsfonts,amsthm,amssymb,amsmath, mathdots, bbm,mathabx, bm, cite}
\usepackage{amssymb, amscd}
\usepackage{authblk}

\usepackage[makeroom]{cancel}
\usepackage[pdftex]{color,graphicx}
\usepackage[all,cmtip]{xy}
\usepackage{enumitem}
\usepackage[pagebackref, colorlinks = true, linkcolor = black, urlcolor  = black, citecolor = red]{hyperref}
\numberwithin{equation}{section}

\newcommand{\un}{\mathbbm{1}}

\newtheorem{theorem}{Theorem}[section]

\newtheorem{proposition}[theorem]{Proposition}
\newtheorem{lemma}[theorem]{Lemma}

\newtheorem{remark}[theorem]{Remark}
\newtheorem{conjecture}[theorem]{Conjecture}

\newcommand{\be}{\begin{equation}}
\newcommand{\ee}{\end{equation}}

\newcommand{\bE}{{\mathbb{E}}}

\newcommand{\bfS}{{\mathbf{S}}}

\newcommand{\cP}{{\mathcal{P}}}

\newcommand{\Tr}{{\mathrm{Tr}}}

\newcommand{\bsig}{ {\bm{\sigma}}}
\newcommand{\btau}{{\bm{\tau}}}
\newcommand{\bnu}{{\bm{\nu}}}

\newcommand{\brho}{ {\bm{\rho}} }

\numberwithin{equation}{section}

\definecolor{armygreen}{rgb}{0.29, 0.33, 0.13}

\renewcommand\leq\leqslant

\renewcommand\geq\geqslant

\title{Free cumulants and freeness for unitarily invariant random tensors}
\author[1]{Beno\^it Collins\footnote{collins@math.kyoto-u.ac.jp}}
\author[2]{Razvan Gurau\footnote{gurau@thphys.uni-heidelberg.de}}
\author[2,3]{Luca Lionni\footnote{luca.lionni@ens-lyon.fr}}

\affil[1]{Mathematics Department, Kyoto University, Kyoto, Japan.}
\affil[2]{Heidelberg University, Institut für Theoretische Physik, Philosophenweg 19, 69120 Heidelberg, Germany.}
\affil[3]{ENS de Lyon, Laboratoire de Physique, CNRS, 46 allée d’Italie, 69364 Lyon Cedex 07, France.}
\date{}

\begin{document}

\maketitle
\begin{abstract}
We address the question of the asymptotic description of random tensors that are local-unitary invariant, that is, invariant by conjugation by tensor products of independent unitary matrices. We consider both the mixed case of a tensor with $D$ inputs and $D$ outputs, and the case where there is a factorization between the inputs and outputs, called pure, which includes the random tensor models extensively studied in the physics literature. 

The finite size and asymptotic moments are defined using correlations of certain invariant polynomials encoded by $D$-tuples of permutations, up to relabeling equivalence. Finite size free cumulants associated to the expectations of these invariants are defined through invertible finite size moment-cumulants formulas. 

Two important cases are considered asymptotically: pure random tensors that scale like a complex Gaussian, and mixed random tensors that scale like a Wishart tensor. In both cases, we derive a notion of tensorial free cumulants associated to first order invariants, through moment-cumulant formulas involving summations over non-crossing permutations. The pure and mixed cases involve the same combinatorics, but differ by the invariants that define the distribution at first order. In both cases, the tensorial free-cumulants of a sum of two independent tensors are shown to be additive. A preliminary discussion of higher orders is provided. 

Tensor freeness is then defined as the vanishing of mixed first order tensorial free cumulants. The equivalent formulation at the level of asymptotic moments is derived in the pure and mixed cases, and we provide an algebraic construction of tensorial probability spaces, which generalize non-commutative probability spaces: random tensors converge in distribution to elements of these spaces, and tensor freeness of random variables corresponds to tensor freeness of the subspaces they generate.
\end{abstract}
\newpage

\tableofcontents

\section{Introduction}

In this paper, tensors $A$ are arrays of complex numbers of the form $A_{i^1 \ldots i^D ;\; j^1 \ldots j^D}$, where all indices take value in $\{1, \ldots, N\}$ for some $N\ge 1$. Tensors for which the components factor as $T_{i^1 \ldots i^D}\bar T_{j^1 \ldots j^D}$ are called \emph{pure}, in which case $A=T\otimes \bar T$, and tensors for which they do not are called \emph{mixed}. We study in this paper random tensors whose probability measure is \emph{local-unitary invariant}, or \textsf{LU}-\emph{invariant}, that is, invariant upon conjugation by a tensor product of $D$ unitary matrices: $A$ and $(U_1 \otimes \cdots \otimes U_D)\cdot A \cdot (U_1^\dagger \otimes \cdots \otimes U_D^\dagger)$ have the same distribution for any $U_1, \ldots, U_D \in U(N)$, where $U(N)$ is the set of $N\times N$ unitary matrices. For $D=1$, the case of unitarily invariant random matrices is recovered \cite{NicaSpeicher, Collins03, CMSS}.  

\

Examples of such probability distributions include the pure complex Gaussian \cite{1Nexpansion1, 1Nexpansion2, 1Nexpansion3, critical} as well as perturbed Gaussian models, for which the probability density function of the Gaussian is altered by the addition of an invariant potential, see e.g. \cite{Gurau-universality, uncoloring, bonzom-SD, Gurau-book, enhanced-1, Walsh-maps, multicritical, bonzom-review, Lionni-thesis, Gurau:2012ix, Gurau:2013pca}. Regarding  random tensors, see also \cite{Tensors1, Tensors12,Tensors2,Tensors3,Tensors4,Tensors5,Tensors6,Tensors7,Tensors8,Tensors9,Tensors10,Tensors11}. Such models were initially studied because they provide generating functions for random triangulations, relevant in random geometry and quantum gravity \cite{critical, Gurau-book, Lionni-thesis}. Example of mixed distributions include Wishart tensors formed by partially tracing pure tensors \cite{IonReview, DNL}, tensor product of $D$ random matrices \cite{CGL, CGL2}, or a GUE random matrix with subdivided indices. 

\

Another important motivation for studying \textsf{LU}-\emph{invariant} distributions comes from quantum entanglement. Consider a $D$-partite quantum system. Its state space is a Hilbert space with a tensor product structure $\mathcal{H} =\mathcal H _1 \otimes \cdots \otimes \mathcal H _D$, and we assume that for all $c\in \{1,\dots D\}$,  $\mathcal{H}_c \cong  \mathbb{C}^N$ for some $N\ge 1$. Seen as density matrices, \emph{pure} quantum state are projectors $\lvert \psi \rangle\langle \psi\rvert$, and fixing an orthonormal basis $\{\lvert i_c\rvert\}_{1\le i_c\le N}$ in each factor $\mathcal H _c$,   $\lvert \psi \rangle\in \mathcal{H}$ decomposes as:
\[
\lvert \psi \rangle =\sum_{i_1, \ldots, i_D=1}^N T_{i^1 \ldots i^D} \lvert i_1\rangle\otimes \cdots \otimes \lvert i_D\rangle \;.
\]
Component-wise, a pure density matrix $\lvert \psi \rangle\langle \psi\rvert$ corresponds to a pure tensor $T\otimes \bar T$, with a normalization condition. A \emph{mixed} density matrix $\rho$ is a positive semi definite operator:
\[
\rho  =\sum_{\textrm{all indices}}^N A_{i^1 \ldots i^D;\;j^1\ldots j^D}\ \lvert i_1\rangle\otimes \cdots \otimes \lvert i_D\rangle\otimes \langle j_1\rvert\otimes \cdots \otimes \langle j_D\rvert \;,
\]
with a normalization condition. Thus, mixed density matrices correspond to certain normalized mixed tensors. This explains the terminology. 

Local-unitary invariance plays an important role in quantum information: such transformations - which correspond to local changes of orthonormal basis in each $\mathcal H _c$ - do not affect the  entanglement  between the $D$ subsystems. In fact, local-unitary invariance is sometimes introduced as a definition of equivalent entanglement \cite{LU-1, LU-2, LU-3, LU-4}, although depending on the context, a more operational definition of equivalent entanglement is often preferred. In any case, local-unitary invariance always plays an important role in multipartite entanglement: entanglement measures are  \textsf{LU}-invariant, two-ways LOCC is equivalent to  \textsf{LU}-equivalence, and so on. In order to study the entanglement properties  of classes of quantum states that are equivalently entangled, one may introduce distributions on the tensor coefficients of the states and, since \textsf{LU}-transformations do not affect entanglement, it is necessary to require that these distributions are \textsf{LU}-invariant. Regarding random quantum states, see e.g.~\cite{IonReview, DNL, Tensors6, Tensors9, random-state1, random-state2, random-state3,  random-state4, random-state5, random-state6, random-state7, random-state8, MingoPopa1, MingoPopa2, funwithreplicas, random-state-52}.

\

In this paper, we address the question of the appropriate notion of moments for characterizing \textsf{LU}-invariant tensor distributions. Results in invariance theory justify our choice to define the moments for such \emph{finite size} distributions  as the expectations of the \emph{invariant homogeneous polynomials} \cite{poly-LU-inQI,Heroetal, Vrana1, Vrana2, TuckNort, BGR1, BGR2, BGR3}, which we call \emph{trace-invariants}. Our purpose is however to ultimately study the distributions \emph{asymptotically}, when the size of the index set (or the dimension) $N$ goes to infinity, while the number of indices (or subsystems) $D$ stays constant. Generalizing the approach of \cite{CMSS}, we discuss which quantities play the role of \emph{asymptotic moments}, a choice which generally differs for mixed and pure random tensors. A notion of \emph{order} is introduced, related to how the finite quantities must be rescaled in order to obtain the asymptotic moments. \emph{First order asymptotic moments} are those which require the strongest rescaling. For classical random matrices, they are given by the expectations of powers of the matrix, while for the tensors considered, a subset of invariants called \emph{melonic} is singled out \cite{1Nexpansion1, 1Nexpansion2, 1Nexpansion3, critical, Gurau-universality, uncoloring, bonzom-SD, Gurau-book, enhanced-1, Walsh-maps, multicritical, bonzom-review, Lionni-thesis, Gurau:2012ix, Gurau:2013pca}. 

\

For non-commutative random variables, \emph{freeness}, or \emph{free independence}, is a notion of independence which is weaker than the usual independence \cite{Voiculescu, Biane, NicaSpeicher, Speicher94, CMSS, MS04}. Usual independence can be formulated as the vanishing of classical cumulants. For random matrices, \emph{asymptotic freeness} in the limit $N\rightarrow \infty$  can be characterized by the vanishing of \emph{free cumulants} (or first order free cumulants) \cite{Speicher94}, a sequence of numbers which represent the same information  as the sequence of  first order asymptotic moments. Such free cumulants are defined as sums of products of first order asymptotic moments. The summations involve non-crossing partitions, and  by M\"oebius inversion one can write the first order asymptotic moments in terms of the free cumulants \cite{Biane, NicaSpeicher, Speicher94}.

Independent unitarily invariant random matrices are known to be asymptotically free, which implies that their free cumulants are additive. This gives access to the spectrum of the sum of two independent unitarily invariant random matrices whose spectra are known (asymptotic random Horn problem). 
Freeness plays a role in the study of random quantum states, see \cite{IonReview, DNL, MingoPopa1, MingoPopa2, random-state-52}. 

\

The pure complex Gaussian tensor ensemble $T, \bar T$ is obtained by requesting the coefficients of $T$ to be independent and identically distributed complex Gaussian variables. Appropriately normalized, the vector $\lvert \psi\rangle\; \propto\; \sum_{i_1, \ldots, i_D=1}^N T_{i^1 \ldots i^D} \lvert i_1\rangle\otimes \cdots \otimes \lvert i_D\rangle$ can be written as $\lvert \psi\rangle=U\lvert 0 \rangle$ where $\lvert 0 \rangle$ is a fixed pure state and $U$ is a Haar distributed random unitary matrix: the pure complex Gaussian tensor corresponds in this interpretation to taking a uniform distribution on pure quantum states. We call Wishart tensor the mixed tensor obtained by summing one of the indices of $T$ with the corresponding index of $\bar T$ of the same position (color) $c$ . This corresponds to partially tracing the associated density matrix over one of the subsystems $\mathcal{H}_c$, obtaining the density matrix induced on the remaining subsystems. 

\emph{The pure random tensors distributions whose asymptotic behavior we study here are assumed to scale at large $N$ like the pure complex Gaussian, in the sense that their asymptotic moments are obtained with the same rescaling at large $N$. However, no other assumption is made on the  asymptotic moments. The mixed tensors considered are those which scale like the Wishart tensor.} 

\ 

With these assumptions, using Weingarten calculus \cite{Weingarten, ColSni, Collins03}, we derive a notion of finite  size  free cumulants by studying the expansion of the generating function of the classical cumulants (related to the tensor HCIZ integral, see \cite{CGL, CGL2}) on the trace-invariants which play the role of moments for \textsf{LU}-invariant random tensors. Taking the asymptotics of these relations, we obtain combinations of  first order asymptotic moments which \emph{define  first order free cumulants for any random tensor} satisfying our scaling assumptions. By construction, these \emph{tensorial free cumulants are additive for independent random tensors}. These formulas involve the same combinatorial restriction in the pure and the mixed case considered, but the asymptotic distributions differ by the classes of invariants that arise at first order. 
These relations can be \emph{inverted} in the lattice product of non-crossing partitions: \emph{knowing the first order free cumulants is equivalent to knowing the  asymptotic distribution at first order.} 

Having derived the tensorial free cumulants, asymptotic tensor freeness can then be defined (at first order) as the vanishing of (first order) free cumulants involving free tensors. We then study the formulation of asymptotic tensor freeness in terms of asymptotic moments, both in the pure and mixed cases which requires  centering certain subgraphs  of the first-order trace-invariants by subtracting their asymptotic expectations. Due to their tree structure, multiplications of tensors corresponding to these subgraphs are seen to be elements of generalizations of unital algebras introduced in the last section.  Tensorial probability spaces are defined as a pair consisting in such a space, together with a trace. By construction, the pure and mixed random tensors we consider converge in distribution to elements of tensorial probability spaces, and tensor freeness can be formulated for elements of tensorial probability spaces, or for the spaces they generate. 

\

This paper is just a first step in a vast program of study of freeness for random tensors. Further work is needed to pursue the generalization of free probability to random tensors, as well as its applications to quantum information theory, and so on. See the related works \cite{Tensors11, Tensors12} addressing the question of free cumulants for real symmetric tensors (and freeness, for \cite{Tensors12}).

\section*{Acknowledgments}

R. G. is and L. L. has been supported by the European Research Council (ERC) under the European Union’s Horizon 2020 research and innovation program (grant agreement No 818066) and by the Deutsche Forschungs-gemeinschaft (DFG) under Germany’s Excellence Strategy EXC– 2181/1 – 390900948 (the Heidelberg STRUCTURES Cluster of Excellence). B. C. is supported by JSPS Grant-in-Aid Scientific Research
(B) no. 21H00987, and Challenging Research (Exploratory) no. 20K20882 and 23K17299. L. L. receives support from CNRS grant FEI 2024 - AAP TREMPLIN-INP.

\section{Notations and prerequisites}
\label{sec:notations-prerequ}

\paragraph{Partitions and permutations.} 
We will denote by $S_n$ be the group of permutations of $n$ elements, and $S_n^D$ the set of $D$-tuples of permutations $(\sigma_1, \ldots, \sigma_D)$, $\sigma_c\in S_n$. 
The cyclic permutation $(1, \cdots n)$ will be denoted by $\gamma_n$, and $\mathrm{id} $ denotes the identity permutation.
For $\sigma\in S_n$,   $\#\sigma$ denotes the number of disjoint cycles of $\sigma$  and $\lvert \sigma \rvert$  the minimal number of transpositions required to obtain $\sigma$ and $ \#\sigma + \lvert \sigma \rvert = n$.

\

We let $\cP(n)$ be the set of all partitions $\pi$ of $n$ elements. The notation $\#\pi$ is used for the number of blocks of $\pi\in S_n$, $B\in \pi$ denotes the blocks, and $|B|$ the cardinal of  $B$. The refinement partial order is denoted by $\le$:  $\pi'\le \pi$ if all the blocks of $\pi'$ are subsets of blocks of $\pi$. Furthermore,  $\vee$ denotes the joining of partitions: $\pi\vee\pi'$ is the finest partition which is coarser than both $\pi$ and $\pi'$. $1_n$ and $0_n$ respectively denote the one-block and the $n$ blocks partitions of $\{1,\dots, n\}$. The partition induced by the cycles of the permutation $\nu$ is denoted by $\Pi(\nu)$, hence $ \# \Pi(\nu)= \#\nu$.   
If $\pi \in \mathcal{P}(n)$ is such that $\pi \ge \Pi(\sigma)$ and if $B\in \pi$, $\sigma_{\lvert_{B}}$ refers to the permutation induced by $\sigma$ on $B$. The number of partitions of $n=\sum_{i\ge 1} id_i $ elements with $d_i$ parts of size $i$ is 
$ \frac{n!}{\prod_{i\ge 1} d_i! (i!)^{d_i} } $.

\

In particular we will encounter \emph{bipartite partitions} of bipartite sets $\{1,\dots n,\bar 1,\dots \bar n\}$. They are partitions of  such sets for which each block has the same number of elements in $\{1, \dots n\}$ and in $\{\bar 1,\dots \bar n \}$. We denote $\mathcal{P}(n ,\bar n)$ the set of such partitions. Due to this last condition, a bipartite partition $\Pi$ can also be seen as a partition $\pi\in \mathcal{P}(n)$ having the same number of blocks as $\Pi$, together with a permutation $\eta\in S_n$ which to each $s$  associates an element $ \overline{\eta(s)}$, up to relabelings of the elements in the preimages of the blocks of $\pi$. 
In detail, defining $\eta \sim_{\pi}\eta '$ if 
$\eta = \eta'\nu $ for some permutation $\nu$ with $\Pi(\nu) \le \pi$ and denoting by ${S_n} {/ {\sim_\pi}}$ the set of equivalence classes of permutations under the relation $ \sim_\pi$, we have $\Pi\Leftrightarrow(\pi, [\eta])$, where the blocks are mapped as:
\be
\label{eq:correspondance-bipartite-part-vs-perm}
 G \in \Pi \ \Leftrightarrow\ 
 G = B \cup \bar B \;, \;\; 
 \text{with}  \;\; B \in \pi, \; \bar B =
  \bigl\{ \overline{\eta(s)} \,,\; s\in B \bigr\} \;,
\ee
independently on the choice of representative $\eta$ of the class $[\eta]\in {S_n} {/ {\sim_\pi}}$.
As the number of partitions of $n=\sum_{i\ge 1} id_i $ elements with $d_i$ parts of size $i$ is 
$ \frac{n!}{\prod_{i\ge 1} d_i! (i!)^{d_i} } $ and the cardinal of 
${S_n} {/ {\sim_\pi}}$ is $\frac{n!}{\prod_{i\ge 1} (i!)^{d_i} } $, the number of bipartite partitions of a bipartite set with $n+n$ elements having $d_i$ parts with $i+i$ elements is
$ \frac{n!}{\prod_{i\ge 1} d_i! (i!)^{d_i} }  \, \frac{n!}{\prod_{i\ge 1} (i!)^{d_i} } $.
We denote by $1_{n,\bar n}$ the  bipartite partition with a single block $\{1,\dots n,\bar 1,\dots \bar n\}$.

\

We denote $D$-tuples of permutations by bold face greek symbols, $ (\sigma_1, \ldots, \sigma_D) =\bsig \in S_n^D$ and  we let $\Pi(\bsig)=\Pi(\sigma_1)\vee \ldots \vee \Pi(\sigma_D)$. We call the blocks of $\Pi(\bsig)$  the \emph{connected components} of $\bsig$ for reasons that will become clear below, and we denote  $K_\mathrm{m}(\bsig)=\#\Pi(\bsig)$ and respectively $K_\mathrm{m}(\bsig, \btau)=\#( \Pi(\bsig)\vee\Pi(\btau))$ if $\bsig, \btau\in S_n^D$. We say that $\bsig$ is \emph{connected} if $K_\mathrm{m}(\bsig)=1$. 

\

Given a $D$-tuple of permutations $(\sigma_1, \ldots, \sigma_D) = \bsig \in S_n^D$, we sometimes distinguish their domain and co-domain and consider the permutations as maps from $\{1, \ldots,  n\}$ to $\{\bar 1, \ldots, \bar n\}$, 
$s\mapsto \overline{\sigma_c(s)}$. 
In this case, for every color $c$, $\sigma_c$ yields a bipartite partition of the set  $\{1, \dots n, \bar 1 \ldots \bar n\}$ into $n$ pairs:
\be
\Pi_\mathrm{p}(\sigma_c) = 
\Bigl\{ \bigl\{ s, \overline{\sigma_c(s)} \bigr\} \, {\Big |}\,  1\le s \le n \Bigr\} \;,
\ee
corresponding with the point of view \eqref{eq:correspondance-bipartite-part-vs-perm} to the partition $0_n$ of $\{1,\dots n\}$ and the permutation $\sigma_c$.
We denote by $\Pi_\mathrm{p}(\bsig)$ the join of this partitions, which is also bipartite:
$ \Pi_\mathrm{p}(\bsig) = \bigvee_{ 1\le c\le D}
 \Pi_\mathrm{p}(\sigma_c)$.
We refer to the parts of  $\Pi_\mathrm{p}(\bsig)$
as the \emph{pure connected components of $\bsig$} and we denote by $K_\mathrm{p}(\bsig)$ their number. We say that $\bsig$  is \emph{purely connected} if $K_\mathrm{p}(\bsig)=1$.  The parts of $\Pi(\bsig)$ and those of $\Pi_\mathrm{p}(\bsig , \mathrm{id})$, where $(\bsig , \mathrm{id})=(\sigma_1, \ldots, \sigma_D, \mathrm{id})$ are in  one-to-one correspondence: one passes from the parts of $\Pi_\mathrm{p}(\bsig , \mathrm{id})$ to the ones of $\Pi(\bsig) $ by identifying $s=\bar s$.

\

The notation $\lambda\vdash n$  denotes that $\lambda$ is an integer partition of the integer $n$, that is, a multiplet of integers $\lambda_1\ge \ldots \ge\lambda_p>0$ such that $\sum_{i=1}^p\lambda_i =n$.  The $\lambda_i$ are the parts of $\lambda$, and we denote $\lambda=(\lambda_1, \ldots, \lambda_p)$, $p=\#\lambda$  the number of parts of $\lambda$, and $d_i(\lambda)$ is the number of parts of $\lambda$ equal to $i$, so that $n = \sum_{i=1}^n i d_i(\lambda) $ and $p=\#\lambda = \sum_{i=1}^n  d_i(\lambda) $.  If $\sigma \in S_n$, $\Lambda(\sigma)$  is the partition of $n$ given by the number of elements of the disjoint cycles  of $\sigma$.

\paragraph{Distance between permutations.}
The map $d:S_n^2 \to \mathbb{R}_+$, 
$d(\sigma, \tau) = \lvert \sigma \tau^{-1}\rvert$ defines a distance between permutations (see \cite{Biane} and e.g.~\cite{NicaSpeicher},  Lecture 23). For $\bsig, \btau\in S_n^D$ and $\eta\in S_n$, we let  $d(\bsig, \btau)=\sum_c d(\sigma_c, \tau_c)$ and $d(\bsig, \eta)=\sum_c d(\sigma_c, \eta)$.  Considering some permutations $\alpha_1, \ldots, \alpha_n$, since $d$ is a distance:
 \be
\label{eq:triangular-dist}
  \lvert \alpha_1 \alpha_{2}^{-1}\rvert + \cdots + \lvert \alpha_{n-1} \alpha_{n}^{-1}\rvert   \ge  \lvert \alpha_1 \alpha_{n}^{-1}\rvert , 
 \ee
 with equality  if and only if these permutations lie on a geodesic $\alpha_1 \rightarrow \alpha_2 \ldots \rightarrow  \alpha_n$.  
A permutation $\tau\in S_n$ satisfying $\lvert \tau\rvert + \lvert \tau\gamma_n^{-1}\rvert = n-1$ is said to be \emph{geodesic} or \emph{non-crossing} on $\gamma_n$.
More generally, we will use the notation $\tau \preceq \sigma$ for two permutations saturating the triangular inequality $\lvert \tau\rvert + \lvert \tau\sigma^{-1}\rvert = \lvert \sigma\rvert$, that is such that $\tau$ lies on a geodesic from the identity to $\sigma$. 

\

Fixing an ordering  of $n$ elements,  a non-crossing partition $\pi$ has no four elements $p_1<q_1<p_2<q_2 \in C$ such that $p_1, p_2\in B$ and $q_1, q_2\in B'$ for $B\neq B'$ two  blocks of $\pi$. As a consequence of \eqref{eq:Euler-charact-bipartite-map1} below,
$\tau \preceq \gamma_n$ if and only if
the partition $\Pi(\tau)$ is non-crossing on $1<\ldots < n$, and the cyclic ordering of the elements of each cycle of $\tau$ agrees with $\gamma_n$ in the sense that they are cyclically increasing.
Therefore, the sets $NC(n)$ of non-crossing partitions on $n$ ordered elements and $S_{\mathrm{NC}}(\gamma_n)=\{\tau\preceq\gamma_n\}$ are isomorphic posets (see Prop.~23.23 in \cite{NicaSpeicher}).  Changing $\gamma_n$ to another cycle with $n$ elements amounts to changing the ordering of the $n$ elements. 

\

From  \eqref{eq:Euler-charact-bipartite-map1} it also follows that the condition $\tau \preceq \sigma$ is equivalent to $\Pi(\tau)\le \Pi(\sigma)$, and for each cycle $B\in \Pi(\sigma)$, $\tau_{\lvert_B}$ is non-crossing on the cycle $\sigma_{\lvert_B}$.
If $\bsig, \btau \in S_n^D$, the notation $\btau \preceq \bsig$ indicates that for all $1\le c \le D$, $\tau_c\preceq \sigma_c$ and if $\eta\in S_n$, the notation $\btau\eta\preceq\bsig\eta$ means that for all $1\le c \le D$, $\tau_c\eta\preceq \sigma_c\eta$. 

\

The number of $\tau\in S_n$ such that $\tau \preceq \gamma_n$ is the \emph{Catalan number} $C_n$: 
\be
\label{eq:Catalan}
C_n=  \frac 1 {n+1}\binom{2n}{n} \;. 
\ee
This is also the number of non-crossing pairings (partitions into blocks of two elements) of $2n$ elements, and there exists an explicit bijection between non-crossing pairings of $2n$ elements and non-crossing partitions of $n$ elements.\footnote{Listing the $n$ elements along a cycle and adding two marks after each element, a non-crossing partition of the $n$ elements is bijectively mapped on a planar chord diagram among the $2n$ marks.} The M\"{o}bius function on the lattice of non-crossing partitions is:
\be
\label{eq:Moebius-on-NC}
\mathsf{M}(\pi) = \prod_{p\ge 1 }  \left[  (-1)^{p-1}C_{p-1} \right]^{d_p(\pi)} \; ,
\ee
where $d_p(\pi)$ is the number of blocks of $\pi$ with $p$ elements and we let $\mathsf{M}(\nu) = \mathsf{M}\bigl(\Pi(\nu)\bigr)$.

\paragraph{Genus.} 
If $\sigma, \tau \in S_n$, the Euler characteristics of $(\sigma, \tau)$ is:
\be
\label{eq:Euler-charact-bipartite-map}
\#\sigma+ \#\tau +\#(\sigma\tau^{-1}) -n =2 K(\sigma, \tau) - 2g(\sigma, \tau),
\ee
where $g(\sigma, \tau)\ge0$ is the \emph{genus} of $(\sigma, \tau)$. Such a pair of permutations is called a ``bipartite map''. Up to simultaneous conjugation of $\sigma$ and $\tau$, bipartite maps bijectively encode isomorphism classes of embeddings of bipartite graphs on orientable surfaces of genus $g$, where white and black vertices are respectively associated to the cycles of $\sigma$ and $\tau$ and the edges are labeled from 1 to $n$.  
One may express  \eqref{eq:Euler-charact-bipartite-map} as:
\be\label{eq:Euler-charact-bipartite-map1}
\lvert \tau\rvert + \lvert \tau\sigma^{-1}\rvert - \lvert \sigma\rvert = 2g(\sigma, \tau) + 2\bigl(\#\sigma - K(\sigma, \tau)\bigr), 
\ee
and since $\#\sigma\ge K(\sigma, \tau)$ with equality if and only if $\Pi(\tau)\le \Pi(\sigma)$, we see that a non-crossing permutation $\tau \preceq \gamma_n$ corresponds to a planar bipartite map with one white vertex while $\tau \preceq \sigma$ corresponds to a planar bipartite map $(\sigma, \tau)$ with only one  white vertex per connected component. 

\begin{figure}[!h]
\centering
\includegraphics[scale=1.3]{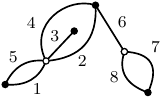}\hspace{2.5cm}\includegraphics[scale=1.3]{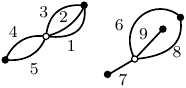}
\caption{
Left: The labeled planar bipartite map encoded by $\sigma=(12345)(876)$ and $\tau=(15)(3)(624)(78)$. Right: Two permutations satisfying $\tau\preceq \sigma$: $\sigma=(12345)(6789)$ and $\tau=(123)(45)(68)(7)(9)$. For each cycle of $\sigma$, the restriction of $\tau$ induces a non-crossing permutation}
\label{fig:planarmaps}
\end{figure}

\paragraph{Classical moment-cumulant formula.} The classical cumulants
of some random variables $x_1, \ldots, x_p$ are defined in terms of the expectations as:
\be
\label{eq:classical-cumulant}
k_n(x_1, \ldots, x_n) = \sum_{\pi\in \mathcal{P}(n)} \lambda_\pi \prod_{G\in \pi} \bE\Bigl[\prod_{i\in G}x_i\Bigr] \,,
\ee
with $\lambda_\pi= (-1)^{\#\pi-1} (\#\pi-1)! $ the M\"oebius function on the lattice of partitions.
If all $x_i$ are equal to $x$, we use the notation $k_n(x)$. 
By M\"{o}bius inversion in the lattice of partitions, we have: 
\be
\label{eq:classical-moment-cumulant}
\bE[x_1\cdots x_n]   = \sum_{\pi\in \mathcal{P}(n)} \prod_{G\in \pi}  k_{\lvert G \rvert}\Bigl(\{x_i\}, i\in G\Bigr) \,.
\ee

A particular case we need to consider is the case of a complex variable $f_s,\bar f_s$, with the additional assumption that only the expectations and cumulants with the same number of $f$s and $\bar f$s are non zero. In that case, the sums above are restricted to bipartite partitions of bipartite sets:\footnote{In terms of partitions of $n$ elements and equivalence classes of permutations this is:
\[
\begin{split}
 \bE[f_1\dots f_{n} \bar f_{\bar 1}\dots \bar f_{\bar n}] 
 = & \sum_{\pi \in \mathcal{P}(n) } \sum_{[\eta] \in S_n /{\sim_\pi} } \prod_{B \in \pi} k_{|B|}\bigl(\{ f_i , \bar f_{\overline{\eta(i)}}\}_{i\in B}  \bigr)  \; ,\crcr
k_{n }[f_1,\dots f_{n}, \bar f_{\bar 1},\dots \bar f_{\bar n}] 
 = & \sum_{\pi \in \mathcal{P}(n) }
 \sum_{ [\eta ] \in S_n/{\sim_\pi} } \lambda_{\pi}\prod_{B\in \pi} \bE\Bigl[ \prod_{i\in B} f_i \; \bar f_{\overline{\eta(i)} } \Bigr] \; .
\end{split}
\]}
\be
\begin{split}
 \bE[f_1\dots f_{n} \bar f_{\bar 1}\dots \bar f_{\bar n}] 
 = & \sum_{\Pi \in \mathcal{P}(n,\bar n) } \prod_{G  = B\cup \bar B\in \Pi} k_{|B|}(\{ f_i \}_{i\in B} , \{\bar f_{\bar j }\}_{\bar j\in \bar B}  )  \; ,\crcr
k_{n }[f_1,\dots f_{n}, \bar f_{\bar 1},\dots \bar f_{\bar n}] 
 = & \sum_{\Pi \in \mathcal{P}(n,\bar n) } \lambda_{\Pi}\prod_{G= B\cup \bar B\in \Pi} \bE\Bigl[ \prod_{i\in B} f_i  \prod_{\bar j\in \bar B} \bar f_{\bar j  } \Bigr] \; .
\end{split}
\ee
The exact same holds for bipartite distributions of a couple of random variables. Below we will often use the notation $T,\bar T$ to designate such bipartite couples of variables: while sometimes we will specify $\bar T$ to be the complex conjugate of $T$, unless otherwise specified, the formulas apply for both cases. 

\paragraph{Weingarten functions.}The Weingarten functions appear when integrating over unitary matrices \cite{Collins03, ColSni}. One has for $n\le N$: 
\be
\label{eq:WeinDef}
\int dU \;   U_{i_1a_1}\cdots U_{i_na_n}\overline{U_{j_1b_1}\cdots U_{j_nb_n}}=
\sum_{\sigma,\tau\in S_n}
 \left(  \prod_{s=1}^n \delta_{i_s , j_{\sigma(s)}}  \right)
  \left( \prod_{s=1}^n \delta_{a_s , b_{\tau(s)}} \right)  \, W^{(N)} (\sigma\tau^{-1})  \; ,
\ee
where $dU$ is the normalized Haar measure on the group $U(N)$ of $N\times N$ unitary matrices, and the $W^{(N)}$ are the Weingarten functions  \cite{ColSni}:
\begin{equation}
\label{eq:expansion-wg-1}
W^{(N)}(\nu) = N^{-n}\sum_{k\ge 0}\;\sum_{\substack{{\rho_1,\ldots , \rho_k\,\in\, S_n\setminus{\mathrm{id}},}\\{ 
 \rho_1\cdots \rho_k =\nu}}}\;(-1)^k\,N^{-\sum_{i=1}^k\lvert\rho_i\rvert } \; ,
\end{equation}
which at large $N$ behave as:
\be
W^{(N)} ( \nu) = \mathsf{M}(\nu) \,  N^{- n - \lvert\nu\rvert} (1+O(N^{-2})) \; .
\ee

We note that the Weingarten functions are class functions $W(\nu) = W(\eta\nu \eta^{-1})$ for any $\eta \in S_n$.

\newpage

\section{Unitarily invariant random matrices}

In this section we follow \cite{CMSS}. A \emph{unitarily invariant random matrix} $M=\{M_{i,j}\}_{1\le i,j\le N}$ is such that for any $U\in U(N)$, $M$ and $UMU^\dagger$ have the same distribution.

\subsection{Moments of unitarily invariant random matrices}
\label{sec:moments-unitarily-invariant-random-matrices}

\paragraph{Trace-invariants of matrices.} The trace-invariants of a matrix $M=\{M_{i,j}\}_{1\le i,j\le N}$ are the products of traces of powers of that matrix:
$
\Tr\bigl(M^{\lambda_1}\bigr) \cdots \Tr\bigl(M^{\lambda_p}\bigr)$, $\lambda_1\ge \ldots \ge\lambda_p>0. 
$
We denote this product in a more compact manner as $\Tr_\lambda(M)$, where $\lambda_1\ge \ldots \ge\lambda_p>0$ are the parts of the integer partition $\lambda\vdash n$. If $\sigma\in S_n$ and the cycle-type of $\sigma$ is $\Lambda(\sigma)=(\lambda_1, \ldots, \lambda_p)$, then we let $\Tr_\sigma(M) = \Tr_{\Lambda(\sigma)}(M)$.   

The trace-invariants are unitarily invariant homogeneous polynomials in the matrix entries $M_{i,j}$ and $\Tr_\lambda(M)$, $\lambda\vdash n\le N$ form a basis in the ring of unitarily invariant polynomials of $M$, and by extension of sufficiently regular unitarily invariant functions\footnote{For normal matrices, by diagonalization this is equivalent to the fact that products of power sums of degree $n\le N$ of the eigenvalues form a basis of the ring of symmetric polynomials in the eigenvalues.}.

\paragraph{Finite size moments.}Due to the unitary invariance, the appropriate moments for a $N\times N$ unitarily invariant random matrix $M$ are the expectations of the trace-invariants:
\be
\label{eq:expectation-trace-invariants-matrix}
\mathbb{E}\bigl[\Tr_\lambda(M)\bigr], \quad \lambda\vdash n\le N. 
\ee
This follows from the remark that if $f$ is a polynomial (or by extension, a sufficiently regular function, see Sec.~\ref{sec:basis}) in the entries of $M$, then there exists a unique set of coefficients $\{c_\lambda\}$ such that:\footnote{This follows from 
$\mathbb{E}\bigl[ f(M) \bigr] = \int dU \; \mathbb{E}\bigl[ f( UMU^\dagger ) \bigr] 
 = \mathbb{E}\Bigl[  \int dU \; f( UMU^\dagger ) \Bigr]$,
with $dU$ the Haar measure. If the polynomial is of degree at most $N$, computing the integral using the Weingarten formula \eqref{eq:WeinDef} yields the desired decomposition. If the polynomial is of degree larger than $N$, one may apply the Cayley-Hamilton theorem.}
\be
\mathbb{E}\bigl[f(M)\bigr]= \sum_{n=1}^N \sum_{\lambda\vdash n} c_\lambda\; \mathbb E\bigl[\Tr_\lambda(M)\bigr] \; .
\ee

\subsection{Asymptotic moments}
\label{sec:asymptotic-moments-unitarily-invariant-random-matrices}

\paragraph{Asymptotic characterization of the distribution.} In principle one needs the expectations $\mathbb E\bigl[\Tr_\lambda(M)\bigr]$ to describe unitarily invariant formal series. On the other hand, the dominant contributions in the $N\rightarrow \infty$  limit
to these expectations usually do not contain more information on the asymptotic distribution than the subset of expectations $\mathbb E[\Tr(M^n)]$, $n\ge 1$. This is due to the fact that the random matrix ensembles classically studied, like the GUE\footnote{Defined by the probability measure $  e^{-\frac{N}{2} \Tr(M^2)} \; \prod_{a,b} dM_{ab}$ for $M$ Hermitian.}, Ginibre or Wishart ensembles, share the property that the expectations of trace-invariants factorize asymptotically:
\be
\mathbb E\Bigl[ \prod_{i=1}^p\Tr\bigl(M^{\lambda_i}\bigr) \Bigr] \sim_{N\rightarrow \infty}  \prod_{i=1}^p \mathbb E\Bigl[\Tr\bigl(M^{\lambda_i}\bigr) \Bigr].\ee
In fact, for $\lambda$ having more than one part, one has to dig quite far in the $1/N$ expansion of $\mathbb{E}\bigl[\Tr_\lambda(M)\bigr]$ to recover information on the asymptotic correlations between the $\Tr(M^\lambda_i)$, which is captured by the dominant contribution to the classical cumulants $k_p(\Tr M^{\lambda_1}, \ldots, \Tr M^{\lambda_p})$.

More precisely, for classical random matrix ensembles such as the ones mentionned above, the dominant contribution to the classical cumulants takes the form (see e.g.~\cite{MS04, CMSS}): 
\be
\label{eq:asympt-class-cum-matrix}
\lim_{N\rightarrow \infty}\, \frac{1}{N^{2-p}}\,k_p\bigl(\Tr(M^{\lambda_1}), \ldots, \Tr(M^{\lambda_p})\bigr) =  \varphi_\lambda(m), 
\ee
that is, the classical cumulants scale asymptotically as $N^{2-p}$, and we denoted the asymptotic coefficient as:
\be
\label{eq:rescaled-cumulants-matrix}
\varphi_\lambda(m) = \varphi_{\lambda_1, \ldots, \lambda_p}(m),
\ee 
where $\lambda = (\lambda_1, \ldots ,\lambda_p)\vdash n$. If $\lambda$ has only one part we use the notation  $\varphi_n(m)= \varphi_\lambda(m)$. 

Recalling that 
$\Lambda(\sigma)$ denotes the cycle-type of the permutation $\sigma$, we define for 
$\pi\ge \Pi(\sigma)$ the multiplicative extension:
\be
\label{eq:mult-expansion-asympt-moments}
\varphi_{\pi, \sigma}(m) = \prod_{B\in \pi} \varphi_{\Lambda( \sigma_{\lvert_{B}}) }(m) \; ,
\ee
and we stress that with this notation, if the cycle type of $\sigma$ is $\Lambda(\sigma) =\lambda=(\lambda_1,\dots\lambda_p)$, then:
\be
\label{eq:extcumulants}
 \varphi_{\Pi(\sigma),\sigma}(m) = \prod_{B\in \Pi(\sigma)} \varphi_{\Lambda( \sigma_{\lvert_{B}}) }(m) =\prod_{i=1}^p \varphi_{\lambda_i}(m) \neq
 \varphi_{\lambda_1,\dots , \lambda_p}(m)  = \varphi_{\Lambda(\sigma)}(m) = 
 \varphi_{1_n,\sigma} (m)\;.
\ee
Below we will also use the notation $\varphi_{\sigma}=  \varphi_{1_n,\sigma} $.

In the case of several matrices $M_1, \ldots, M_n$, we generalize the notation in the obvious manner.
For instance, for $\tau\in S_n$, $\pi\in \mathcal{P}(n)$
with $\pi\ge \Pi(\tau)$ we have the following scaling: 
\be
\label{eq:Phi-higher-order-0}
\lim_{N\to \infty} \frac{1}{N^{2\#\pi - \#\tau}  }
\prod_{ G\in \pi} k_{\#(\tau \lvert_{_G})} \biggl(\Bigl\{\Tr\Bigl(\prod _{s\in \tau_i}M_s\Bigr)\Bigr\}_{\tau_i \textrm{ cycle of }\tau \lvert_{_G}}\biggr) = \varphi_{\pi, \tau} (m_1, \ldots, m_n) \;,
\ee
where the matrix product inside the trace follows the cyclic ordering of the elements in the $\tau_i $.

\paragraph{Order of dominance.} We define the \emph{order of dominance} of a classical cumulant as $2$ minus its dominant scaling in $N$. If \eqref{eq:asympt-class-cum-matrix} is satisfied, the order of dominance of  $k_p\bigl(\Tr(M^{\lambda_1}), \ldots, \Tr(M^{\lambda_p})\bigr) $ is $p$. By definition the invariants with the largest scaling are of order $1$, while the order of the other invariants indicates how much more they are suppressed in scaling than the dominant invariants. We have that:
\begin{itemize}
 \item[-] the asymptotic distribution is described at first order by $\{\varphi_n(m )\}_{n\ge 1} $, the  limits of the rescaled expectations $\frac{1}{N}\mathbb E[\Tr(M^n)]$. This first order information fixes for instance the asymptotic spectrum of $M$, if $M$ is a normal matrix;
 \item[-]  the fluctuations of order $p$ (the correlations between $p$ eigenvalues) are encoded in  the order $p$ invariants $\varphi_{\lambda_1, \ldots, \lambda_p}(m)$.
\end{itemize}

The factorization of the expectations is seen as follows.
The expectations of products of appropriately normalized traces admit $N\rightarrow \infty$ limits:
\be
\label{eq:asympt-mom-matrix}
\lim_{N\rightarrow \infty} \frac 1 N\, \bE \bigl[\Tr (M^n)\bigr] = \varphi_{n} (m) \qquad \mathrm{and}\qquad \lim_{N\rightarrow \infty} \frac 1 {N^p}\, \bE \bigl[\Tr_{\lambda} (M)\bigr]  = \prod_{i}\varphi_{\lambda_i} (m) \; ,
\ee
and going further in the $1/N$ expansion, we find that the classical cumulant contributes to the corresponding normalized expectation at a lower $1/N$ order:
\be
 \frac 1 {N^p}\, \bE \bigl[\Tr_{\lambda} (M)\bigr]  = \prod_{i}\varphi_{\lambda_i} (m) + \dots
  +  \frac{1}{N^{p}} k_p \bigl(\Tr(M^{\lambda_1}), \ldots, \Tr(M^{\lambda_p})\bigr),
  \; .
\ee
where all the terms apart from the first one are $O(1/N)$, and the rightmost term is the most suppressed (of order $\frac{1}{N^{2(p-1)}}$). 
\paragraph{The Wishart ensemble.} 
A particular example we will be interested in below is that of a Wishart random matrix.
Let $X$ be a $N\times N'$ random matrix with independent and identically distributed complex Gaussian entries (a Ginibre matrix) satisfying\footnote{Equivalently with probability measure $e^{-N\Tr(XX^\dagger)} \prod_{ab}dX_{ab} d\overline{X_{ab}}$.}  $\bE \left[ X_{i_1, i_2} \bar X_{j_1, j_2} \right]  =  \delta_{i_1, j_1} \delta_{i_2, j_2} /N $, and assume that $N' \sim t N$ asymptotically for some $t\in(0,\infty)$.  The moments of the Wishart random matrix $W = XX^\dagger$ are (see e.g.~\cite{DNL}):
\be
\label{eq:map-exp-wishart}
\bE \bigl[\Tr \, W^n \bigr] = \sum_{\tau\in S_n}N^{\#(\gamma_n\tau^{-1})-n} {N'}^{\,\#\tau } = \sum_{\tau\in S_n}N^{\#(\gamma_n\tau^{-1}) + \#\tau -n } {t}^{\,\#\tau }  \;,
\ee
where we recall that $\gamma_n$ is the cycle $(12 \ldots n)$. From Sec.~\ref{sec:notations-prerequ},  we have that $\#(\gamma_n\tau^{-1}) + \#\tau \le 1+n$, with equality if and only if $\tau \preceq \gamma_n$, hence $\frac{1}{N}\bE \bigl[\Tr \, W^n \bigr] \sim \varphi_n(w_t)$,
where:
\be
\label{eq:asymptotic-moments-wishart-D1}
\varphi_n(w_t)  = \sum_{\tau \preceq \gamma_n} t^{\#\tau} \;,
\ee
 is the $n$th moment of the Mar\v{c}enko-Pastur law of parameter $t$. For the square Wishart, $t=1$, one gets the Catalan number:
 \be
 \label{eq:asymptotic-moments-square-wishart-D1}
 \varphi_n(w) =C_n,
 \ee 
where $w=w_1$. The asymptotic behavior of the cumulants $k_p(\Tr W^{\lambda_1}, \ldots, \Tr W^{\lambda_p})$ is the sum \eqref{eq:map-exp-wishart} with a connectivity condition (see for instance \cite{MS04}):
\be
 k_p\Bigl(\Tr(W^{\lambda_1}), \ldots, \Tr(W^{\lambda_p})\Bigr) = \sum_{\substack{{\tau\in S_n}\\{\Pi(\tau) \vee \Pi(\gamma_\lambda)=1_n}}}N^{\#(\gamma_\lambda\tau^{-1}) + \#\tau - n } {t}^{\,\#\tau} \;,
 \ee
where $\lambda=(\lambda_1, \ldots, \lambda_p)$ and $\gamma_\lambda$ is a permutation of cycle-type $\lambda$. Applying \eqref{eq:Euler-charact-bipartite-map} to the map $(\tau, \gamma_\lambda)$, which is connected, yields $\#(\gamma_\lambda\tau^{-1}) + \#\tau-n= 2-p-2g(\tau, \gamma_\lambda) \le 2 - p$, hence we reproduce the scaling advertised in \eqref{eq:asympt-class-cum-matrix} with:
\be
\varphi_{\lambda_1, \ldots, \lambda_p}(w_t) = \sum_{\substack{{\tau\in S_n }\\[+0.2ex]{\Pi(\tau) \vee \Pi(\gamma_\lambda)=1_n }\\[+0.2ex]{  g(\tau, \gamma_\lambda)=0}}} {t}^{\,\# \tau } \;.
\ee

\subsection{Free cumulants and first order freeness}
\label{subsub:free-cumulants}

\paragraph{Free cumulants.} The free cumulants are the central tools of free probability. For a $N\times N$ unitarily invariant random matrix $M$, they are defined  through the relations:  
\be
\label{eq:cum-mom}
\kappa_n(m) = \sum_{\tau \preceq \gamma_n} \varphi _{\,\Pi(\tau), \tau}(m)\;  \mathsf{M}(\gamma_n\tau^{-1}) \; , 
\ee
where we recall that if $\tau$ has cycle-type $(\lambda_1, \ldots, \lambda_p)$ then $\varphi _{\,\Pi(\tau), \tau}(m)=\prod_i \varphi _{\lambda_i}(m)$, and $\mathsf{M}$ is the M\"{o}bius  function on the lattice of non-crossing partitions  \eqref{eq:Moebius-on-NC}. By M\"{o}bius inversion  \cite{NicaSpeicher}, we have: 
\be
\label{eq:mom-cum}
\varphi_n(m) = \sum_{\tau \preceq \gamma_n} \kappa_{\,\Pi(\tau), \tau}(m) \;,
\ee
where, for $\tau$ having cycle-type $(\lambda_1, \ldots, \lambda_p)$, we let $\kappa _{\,\Pi(\tau), \tau}(m)=\prod_{i=1}^p \kappa_{\lambda_i}(m)$. 
These are the so-called \emph{free moment-cumulant formulas}. They generalize in the obvious manner to $n$ distinct matrices $M_1, \ldots, M_n$, yielding $\kappa_n (m_1, \ldots, m_n)$.

The information encoded in the free cumulants $\{\kappa_n(m)\}_{n\ge 1}$ is equivalent to that encoded in the first order asymptotic moments $\{\varphi_n(m)\}_{n\ge 1}$. 
For a Wishart random matrix for instance, comparing \eqref{eq:asymptotic-moments-wishart-D1} and \eqref{eq:asymptotic-moments-square-wishart-D1} to \eqref{eq:mom-cum} we have:
\be
\label{eq:cumulants-wishart-D1}
\kappa_n(w_t)=t,\qquad \mathrm{and}\qquad \kappa_n(w)=1 \;,
\ee
while for the GUE we have $\kappa_n(m)=\delta_{n,2}$.

\paragraph{Freeness.} 
Just like two random variables $x,y$ are independent if and only if their mixed (classical) cumulants vanish, some \emph{non-commutative} random variables $a,b,\ldots$ 
are said to be \emph{free}, if the free cumulants involving two different variables from this set at least vanish, that is, $\kappa_n(a_1, \ldots, a_n) =0$ if $a_i\neq a_j$ for some $i,j$. 
From the free moment-cumulant formulas, one can equivalently formulate freeness by the vanishing of centered mixed moments, see Sec.~\ref{sub:matrix-freeness} for more details. The original definition in the moments formulation is due to Voiculescu \cite{Voiculescu}, and the equivalent definition in terms of free cumulants to Speicher  \cite{Speicher94}.
Two random matrices $A,B$ converging to free non-commutative variables are said to be \emph{asymptotically free}. 
As an easy consequence, since the free cumulants are multilinear, just like classical cumulants are additive for independent random variables, 
 two free random variables $a,b$ satisfy:
\be
\label{eq:additiv-cum}
\kappa_n(a + b) =  \kappa_n(a) + \kappa_n( b). 
\ee

Voiculescu proved \cite{Voiculescu} that   two independent random matrices $A,B$ which almost surely have asymptotic spectra, such that the distribution of $B$ is invariant under conjugation by unitary matrices are asymptotically free. As a consequence, their free cumulants satisfy \eqref{eq:additiv-cum}.

\paragraph{Higher order free cumulants.}

Higher order free cumulants play the same role as $\{\kappa_n(m)\}_{n\ge 1}$ for $\{\varphi_n(m)\}_{n\ge 1}$, but for higher order asymptotic moments $\{\varphi_\lambda(m)\}_{\lambda\vdash n}$. Higher order free moment-cumulant formulas \cite{CMSS} define the higher order free cumulants
 $\kappa_\lambda(m) = \kappa_{\lambda_1, \ldots, \lambda_p}(m)$ for $\lambda=(\lambda_1, \ldots, \lambda_p)\vdash n$, and fixing some $q\ge 2$, the sets of cumulants 
$\{\kappa_{\lambda_1, \ldots, \lambda_p}(m)\}_{\lambda_1\ge \ldots\ge \lambda_p>0} $ for $p\le q$ and corresponding moments $\{\varphi_{\lambda_1, \ldots, \lambda_p}(m)\}_{\lambda_1\ge \ldots\ge \lambda_p>0}$   for $p\le q$  encode equivalent data on the asymptotic distribution of $M$.
The higher order free moment-cumulant formulas  are more complicated than the first order ones, see \cite{CMSS,LL-higher}.

\subsection{Moment-cumulant relations at finite $N$}

We have so far discussed the asymptotic moments and free cumulants. The free moment-cumulant formulas of first and higher orders are obtained as the $N\rightarrow \infty$ limits of finite-size moment-cumulant relations.

Let $\tau\in S_n$,  $\pi\ge \Pi(\tau)$, and for $G\in \pi$, denote by $p_G$ the number of cycles of $\tau_{\lvert_{G}}$ and $\lambda_1^G\ge \ldots \ge \lambda_{p_G}^G$ the ordered list of the number of elements of these cycles. Following
\cite{CMSS}, we define the finite $N$ multiplicative extension of the classical cumulants:
\be
\label{eq:Phi-higher-order}
\Phi_{\pi, \tau} [M]=  \prod_{ G\in \pi} k_{p{_G}}\left(\Tr(M^{\lambda_1^G}), \ldots, \Tr(M^{\lambda_{p_G}^G})\right) \; .
\ee
With this notation, the classical moment-cumulant formulas read:
\be
\bE\left[\Tr_{\tau} (M)\right] =
\sum_{\substack{ {\pi \in  \mathcal{P}(n) }\\ {\pi \ge \Pi(\tau)} } } \Phi_{\pi,\tau} [M] \;, \qquad \Phi_{1_n,\tau}[M] = 
\sum _{\substack{ {\pi \in  \mathcal{P}(n) }\\ {\pi \ge \Pi(\tau)} } } 
\lambda_\pi  \prod_{B\in \pi} \bE\left[\Tr_{\tau|_B} (M)\right] \;,
\ee
with $\lambda_{\pi}$ the M\"oebius function on the lattice of partitions. Naturally, $\Phi_{1_n,\tau}[M] = \Phi_{\tau}[M] $ is just the classical cumulant corresponding to the integer partition $\Lambda(\tau)$.
For $n$ distinct matrices $M_1, \ldots, M_n$, denoting $\vec M=(M_1, \ldots, M_n)$, we extend the notation to: 
\be
\label{eq:Phi-higher-order-0}
\Phi_ {\pi, \tau} [\vec M]=  \prod_{ G\in \pi} k_{p_G}\left(\left\{\Tr\left(\vec \prod _{j\in \tau_i}M_j\right)\right\}_{\tau_i \textrm{ cycle of }\tau \lvert_{_G}}\right) \;.
\ee

The finite $N$ precursors of the free cumulants are then defined  for $\sigma\in S_n$ and $\pi\ge \Pi(\sigma)$ 
(see  \cite{CMSS},  (24) and (14)) via:
\be
\label{eq:Cum-higher-order-fin-N-gen}
\mathcal{K}_{\pi, \sigma} [\vec M]=  \sum_{\tau\in S_n} \sum_{ \substack{ {\pi'\in \mathcal{P}(n)}\\{\pi\ge \pi' \ge \Pi(\tau)}} } \Phi_{\pi', \tau} [\vec M] \sum_{\substack{{\pi''\in \mathcal{P}(n)}\\{\pi\ge \pi''  \ge \Pi(\sigma)\vee\pi'}}} \lambda_{\pi'', \pi} \prod_{G\in \pi''} W^{(N)}\left(\sigma_{\lvert_{_G}}\tau_{\lvert_{_G}}^{-1}\right) \;,
\ee
where the $W^{(N)}$ are the Weingarten functions 
and $ \lambda_{\pi'',\pi}$ for $\pi''\le \pi$ is again the M\"oebius function $\lambda_{\pi'',\pi} = \prod_{B\in \pi} (|\pi''|_{B}| -1 )!(-1)^{|\pi''|_{B}| -1}$. The precursors of the free cumulants can also  be written directly in terms of expectations of traces. For instance, for $\mathcal{K}_{1_n, \sigma}[M] = \mathcal{K}_{\sigma}[M]$,  we have:
\be
\mathcal{K}_{\sigma}[M] = \sum_{\tau \in S_{n}}  \sum_{ \substack{ {\pi'\in\mathcal{P}(n) } \\ {\pi'\ge \Pi(\tau)} } }  \Phi _{\pi', \tau} [M] \sum_{ \substack{ {\pi'' \in \mathcal{P}(n) } \\ {\pi'' \ge \Pi( \sigma)\vee \pi' } } } \lambda_{\pi''} \prod_{G\in \pi''}  W^{(N)} 
\Bigl(\sigma_{\lvert_{_G}}\tau_{\lvert_{_G}}^{-1}\Bigr) \;,
\ee
and exchanging the sums we get:
\be\label{eq:freecumumome}
\begin{split}
  \mathcal{K}_{\sigma}[M] & = \sum_{\tau \in S_{n}} \sum_{ \substack{  {\pi'' \in \mathcal{P}(n) }\\{\pi'' \ge \Pi( \sigma)\vee \Pi(\tau) } } } \lambda_{\pi''} \prod_{G\in \pi''}  W^{(N)} 
\Bigl(\sigma_{\lvert_{_G}}\tau_{\lvert_{_G}}^{-1} \Bigr) 
 \sum_{ \substack{ {\pi' \in \mathcal{P}(n)} \\{\pi''\ge \pi'\ge \Pi(\tau) } }}  \Phi _{\pi', \tau} [M] \crcr
 & =  \sum_{\tau \in S_{n}} \sum_{ \substack{  {\pi'' \in \mathcal{P}(n) }\\{\pi'' \ge \Pi( \sigma)\vee \Pi(\tau) } } } \lambda_{\pi''} \prod_{G\in \pi''}  W^{(N)} 
\Bigl(\sigma_{\lvert_{_G}}\tau_{\lvert_{_G}}^{-1}\Bigr)\; \bE\left[\Tr_{\tau|_G} (M)\right]
\; .
\end{split}
\ee

The relations between classical and free cumulants in  \eqref{eq:Cum-higher-order-fin-N-gen} can be inverted as \cite{CMSS}:
\be
\label{eq:Mom-higher-order-fin-N-gen}
\Phi_{\pi, \sigma} [\vec M]=  \sum_{\substack{{\tau\in S_n, \,\pi'\in \mathcal{P}(n)}\\[+0.2ex]{\pi\ge \pi' \ge \Pi(\tau)}\\{\pi'\vee \Pi(\sigma\tau^{-1}) = \pi}}} \mathcal{K}_{\pi', \tau} [\vec M] \, \cdot \, N^{\#(\sigma\tau^{-1})} \;.
\ee
The finite $N$ classical cumulant $\Phi_{\pi, \sigma} [\vec M]$ is of order $N^{2\#\pi - \#\sigma }$ and the precursor $\mathcal{K}_{\pi, \sigma} [\vec M]$ of the free cumulant $\kappa_{\pi, \sigma} (\vec m)$ is of order $N^{2\#\pi -\#\sigma-n} $. Their  rescaled limits converge to the asymptotic versions we defined previously:
 \be
\label{eq:rescaled-higher-moments-and-cumulants}
 \varphi_{ \pi, \sigma} (\vec m) = \lim_{N\rightarrow \infty } N^{\#\sigma - 2\#\pi  }  \Phi_{\pi, \sigma} [\vec M] \; , 
 \qquad  \kappa_{\pi, \sigma} (\vec m) = \lim_{N\rightarrow \infty } N^{\#\sigma + n- 2\#\pi }  \mathcal{K}_{\pi, \sigma} [\vec M]. 
\ee
The free moment-cumulant formulas of arbitrary order are recovered \cite{CMSS, LL-higher} by taking the $N\rightarrow \infty$ limit of the finite $N$ relations  \eqref{eq:Cum-higher-order-fin-N-gen} and  \eqref{eq:Mom-higher-order-fin-N-gen}.

\newpage

\section{Local-unitary invariant random tensors}
\label{sec:char-lu-inv-tensor-distributions}

We will consider throughout this paper two classes of random tensors ensembles:
\begin{itemize}
\item \textbf{Mixed ensembles.}
The first case consists in a complex random tensor of the form $A=\{A_{i^1 \ldots i^D ; j^1 \ldots j^D}\}$ with $1\le i^c, j^c \le N$. The distribution is local-unitary invariant if for any $U_1, \ldots, U_D\in U(N)$, the random tensors $A$ and $(U_1 \otimes \cdots \otimes U_D)A (U_1^\dagger \otimes \cdots \otimes U_D^\dagger)$ have the same distribution. Mixed random tensors are adapted for describing random operators, or random mixed
quantum states, on a $D$-partite Hilbert space $\mathcal H _1 \otimes \cdots \otimes \mathcal H _D.$ Note that $A$ is not a priori assumed to be Hermitian.
\item \textbf{Pure ensembles.}
The second case is that of a pair of tensors 
$T_{i^1\dots i^D}$ and $\bar T_{j^1 \dots j^D}$. Despite the notation, just as $A$ is not a priori assumed to be Hermitian,  we generally assume that $T_{i^1 \ldots i^D}$ and $\bar T_{i^1 \ldots i^D}$ are independent, and explicitly say so otherwise\footnote{Note that even in the latter situation, these variables behave as independent ones, for instance when deriving saddle point or Schwinger-Dyson equations  involving  differentiation with respect to $T_{i^1 \ldots i^D}$ and $\bar T_{i^1 \ldots i^D}$.}. The distribution is local-unitary invariant if $T,\bar T$ and 
$ (U_1 \otimes \cdots \otimes U_D)T , 
\bar T (U_1^\dagger \otimes \cdots \otimes U_D^\dagger)$ have the same distributions.
Pure random tensors  are adapted for describing  random pure quantum states on a $D$-partite Hilbert space. The case  where $T$ is a complex Gaussian  (Ginibre) tensor and the components of $\bar T$ are the conjugate of those of $T$ is the ensemble that has been studied the most in the literature.  
\end{itemize}

Both in the mixed and in the pure case, we call the indices $i$ on which the unitary matrices $U$ act \emph{output indices}, and the indices $j$ on which the Hermitian conjugates $U^\dagger$ act the \emph{input indices}. We sometimes use the shorthand notation $\vec i$ to denote the $D$-tuple of indices $i^1\dots i^D$.

It is sometimes useful to regard the pure case as a factorized version of the mixed one $A=T\otimes \bar T$, or component wise $A_{i^1 \ldots i^D ;\; j^1 \ldots j^D} = T_{i^1 \ldots i^D} \; \bar T_{ j^1 \ldots j^D}$. Below we will often first discuss the mixed case and then adapt the various notions to the pure one.

\subsection{Trace-invariants and moments of invariant distributions of finite size}
\label{sec:trace-inv}

We discuss a family of invariant polynomials which generalize the traces of powers of matrices to tensors. 

\subsubsection{Trace-invariants}
\label{subsub:trace-invariants}

Similarly to invariant matrix functions, a function $f$ of the tensor $A$ is said to be local-unitary invariant (\textsf{LU}-invariant), if for any $U_1, \ldots, U_D \in U(N)$, 
\be
\label{eq:LU-invariance-function}
f(A) = f\Bigl((U_1 \otimes \cdots \otimes U_D)A (U_1^\dagger \otimes \cdots \otimes U_D^\dagger)\Bigr).
\ee

We call \emph{trace-invariants} the homogeneous local-unitary invariant polynomials in the tensor components \cite{Gurau-book} such that all the
 output indices $i^c$ are identified and summed (contracted) with input indices $j^c$ respecting the color $c=1, \dots D$.

\paragraph{Graphical representation.} The trace-invariants can be represented as edge colored graphs.  
In the mixed case we represent a tensor $A$ by a pair consisting in a white and a black vertex
connected by a thick edge to which we assign a color $D+1$. The input (resp.~output) indices of color $c$, $j^c$ (resp.~$i^c$) are represented by half-edges of color $c$ connected to the black (resp.~white) vertex. This is depicted in Fig.~\ref{fig:graph-repr-mixed-and-pure}, on the left. In the pure case we represent a tensor $T$ as a white vertex with output half-edges $i^c$, and a tensor $\bar T$ as a black vertex with input half-edges $j^c$, as depicted in Fig.~\ref{fig:graph-repr-mixed-and-pure}, on the right.

\begin{figure}[!h]
\centering
\includegraphics[scale=1.3]{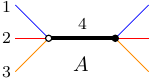}\hspace{2.5cm}\includegraphics[scale=1.3]{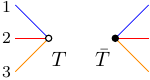}
\caption{
Left: a mixed tensor $A_{i^1\ldots i^D ; j^1 \ldots j^D}$ with input indices $j^c$  (resp.~output $i^c$) represented as half-edges of color $c$ attached to the black (resp.~white) vertex. Right: the pure case with $T_{i^1\dots i^D}$ and $\bar T_{j^1 \dots j^D}$.}
\label{fig:graph-repr-mixed-and-pure}
\end{figure}

The pairing of a half-edge of color $c$ on a white vertex with a half-edge of color $c$ on a black one to form an edge of color $c$ represents the identification and summation of the two corresponding indices $i^c$ and $j^c$. 

As illustrated in Fig.~\ref{fig:graph-repr-trace-invariants}, the resulting graphs are \emph{bipartite $(D+1)$-edge-colored graphs} in the mixed case (simply $(D+1)$-colored graphs below), 
while in the pure case they are \emph{$D$-colored}, as there are no thick edges.

\begin{figure}[!h]
\centering
\includegraphics[scale=1]{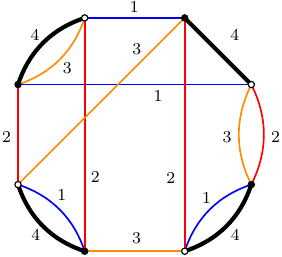}\hspace{2cm}\includegraphics[scale=1]{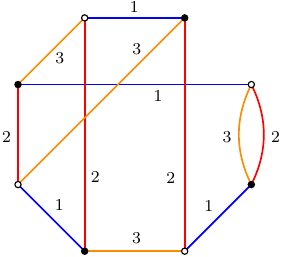}
\caption{Left: a trace-invariant for a mixed tensor $A$ with $D=3$ inputs and 3 outputs. Right: a similar invariant for the pure case with two tensors, $T$ with $D=3$ outputs and $\bar T$ with $3$ inputs.}
\label{fig:graph-repr-trace-invariants}
\end{figure}

The trace-invariants are in one-to-one correspondence with the non-isomorphic non-labeled bipartite edge colored graphs with $n$ white and $n$ black vertices. Due to the presence of the extra edges, the invariance under relabeling in the mixed case differs from the one of the pure case.

\paragraph{Encoding via permutations.} 
Trace-invariants are encoded by permutations. 
In the mixed case we label $1$ to $n$ the copies of the tensor $A$, and we consider the $D$ permutations $\bsig = (\sigma_1, \ldots, \sigma_D)$, $\sigma_c \in S_n$ obtained by setting $\sigma_c(s)=s'$ if the output index $i^c $ of the tensor labeled $s$ is identified and summed with the input index $j^c$ of the tensor labeled $s'$. Denoting $\Tr_{\bsig}(A)$ the labeled trace-invariants, we have:
\be
\label{def:trace-invariants}
\Tr_{\bsig}(A)=\sum_{\rm{all\ indices}} 
 \left( \prod_{s=1}^n A_{i^1_s\ldots i^D_s; j^1_s\ldots j^D_s } \right) 
\prod_{   c =1}^D \left( \prod_{s=1}^n \delta_{i^c_s,j^c_{\sigma_c(s)}}  \right) \; .
\ee
In the pure case we label $1$ to $n$ the tensors 
$T$ and $\bar 1$ to $\bar n$ the tensors $\bar T$, and we set $\sigma_c(s)=s'$ if the output index $i^c$ on the white vertex $s$ is connected with the input index $j^c$ on the black vertex $\overline{s'}$: 
\be
\label{def:trace-invariantsp}
\Tr_{\bsig}(T,\bar T)=\sum_{\rm{all\ indices}} 
 \left( \prod_{s=1}^n T_{i^1_s\ldots i^D_s}  \bar T_{ j^1_{\bar s} \ldots j^D_{\bar s} } \right)  
\prod_{   c =1}^D \left( \prod_{s=1}^n \delta_{i^c_s,j^c_{\overline{ \sigma_c(s) } }}  \right) \; .
\ee
The summation convention is represented in Fig.~\ref{fig:graph-convention}. 
\begin{figure}[!h]
\centering
\includegraphics[scale=1.1]{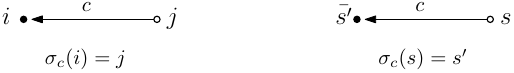}
\caption{Convention adopted  for encoding index summations with permutations, mixed on the left and pure on the right.}
\label{fig:graph-convention}
\end{figure}
Remark that if $A=T\otimes \bar T$, then 
$\Tr_{\bsig}(A)=\Tr_{\bsig}(T,\bar T)$.

The connected components of the $(D+1)$-edge colored graph obtained in the mixed case correspond to
the blocks of the partition $\Pi(\bsig) = \Pi(\sigma_1) \vee \dots \vee \Pi(\sigma_D)$, and we denoted their number by $K_\mathrm{m}(\bsig)=\#\Pi(\bsig)$. This follows by observing that the cycles of the permutation $\sigma_c$ correspond to the bi-colored cycles of edges with colors $(c,D+1)$ in the graph. We call these the \emph{mixed connected components of} $\bsig$.

The connected components of the $D$-edge colored graph obtained in the pure case (also called \emph{pure connected components})
correspond to the blocks of the bipartite partition $\Pi_\mathrm{p}(\bsig) = \bigvee_{1\le c\le D} \Pi_\mathrm{p}(\sigma_c)$ discussed in Sec.~\ref{sec:notations-prerequ}, whose number we denoted by $K_\mathrm{p}(\bsig)$. If $K_\mathrm{p}(\bsig)=1$, we say that $\bsig$ is  purely connected.

\paragraph{Relabeling.} 
We let for $\eta, \nu\in S_n$:
\be
\label{eq:left-and-right-composition}
\eta \bsig \nu =(\eta \sigma_1 \nu , \ldots, \eta \sigma_D \nu  ).
\ee

The trace-invariants do not depend on the labels of the tensors. The mixed and pure cases behave slightly differently:
\begin{itemize}
 \item {\it The mixed case.} In the mixed case a relabeling of the $n$ tensors $A$ corresponds to a simultaneous conjugation of all the $\sigma_c$ by the same permutation $\eta\in S_n$, that is:
\be
\label{eq:invariance-traceinvariant-mixed}
\Tr_{\bsig}(A) = \Tr_{\eta\, \bsig \eta^{-1}}(A). 
\ee

In this case a non-labeled trace-invariant is an equivalence class (orbit in $S^D_n /{\sim_\mathrm{m}}$) of labeled bipartite $(D+1)$-edge colored graphs with the equivalence relation:
\be
\label{eq:eq-classes-mixed}
\bsig \sim_\mathrm{m} \bsig' \quad \Leftrightarrow \quad \exists \eta\in S_n \mid \bsig=\eta \bsig' \eta^{-1} \; . 
\ee

\item {\it The pure case.} In the pure case the white and black vertices can be relabeled independently, that is for any $\eta, \nu\in S_n$:
\be
\label{eq:invariance-traceinvariant-pure}
\Tr_{\bsig}(T, \bar T) = \Tr_{\eta \bsig\nu}(T, \bar T) \; . 
\ee 
In this case a non-labeled trace-invariant is an equivalence class (orbit in $S_n^D /{\sim_\mathrm{p}}$) of labeled bipartite $D$-edge colored graphs (with no thick edges) with the equivalence relation:
\be
\label{eq:eq-classes-pure}
\bsig \sim_\mathrm{p} \bsig' \quad \Leftrightarrow \quad \exists \eta, \nu\in S_n \mid \bsig=\eta \bsig' \nu . 
\ee

\end{itemize}

\begin{remark} It is obvious that $\bsig \sim_\mathrm{m} \bsig'$ implies $\bsig \sim_\mathrm{p} \bsig'$, but the converse is not true in general. What holds however is the following:
\be
\bsig \sim_\mathrm{p}\bsig' \ \ \mathrm{in} \ \ S_n^D\quad  \Leftrightarrow\quad  \tilde \bsig \sim_\mathrm{m} \tilde \bsig' \ \ \mathrm{in}\ \  S_n^{D-1} \;, 
\ee
where $\tilde \bsig =(\sigma_1 \sigma_D^{-1}, \ldots, \sigma_{D-1}\sigma_D^{-1})\in S_n^{D-1}$, that is one can treat a pure invariant like a mixed one with one less color by declaring 
the edges of color $D$ of the pure invariant to be thick.
\end{remark}

\paragraph{Matrices.} The $D=1$ mixed case reproduces the matrix trace-invariants discussed in Sec.~\ref{sec:moments-unitarily-invariant-random-matrices}. Indeed, for $D=1$ a mixed trace-invariant is a collection of cycles with thick edges of color 2 and thin edges of color 1. Labeling the matrices $M$ corresponds to labeling the thick edges and  \eqref{def:trace-invariants} becomes $\Tr_\sigma(M)$. Relabeling the matrices corresponds to conjugating $\sigma$ by some permutation and a non-labeled trace-invariant corresponds to the conjugacy class of $S_n/{\sim_\mathrm{m}}$ of permutations of fixed cycle type.

The $D=2$ pure case  is also matricial: one has matrices $X$ and $\bar X$, typically the conjugate of $X$. A trace-invariant corresponds to a collection of alternating cycles with white and black vertices and thin edges of color 1 and 2. The trace-invariants are functions of $XX^\dagger$.

\subsubsection{Theory of invariants}
\label{sec:basis}

In this section, we show that asymptotically at large $N$ the trace-invariants generate the set local-unitary invariant polynomials and are linearly independent. The results we derive are weaker than the results for unitarily invariant matrices, and it should be possible to improve them. While similar results have already been obtained in the literature \cite{Heroetal, Vrana1, Vrana2, BGR1, BGR2, BGR3, Deksen, TuckNort}, due to differences in vocabulary and methods we find it useful to re-derive them here.

\paragraph{Distance between orbits.} In the \emph{mixed case}, a distance $d_\mathrm{m}([\bsig]_{\mathrm{m}}, [ \btau]_{\mathrm{m}})$  between two equivalence classes $[\bsig]_{\mathrm{m}},[\btau]_{\mathrm{m}} \in S_n^D /{\sim_\mathrm{m}}$ is defined by:
\be
\label{eq:def-distance-orbits-mixte}
d_\mathrm{m}([\bsig]_{\mathrm{m}}, [ \btau]_{\mathrm{m}}) = \min_{\eta_1, \eta_2 \in S_n}\sum_{c=1}^D \lvert\eta_1  \sigma_c \eta_1^{-1} \eta_2 \tau_c^{-1}\eta_2^{-1} \rvert  = \min_{\eta \in S_n}\sum_{c=1}^D \lvert \sigma_c \eta \tau_c^{-1}\eta^{-1}\rvert  \ge 0 \; ,
\ee
and it is well defined as it is independent on the particular representatives $\bsig$ and $\btau$ of the classes $[\bsig]_{\mathrm{m}}$ and $[\btau]_{\mathrm{m}}$.

In the \emph{pure case}, a distance $d_\mathrm{p}([\bsig]_{\mathrm{p}}, [\btau]_{\mathrm{p}})$  between two equivalence classes $[\bsig]_{\mathrm{p}},[ \btau]_{\mathrm{p}} \in S_n^D /{\sim_\mathrm{p}}$ is defined by:
\be
\label{eq:def-distance-orbits-pure}
d_\mathrm{p}([\bsig]_{\mathrm{p}}, [ \btau]_{\mathrm{p}} ) 
= \min_{\eta, \nu \in S_n}\sum_{c=1}^D \lvert \sigma_c \eta \tau_c^{-1}\nu\rvert  \ge 0 \;, 
\ee
which is again a class function. The two functions are non-negative, symmetric and respect the triangle inequality. For instance, for any $\alpha$ and any $\gamma$ we have 
$|\sigma \eta \tau^{-1}\eta^{-1}| 
\le |\sigma \gamma \alpha^{-1}\gamma^{-1}| + |\gamma \alpha \gamma^{-1} \eta\tau^{-1}\eta^{-1}|$ and we take the minimum first over $\gamma$ and then over $\eta$. Furthermore, if $\min_{\eta\in S_n}\sum_c \lvert \sigma_c \eta \tau_c^{-1}\eta^{-1}\rvert$ vanishes, then there exists $\eta$ such that for all $c$, $\lvert \sigma_c \eta \tau_c^{-1}\eta^{-1}\rvert=0$, that is,  $\bsig=\eta \btau \eta^{-1}$, and similarly in the pure case. Therefore:

\begin{lemma} 
\label{lem:distance-between-orbits}The functions $d_\mathrm{m}( [ \bsig ]_{\mathrm{m}}, [\btau]_{\mathrm{m}})\in \mathbb{N}$ and $d_\mathrm{p}([\bsig]_{\mathrm{p}}, [\btau]_{\mathrm{p}})\in \mathbb{N}$ are distance functions between the equivalence classes, in particular:
\[
d_\mathrm{m}([\bsig]_{\mathrm{m}}, [\btau]_{\mathrm{m}})  = 0 \quad \Leftrightarrow \quad  \exists \eta,\  \bsig = \eta \btau \eta^{-1}  \;,\qquad \text{that is} \quad \bsig\sim_{\mathrm{m}} \btau \; , 
\]
and analogously for $d_\mathrm{p}$ and $S_n^D /{\sim_\mathrm{p}}$. 
\end{lemma}

\paragraph{A spanning family.} Let $P$ be a
local-unitary invariant polynomial of degree $n$ in the tensor entries:
\[
 P(A) = \sum_{q=0}^n \sum_{i,j} P_{ i_1^1\dots i_1^D; j_1^1 \dots j_1^D  |  \dots  | i_q^1\dots i_q^D; j_q^1 \dots j_q^D  } \prod_{v=1}^q   A_{i_v^1\dots i_v^D; j_v^1 \dots j_v^D} \;,
\]
where the coefficients are complex numbers. Since $P(A)$ is local-unitary invariant: 
\[
P(A) = \int d U_1 \cdots d U_D \; P\Bigl((U_1 \otimes \cdots \otimes U_D)A (U_1^\dagger \otimes \cdots \otimes U_D^\dagger)\Bigr) \;,
\]
with $dU_c$ the Haar measure on $N\times N$ unitary matrices and using the Weingarten formula \eqref{eq:WeinDef} we get:
\[
 P(A) = 
 \sum_{q=0}^n \sum_{\btau\in S^D_q } \Tr_{\btau}(A) 
  \sum_{\bsig\in S^D_q  } \prod_{c=1}^D
  W^{(N)} (\sigma_c\tau_c^{-1})
 \sum_{i,j} P_{ i_1^1\dots i_1^D; j_1^1 \dots j_1^D  |  \dots  | i_q^1\dots i_q^D; j_q^1 \dots j_q^D  } \prod_{v=1}^q \prod_{c=1}^D  \delta_{i_v^cj^c_{\sigma^c(v)}} \; . 
\]

Proceeding similarly in the pure case we establish the following lemma.
\begin{lemma}
\label{lem:generating}
The trace-invariants of a tensor $A\in M_N(\mathbb{C})^{\otimes D}$ are a spanning family for the local-unitary invariant polynomials of degree $n\le N$. More precisely,  for any \textsf{LU}-invariant polynomial $P$ of degree $n\le N$, there exists a set of complex coefficients $\{ C_{[\btau]_{\mathrm{m}} } \}_{[ \btau]_{\mathrm{m}} \in S_q^D/{\sim_\mathrm{m}} , q\le n}$ such that:
\[
 P(A) = \sum_{q=1}^n\  \sum_{ [\btau]_{\mathrm{m}} \in S_q^D /{\sim_\mathrm{m}}  } C_{[\btau]_{\mathrm{m}}}  \Tr_{[\btau]_{\mathrm{m}}} (A) \; ,
\]
where $\Tr_{[\btau]_{\mathrm{m}} } (A) = \Tr_{\btau } (A) $ for any representative $\btau \in[\btau]_{\mathrm{m}} $.

In the pure case, for any \textsf{LU}-invariant polynomial $P(T,\bar T)$ of degree $n\le N$ in each of $T$ and $\bar T$, there exists a set of complex coefficients $\{ C_{[\btau]_{\mathrm{p}} } \}_{[ \btau]_{\mathrm{p}} \in S_q^D/{\sim_\mathrm{p}} , q\le n}$ such that:
\[
 P(T,\bar  T) = \sum_{q=1}^n\  \sum_{ [\btau]_{\mathrm{p}} \in S_q^D /{\sim_\mathrm{p}}  } C_{[\btau]_{\mathrm{p}}}  \Tr_{[\btau]_{\mathrm{p}}} (T,\bar  T) \; ,
\]
where $\Tr_{[\btau]_{\mathrm{p}} } (T,\bar  T) = \Tr_{\btau } (T,\bar T) $ for any representative $\btau \in[\btau]_{\mathrm{p}} $.
\end{lemma}

\begin{remark} For a matrix, the statement is easily extended to polynomials of any degree using the Cayley-Hamilton Theorem. For tensors, the statement can be extended to polynomials whose degree is exponential in the system size, by application of a result by Deksen \cite{Deksen}, see also\cite{TuckNort}, where the bound is explicitly computed for the local-unitary case\footnote{L.L. thanks Michael Walter for mentioning this fact. }.
\end{remark}

\paragraph{Linear independence.} In what concerns the linear independence of the trace-invariants, the following holds.
\begin{theorem}
\label{thm:independence}
For any $n_{\rm max}$, there exists $N_{n_{\rm max} }$ such that for any $N\ge N_{n_{\rm max} } $ the families of functionals below are linearly independent on $\mathbb{C}$:
\begin{itemize}
\item\emph{Mixed case:}
\[
\Tr_{[\bsig]_{\mathrm{m}}} : M_N(\mathbb{C})^{\otimes D}\to \mathbb{C}, \qquad [ \bsig ]_{\mathrm{m}} \in S_n^D /{\sim_\mathrm{m}} \; ,\; \; n\le n_{\rm max} \;.
\]
\item\emph{Pure case:}
\[
\Tr_{[\bsig]_{\mathrm{p}}} : (\mathbb C^N)^{\otimes D}\times  (\mathbb C^N)^{\otimes D} \mapsto \mathbb{C} 
\;, \qquad [\bsig]_{\mathrm{p}} \in S_n^D /{\sim_\mathrm{p}} \; ,\;\; n\le n_{\rm max} \;.
\]
\end{itemize}
\end{theorem}
\begin{proof}
The theorem is proved in Appendix~\ref{app:proofs-of-section-Invariance-itself}. 
\end{proof}

Note that, in order to include all the trace-invariants of degree up to $ n_{\rm max}$, one needs to choose a $ N_{n_{\rm max} }$ larger than the number of all such invariants, which scales super exponentially with $n_{\rm max}$.

\paragraph{Scalar products between orbits.} In the course of the proof of  Thm.~\ref{thm:independence} in the appendix we establish the following two instructive propositions. 

\begin{proposition}[Mixed case]
\label{prop:asympt-scalar-product-mixed}
We consider $A$ a Ginibre random tensor, that is, the $N^{2D}$ components $A_{i^1 \ldots i^D; j^1 \ldots j^D}$ are independent complex Gaussian variables\footnote{The joint probability measure writes $ e^{- N^ D\Tr(A A^\dagger)} dAd\bar A$.} with covariance
$\mathbb{E}[\bar A_{\vec{i} ; \vec j}A_{\vec k ; \vec l}] =\prod_{c=1}^D \delta_{i^c, k^c}\delta_{j^c, l^c}/ N^D 
$. For $[\bsig]_{\mathrm{m}},[\btau]_{\mathrm{m}} \in S_n^D /{\sim_\mathrm{m}}$, we define the matrix:
\be
\bigl\langle [\bsig]_{\mathrm{m}}, [ \btau]_{\mathrm{m}} \bigr\rangle_\mathrm{m} = \mathbb E_M\bigl[\Tr_{[\bsig]_{\mathrm{m}}} (\bar A)  \Tr_{[\btau]_{\mathrm{m}}} (A)  \bigr] = C_{[\bsig]_{\mathrm{m}},  
   [\btau]_{\mathrm{m}} } 
N^{- d_\mathrm{m}([\bsig]_{\mathrm{m}},  
   [\btau]_{\mathrm{m}} )}(1+O(N^{-1})) \; , 
\ee
where $d_\mathrm{m}([\bsig]_{\mathrm{m}},  
[\btau]_{\mathrm{m}} ) =  \min_{\eta \in S_n}\sum_{c=1}^D \lvert \sigma_c \eta \tau_c^{-1}\eta^{-1}\rvert $ is the distance between orbits in the mixed case; $C_{[\bsig]_{\mathrm{m}},  
[\btau]_{\mathrm{m}} } >0$ is the number of permutations $\eta\in S_n$ for which 
$\sum_{c=1}^D \lvert \sigma_c \eta \tau_c^{-1}\eta^{-1}\rvert = d_\mathrm{m}([\bsig]_{\mathrm{m}},  
[\btau]_{\mathrm{m}} ) $
and $C_{[\bsig]_{\mathrm{m}}  } = C_{[\bsig]_{\mathrm{m}} ,[\bsig]_{\mathrm{m}}  }$ is the cardinal of the centralizer of $\bsig$.   For $N$ large enough, the form $\bigl\langle [\bsig]_{\mathrm{m}}, [ \btau]_{\mathrm{m}} \bigr\rangle_\mathrm{m}$  is symmetric positive-definite and in the limit $N\to \infty$ the orbits are asymptotically orthogonal:
\be
\lim_{N\rightarrow \infty }\bigl\langle [\bsig]_{\mathrm{m}}, [ \btau]_{\mathrm{m}} \bigr\rangle_\mathrm{m}  =C_{[\bsig]_{\mathrm{m}}  }   \;\delta_{[\bsig]_{\mathrm{m}}  , [\btau]_{\mathrm{m}}  } \; .
\ee
\end{proposition}

\begin{proposition}[Pure case]
\label{eq:asympt-scalar-product-pure}
Let $T_1, T_2\in (\mathbb C^N)^{\otimes D}$ be two independent Ginibre random tensors with  $\mathbb{E}[\bar T_{\vec{i}}\ T_{\vec k }]=\prod_{c=1}^D \delta_{i^c, k^c}/ N^{D/2}$. For $[\bsig]_{\mathrm{p}},[ \btau]_{\mathrm{p}} \in  S_n^D /{\sim_\mathrm{p}}$ the matrix: 
\be
\begin{split}
\bigl\langle [\bsig]_{\mathrm{p} },[ \btau]_{\mathrm{p}}  \rangle_\mathrm{p} :=& \mathbb E_{T_1, T_2}
\bigl[
\Tr_{ [\bsig]_{\mathrm{p} } } (\bar T_2, \bar T_1) 
\Tr_{ [\btau]_{\mathrm{p} } } (T_1, T_2) 
\bigr] \crcr
= & D_{[\bsig]_{\mathrm{p} },[ \btau]_{\mathrm{p}}   } 
N^{- d_\mathrm{p}([\bsig]_{\mathrm{p} },[ \btau]_{\mathrm{p}}   )}(1+O(N^{-1})) \;,
\end{split}
\ee
where $ d_\mathrm{p}([\bsig]_{\mathrm{p} },[ \btau]_{\mathrm{p}} =\min_{\eta, \nu \in S_n}\sum_{c=1}^D \lvert \sigma_c \eta \tau_c^{-1}\nu\rvert   $ is the distance function in the pure case;
$ D_{[\bsig]_{\mathrm{p} },[ \btau]_{\mathrm{p}}   }  $ is the number of couples of permutations $\eta, \nu\in S_n$ 
for which $ d_\mathrm{p}([\bsig]_{\mathrm{p} },[ \btau]_{\mathrm{p}} = \sum_{c=1}^D \lvert \sigma_c \eta \tau_c^{-1}\nu\rvert  $
and we denote $D_{[\bsig]_{\mathrm{p} } } := D_{[\bsig]_{\mathrm{p} },[ \bsig]_{\mathrm{p}}   }  $.  For $N$ large enough, the form $ \bigl\langle [\bsig]_{\mathrm{p} },[ \btau]_{\mathrm{p}}  \rangle_\mathrm{p}  $  is symmetric positive-definite,  and in the limit $N\to \infty$ the orbits are asymptotically orthogonal:
\be
\lim_{N\rightarrow \infty } \bigl\langle [\bsig]_{\mathrm{p} },[ \btau]_{\mathrm{p}}  \rangle_\mathrm{p}  = D_{[\bsig]_{\mathrm{p} } } \;\delta_{[\bsig]_{\mathrm{p} },[ \btau]_{\mathrm{p}} } \; .
\ee
\end{proposition}

We have the following remark.

\begin{remark}In \cite{BGR2} a similar scalar product is considered in the pure case, but for only one Ginibre tensor $T$, that is when $T_1 = T$ and $T_2 = \bar T$ the complex conjugate, instead of the two independent ones $T_1, T_2$. In order to eliminate the self contractions on $\Tr_{\bsig}(T, \bar T)$, the authors introduce a ``normal-ordering'' $\mathcal{N}$ of the trace-invariants:
\[
\bigl\langle [\bsig]_\mathrm{p} , [\btau]_\mathrm{p}  \bigr\rangle^{\mathcal{N}}_\mathrm{p} = \mathbb E \bigl[ \overline{\mathcal{N}(  \Tr_{ \bsig} (T, \bar T) ) } \mathcal{N}( \Tr_{\btau} (\bar T, T) )  \bigr] \;. 
\]
The use of two independent tensors $T_1, T_2$ has the effect of precisely eliminating these self contractions without the need for normal-ordering. This small point is essentially the only difference between the complex tensor case and the bipartite distribution of two random tensors.\footnote{Regarding analogous scalar products in the context of real tensors, we also refer to \cite{ONproduct} for  invariance corresponding to tensor products of independent orthogonal matrices,  and to \cite{Tensors11} for invariance corresponding to tensor products of a single orthogonal matrix. } 
\end{remark}

\paragraph{Non-polynomial invariants.}
As $N$ grows, Lemma~\ref{lem:generating} and Thm.~\ref{thm:independence} imply that 
the trace-invariants are both spanning and linearly independent, hence a basis in the space of invariant polynomials. In the limit $N\rightarrow \infty$, the statement formally extends to sufficiently nice functions.
Loosely speaking, the functions we are interested in are series in the components of the tensor: 
\[
f(A) = \sum_{q= 0}^{+ \infty} \; \sum_{i,j} P_{ i_1^1\dots i_1^D; j_1^1 \dots j_1^D  |  \dots  | i_q^1\dots i_q^D; j_q^1 \dots j_q^D  } \prod_{v=1}^q   A_{i_v^1\dots i_v^D; j_v^1 \dots j_v^D} \;,
\]
which, rescaling $A$ to $zA$ with $z\in \mathbb{C}$, admit a formal expansion of the form: 
\[
f(z;A)=\sum_{k\ge 0} \frac 1 {N^k} f_k(z;A), \qquad f_k(z;A)=\sum_{n\ge0} c_{k,n,N}(A) z^n \;, 
\]
where for every $k$ and $A$, $f_k(z;A)$ is a convergent series in some neighborhood of $z=0$ and each $c_{k,n,N}(A)$ is a homogeneous \textsf{LU}-invariant polynomial of degree $n$. Note that the order $k$ is just the dominant order in $1/N$ of $c_{k,n,N}(A)$, but $c_{k,n,N}(A)$ can still contain some terms in $O(1/N)$ which wash out in the limit $N\rightarrow \infty$. In this case there exists $N_n$ such that for $N\ge N_n$, $c_{k,n,N}(A)$ admits a unique expansion in trace-invariants 
$ c_{k,n,N}(A)=\sum_{ [\btau]_{\mathrm{m}} \in S_q^D/{\sim_\mathrm{m}} } C_{k, [\btau]_{\mathrm{m}}, N}  \Tr_{[\btau]_{\mathrm{m}}}  (A)$.
Taking a sequence of tensors $A^{(N)}\in M_N(\mathbb{C})^{\otimes D}$ indexed by the size $N$, for every $n\ge 0$, we assume that $c_{k,n,N}(A^{(N)})$ converges to some finite limit:
\[
 c_{k,n}(a) = \sum_{ [\btau]_{\mathrm{m}} \in S_q^D/{\sim_\mathrm{m}} } \lim_{N\to \infty} C_{k, [\btau]_{\mathrm{m}}, N}  \Tr_{[\btau]_{\mathrm{m}}}  (A^{(N)}) \; ,
\]
where $a$ captures some information about $A^{(N)} $ at infinite $N$, and for each $k$, the series $f_k(z;a) = \sum_{n\ge 0} c_{k,n}(a) z^n$ is convergent in some neighborhood of $z=0$. The full expansion $f(z;A)$ is usually divergent for $z\neq 0$, but for any $k\ge 0$, one has a unique convergent expansion of $f_k(z,a)$ on asymptotic trace-invariants.

\subsubsection{Finite size moments}
\label{sec:tensor-moments}

The appropriate invariant moments for a local-unitary invariant mixed random tensor $A$ for finite $N$ are the expectations of  trace-invariants \cite{1Nexpansion1, 1Nexpansion2, 1Nexpansion3,critical,  Gurau-universality, uncoloring, bonzom-SD, Gurau-book, enhanced-1, Walsh-maps, multicritical, bonzom-review, Lionni-thesis}:
\be
\label{eq:expectation-trace-invariants-tensor}
\mathbb{E}\bigl[\Tr_{[\bsig]_m}(A)\bigr], \quad [\bsig] \in S_n^D/{\sim_{\mathrm{m}}} \;,
\ee
where we do not require any connectivity condition on $[\bsig]_\mathrm{m}$. For a pure tensor this is replaced by  $\mathbb{E}\bigl[\Tr_{[\bsig]_{\mathrm{p}}}(T, \bar T)\bigr]$ with $[\bsig]_{\mathrm{p}} \in S_n^D /{\sim_\mathrm{p}}$. Only a finite number of these moments are needed to fully describe the distribution in the sense below.  From Lemma~\ref{lem:generating} and Thm.~\ref{thm:independence}, one can compute the expectation of any polynomial  $P(A)$  in the components of $A$ which is of degree $n\le N$ as:
\be
\label{eq:exp-poly-finite-moments-tensor}
\mathbb{E}\bigl[ P (A)\bigr]= \sum_{q=1}^n \sum_{ [ \bsig ]_{\mathrm{m}} \in S_q^D /{\sim_\mathrm{m}}} c_{[ \bsig ]_{\mathrm{m} } }\; \mathbb E\bigl[\Tr_{[ \bsig ]_{\mathrm{m} } } (A) \bigr] \;, 
\ee
and the coefficients of this expansion are unique for $N$ large enough, and similarly in the pure case.   Indeed, this derives from applying Lemma~\ref{lem:generating} to the \textsf{LU}-invariant function:
$
\tilde P(A)= \mathbb \int dU_1 \cdots dU_D f(UAU^\dagger)$  where $U=U_1\otimes \cdots \otimes U_D, 
$
together with the fact that $A$ being \textsf{LU}-invariant, $\mathbb{E}_A[\tilde P( A)] = \mathbb{E}[P( A)]. $

This extends to the computation of $\mathbb{E}[ f (A)]$ for non-necessarily invariant but sufficiently regular functions $f$ which can be obtained as the limits of sequences of polynomials, but precisely characterizing this class of functions is beyond the scope of this paper. 

\subsection{Finite size free cumulants}
\label{sec:finite-size-cumulants}

\subsubsection{Linearization: intuitive picture}

Our aim is to construct the equivalent of the free cumulants of random matrices discussed in Sec.~\ref{subsub:free-cumulants} in the tensor case. In particular, the tensor free cumulants should be asymptotically additive for sums of independent random tensors. 

The moment-generating function in the mixed and pure cases are:
\be\label{eq:genfunc}
 Z_A(B) = \bE[e^{\Tr(B^T A)}] \;, \qquad 
 Z_{T,\bar T} (J,\bar J) = \bE [e^{  J \cdot T +  \bar J\cdot \bar T }] \;,
\ee
where $\Tr(B^TA) = \sum_{\vec i,\vec j} A_{\vec i;\vec j} B_{\vec i;\vec j}$ and $J \cdot T = \sum_{\vec i} J_{\vec i} T_{\vec i}$ respectively $ \bar J \cdot \bar T = \sum_{\vec j} \bar J_{\vec j} \bar T_{\vec j}$
and the tensors $B$ respectively $J,\bar J$ are fixed (commonly called sources in the physics literature). Note that the source terms are chosen so as to ensure that the indices of the sources and the random tensors have the same nature, for instance the indices of $A$ in the first position  are summed together with the indices of $B$ in the first position,  etc. 

From this formulas, it is apparent why the pure case is not just the substitution $A = T\otimes \bar T$: performing this substitution in $Z_{A}(B)$ leads to a generating function with a bi-linear source for $T$ and $\bar T$ which is different from the moment-generating function in the pure case (the latter has linear sources).

Let us consider two independent \textsf{LU}-invariant mixed random tensors $A_1$ and  $A_2$ and the two moment-generating functions
$Z_{A_a}(B)$. As $A_1$ and $A_2$ are independent, the generating function factors:
\be
\label{eq:facto}
Z_{A_1 + A_2}(B) =  Z_{A_1}(B) Z_{A_2}(B) \;, 
\ee
so that the moment-generating functions are additive:
\be
\label{eq:log-is-loglog}
\log Z_{A_1 + A_2}(B) = \log Z_{A_1}(B) + \log Z_{A_2}(B) \;. 
\ee
Replacing $B$ by $zB$, differentiating $n$ times with respect to $z$ and setting $z$ to zero, we have:
\be
\label{eq:log-is-loglog-n}
\frac{\partial^{n}}{\partial z^n}\log Z_{A_1 + A_2}(zB) \bigr \rvert_{z=0} = \frac{\partial^{n}}{\partial z^n}\log Z_{A_1}(zB) \bigr \rvert_{z=0} + \frac{\partial^{n}}{\partial z^n}\log Z_{A_2}(zB) \bigr \rvert_{z=0} \;. 
\ee

From the \textsf{LU}-invariance of the measures of $A_1$ and  $A_2$ it follows that, as functions of $B$, $Z_{A_a}(B)$ are \textsf{LU}-invariant, $Z_{A_1 + A_2}$ is \textsf{LU}-invariant, and the three terms in \eqref{eq:log-is-loglog-n} are \textsf{LU}-invariant as well. Since they are homogeneous polynomials in $B$ of degree $n$, for $N_n$ sufficiently large,  they expand uniquely on the set of trace-invariants: 
\be
\label{eq:dvt-basis-A-regroup}
 \frac{\partial^{n}}{\partial z^n}\log Z_{A_1}(zB) \bigr \rvert_{z=0} = \sum_{[\bsig]_{\mathrm{m}} \in  S_n^D /{\sim_\mathrm{m}} } \Tr_{[\bsig]_{\mathrm{m}}}(B)\ \mathcal{K}^{\mathrm{m}}_{[\bsig]_{\mathrm{m}}}[A_1]  \;, 
\ee 
where the coefficients $\mathcal{K}^{\mathrm{m}}_{[\bsig]_{\mathrm{m}}}[A] $ are explicitly computed below.
Due to the linear independence, 
Eq~\eqref{eq:log-is-loglog-n}
implies that $\mathcal{K}^{\mathrm{m}}_{[\bsig]_{\mathrm{m}}}[A] $ are additive, as formalized in the proposition below.
\begin{proposition} 
\label{prop:large-N-linear-gen}
Let $A_1$ and $A_2$ be two sufficiently regular independent random tensors, and consider coefficients $\mathcal{K}^{\mathrm{m}}_{[\bsig]_{\mathrm{m}}}[A_a] $ for $a=1,2$ as well as $\mathcal{K}^{\mathrm{m}}_{[\bsig]_{\mathrm{m}}}[A_1+A_2]$ defined in \eqref{eq:dvt-basis-A-regroup}. Then there exists $N_n$ such that for any $N>N_n$:
\be
 \mathcal{K}^{\mathrm{m}}_{ [\bsig]_{\mathrm{m}} }[A_1+A_2]  =  \mathcal{K}^{\mathrm{m}}_{[\bsig]_{\mathrm{m} } } [A_1]  + \mathcal{K}^{\mathrm{m}}_{[\bsig]_{\mathrm{m}}}[A_2]  \; . 
\ee
\end{proposition}

Everything goes through, mutatis mutandis, in the pure case:
\be
 \frac{\partial^{n}}{\partial z^n}
 \frac{\partial^{n}}{\partial \bar z^n}
 \log Z_{T,\bar T}(zJ,\bar z \bar J) \bigr \rvert_{z=0} = \sum_{[\bsig]_{\mathrm{p}} \in  S_n^D /{\sim_\mathrm{p}} } \Tr_{[\bsig]_{\mathrm{p}}}(J,\bar J)\, \mathcal{K}_{[\bsig]_{\mathrm{p}}}[T,\bar T]  \;, 
\ee 
and for two independent pairs $T_1,\bar T_1$ and $T_2,\bar T_2$ the cumulants are additive for any $N>N_n$:
\be
\label{eq:additivity-pure}
\mathcal{K}_{ [\bsig]_{\mathrm{p}} }[T_1+T_2, \bar T_1+\bar T_2]  =  \mathcal{K}_{[\bsig]_{\mathrm{p} } } [T_1,\bar T_1]  + \mathcal{K}_{[\bsig]_{\mathrm{p}}}[T_2,\bar T_2]  \; . 
\ee

\subsubsection{Moment-cumulant relations at finite $N$}
\label{subsub:cumulants-for-finite-N}
 A trick to derive the finite $N$ versions of the free cumulants in the matrix case is to take the generating function of classical cumulants and to average over the unitary group using Weingarten calculus \cite{CMSS}. We will employ the exact same strategy for tensors, and derive explicit expressions for the coefficients $ \mathcal{K}^{\mathrm{m}}_{ [\bsig]_{\mathrm{m}} }[A]$ and $\mathcal{K}_{[\bsig]_{\mathrm{p}}}[T,\bar T]  $. As these expressions turn out to be invertible, we posit that this coefficients yield the correct generalization of finite size free cumulants for \textsf{LU}-invariant  random tensors\footnote{Here, we mean appropriate quantities to obtain the generalizations of free cumulants of arbitrary orders in the limit $N\rightarrow \infty$. There is a notion of finite free cumulants \cite{finite-free-cum}, which appears to be different, but both converge to free cumulants in the matrix case. See also the different notion of finite free cumulants  for real symmetric random tensors in \cite{Tensors11}, obtained from a different approach. }. 

\begin{theorem}
\label{thm:finite-free-cumulants}
Consider $A$ a mixed $\mathsf{LU}$-invariant random tensor. Then for any fixed tensor $B$, the Taylor coefficients of the logarithm of $Z_A(zB)=\bE [e^{z\Tr B^T A}]$ admit the expansion:
\be
\label{eq:exp-logZ-trK}
\frac{\partial^{n}}{\partial z^n}\log Z_A(zB) \bigr \rvert_{z=0} = \sum_{\bsig \in S_n^D }  \Tr_{\bsig}(B) \; \mathcal{K}^{\mathrm{m}}_\bsig[A] \; ,
\ee
where the mixed finite  size  free cumulants $ \mathcal{K}^{\mathrm{m}}_\bsig[A]$ are  (the natural generalization of \eqref{eq:freecumumome}):
\be
\label{eq:cum-finN-int}
\begin{split}
 \mathcal{K}^{\mathrm{m}}_\bsig[A] & = \sum_{\btau \in S^D_{n}} \sum_{\substack{{\pi \in \mathcal{P}(n)}\\{\pi \ge \Pi(\bsig)\vee \Pi( \btau)}}} \lambda_\pi  \prod_{G\in \pi}   \bE\left[\Tr_{\btau_{|_G} } (A)\right]\prod_{c=1}^D W^{(N)} \left( \sigma_{ c|_G}
 \tau_{c|_G}^{-1} \right) \crcr
 &  = \sum_{\substack{{\pi \in \mathcal{P}(n)}\\{\pi \ge \Pi(\bsig) }}} \lambda_\pi 
 \sum_{\substack{ { \btau \in S^D_{n}  } \\ {
    \Pi(\btau) \le \pi
 } }}   \prod_{G\in \pi}   \bE\left[\Tr_{\btau_{|_G} } (A)\right]\prod_{c=1}^D W^{(N)} \left( \sigma_{ c|_G}
 \tau_{c|_G}^{-1} \right) 
 \; . 
\end{split}
\ee

In the pure case, we view permutations over $n$ elements as bijections from 
the set $\{1,\dots n\}$ of white elements to the set $\{\bar 1,\dots \bar n\}$ of black elements, $s\to \overline{\sigma_c(s)}$ and denoting $S_{n,\bar n}$ the set of such bijections, we have:
\be
\partial^n_{z}\partial^n_{\bar z}
\log Z_{T,\bar T} ( zJ,\bar z \bar J)\bigr\rvert_{z=\bar z=0}  = \sum_{\bsig \in S_{n,\bar n}^D }  \Tr_{\bsig}(J,\bar J) \; \mathcal{K}_\bsig[T,\bar T] \; , 
\ee
where the pure finite  size  free cumulants are written in terms of bipartite partitions as:
\be
\begin{split}
\label{eq:cum-finN-int-pur}
 \mathcal{K}_\bsig[T,\bar T] & = \sum_{\btau \in S^D_{n,\bar n}} \sum_{\substack{{\Pi \in \mathcal{P}(n,\bar n)}\\{\Pi \ge \Pi_{\mathrm{p}}(\bsig) \vee 
 \Pi_{\mathrm{p}} (\btau)}}} \lambda_\Pi  \prod_{G= B\cup \bar B\in \Pi}   \bE\left[\Tr_{\btau_{|_B} } (T,\bar T)\right]\prod_{c=1}^D W^{(N)} \left( \sigma_{ c|_B}
 \tau_{c|_B}^{-1} \right) \crcr
 &  = \sum_{\substack{{\Pi \in \mathcal{P}(n,\bar n)}\\{\Pi \ge \Pi_{\mathrm{p}}(\bsig) }}} \lambda_\Pi 
 \sum_{\substack{ { \btau \in S^D_{n,\bar n}  } \\ { \Pi_{\mathrm{p}}(\btau) \le \Pi } }}   \prod_{G= B\cup \bar B\in \Pi}   \bE\left[\Tr_{\btau_{|_B} } (T,\bar T)\right]\prod_{c=1}^D W^{(N)} \left( \sigma_{ c|_B}
 \tau_{c|_B}^{-1} \right) 
 \; , 
\end{split}
\ee
where $\tau_{|_B}$ is the restriction\footnote{This restriction is compatible with the block, as $\Pi \ge \Pi_{\mathrm{p}}(\btau)$.} of the bijection $\tau:\{1,\dots n\} \to \{\bar 1,\dots \bar n\}$ to the block $B$,
which  to the labels in the set $B$ maps those in the set $\bar B$, $ \sigma_{ c|_B}
 \tau_{c|_B}^{-1} $ is a permutation of the elements of $\bar B$ having the same cycle type as the permutation $ \tau_{c|_B}^{-1}\sigma_{ c|_B} $ of the elements of $B$, and $\Pi_{\mathrm{p}}(\bsig)$ denotes the bipartite partition into pure connected components corresponding to $\bsig$.

The equations~\eqref{eq:cum-finN-int}
and~\eqref{eq:cum-finN-int-pur} are inverted as:
\be
\label{eq:mom-finN}
\begin{split}
\bE\left[\Tr_{\bsig}(A)\right]& = \sum_{\btau \in S^D_n}\  \sum_{\substack{{\pi\in \mathcal{P}(n)}\\{ \pi \ge \Pi(\btau)}}}  \ 
\mathcal{K}^{\mathrm{m} }_{\pi,\btau}[A] \  N^{nD - d(\bsig, \btau)} \; , \crcr
\bE\left[\Tr_{\bsig}(T,\bar T)\right] &  = \sum_{\btau \in S^D_{n,\bar n} }\  \sum_{\substack{{\Pi\in \mathcal{P}(n,\bar n)}\\{ \Pi \ge \Pi_{\mathrm{p}}(\btau)}}}  \ 
\mathcal{K}_{\Pi,\btau}[T,\bar T] \  N^{nD - d(\bsig, \btau)}  \;,
\end{split}
\ee
where we respectively denoted $\mathcal{K}^{\mathrm{m}}_{\pi,\btau}[A]  = \prod_{G\in \pi} \mathcal{K}^{\mathrm{m}}_{\btau_{\lvert_G}}[A]$ and  $\mathcal{K}_{\Pi,\btau}[T,\bar T]  = \prod_{G=B\cup \bar B\in \Pi} \mathcal{K}_{\btau_{\lvert_B}}[T,\bar T]$ the multiplicative extensions of the finite cumulants, and where we recall that $d(\bsig, \btau)=\sum_{c=1}^D \lvert \sigma_c\tau_c^{-1}\rvert$.
 
It is self evident that $ \mathcal{K}^{\mathrm{m}}_\bsig[A] $ and $ \mathcal{K}_\bsig[T,\bar T] $ are class functions for the equivalence relations $\sim_{\mathrm{m}}$ respectively $\sim_{\mathrm{p}}$, and regrouping the sums over $\bsig$ into the corresponding equivalence classes and using Proposition~\ref{prop:large-N-linear-gen}, it follows that $\mathcal{K}^{\mathrm{m}}_\bsig[A] $ and 
$ \mathcal{K}_\bsig[T,\bar T] $ are \emph{additive} for $N$ large enough.
\end{theorem}
\begin{proof}
See  Appendix~\ref{app:proofs-of-finite-free}. 
\end{proof}

Observe that in terms of classical cumulants, we may express  the relations above in terms of the  $\Phi$ (leading to the formulation considered for taking asymptotics):
\[
 \prod_{G= B\cup \bar B\in \Pi}   \bE\left[\Tr_{\btau_{|_B} } (T,\bar T)\right] = 
 \sum_{\substack{ {\pi \in \mathcal{P}(n,\bar n)}\\ { \Pi \ge \pi \ge \Pi_{\mathrm{p}}(\btau)} } }  
  \Phi_{\pi,\btau}[T,\bar T] \; ,
\]
and inverse formula (and similar formulas in the mixed case).

Note that for a \emph{purely connected} $\bsig$ we have not only $\Pi_{\mathrm{p}}(\bsig)=1_{n,\bar n}$ (the one set bipartite partition), but also $\Pi(\bsig) = 1_n$, and therefore:
\be
\label{eq:cum-finN-int-pure-conn}
\begin{split}
 \mathcal{K}^{\mathrm{m}}_\bsig[A] & = \sum_{\btau \in S^D_{n}} \bE\left[\Tr_{\btau} (A)\right]\prod_{c=1}^D W^{(N)} ( \sigma_{ c} \tau_{c}^{-1}) \;, \crcr
  \mathcal{K}_\bsig[T,\bar T] &= \sum_{\btau \in S^D_{n}} \bE\left[\Tr_{\btau} (T,\bar T)\right]\prod_{c=1}^D W^{(N)} ( \sigma_{ c} \tau_{c}^{-1})  \; ,
\end{split}
\ee
that is the pure cumulants can be obtained from the mixed ones by simply substituting $A = T \otimes \bar T$, but this is in general no longer true for $\bsig$ which is not purely connected (and for non-connected $\bsig$, $\mathcal{K}^\mathrm{m}_{\bsig}[T\otimes \bar T]$ is non in general invariant under $\sim_\mathrm{p}$. 

For instance  for $n=D=2$, consider a complex matrix $M$ and let  $A=M\otimes \bar M$, $\mathbf{id}_2 = \bigl(\mathrm{id}, \mathrm{id}\bigr)\in S_2^2 $ (purely connected), and $\btau_{(12)}=\bigl((12), (12)\bigr)\in S_2^2$ (connected but not purely connected). Then one explicitly computes:
\begin{align}
\label{eq:cum-finN-int-pure-conn-ex}
\mathcal{K}_{\mathbf{id}_2}[M, \bar M] & = \mathcal{K}^\mathrm{m}_{\mathbf{id}_2}[M\otimes \bar M] =   \mathcal{K}^\mathrm{m}_{\btau_{(12)}}[M\otimes \bar M]    - \frac 1 {N^4} \mathbb{E} \bigl[ \Tr(MM^\dagger)\bigr]^2  \; ,  \\
\mathcal{K}^\mathrm{m}_{\btau_{(12)}}[M\otimes \bar M] & = \frac {N^2 + 1} {N^2(N^2-1)^2}\mathbb{E} \bigl[ \Tr(MM^\dagger)\Tr(MM^\dagger)\bigr] -\frac{ 2\; \mathbb{E} \bigl[ \Tr(MM^\dagger MM^\dagger)\bigr]}{N(N^2-1)^2} \;,\nonumber
\end{align}
and:
\be
\mathcal{K}_{\btau_{(12)} }[M, \bar M] =   \mathcal{K}^\mathrm{m}_{\btau_{(12)}}[M\otimes \bar M]    - \frac 1 {N^4} \mathbb{E} \bigl[ \Tr(MM^\dagger)\bigr]^2 \neq  \mathcal{K}^\mathrm{m}_{\btau_{(12)}}[M\otimes \bar M] \; . \ee

\newpage

Finally, the following microscopic formulation holds, analogous to e.g.~Remark.~4.3 of \cite{CMSS}. The same formula holds considering possibly different tensors. 
\begin{proposition}
\label{prop:finite-free-prop}
For each $c\in \{1,\ldots D\}$, choose $\sigma_c\in S_n$, and distinct $1\le i_c(1), \ldots, i_c(n)\le N$, then the finite size free cumulants may be expressed in terms of the classical cumulants of entries of the tensors with distinct indices as:
$$
\mathcal{K}_\bsig^\mathrm{m}[ A] = k_n\Bigl(\bigl\{A_{i_1(\sigma_1(s)), \ldots, i_D(\sigma_D(s))\; ;\;   i_1(s), \ldots, i_D(s)  } \bigr\}_{1\le s \le n}\Bigr)\;.
$$
and
$$
\mathcal{K}_\bsig[T,\bar T] =  k_n\Bigl(\bigl\{T_{i_1(\sigma_1(s)), \ldots, i_D(\sigma_D(s))} , \bar T_{ i_1(s), \ldots, i_D(s)  } \bigr\}_{1\le s \le n}\Bigr)\;,
$$
\end{proposition}
\begin{proof}
See  Appendix~\ref{app:proofs-of-finite-free-prop}. 
\end{proof}

\subsection{Asymptotic moments}
\label{sec:asymptotic-moments-tensor}

While one needs all the trace-invariants to describe unitarily invariant functions in the limit $N\rightarrow \infty$, similar to random matrices, the dominant contribution of the expectations of these quantities in the limit  $N\rightarrow \infty$ should 
not contain more information on the asymptotic distribution than the expectations of the connected ones.  This is because one expects again an asymptotic factorization of the expectations of trace-invariants over their connected components, with the non-factorized parts playing a role only at sub-dominant orders. 

\subsubsection{Mixed case}
\label{subsub:mixed}

We consider a mixed invariant $\bsig$ consisting in several mixed connected components $\bsig_1^\mathrm{m},\dots \bsig^\mathrm{m}_q$
corresponding to the blocks of $\Pi(\bsig)$
and we denote:
\be
\Phi^{\mathrm{m}}_{ \bsig} [A]:=k_q \bigl(\Tr_{\bsig^\mathrm{m}_1}(A), \ldots, \Tr_{\bsig^\mathrm{m}_q}(A)\bigr) \;,
\ee
the associated finite $N$ classical cumulant. The superscript ``$\mathrm{m}$'' indicates  that the components are connected in the ``mixed sense'': these are the connected components of the $(D+1)$-edge colored graph, which includes the thick edges, in the representation introduced in Sec.~\ref{subsub:trace-invariants}.

The approach we pursue is to study invariant tensor distributions for which an asymptotic scaling function of the classical cumulants $r_A:S_n^D\rightarrow \mathbb{R}$ is given:
\be
\label{eq:asympt-class-cum-hyp}
\lim_{N\rightarrow \infty} \, \frac 1 {N^{r_A(\bsig) }}\ \Phi^{\mathrm{m}}_{ \bsig} [A] =    \varphi^{\mathrm{m}}_\bsig(a)\;,
\ee
where the $\varphi^{\mathrm{m}}_\bsig(a)$ are not all vanishing and are called \emph{asymptotic moments}. The aim is to obtain a theory of (free) probability for all invariant distributions $A, A',\ldots$ that share the same scaling function $r_A=r_{A'}$, but  might have different asymptotic moments. 
For classical random matrix ensembles for instance, the scaling function is $r_A(\sigma) = 2-\#\sigma$.

If $\bsig$ is connected (in the mixed sense), the classical cumulant equals the expectation and we have:
\be
\label{eq:asympt-conn-mom-tens}
\lim_{N\rightarrow \infty }\,\frac 1 {N^{r_A(\bsig) } }\ \mathbb E \left[\Tr_{\bsig}(A)\right]=  \varphi^{\mathrm{m}}_{\bsig}(a) \; ,
\ee
while if $\bsig$ is not connected ($K_\mathrm{m}(\bsig)>1$), we sometimes render its connected components explicit and denote:
\be
\label{eq:rescaled-cumulants-tensor}
\varphi^{\mathrm{m}}_\bsig(a) = \varphi^{\mathrm{m}}_{\bsig^\mathrm{m}_1, \ldots, \bsig^\mathrm{m}_q}(a) \; .
\ee 
For $\pi\ge \Pi(\bsig)$,
we define the multiplicative extension of the asymptotic moments: 
\be
\label{eq:part-cases-phi-bsig}
\varphi^{\mathrm{m}}_{\pi, \bsig} (a) = \prod_{ G\in \pi} \varphi^{\mathrm{m}}_{ \bsig_{\lvert_{G}}} (a), \quad \mathrm{so\; that }\quad \varphi^{\mathrm{m}}_{1_n, \bsig} (a) = \varphi^{\mathrm{m}}_{\bsig} (a) \; ,\quad  \varphi^{\mathrm{m}}_{\Pi(\bsig), \bsig}(a) = \prod_{i=1}^q \varphi^{\mathrm{m}}_{\bsig^\mathrm{m}_i}(a) \;,
\ee 
where $\bsig^\mathrm{m}_i$ are the mixed connected components of $\bsig$. We use similar notations for the finite $N$ versions $\Phi^{\mathrm{m}}_{\pi, \bsig} [A]$, and for different ensembles $A_1, \ldots, A_n$ we define in the obvious manner the notations $\Phi^{\mathrm{m}}_{\bsig} [\vec A] = \Phi^{\mathrm{m}}_{ \bsig} [A_1, \ldots, A_n]$ and $\varphi^{\mathrm{m}}_\bsig (\vec a)$ and their multiplicative extensions.

\subsubsection{Pure case}
\label{subsub:pure}

In the pure case, the trace-invariants factor over the pure connected components, that is, 
the connected components of the $D$-edge colored graph described in Sec.~\ref{subsub:trace-invariants}. Denoting  $\bsig_1^{\mathrm{p}}, \ldots, \bsig^{\mathrm{p}}_q$ the pure connected components of $\bsig$ corresponding to the blocks of the bipartite partition $\Pi_{\mathrm{p}}(\bsig)$, we have  $\Tr_{\bsig} (T, \bar  T) = \prod_{i=1}^q \Tr_{\bsig^{\mathrm{p}}_i} (T, \bar T)$,
and similar to the mixed case, we denote the pure classical cumulants:
\be
\Phi_{ \bsig} [T, \bar T]=k_q \bigl(\Tr_{\bsig^\mathrm{p}_1 } (T, \bar T), \ldots, \Tr_{\bsig^\mathrm{p}_q } (T, \bar T) \bigr) \; . 
\ee

We assume that a scaling function $r_{T,\bar T}:S_n^D\rightarrow \mathbb{R}$ is given such that:
\be
\label{eq:asympt-class-cum-hyp-pure}
\lim_{N\rightarrow \infty} \, \frac{1} 
{N^{r_{T, \bar T}(\bsig) } }\ \Phi_{ \bsig} [T, \bar T] =   \varphi_{ \bsig} (t, \bar t) \;,
\ee
where the $ \varphi_{ \bsig} (t, \bar t) $ are not all vanishing.
If $\bsig$ is purely connected, $K_\mathrm{p}(\bsig)=1$, then:
   \be
\label{eq:asympt-conn-mom-tens-pure}
\lim_{N\rightarrow \infty} \, \frac 1 {N^{r_{T, \bar T}(\bsig) }}\  \mathbb E \left[\Tr_{\bsig}(T, \bar T)\right] =   \varphi_{\bsig}(t, \bar t) \;,
\ee
and if it is not, we sometimes use the notation
$
\varphi_\bsig(t, \bar t)   = \varphi_{ \bsig_1^{\mathrm{p}}, \ldots \bsig^{\mathrm{p}}_q }(t, \bar t)$.
If $\Pi$ is a bipartite partition such that $\Pi\ge \Pi_\mathrm{p}(\bsig)$, we define the
multiplicative extension: 
\be
\label{eq:part-cases-phi-bsig-pure}
\begin{split}
\varphi_{\Pi, \bsig} (t, \bar t)  &= 
\prod_{ G = B\cup \bar B\in \Pi} \varphi_{ \bsig_{\lvert_B } } (t, \bar t) \; , \crcr
\varphi_{1_{n,\bar n}, \bsig} (t, \bar t) = \varphi_\bsig (t, \bar t),& \qquad \quad \varphi_{\Pi_\mathrm{p}(\bsig), \bsig}(t,\bar  t) = \prod_{i=1}^q \varphi_{\bsig^\mathrm{p}_i }(t, \bar t) \;,
\end{split}
\ee 
where $\bsig_{\lvert_B } $,
the restriction of the map $\bsig$ to the set $B$ (which is such that $\bsig_{\lvert_B } (B)=\bar B$) is well defined, as $\Pi\ge \Pi_\mathrm{p}(\bsig)$.

\ 

The $D=2$ pure case corresponds to matrices $M$ and $\bar M$ (where $\bar M$ can be the complex conjugate  of $M$ or not). The Gaussian scaling is in this case    
$r_{M,\bar M}(\bsig) = 2 - \#(\sigma_1\sigma_2^{-1})$.
The case where $M$ is the pure complex Gaussian ($D=2$) corresponds to the square Wishart random matrix ($D=1$), yielding the order $n$ purely connected asymptotic moments:
\be
\label{eq:asymptotic-moments-pure-gaussian} \varphi_{(\sigma_1, \sigma_2) |_{ 
  K_\mathrm{p}(\sigma_1, \sigma_2)=1 
} }= \,C_n \;.
\ee

For $D\ge 3$,  pure complex random tensor ensembles with $\bar T$ the complex conjugate of $T$ corresponding to  \textsf{LU}-invariant perturbed Gaussian distributions have been studied extensively \cite{1Nexpansion1, 1Nexpansion2, 1Nexpansion3, critical, Gurau-universality, uncoloring, bonzom-SD, Gurau-book, enhanced-1, Walsh-maps, multicritical, bonzom-review, Lionni-thesis, Gurau:2012ix, Gurau:2013pca}. We review them below in Sec.~\ref{sub:Complex-Gaussian-Tensor} and Sec.~\ref{sec:random-tensor-models}.  
For such models it can be proven \cite{Gurau-book} that the classical cumulants admit an asymptotic behavior:
\be
\label{eq:ansatz-asympt-classical-cumulants-pure}
\Phi_{ \bsig} [T, \bar T]  \sim N^{D - \Omega(\bsig)}  \varphi_\bsig(t, \bar t) \; ,
\ee
with $\Omega(\bsig)\ge 0$. Contrary to the $D=2$ case, for higher $D$ the scaling factor $\Omega(\bsig)\ge 0$ is not necessarily additive over the pure connected components, with explicit counterexamples are known for $D=6$ (see Sec.~\ref{sec:order-of-dominance-gaussian-scaling}).

\subsubsection{Order of dominance and asymptotic factorization}
\label{sec:order-general}

Both in the mixed and in the pure case, denoting 
$\bsig_1,\dots \bsig_q $ some connected (in the appropriate sense) invariants and $\bsig_G= \bigcup_{j\in G} \bsig_j$, the classical moment cumulant formula writes:
\be
 \frac{1}{N^{\sum_{i=1}^q r(\bsig_i) } } 
 \bE \left[ \prod_{i=1}^q \Tr_{\bsig_i}[\cdot] \right]  =  \sum_{ \pi\in \mathcal{P}(q)} 
  \frac{1}{ N^{  \sum_{i=1}^q r(\bsig_i)  -\sum_{G\in \pi}  r( \bsig_G)  } } 
 \prod_{G \in \pi}  \;
  \frac{1}{N^{ r(   \bsig_G) } } 
  \Phi_{\bsig_G}[\cdot]  \;.
\ee 
By our scaling assumption, $ N^{-   r( \bsig_G) }    \Phi_{\bsig_G }[\cdot] $
converges  to some finite value in the limit $N\rightarrow \infty$, and one may therefore deduce from the equation above that:
\begin{itemize}
 \item all the rescaled expectations generally\footnote{Here we are interested in characterizing appropriate choices of scaling functions that can be used to describe classes of distributions scaling in this manner, for which the values of the asymptotic moments $\varphi$ can a priori take any values. The following statements are meant in this sense, but other choices of scaling could be considered e.g. for distributions whose higher order asymptotic moments vanish.} have well defined  limits when $N\rightarrow \infty$ if and only if the scaling function is \emph{subadditive} on connected components, that is, for any set of connected components $\{\bsig_i\}_{i\in I}$:
\be
  r\Bigl(\bigcup_{i\in I} \bsig_i \Bigr) \le  \sum_{i\in I} r(\bsig_i ) \;.
\ee
\item the rescaled expectations generally factor at first order:
\be
\lim_{N\to \infty}  \frac{1}{N^{\sum_i r(\bsig_i)} } 
 \bE \left[ \prod_{i} \Tr_{\bsig_i}[\cdot] \right] = \prod_{i} \varphi_{\bsig_i} (\cdot) \;,
\ee
if and only if the scaling function is \emph{strictly subadditive}, that is, for any set of connected components $\{\bsig_i\}_{i\in I}$:
\be
  r\Bigl(\bigcup_{i\in I} \bsig_i \Bigr) <  \sum_{i\in I} r(\bsig_i ) \;.
\ee
\end{itemize}

A scaling function $r(\bsig) = D - \Omega(\bsig)$ as in \eqref{eq:ansatz-asympt-classical-cumulants-pure} is clearly strictly subadditive if $\Omega(\bsig)$ is additive (as it is the case for $D=2$), but because of the presence of the $D$ factor, it may very well be subadditive even if $\Omega(\bsig)$ is not additive, as it is the case for $D=6$ (see Sec.~\ref{sec:order-of-dominance-gaussian-scaling}).

\

The \emph{dominant}, or \emph{first order} invariants are the invariants with maximal $r(\bsig)$; the \emph{order of dominance} of an invariant is the amount by which its scaling is supressed with respect to the first order ones, that is $1 + \max_{\bsig'}r(\bsig') - r(\bsig)$.

\newpage

\section{The Gaussian scaling}
\label{sec:Gaussian-scaling}

We discuss the scaling function $r(\bsig)$ for some pure \textsf{LU}-invariant random tensors $T,\bar T$ with $\bar T$ the complex conjugate of $T$. We first
 discuss the Gaussian case, when the components of $T$ are i.i.d.~complex Gaussian random variables (a Ginibre like tensor) and subsequently the case of a \textsf{LU}-invariant perturbed Gaussian distribution.  In both cases, the first  order asymptotic moments $\varphi_{\bsig}(t,\bar t)$ corresponding to invariants with maximal scaling $r(\bsig)$ are a subclass of the purely connected invariants, called \emph{melonic}, which we will discuss in detail.

This discussion is also crucial for the rest of the paper. In Sec.~\ref{sec:free-cum},  we will  derive the first order free cumulants for generic ensembles of either pure random tensors that scale like a complex Gaussian tensor, or mixed random tensors that scale like a Wishart tensor. Due to these scaling assumption, such distributions have the same first order invariants as the pure complex Gaussian and Wishart tensors discussed here. 

\subsection{Asymptotic scaling of pure Gaussian tensors}
\label{sub:Complex-Gaussian-Tensor}

We start by recalling some well-known (see e.g.~\cite{Gurau-universality}) results for a Gaussian pure random tensor $T,\bar T$, that is, the components of $T$ are i.i.d.~complex Gaussian random variables.
As shown in  \cite{Gurau-universality}, a large class of perturbed Gaussian tensor measures fall in the same universality class.
The expectation of a function $f$ is:
\be
\label{eq:distrib-random-tens}
\bE \left[ f(T,\bar T) \right] = \int 
\frac{ dT d\bar T}{{\mathcal{N}}} \; e^{- N^{D-1} T\cdot\bar T
} f(T,\bar T), 
\ee
where $dT=\prod_{\{i_c\}_{1\le c \le D}} dT_{i^1 \ldots i^D}$   and similarly for the complex conjugate, the normalization $\mathcal{N}$ is chosen so that $\bE \left[ 1 \right]=1$, and 
$T\cdot\bar T =\sum_{i^1,\dots i^D =1 }^N T_{i^1 \ldots i^D}\bar T_{i^1 \ldots i^D}$.

We are interested in studying the asymptotic moments $ \varphi_\bsig(t, \bar t)$, which we denote from now on simply $ \varphi_\bsig$, 
and the scaling function $r_{T,\bar T} (\bsig)$,  henceforth denoted $r(\bsig)$.  The expectations of trace-invariants are computed using Wick theorem, that is, the classical moment-cumulant formula for a centered Gaussian random variable: 
\be
\label{eq:Wick}
\mathbb{E}\left[  T_{i^1_1, \ldots i^D_1} \bar T_{j^1_{\bar 1}, \ldots j^D_{\bar 1} } \cdots  T_{i^1_n, \ldots  i^D_n} \bar T_{j^1_{\bar n}, \ldots j^D_{\bar n} }  \right] = \sum_{\eta \in S_n} \prod_{s=1}^n \bE \left[ T_{i^1_{s} , \ldots i^D_{s} } \bar T_{j^1_{ 
\overline{ \eta(s) } } , \ldots j^D_{ \overline{ \eta(s)}  } }    \right]  \;,
\ee
where $\eta$ defines the ``Wick pairing''. Since $\bE \left[ T_{i^1 \ldots i^D} \bar T_{j^1 \ldots j^D} \right]  = N^{1-D} \prod_{c=1}^D   \delta_{i^c, j^c} $, we get:
\be
\mathbb{E}\left[\Tr_{\bsig}(T, \bar T)\right]  = \sum_{\eta \in S_n} N^{n - d(\bsig, \eta)}
 \; ,\qquad d(\bsig, \eta)=\sum_{c=1}^D\lvert \sigma_c \eta^{-1}\rvert \;. 
\ee

The classical cumulants are given by similar expressions, but with an additional connectivity condition. To be precise, for $\bsig_1, \ldots, \bsig_q$ a collection of purely connected trace-invariants, $K_\mathrm{p}(\bsig_i)=1$, denoting $\bsig\in S_n^D$ their disjoint union, we have: 
\be
\label{eq:sum-over-wick-pairings}
\Phi_{ \bsig} [T, \bar T] =k_q \bigl(\Tr_{\bsig_1} (T, \bar T), \ldots, \Tr_{\bsig_q}(T,\bar T)\bigr)  = \sum_{\substack{{\eta \in S_n, \textrm{ s.t.}}\\{K_\mathrm{p}(\bsig, \eta) = 1}}} N^{n - d(\bsig, \eta)},
\ee
where $K_\mathrm{p}(\bsig, \eta)$ is the number of pure connected components of the trace-invariant defined by the $D+1$ permutations $(\eta,\sigma_1, \ldots, \sigma_D)$. This sum is dominated by the terms which minimize $d(\bsig, \eta)$, that is:
\be
\label{eq:asympt-gaussian}
\Phi_{ \bsig} [T, \bar T] \sim N^{n - \min_{\eta\in S_n, K_\mathrm{p}(\bsig, \eta) = 1 } d(\bsig, \eta)} \, \varphi_\bsig  \; ,
\ee
where:
\be
\label{eq:asympt-moment-gaussian}
\varphi_\bsig  = \mathrm{Card}\bigl\{\eta \in S_n\  \mid\   K_\mathrm{p}(\bsig, \eta) =  1\textrm{ and } d(\bsig, \eta)   \textrm{ is minimal}\bigr\} \;.
\ee
It follows that the scaling function for a complex pure Gaussian tensor is:
\be
\label{eq:def-scling-purely-connected}
r(\bsig)  = n - \min_{\eta\in S_n \,, \; K_\mathrm{p}(\bsig, \eta) =  1} d(\bsig, \eta)
\; . 
\ee 

We will review below various known properties of this Gaussian scaling function. However, we emphasize from the beginning that one main question remains open: it is not known whether this scaling function is subadditive or not. We conjecture this to be the case.

\begin{conjecture}\label{conj:subadi}
 The Gaussian scaling function:
 \[
  r(\bsig)  = n - \min_{\eta\in S_n \,, \; K_\mathrm{p}(\bsig, \eta) =  1} d(\bsig, \eta)
\; ,
 \]
is strictly subadditive on the pure connected components, that is for any $\bsig$ with pure connected components $\bsig_i$ for $i \in I$, $r(\bsig) < \sum_{i\in I} r(\bsig_i) $.
\end{conjecture}

\subsection{Melonic and compatible invariants}
\label{sub:Melo}

Two classes of trace-invariants will play an important role in the following. 

\paragraph{Melonic invariants.}Melonic invariants \cite{critical} dominate the asymptotic moments in the  Gaussian and perturbed Gaussian cases, and more generally, in the case of pure random tensors for which the $\Phi_{ \bsig} [T, \bar T] $ exhibit Gaussian  scaling. The following definitions and results are folklore in the random tensor literature, see \cite{Gurau-book} and references therein. 

Let us fix $D$.  Melonic invariants are defined recursively in the graphical representation: the only invariant with two vertices  is melonic (represented on the left in Fig.~\ref{fig:melon-recursive}) and corresponds to the unique element of $S_1^D$.
If a connected trace-invariant is melonic and has more than two vertices, then it contains a black and a white vertex linked by $D-1$ edges, representing a tensor $T$ and a tensor $\bar T$ sharing precisely $D-1$ indices summed together. If the two remaining indices have color $c$, replacing this pair by an edge of color $c$ as in Fig.~\ref{fig:melon-recursive}, the resulting invariant is itself melonic. Fig.~\ref{fig:ex-melo} depicts some example of melonic graph.\footnote{Conversely, melonic graphs are constructed by recursive insertions of pairs of vertices connected by $D-1$ edges, respecting the colorings, starting from the graph with two vertices.}

\begin{figure}[!h]
\centering
\raisebox{2mm}{\includegraphics[scale=1.3]{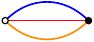}}\hspace{2cm}\includegraphics[scale=1.3]{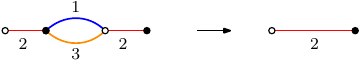}
\caption{Left: the only $D$-colored graph (here $D=3$) with two vertices is melonic. Right: a pair of vertices linked by $D-1$ edges is replaced by an edge. 
}
\label{fig:melon-recursive}
\end{figure}

This recursive construction induces a pairing of the black and white vertices, corresponding to the list of pairs of vertices recursively removed. The pairing of vertices does not depend on the order in which the removals are performed: for each melonic trace-invariant this pairing is unique and will be called its \emph{canonical pairing}. Another characterization is given in Thm.~\ref{thm:second-degree}.

An alternating cycle in the colored graph consisting in edges of a color $c$ and canonical pairs is said to be a  separating cycle, if cutting any pair of edges of color $c$ in the cycle, the number of pure connected components of the graph is raised by one. An equivalent characterization of melonic invariants is that any cycle alternating edges of a fixed color and canonical pairs either contains a single colored edge, or is separating.

\begin{figure}[!h]
\centering
\includegraphics[scale=1.3]{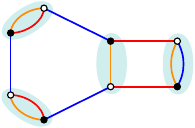}
\caption{Example of melonic graph for $D=3$. The canonical pairs are highligted in grey. 
}
\label{fig:ex-melo}
\end{figure}

The following shows that among connected invariants in $S_n^D$ , melonic invariants are those which minimize the sum of distances between all the pairs of permutations $\sum_{1\le c_1<c_2\le D} \lvert \sigma_{c_1}\sigma_{c_2}^{-1}\rvert$. 

\begin{theorem}[Gurau, Rivasseau \cite{1Nexpansion1, 1Nexpansion2, 1Nexpansion3}]
\label{thm:degree}
Consider a trace-invariant $\bsig\in S_n^D$. The \emph{degree of the invariant}:
\[
\omega(\bsig)=  \sum_{1\le c_1<c_2\le D} \lvert \sigma_{c_1}\sigma_{c_2}^{-1}\rvert  -  (D-1)\bigl( n - K_\mathrm{p}(\bsig)\bigr)  \;,
\]
is a non-negative integer $\omega(\bsig)\ge 0$. For $D\ge 4$, the degree vanishes if and only if $\bsig$ is melonic. 
\end{theorem}

We have the following associated result (see e.g.~\cite{Gurau-universality}).

\begin{theorem}
\label{thm:second-degree}
For $D\ge 3$, consider a trace-invariant $\bsig\in S_n^D$ and $\eta\in S_n$.  Then:
\be
\omega(\bsig,\eta)  -  \omega(\bsig) 
= \bar \omega(\bsig;\eta) = D K_\mathrm{p}(\bsig, \eta) - (D-1) K_\mathrm{p}(\bsig) - n  +  d(\bsig, \eta) \ge 0,
\ee
with equality if and only if the trace-invariant $(\eta,\bsig)$ is melonic, which implies that $\bsig$ is \emph{melonic}.

If  $\bsig$ is \emph{melonic}, then there exists a unique $\eta\in S_n$ such that $K_\mathrm{p}(\bsig, \eta)=K_\mathrm{p}(\bsig)$ and $\bar \omega(\bsig;\eta) =0$: it corresponds to the canonical pairing of $\bsig$. This holds in particular if $\bsig$ is melonic and purely connected. 

If $\bsig$ is melonic but not purely connected, there are other $\eta\in S_n$ (different from the canonical pairing) such that $\bar \omega(\bsig;\eta)=0$ and $K_\mathrm{p}(\bsig, \eta)<K_\mathrm{p}(\bsig)$. 

\end{theorem}

\paragraph{Compatible invariants.}
Compatible invariants $\bsig$  (see e.g.~\cite{funwithreplicas})  are those for which there exists a permutation $\eta$ lying on a geodesic linking every pair of permutations of $\bsig$. Defining:
\be\label{eq:compat-inv}
\nabla( \bsig ; \eta) = \sum_{1\le c_1<c_2\le D} \Big(  \lvert \sigma_{c_1} \eta^{-1}\rvert + \lvert \sigma_{c_2} \eta^{-1} \rvert  - \lvert \sigma_{c_1}\sigma_{c_2}^{-1}\rvert \Big) \ge 0 ,
\ee
then 
$ 
\min_{\eta\in S_n} \nabla( \bsig ; \eta) \ge 0
$
 vanishes if and only if $\bsig$ is compatible and if $\nabla( \bsig ; \eta)=0$, we say that $\eta$ \emph{renders $\bsig$ compatible}. 
\begin{lemma}
 Melonic invariants are compatible. If $\bsig\in S_n^D$ is melonic, there is a unique permutation $\eta$ which realizes $\nabla( \bsig ; \eta)=0$,  and it corresponds to the canonical pairing of $\bsig$. 
\end{lemma}

\proof 
Since:
 \be
 \label{eq:two-perm-to-one-perm}
 \sum_{1\le c_1<c_2\le D}  \Big( \lvert \sigma_{c_1} \eta^{-1}\rvert + \lvert \sigma_{c_2} \eta^{-1} \rvert\Big)  = (D-1)\sum_{c=1}^D\lvert \sigma_{c} \eta^{-1}\rvert \;, 
 \ee
we have:
 \be
 \label{eq:nabla-and-degrees}
\frac{1}{D-1}\nabla( \bsig ; \eta) = 
\bar \omega(\bsig; \eta) - \frac{\omega(\bsig)}{D-1} + D\bigl[K_\mathrm{p}(\bsig) - K_\mathrm{p}(\bsig, \eta)\bigr].
\ee
For a melonic $\bsig$, Thm.~\ref{thm:degree}, 
and Thm.~\ref{thm:second-degree} imply that first 
$\omega(\bsig)=0$, and second there exists a unique $\eta$ corresponding to the canonical pairing such that $K_\mathrm{p}(\bsig) =K_\mathrm{p}(\bsig, \eta)$ and  $\bar \omega(\bsig; \eta) = 0$.\qed

\subsection{Order of dominance for the Gaussian scaling}
\label{sec:order-of-dominance-gaussian-scaling}

Rewriting the Gaussian scaling function in $r(\bsig) = n-\min d(\bsig,\eta)$ using Thm.~\ref{thm:second-degree} we conclude that the finite $N$ cumulants scale like: 
\be
\label{eq:Gaussian-asymptotics-with-degree}
\Phi_{ \bsig} [T, \bar T] \sim N^{ 1 - (D-1) (K_\mathrm{p}(\bsig)-1) - \min \bar \omega(\bsig;\eta)} \, \varphi_\bsig  \; ,
\ee
where the minimum is taken over the $\eta \in S_n$ such that $K_\mathrm{p}(\bsig, \eta) = 1$ and  $\varphi_{\bsig}$ is the number of $\eta\in S_n$ with $K_\mathrm{p}(\bsig, \eta) = 1$ which minimize $\bar \omega(\bsig;\eta)$.

\paragraph{First order invariants.} The first order contributions are the invariants with maximal scaling.  
\begin{theorem}
\label{thm:melonic-gaussian}
For $D\ge 3$ and $T$ the order $D$ complex pure Gaussian tensor, the invariants $\bsig$ with maximal scaling 
$r(\bsig)$ are the \emph{purely connected, melonic} invariants. Furthermore, the corresponding \emph{asymptotic moment $\varphi_\bsig$ is one}, and the first correction is of order $N^{3-D}$:
\[
\Phi_{ \bsig} [T, \bar T]_{\bigl\lvert{\substack{{K_\mathrm{p}(\bsig)=1}\\{\omega(\bsig)=0}}}} = N + O({1}/{N^{D-3}}) \;.
\]
\end{theorem}

The only point not contained in Thm.~\ref{thm:second-degree} is the order of the correction, see e.g.~\cite{Gurau-Schaeffer}.
The fact that $\varphi_\bsig=1$ is to be compared with the $D=2$ case \eqref{eq:asymptotic-moments-pure-gaussian}, for which one obtains the Catalan number $C_n$. The melonic invariants can be enumerated and the number of connected melonic invariants is the Fuss-Catalan number \cite{critical}, which is to be
compared with the $D=2$ case for which there is only one such invariant at each $n$.
 
\paragraph{Purely connected invariants.}The order of dominance $1 +  (D-1) (K_\mathrm{p}(\bsig)-1) + \min \bar \omega(\bsig;\eta)$
of an invariant captures the amount by which its scaling is suppressed with respect to the first order invariants. The invariants of orders $k \in \{2,\ldots, D-1\}$ are purely connected and such that $\min \bar \omega = k-1$ \cite{SYK}. For order $D$ and above, both $K_\mathrm{p}(\bsig)$ and $\min \bar \omega(\bsig,\eta)$ play a role in determining the order.
This is new with respect to the $D=2$ case, for which there is no equivalent of $\min \bar \omega(\bsig,\eta)$. For purely connected trace-invariants, \eqref{eq:Gaussian-asymptotics-with-degree} simplifies to:  
\be
\label{eq:Gaussian-asymptotics-connected}
\Phi_{ \bsig} [T, \bar T] _{\bigl\lvert_{K_\mathrm{p}(\bsig)=1}} \sim N^{ 1 - \min_{\eta\in S_n} \bar \omega(\bsig;\eta)} \, \varphi_\bsig \;,
\ee
and one can in principle \cite{SYK, Lionni-thesis} identify for any order $k\ge 2$ the purely connected trace-invariants with $\min_{\eta\in S_n} \bar \omega(\bsig;\eta) =k-1$. The contributions of order 2 and 3 are for  instance given in (19) and (25), (26) and (27) of \cite{SYK} and are asymptotically enumerated in \cite{SYK2}. In practice, this becomes rapidly quite tedious. Given a purely connected trace-invariant $\bsig$, short of a full computation, it is not obvious how to relate $\min_{\eta\in S_n} \bar \omega(  \bsig;\eta)$ to the properties of $\bsig$. 

\begin{proposition}
\label{prop:lower-bound-on-bar-omega}
From  \eqref{eq:nabla-and-degrees} it follows that
for $\bsig\in S_n^D$ purely connected we have:
\[
\min_{\eta\in S_n} \bar \omega(\bsig;\eta)  \ge \frac{\omega(\bsig)}{D-1}, 
\]
with equality if and only if $\bsig$ is compatible. 
\end{proposition}

This gives a lower bound on the order of contribution of an invariant $\bsig$.  There are compatible and non-compatible invariants at any order,  as for every $k\ge 2$ it is easy to construct both compatible and non-compatible purely connected invariants satisfying $\min_{\eta\in S_n} \bar \omega( \bsig;\eta) =k-1$.\footnote{This relies on the method of \cite{Gurau-Schaeffer}. One can also prove that the probability  for an invariant $\bsig\in S_n^D$ of degree $(k-1)(D-1)$ to be compatible tends to one when $n$ goes to infinity.}

\paragraph{Non-connected invariants.} The first contribution with $K_\mathrm{p}(\bsig)>1$ arises at order $D$ and corresponds to $\bsig$ consisting of two purely connected melonic graphs. This generalizes to arbitrary number of pure connected components.
\begin{theorem}
For $D\ge 3$ and $T$ the order $D$ complex pure Gaussian tensor, the invariants $\bsig$ with maximal scaling $r(\bsig)$ at fixed number of connected components $q=K_\mathrm{p}(\bsig)$ are the disjoint unions of $q$ melonic purely connected components:
\[
\Phi_{ \bsig} [T, \bar T]_{\bigl\lvert{\substack{{K_\mathrm{p}(\bsig)=q}\\{\omega(\bsig)=0}}}} = N^{1-(D-1)(q-1)}\Bigl(\varphi_\bsig + O({1}/{N^{D-2}})\Bigr),
\]
where $\varphi_\bsig$ is the number of $\eta$ such that $(\bsig,\eta)$ is melonic.\footnote{It should be straightforward to enumerate this family.} 
\end{theorem}

For non-melonic, non purely connected invariants, the situation is worse. The task of computing the Gaussian scaling function for an arbitrary invariant is computationally hard: beyond checking all the Wick pairings, no procedure is known to read off $\min \bar \omega(\bsig,\eta)$ or to construct the $\eta$ minimizing $\bar \omega(\bsig;\eta)$ under the constraint $K_\mathrm{p}(\bsig, \eta) = 1$.
As general exact results are lacking, the best one can  
do is to search for convenient bounds on the scaling function. For any invariant $\bsig$ an obvious bound on the Gaussian scaling is:
\be
\label{eq:naive-bound}
 r( \bsig ) = 1 - (D-1) (K_\mathrm{p}(\bsig)-1) - 
 \min_{\eta_\in S_n,K_{\mathrm{p}} (\bsig,\eta) =1}
\bar \omega(\bsig,\eta)\le D - \sum_{i=1}^{ K_\mathrm{p}(\bsig)} (D-1) \; .
\ee
The bound is saturated by the melonic family and is additive on the pure connected components.

The whole idea is to try to improve this naive additive bound as much as possible. The strategy is to search for bounds consisting in a constant term $D$,\footnote{ See Thm 4.2.2 in the published version of \cite{Lionni-thesis}. } the scaling of the expectation of $1$, plus a piece which is additive on the pure connected components. The whole point is to find optimal upper bounds: for various families $\mathbb{B}$ of purely connected invariants one searches for some numbers $b(\bsig_i)$ for $\bsig_i\in \mathbb{B}$ such that:
\be
\label{eq:asympt-gaussian-expected1}
r\Bigl(\bigcup_{i} \bsig_i \Bigr) = D - \sum_{i} (D-1) - 
 \min_{\eta_\in S_n,K_{\mathrm{p}} (\bsig,\eta) =1}
\bar \omega(\bsig,\eta)\le D - \sum_i b(\bsig_i)\;,
\ee 
and the bound is \emph{tight},\footnote{It is immediate to prove that if such a bound holds and it is saturated for all the connected invariants then the scaling is strictly subadditive.} in the sense that for any $\bsig_l \in \mathbb{B}$, there exists some $\bigcup_{j}\bsig_j$ with $\bsig_j\in \mathbb{B}$, such that the bound is saturated for $ \bigcup_{j}\bsig_j \cup \bsig_l$.
The naive bound \eqref{eq:naive-bound} would correspond to $b(\bsig_i)=D-1$, but it is not tight, and the point is to improve this by increasing $b(\bsig_i)$ as much as possible: while for purely connected melonic invariants $b(\bsig_i) = D-1$ cannot be improved, it turns out that for other invariants this can be improved to some $b(\bsig_i)> D-1$, called  \emph{optimal} (due to the tightness)
\cite{enhanced-1, Walsh-maps, octa, multicritical, bonzom-review, Lionni-thesis, Bonzom-balls}.

For $D<6$, the maximal $b(\bsig_i)$ allowed by the scaling of the purely connected invariants, $b(\bsig_i) = D-1 +  \min_{\eta \in S_n} \bar \omega( \bsig_i;\eta)$, 
works in all the examples treated so far\footnote{Rigorously, one should denote this $ b_{ \mathbb{B}}(\bsig_i)$, as one proves the bound for a fixed family of invariants. It is an open question whether these improved bounds change by including more invariants.}, and \eqref{eq:asympt-gaussian-expected1} is always saturated.  However, as detailed in the last section of \cite{Lionni-thesis} this cannot be true in general. For $D=6$ there exists a connected $\bsig_0\in S_3^6$ such that $r(\bsig_0)=-4$ but $r(\bsig_0\cup\bsig_0)=-12 $, inconsistent with this choice of $b(\bsig_i)$.\footnote{It would give $b(\bsig_0)=10$ but then $r(\bsig_0\cup\bsig_0)=-12 > 6-20=-14$ contradicts \eqref{eq:asympt-gaussian-expected1}. Instead the correct choice is likely $b(\bsig_0)=9$ \cite{Lionni-thesis}, so that \eqref{eq:asympt-gaussian-expected1} is a strict inequality for $\bsig_0$. This phenomenon is also found for the real Gaussian tensor for $n=2$ and $D=3$, for instance. } In particular: 
\[
\frac {k_2\bigl(\Tr_{\bsig_0}(T, \bar T), \Tr_{\bsig_0}(T, \bar T)\bigr)}{\mathbb{E}\bigl[\Tr_{\bsig_0}(T, \bar T)\bigr]^2} \asymp N^{-4} \;,
\] 
which is less suppressed than what the intuition from measure  concentration estimates would suggest \cite{ledoux}. 
To our knowledge this counter-intuitive scaling is specific to the tensor realm.  
However, we stress that the Gaussian scaling function is clearly subadditive in this case, that is, this example does not contradict our conjecture~\ref{conj:subadi}. 

\subsection{Random tensor models with invariant potentials}
\label{sec:random-tensor-models}

The expectations of a perturbed Gaussian tensor measure are given by:
\be
\label{eq:distrib-random-tens}
\bE \left[ f(T,\bar T) \right] = \int \frac{dT d\bar T} { \mathcal{N} } \, e^{- N^{D-1}T\cdot\bar T
+  V(T,\bar T)
} f(T,\bar T) \;,
\ee 
where $\mathcal{N}$ is such that $\mathbb{E}(1)=1$,  and  the perturbation potential is an invariant: 
\be
\label{eq:potential}
V(T,\bar T) \bigl[\{z_\bsig\}\bigr]= \sum_{n\ge 2}  \sum_{\bsig\in S^D_n} \frac{N^{ \zeta(\bsig)} z_\bsig}{c(\bsig)}\, \Tr_{\bsig}(T, \bar T) \;,
\ee
with $\zeta(\bsig)\ge 0$ a choice of scaling, $c(\bsig)$ suitable combinatorial factors and $z_\bsig\in \mathbb{C}$ the coupling constants. The convergence of the integral is ensured for some choices of the couplings. We denote $\mathbb{B}$ the (potentially infinite) set of trace-invariants for which $z_\bsig\neq 0$. The Gaussian case is recovered by setting all the couplings to $0$, and will be denoted $T_0$ in this subsection.

\

The characterization of the asymptotic moments of such distributions is usually approached ``perturbatively'', that is by expanding
$\exp V(T,\bar T)$ in Taylor series and exchanging the Gaussian integral and the sum. The cumulants are then expressed as divergent sums over connected $(D+1)$-colored graphs. The limit $N\rightarrow \infty$ selects sub-series of graphs that optimize some combinatorial constraints and are summable when the coupling constants are sufficiently small. The fact that this procedure provides the true $N\rightarrow \infty$ asymptotic of the cumulants of \eqref{eq:distrib-random-tens} is proved by constructive methods \cite{Gurau-universality,Gurau:2013pca,Gurau-book}.

\

The advantage of the perturbative approach is that it reduces the computation of the asymptotic scaling of the cumulants 
$\Phi_{ \bsig} [T, \bar T] =k_q \bigl(\Tr_{\bsig_1 }(T, \bar T), \ldots, \Tr_{\bsig_q}(T,\bar T)\bigr)$ for $T$ distributed as in \eqref{eq:distrib-random-tens}, to the evaluation of the dominant asymptotic contributions of Gaussian  cumulants of the form: 
\be
k_{q+l} \Bigl(\Tr_{\bsig_1 }(T_0, \bar T_0), \ldots, \Tr_{\bsig_q }(T_0,\bar T_0), \Tr_{\btau_1}(T_0, \bar T_0), \ldots, \Tr_{\btau_l}(T_0,\bar T_0)\Bigr)  \;,
\ee
where $T_0$ is the order $D$ pure Gaussian tensor of the previous subsections, and the $\btau_i$ can be any invariant in $\mathbb{B}$.

Denoting by $\bsig$ and $\btau$ the disjoint unions of the  $\bsig_i$ and $\btau_j$, respectively, and letting $n$ be  the total number of tensors $T_0$, so that $\bsig\cup \btau\in S_n^D$, and $\eta\in S_n$ be an arbitrary permutation, the asymptotic scaling $r_V(\bsig)$ of $\Phi_{ \bsig} [T, \bar T] $ is obtained \eqref{eq:asympt-gaussian} by maximizing $n- d(\bsig\cup \btau, \eta)$ over $\eta\in S_n$ such that $K_\mathrm{p}(\bsig\cup \btau, \eta) = 1$. 
As discussed in the previous section, a tight additive bound  \eqref{eq:asympt-gaussian-expected1} on the Gaussian scaling holds in principle for some optimal $b$, and this bound translates here to:
\be
\label{eq:potential-exponent-to-maximize}
\begin{split}
& r_V(\bsig)\le \max_{ \btau_j  \in \mathbb{B }} \bigg\{D - \sum_{i=1}^q b(\bsig_i) + \sum_{j=1}^l\bigl(\zeta(\btau_j) - b(\btau_j)\bigr) \bigg\} \;  , \crcr
& b(\btau_j) \ge D-1 \; , \qquad  b(\btau_j^{\text{melonic} } ) = D-1 \; ,
\end{split}
\ee
where this bound is also optimal for $r_V$.\footnote{Similar issues as for the Gaussian case hold here: the question of the  independence with respect to $\mathbb{B}$ and convergence issues \cite{Lionni-thesis}, which are also related here to the existence of the maximum, as well as the question of whether this bound is always saturated for $D<6$.}

The statements in the theorem below are either reformulations of the results in \cite{Gurau-universality} and \cite{enhanced-1, Walsh-maps,  multicritical, bonzom-review, Lionni-thesis}, or follow from  \eqref{eq:potential-exponent-to-maximize}. It relates the optimal $b$'s from the tight additive upper bound \eqref{eq:asympt-gaussian-expected1} on the Gaussian scaling $r$, to the $\zeta$'s in the potential \eqref{eq:potential}.

\begin{theorem}[Gaussian universality] 
\label{thm:gaussian-universality}
As a function of the scaling $\zeta(\btau)$ of the potential $V$ in the distribution \eqref{eq:distrib-random-tens}, we obtain different large $N$ limits: 
\begin{itemize}
 \item If $\zeta(\btau) = D-1$ for all the invariants $\btau \in \mathbb{B}$, then the distribution is asymptotically Gaussian with a modified covariance  \cite{Gurau-universality}:
\be
\label{eq:gau1}
\Phi_{ \bsig} [T, \bar T] \sim_N \Phi_{ \bsig} [T_0, \bar T_0] \; G^n \;,  \qquad \forall\bsig\in S^D_n \;,
\ee
where the notation $\sim_N$ signifies that the two are identical at first order in the limit $N\rightarrow \infty$.
Furthermore, denoting $n_\btau$ the cardinal of the set on which $\btau$ acts, $G$ is the unique power series solution of the equation:
\be
\label{eq:gau2}
G= 1 - \sum_{\substack{{\btau\, \in \, \mathbb{B}}\\{\mathrm{melonic} }}} n_\btau z_\btau G^{n_\btau} \; . 
\ee

We stress that this holds independently on the existence of a tight additive upper bound. The same holds if $\zeta(\bsig) = D-1$ for the melonic graphs and $\zeta(\bsig) <D-1$ for the non-melonic ones.\footnote{This includes the case where the scaling $\zeta(\bsig)$ is equal to  $-\omega(\bsig)$, which is common in the literature \cite{Gurau-universality, uncoloring}. }
\end{itemize}

Assuming the existence of a tight additive upper bound on the Gaussian scaling \eqref{eq:asympt-gaussian-expected1} involving optimal $b(\btau)>D-1$ for the non-melonic $\btau\in \mathbb{B}$  and  of the upper bound \eqref{eq:potential-exponent-to-maximize}, then:

\begin{itemize}

\item If for all $\btau \in \mathbb{B}$, $\zeta(\btau) < b(\btau)$,\footnote{For instance if for all $\btau \in \mathbb{B}$, $\zeta(\btau) < D-1$.} then the distribution is asymptotically Gaussian identical with the distribution of $T_0$, in the sense that for any trace-invariant $\bsig$:
\be
\Phi_{ \bsig} [T, \bar T] \sim \Phi_{ \bsig} [T_0, \bar T_0] \; . 
\ee

\item  If $\zeta(\bsig) =D-1$ for the melonic graphs and 
 $D-1\le \zeta(\btau) < b(\btau) $ for the non-melonic ones, the Gaussian behavior of 
  \eqref{eq:gau1} and \eqref{eq:gau2} is obtained again.
 
\item If for some non-melonic $\btau \in \mathbb{B}$, $\zeta(\btau) = b(\btau)$, then the distribution is not necessarily asymptotically Gaussian \cite{enhanced-1, Walsh-maps,  multicritical, bonzom-review, Lionni-thesis}.\footnote{Note that the scaling $r_V$ may still coincide with the Gaussian scaling $r$ considered in this paper even when the distribution is not asymptotically Gaussian. Choosing $\zeta(\bsig)=b(\bsig)$ for all $\bsig$, the upper bounds \eqref{eq:asympt-gaussian-expected1} and \eqref{eq:potential-exponent-to-maximize} coincide, and are saturated for all known choices of $\mathbb{B}$ for $D=3,4,5$. }

\item If for some $\btau \in \mathbb{B}$,  $\zeta(\btau) > b(\btau)$, then the distribution does not admit a large $N$ limit.

\end{itemize}
\end{theorem}

\subsection{Scaling of the Wishart tensor}
\label{sub:true-Wishart-tensor}

Similar to a Wishart matrix, the Wishart tensor is  obtained from the pure complex Gaussian tensor $T,\bar T$ by partially tracing over one of its indices: 
\be
\label{eq:wishart-tensor}
W_{i^1 \ldots i^D ; j^1 \ldots j^D} = \sum_{k=1}^N T_{i^1 \ldots i^D k}\bar T_{j^1 \ldots j^D k} \; .
\ee
A priori, this is just an equivalent perspective on the pure case $T, \bar T$ with the difference that there is a fixed labeling of the $T$ and $\bar T$ indicating the pairs whose $(D+1)$th indices are summed together to form a mixed tensor. 

\paragraph{Scaling function.}As $T$ is a pure complex Gaussian, we have for $\bsig\in S_n^D$:
\be
\label{eq:scaling-true-wishart-D}
\lim_{N\rightarrow \infty} \frac 1 { N^{r_W(\bsig)}}\,\Phi^\mathrm{m}_\bsig[W] =\varphi^\mathrm{m}_\bsig(w), \hspace{1.5cm} r_W(\bsig)= n - \min d\bigl( ( \bsig ,  \mathrm{id}), \eta\bigr),
\ee
where the minimum is taken over $\eta\in S_n$ for which $K_\mathrm{m}(\bsig, \eta) =K_\mathrm{p}\bigl(( \bsig ,\mathrm{id}) , \eta\bigr)  =1$, and $(\bsig , \mathrm{id}) =(\sigma_1, \ldots, \sigma_D, \mathrm{id})$ so that $d\bigl(( \bsig ,  \mathrm{id} ), \eta\bigr) = d(\bsig, \eta)  + \lvert \eta \rvert$. As $r_W(\bsig)=r(  \bsig, \mathrm{id})$ with $r$ the Gaussian scaling in \eqref{eq:def-scling-purely-connected}:
\be
 r_W(\bsig) = 1- D(K_\mathrm{m}(\bsig) - 1) - \min \bar \omega \bigl( (\bsig,  \mathrm{id}) ; \eta\bigr). 
\ee

\paragraph{First order.}The first order corresponds to the $\bsig$ which are mixed connected (hence not necessarily purely connected) and for which $(\bsig, \mathrm{id})$ is melonic, that is, they are melonic when including the thick edges corresponding to the permutation $\mathrm{id}$, therefore:
\begin{enumerate}[label=$-$]
\item If $\bsig$ is purely connected, then its canonical pairing must be the identity: $\bar \omega(\bsig;\mathrm{id})=0$. 
\item If $\bsig$ is not purely connected, then each of its pure connected components is melonic, and two situations may occur for  the thick edges: either they connect the vertices of a canonical pair in one of the pure connected components, or if they connect different pure connected components, then the cycle in the graph formed alternatively by thick edges and  canonical pairs is a  separating cycle (Sec.~\ref{sub:Melo}). In that case, the canonical pairing is not the identity. 
\end {enumerate}

\begin{figure}[!h]
\centering
\raisebox{+0.7cm}{\includegraphics[scale=1]{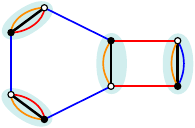}}\hspace{3cm}\includegraphics[scale=1]{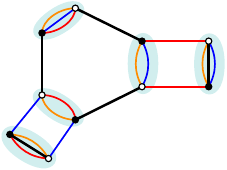}
\caption{A pure random tensor with $D$ indices with scaling function $r$ is defined asymptotically at first order by the purely connected melonic invariants. A mixed random tensor with scaling function $r_W$ is defined asymptotically at first order by the melonic invariants which are  connected when the thick edges  play the role of a color $D+1$. The  purely connected ones (left) coincide in the two situations (Fig.~\ref{fig:ex-melo}), but the others (right) are considered  only in the mixed case. 
}
\label{fig:ex-invariants-purs-mixed}
\end{figure}

At this level it is just a partition of the melonic connected graphs  with $D+1$ colors into two classes, according to whether the deletion of all the edges of color $D+1$ disconnects the graph or not. In Sec.~\ref{sub:mixed-like-pure-gaussian}, we will treat more generally mixed random tensors that scale like a Wishart tensor (without assuming that they derive from a pure tensor with $D+1$ indices). In that case, there is a difference between the invariants that characterize the asymptotic distributions in the pure and mixed cases, see Fig.~\ref{fig:ex-invariants-purs-mixed}.

\newpage

\section{Free cumulants of unitarily invariant random tensors }
\label{sec:free-cum}

In this section, we define and discuss first order free cumulants in two situations: pure random tensors that scale like a pure complex Gaussian, and mixed random tensors that scale like a Wishart tensor. 

\subsection{Free cumulants  for pure Gaussian tensors}
\label{sec:free-cumulants-Gaussian}

We start from the example of a pure complex Gaussian tensor with covariance 
$\bE \left[ T_{i^1 \ldots i^D} \bar T_{j^1 \ldots j^D} \right]  = N^{1-D} \; C \prod_{c=1}^D \delta_{i^c, j^c}$. One then has for the two-tensors invariant $T\cdot \bar T$  (left of Fig.~\ref{fig:melon-recursive}), the only element of $S_1^D$, which we denote by $\textbf{id}_1$:
\be
\Phi_{\mathbf{id}{_1}} \bigl[T, \bar T\bigr]= \mathbb E \bigl[  T\cdot \bar T\bigr]  = N C \;, 
\ee
so that the asymptotic covariance is: 
\be
\label{eq:asymptotic-covariance-gaussian}
\varphi_{\mathbf{id}{_1}}  = 
 \lim_{N\to \infty} \frac{1}{N} \Phi_{\mathbf{id}{_1}} \bigl[T, \bar T\bigr]
= C \;.
\ee
In the case discussed in Thm.~\ref{thm:gaussian-universality} of an asymptotically Gaussian random tensor model for instance, one would have $C=G$. 

\paragraph{Free cumulants.}In the Gaussian case, the moment-generating function is simply:
\be
\log Z_{T,\bar T} (J,\bar J) = \log\bE_{T,\bar T} [e^{  J \cdot T +  \bar J\cdot \bar T }] = C \, N^{1-D} \; J \cdot \bar J  =  C \, N^{1-D} \; \Tr_{\textbf{id}_1}(J,\bar J) \;,
\ee
hence the only finite $N$ free cumulant is $\mathcal{K}_{\textbf{id}_1}[T,\bar T]$, and the only non-trivial asymptotic free cumulant is:
\be
 \kappa_{{\mathbf{id}{_1} } } = 
  \lim_{N\to \infty} \frac{1}{N^{1-D}} \mathcal{K}_{\textbf{id}_1}[T,\bar T] = C  \;.
\ee

The corresponding free cumulants can be generally expressed in terms of the asymptotic moments by: 
\be
\label{eq:free-cumulants-gaussian}
\kappa_{\bsig} =  \varphi_{\mathbf{id}{_1}}  \times  \delta_{\bsig, {\mathbf{id}{_1}} }\;,
\ee
and  \eqref{eq:asympt-moment-gaussian} gives the inverse relation, that is, the asymptotic moments are expressed in terms of the free cumulants through the relation:
\be
\label{eq:asympt-gaussian-bis}
\forall \bsig\in S_n^D \;, \qquad\varphi_{ \bsig}  =  \mathrm{Card}\bigl\{\eta \mid K_\mathrm{p}(\eta, \bsig)=1\ \mathrm{and}\ \bar \omega(\bsig, \eta)\ \mathrm{minimal}\bigr\} \times  \bigl(\kappa_{\mathbf{id}{_1}} \bigr)^n \; .
\ee
In particular for $D\ge 3$, the cardinal is just $1$ if $\bsig$ is purely connected and melonic, see Thm.~\ref{thm:melonic-gaussian}, therefore one has:  
\be
\label{eq:asympt-gaussian-melonic-bis}
{\varphi_{ \bsig}}  _{\bigl\lvert{\substack{{K_\mathrm{p}(\bsig)=1}\\{\omega(\bsig)=0}}}} =  (\kappa_{\mathbf{id}{_1}} )^n \;.
\ee

\begin{remark}
 This is to be compared with the $D=2$ result, for which \eqref{eq:asympt-gaussian-bis} still holds but for $\bsig=(\sigma_1, \sigma_2)$ with $\#(\sigma_1\sigma_2^{-1})=1$, the number of $\eta$ which minimize $\bar\omega(\bsig,\eta)$ is the Catalan number \eqref{eq:Catalan}:
\be
\label{eq:asympt-gaussian-melonic-bis}
\varphi_\bsig =  C_{n} \times (\kappa_{\mathbf{id}{_1}} )^n \;.
\ee

The notion of cumulant depends on whether we consider a pure  complex tensor $T,\bar T$ with $D=2$, or a (mixed) Wishart matrix $W=TT^\dagger$ in $D=1$. The cumulants in the two cases are related by: 
\be
\label{eq:wishart-as-free-cumulantvs-pure-D2}
\kappa_n(w)=\kappa_{\Pi_\mathrm{p}(\gamma_n, \gamma_n), (\gamma_n, \gamma_n)}(t, \bar t) \; , 
\ee
which agree since the pure connected components of $(\sigma, \sigma)$ are a collection of invariants in $S_1^2$. More generally $\kappa_{\Pi(\sigma), \sigma}(w) = \kappa_{\Pi_\mathrm{p}(\sigma, \sigma), (\sigma, \sigma)}(t, \bar t)$.
\end{remark}

\subsection{Free cumulants for pure unitarily invariant random tensors}
\label{sec:free-cumulants-pure}

Our aim is to identify the correct notion of tensorial free-cumulants associated to the first order moments $\Phi_{ \bsig} [T,  \bar T] $ for \textsf{LU}-invariant pure random tensors $T,\bar T$ whose classical cumulants scale with the Gaussian scaling function \eqref{eq:asympt-gaussian}, that is:
\be
\label{eq:scaling-hypothesis-pure}
\lim_{N\rightarrow \infty}\frac 1 {N^{r(\bsig)}}\ \Phi_{ \bsig} [T,  \bar T]  =  \varphi_\bsig(t, \bar t), \hspace{1.5cm} r(\bsig)= n - \min_{\eta\in S_n, K_\mathrm{p}(\bsig, \eta) = 1 } d(\bsig, \eta) \; .
\ee
We do not assume anything regarding the  asymptotic moments $\varphi_{\bsig}(t, \bar t)$: they are an unspecified list of numbers 
that characterize the distribution asymptotically.
This includes the pure complex Gaussian case and some Gaussian measures perturbed by invariant potentials, but the results derived here are a priori more general.  

Our starting point is Thm.~\ref{thm:finite-free-cumulants}, namely the formula of the finite $N$ free cumulant:
\be
\begin{split}
 \mathcal{K}_\bsig[T,\bar T] & = \sum_{\btau \in S^D_{n,\bar n}} \sum_{\substack{{\Pi \in \mathcal{P}(n,\bar n)}\\{\Pi \ge \Pi_{\mathrm{p}}(\bsig) \vee 
 \Pi_{\mathrm{p}} (\btau)}}} \lambda_\Pi  \prod_{G= B\cup \bar B\in \Pi}   \bE\left[\Tr_{\btau_{|_B} } (T,\bar T)\right]\prod_{c=1}^D W^{(N)} \left( \sigma_{ c|_B}
 \tau_{c|_B}^{-1} \right) , 
\end{split}
\ee
supplemented by the finite $N$ moment cumulant expressions for $\Pi\ge \Pi(\btau)$:
\be
\prod_{G = B\cup \bar B \in \Pi} \bE\left[\Tr_{\btau|_{B}} (T,\bar T)\right] =
\sum_{\substack{ {\Pi' \in  \mathcal{P}(n,\bar n) }\\ { \Pi\ge \Pi' \ge \Pi(\btau)} } } \;\; 
\prod_{G' = B'\cup\bar B' \in \Pi'}  \Phi_{\btau|_{B'}} [T,\bar T] \;,
\ee
and the asymptotic of the Weingarten function
$W^{(N)} ( \nu) = \mathsf{M}(\nu) \,  N^{- n - \lvert\nu\rvert} (1+O(N^{-2}))$ for $\nu\in S_n$.

\paragraph{Limit of first order finite size free cumulants.} 

As discussed in Section~\ref{sec:order-of-dominance-gaussian-scaling}, the first order consists of the purely connected melonic trace-invariants. 

In order to state the appropriate moment cumulant relation at first order, we recall that $\btau \preceq \bsig$ means that for all $c\in \{1, \ldots, D\}$, $\tau_c\preceq \sigma_c$, that is, for each cycle of $\sigma_c$, the restriction of $\tau_c$ to the support of this cycle is a non-crossing permutation. Also, for $\bnu\in S_n^D$, we define $\mathsf{M}(\bnu) = \prod_{c=1}^D \mathsf{M}(\nu_c)$, where $\mathsf{M}(\nu)$ is the  M\"{o}bius function on the lattice of non-crossing partitions. Denoting $\mathrm{NC}(n)$ the lattice of non-crossing partitions on $n$ elements, $\mathsf{M}(\bnu)$ is the M\"{o}bius function on the direct product of lattices $\left(\mathrm{NC}(n)\right)^{\times D}$, which is itself a lattice.

\begin{theorem}
\label{thm:limit-of-finite-cumulants}  
Let $D\ge3$ and let $\bsig\in S_n^D$ be a melonic purely connected invariant, $\omega(\bsig)=0$ and   $K_\mathrm{p}(\bsig)=1$, and let $\eta\in S_n$ be the permutation defined by the canonical pairs of $\bsig$. Consider a pure random tensor $T,\bar T$ whose classical cumulants scale as in  \eqref{eq:scaling-hypothesis-pure}. Then the finite  size  free cumulant 
$\mathcal{K}_{\bsig}[T,  \bar T]$ scales as $N^{1 - nD}$ and:
 \be
 \label{eq:imit-free-finite-melonic}
\kappa_{\bsig}(t, \bar t) = \lim_{N\rightarrow \infty} N^{nD-1}\mathcal{K}_{\bsig}[T, \bar T]  = \sum_{\substack{{\btau \in S^D_{n}}\\{\btau\eta^{-1} \preceq \bsig\eta^{-1}}}}   \varphi _{\Pi_\mathrm{p}(\btau), \btau} (t, \bar t)\,  \mathsf{M}(\bsig\btau^{-1}) \;.
\ee
In particular, the $\btau$s in the sum are melonic but not necessarily purely connected, and $\eta$ is the canonical pairing on $\btau$. 
\end{theorem}
\begin{proof}
 The proof is presented in Sec.~\ref{sec:proof-free-cumulants-pure-melonic}.  
\end{proof}

In order to gain some intuition on the $\btau$s contributing to this sum, we consider a melonic graph $\bsig$ with canonical pairing given by the identity (that is, the graph $\bsig\eta^{-1}$ in the previous theorem): due to the composition by $\eta^{-1}$, the non-labeled graphs $[\btau]_{\mathrm{p}} \in S_n^D/{\sim_{\mathrm{p}}}$  in the sum are the same for any choice of $\bsig$ in a given class  $[\bsig]_{\mathrm{p}} \in S_n^D/{\sim_{\mathrm{p}}}$, so we may make this choice.  

For any color $c$, the permutation $\sigma_c$ consists in disjoint cycles which correspond to the cycles in the graph formed alternatively of canonical pairs and edges of color $c$, and for each such cycle, one independently considers all the possible non-crossing permutations. If the cycle is a fixed-point, that is, if an edge of color $c$ connects the two vertices of a canonical pair, then this edge remains unchanged in all the $\btau\preceq\bsig$. Otherwise, one may generate all the $\btau \preceq \bsig$ by a sequence of flips of edges that belong to cycles of length greater than one, where a flip of two edges consists in exchanging the white vertices to which they are connected. See Fig.~\ref{fig:ex-preceq-melo}, where all the $[\btau]_{\mathrm{p}} \in S_n^D/{\sim_{\mathrm{p}}}$  in the sum are represented for the example of Fig.~\ref{fig:ex-melo}.

\begin{figure}[!h]
\centering
\includegraphics[scale=1]{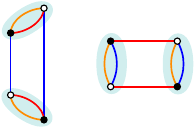}\hspace{2cm}\includegraphics[scale=1]{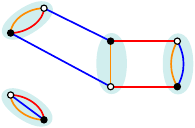}\hspace{2cm}\includegraphics[scale=1]{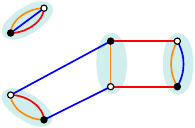}\\[+1cm]
\includegraphics[scale=1]{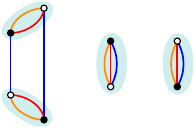}\hspace{2cm}\includegraphics[scale=1]{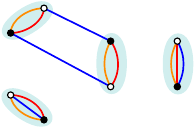}\hspace{2cm}\includegraphics[scale=1]{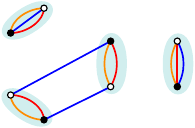}\\[+1cm]
\includegraphics[scale=1]{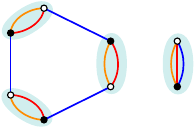}\hspace{2cm}\includegraphics[scale=1]{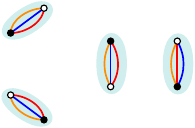}\hspace{2cm}\includegraphics[scale=1]{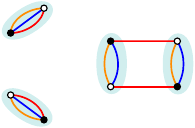}
\caption{The $\btau$s appearing in the formula defining the tensorial free cumulants, for $\bsig$ represented in Fig.~\ref{fig:ex-melo}. The blue blobs indicate the canonical pairing.  The labeling is not indicated, as the same diagrams are obtained regardless of the labeling of $\bsig$.  If the labeling is such that the canonical pairing is the identity, then these diagrams correspond to the $\btau\preceq\bsig$.
}
\label{fig:ex-preceq-melo}
\end{figure}

The theorem generalizes trivially for $n$ pure random tensors. We consider some pure random tensors $(T_a, \bar T_a)$, $(T_b , \bar T_b ) $ etc, with bipartite joint distribution, and we denote $\vec X = (X_1,\dots X_n)$, $X_s\in \{T_a, T_b \ldots\}$  
and $ \vec{X'} = (X'_{\bar 1},\dots X'_{\bar n}) $,  $ X'_{\bar s}\in \{\bar T_a, \bar T_b \ldots\}$. For $\btau\in S_n^D$ and $\Pi\ge \Pi_{\mathrm{p}}(\btau)$, we define:
\be
\label{eq:multi-tensors-pure-scaling}
\Phi_ {\Pi, \btau} [\vec X, \vec{ X'}]=  \prod_{ G = B\cup \bar B \in \Pi} k_{ K_{\mathrm{p}}(\btau_{|_B} ) }\left( \Tr_{\btau_{|_B} } 
\Bigl( \{ X_s \}_{s\in B}, \{ X'_{\bar s} \}_{\bar s\in \bar B}  \Bigr) \right) \;,
\ee
where $B$ collects the labels of the white vertices, and $\bar B = \tau(B)$ the ones of the black vertices. We denote by $\mathcal{K}_\bsig [\vec X, \vec {X'}]$ the appropriate generalization of the pure cumulants in \eqref{eq:cum-finN-int-pur} to $n$ pairs of tensors, and assuming that the generalized scaling ansatz:
\be
\label{eq:scaling-hypothesis-pure-multitensors}
\lim_{N\rightarrow \infty}\frac 1 {N^{r(\bsig)}}\ \Phi_{ \bsig}  [\vec X, \vec{ X'} ]  =  \varphi_\bsig(\vec x, \vec{x'}) \;,
\ee
holds\footnote{\label{footnote:gaussian-mixed}This assumption holds for classical random matrix ensembles, and for instance for pure Gaussian tensors. In that case, the scaling is saturated if among the Wick pairings $\eta$ contributing at first order, there exists some such that for all $s$, $X_s$ and $X'_{\overline{ \eta(s) } }$  are all of the same type $T_1, \bar T_1$ or $T_2, \bar T_2$, and if not the moment is suppressed in scaling ($\varphi_\bsig(\vec x, \vec{x'})$ vanishes).}, then Thm.~\ref{thm:limit-of-finite-cumulants} goes through, thus defining $\kappa_{\bsig}(\vec x, \vec{x'})$. If  for all $n\ge 1$ and $s\in \{1, \ldots, n\}$, $X_s$ and  $X'_s$ are drawn independently from a set of complex pairs (or a bipartite distribution), then $\varphi_\bsig(\vec x, \vec{x'})$ define the joint distribution of the pairs. 

Note that using pure cumulants with pairs  $(T_a,\bar T_a)$, as we do here, or mixed cumulants with substitutions $A_a = T_a\otimes \bar T_a$  in these scaling assumptions, leads to different objects and scaling assumptions. Indeed, in the second case, one considers only trace-invariants in which the tensors  $X_s$ and $X'_{\bar s}$ belong to the same pair: if $X_s=T_a$ for instance, then one must have $X_{\bar s}'=\bar T_a$. As an example, consider the invariant $\btau= \mathbf{id}{_1} \cup \mathbf{id}{_1}$ consisting in the disjoint union 
of two two-vertices invariants, and two pairs $(T_a,\bar T_a)$ and $(T_b , \bar T_b)$ then:
\begin{itemize}
 \item with pure invariants (as we chose to do here), one makes scaling assumptions on the four combinations:
 \[ \Tr_{\mathbf{id}{_1} } [T_a,\bar T_a] \;, \qquad 
    \Tr_{\mathbf{id}{_1} } [T_a,\bar T_b] \;, \qquad 
    \Tr_{\mathbf{id}{_1} } [T_b,\bar T_a] \; , \qquad
    \Tr_{\mathbf{id}{_1} } [T_b,\bar T_b] \; 
\]
 \item with mixed invariants, but substituting 
 $A = T\otimes \bar T$, one only obtains the two combinations:
 \[
  \Tr_{\mathbf{id}{_1} } [T_a \otimes \bar T_a] =
   \Tr_{\mathbf{id}{_1} } [T_b,\bar T_b] \;,\qquad
    \Tr_{\mathbf{id}{_1} } [T_b \otimes \bar T_b] =
    \Tr_{\mathbf{id}{_1} } [T_b,\bar T_b] \;.
 \]
\end{itemize}

\paragraph{Properties of first order free cumulants.} Free-cumulants for arbitrary melonic invariants are defined as multiplicative extensions.  The following theorem (proven in Sec.~\ref{sec:proof-free-cumulants-pure-melonic-multinv}) is  analogous to Prop.~11.4 of \cite{NicaSpeicher}. Since $\varphi _{\Pi_\mathrm{p}(\btau), \btau} (t, \bar t)=\varphi _{\Pi_\mathrm{p}(\btau\eta^{-1}), \btau\eta^{-1}} (t, \bar t)$, the sum defining $\kappa_{\bsig}(t, \bar t)$ in Thm.~\ref{thm:limit-of-finite-cumulants}  is independent on the choice of $\bsig$ in $[\bsig]_\mathrm{p}$, so that for simplicity we fix the labeling to be such that the canonical pairing is the identity: $\eta=\mathrm{id}$.

\begin{theorem}
\label{thm:free-cumulants-melonic}
Let $\bsig\in S_n^D$ be a 
melonic trace-invariant  
($\omega(\bsig)=0$), labeled to set the canonical pairing to be the identity. We denote $\bsig_1^\mathrm{p}, \ldots, \bsig_q^\mathrm{p}$ its pure connected components and we define:
\be
\label{eq:def-mult-ext-free-cumulants-melonic}
\kappa_{\Pi_\mathrm{p}(\bsig), \bsig} (\vec x, \vec {x'}) = \sum_{\btau \preceq \bsig} \varphi_{\Pi_\mathrm{p}(\btau), \btau}(\vec x,\vec {x'} ) \,\mathsf{M}(\bsig\btau^{-1}), 
\ee
where $\vec x=(x_1,\ldots x_n)$, $x_i\in \{t_a, t_b, \ldots\}$, $\vec{x'}=(x'_{\bar 1},\ldots x'_{\bar n})$, $x'_{\bar i}\in \{\bar t_a, \bar t_b, \ldots\}$, and $\varphi_{\;\Pi_\mathrm{p}(\bsig), \bsig}(\vec x,\vec {x'})  =\prod_{i=1}^q \varphi_{\bsig_i^\mathrm{p}} (\vec x, \vec{x'})$. Then:
\begin{itemize}
\item Each $\btau \preceq \bsig$ is itself melonic, with canonical pairing the identity. 

\item Prop.~\ref{prop:finite-free-prop} allows expressing $\kappa_{\bsig}$ as a rescaled limit of classical cumulants of tensor entries. 

\item The family $\kappa_{\Pi_\mathrm{p}(\bsig), \bsig} (\vec x, \vec {x'} )$ is the multiplicative extension of the family 
$\{ \kappa_\bsig (\vec x, \vec{ x'})  \, | \, K_\mathrm{p}(\bsig)=1, \omega(\bsig)=0 
\} $ defined in Thm.~\ref{thm:limit-of-finite-cumulants} :
\be
\label{eq:multiplicativty-melonic-free-cumulants}
\kappa_{\Pi_\mathrm{p}(\bsig), \bsig} (\vec x, \vec {x'}) = \prod_{i=1}^q \kappa_{\bsig_i^\mathrm{p} } (\vec x, \vec {x'})\;.
\ee
\item If $\bsig$ is purely connected and melonic and if $(T_1, \bar T_1)$ and $(T_2, \bar T_2)$ are independent and scale as   \eqref{eq:scaling-hypothesis-pure-multitensors}, then $(T_1+T_2, \bar T_1+\bar T_2)$ scales as \eqref{eq:scaling-hypothesis-pure}, and: 
\be
\kappa_\bsig(t_1+t_2, \bar t_1+\bar t_2)= \kappa_\bsig(t_1, \bar t_1)+\kappa_\bsig(t_2, \bar t_2)\;.
\ee
\item The defining relation \eqref{eq:def-mult-ext-free-cumulants-melonic} can be inverted, so that for any melonic $\bsig$ with canonical pairing the identity: 
\be
\label{eq:inversion-free-cum-melonic}
\varphi_{\Pi_\mathrm{p}(\bsig), \bsig} (\vec x, \vec{x'})= \sum_{\btau \preceq \bsig} \kappa_{\Pi_\mathrm{p}(\btau), \btau} (\vec x, \vec{x'}), 
\ee
so that the data 
$\{\kappa_{\bsig} (\vec x, \vec {x'})  \, | \, K_\mathrm{p}(\bsig)=1, \omega(\bsig)=0 \} $ is equivalent to 
$\{\varphi_{\bsig} (\vec x, \vec {x'})  \, | \, K_\mathrm{p}(\bsig)=1, \omega(\bsig)=0\} $. This 
together with \eqref{eq:multiplicativty-melonic-free-cumulants} can be taken as an alternative definition of melonic free-cumulants.
\end{itemize}
\end{theorem}

 \

For a Gaussian complex tensor $T,\bar T$, we have $\kappa_{\bsig} (t, \bar t)= C\; \delta_{\bsig, {\mathbf{id}{_1}} }$ and, with the notations of the theorem,  the only non-zero term in the sum \eqref{eq:inversion-free-cum-melonic} is $\btau=\mathbf{id}{_n}=(\mathrm{id}, \ldots, \mathrm{id})\in S_n^D$ (the melonic invariant consisting of $n$ disjoint two-vertex graphs whose canonical pairing is the identity), that is:
\be
\label{eq:trivial-moment-cumulant-gaussian}
\varphi_{\Pi_\mathrm{p}(\bsig), \bsig} (t, \bar t)=  \kappa_{\Pi_\mathrm{p}(\mathbf{id}{_n}), \mathbf{id}{_n}}(t,\bar t) \; , 
\ee
and one recovers \eqref{eq:asympt-gaussian-melonic-bis}. 

For $D=1$, the free cumulants of the square Wishart random matrix \eqref{eq:cumulants-wishart-D1} are $\kappa_n(w)=1$, and those of the GUE are $\kappa_n(m)=\delta_{n,2}$. Both lead to similar combinatorics, in the sense that the associated asymptotic moments enumerate graphs embedded on the 2-sphere in the same combinatorial universality classes. 
For the pure $D>1$ case treated here this is no longer the case: as explained above, choosing $\kappa_{\bsig} (t, \bar t)=1$ for purely connected melonic invariants leads to richer combinatorics \eqref{eq:inversion-free-cum-melonic} in comparison to the choice $\kappa_{\bsig} (t, \bar t)= \delta_{\bsig, {\mathbf{id}{_1}} }$, which leads to \eqref{eq:trivial-moment-cumulant-gaussian}.

The choice $\kappa_{\bsig} (t, \bar t)=1$ is quite natural and generalizes the square Wishart distribution, albeit in a different sense from what we will discuss below in Sec.~\ref{sub:mixed-like-pure-gaussian}. If $\bsig$ is purely connected and melonic, with canonical pairing the identity, and if for each $c$, the partition $\Lambda(\sigma_c)\vdash n$ has parts $\lambda_c^1,\ldots, \lambda_c^{k_c}$, with $k_c=\#(\sigma_c)$: 
\be
\label{eq:asymptotic-moments-cumulants-are-one}
\varphi_{\bsig}(t,\bar t) = \prod_{c=1}^D \prod_{i=1}^{k_c} C_{\lambda_{c,i}},
\ee
to be compared with the Catalan number for $D=1$ \eqref{eq:asympt-gaussian-melonic-bis}.

\subsection{A first glimpse at purely connected higher orders}
\label{sec:free-cumulants-pure-compatible}

Non-melonic purely connected invariants are higher-order free cumulants as they are more suppressed than their melonic counterparts in the  limit $N\rightarrow \infty$. This is new with respect to matrices, for which the order of dominance is given only by the number of pure connected components.

In the proof of Thm.~\ref{thm:limit-of-finite-cumulants}  
we have obtained (see  \ref{eq:cum-finN-pure-conn-phi-first-asympt} and subsequent) the general formula 
for a finite $N$ free cumulant associated to a purely connected invariant:
\begin{align}
\nonumber
\mathcal{K}_{\bsig}[T,  \bar T]  = \frac1{N^{nD}}\sum_{\btau \in S^D_{n}}  \sum_{ \pi\ge \Pi_\mathrm{p}(\btau)}  \left [ \varphi _{\pi, \btau} (t, \bar t)\,  \mathsf{M}(\bsig\btau^{-1})  N^{r(\bsig) - \Delta_\pi(\bsig; \btau)} +o\left(N^{r(\bsig) - \Delta_\pi(\bsig; \btau)}\right) \right] \;, 
\end{align}
where $\Delta_\pi(\bsig;\btau) \ge 0 $. Furthermore, we have shown that $ \Delta_\pi(\bsig;\btau) = 0$ if and only if:
\be
\label{eq:general-formula-asymptotics}
\begin{split}
& d(\bsig, \btau) + \min_{\eta \in H_{\btau,\pi}} d(\btau  , \eta ) =\min_{\eta_0\in S_n} d(\bsig, \eta_0) \; , \crcr
& H_{\btau,\pi}=\bigl\{\eta \in S_n \mid \Pi_\mathrm{p}(\eta) \le \pi \quad \mathrm{and}\quad \forall B\in \pi,\ K_\mathrm{p}(\btau_{\lvert_B}, \eta_{\lvert_B})=1\bigr\} \; ,
\end{split}
\ee 
and that $ \Delta_{1_n}(\bsig;\bsig) = 0 $, that is the condition is always satisfied at least for $\btau=\bsig$ and $\pi=1_n$ (the one-block partition). We denote the set of terms that dominate by:
\be
\label{eq:bfSp}
\bfS(\bsig) = \Bigl\{(\pi, \btau) \ \bigl\lvert \  \pi \ge \Pi_\mathrm{p}(\btau)\ \  \mathrm{and}\ \ \Delta_\pi(\bsig,\btau) =0 \Bigr\}.
\ee

\begin{proposition}
\label{prop:limit-of-finite-cumulants-general}  
Let $\bsig\in S_n^D$, $D\ge 1$, be purely connected. Consider a pure random tensor whose cumulants scale as \eqref{eq:scaling-hypothesis-pure}. Then the finite  size  free cumulant  $\mathcal{K}_{\bsig}[T, \bar T]$ scales as $N^{r(\bsig) - nD}$, where $r(\bsig)=n-\min_{\eta\in S_n} d(\bsig, \eta)$, and:
 \be
\lim_{N\rightarrow \infty} N^{nD-r(\bsig)}\mathcal{K}_{\bsig}[T, \bar T]  = \sum_{(\pi, \btau)\in \bfS(\bsig)}   \varphi _{\pi, \btau} (t, \bar t)\,  \mathsf{M}(\bsig\btau^{-1} ) \;.
\ee
\end{proposition}

We note that if there exists a $\tilde \eta\in H_{\btau,\pi}$ for which both $d(\btau  ; \tilde \eta ) = \min_{\eta \in H_{\btau,\pi}} d(\btau  ; \eta )$ and $d(\bsig  ; \tilde \eta ) =\min_{\eta_0\in S_n} d(\bsig, \eta_0)$,
that is $
d(\bsig, \btau) + d(\btau  ,\tilde \eta ) =d(\bsig, \tilde \eta)$, then
$\btau \tilde \eta^{-1}\preceq\, \bsig \tilde \eta^{-1}$, and in this case the cumulant involves a sum in a lattice and the relation can be inverted using M\"{o}bius inversion. However, at this point, it is not clear if such a permutation always exist, or how to organize $\bfS(\bsig)$ as a lattice in general.

The condition \eqref{eq:general-formula-asymptotics} can be written 
in a form adapted to compatible invariants discussed in Section~\ref{sub:Melo}. We first reformulate the \eqref{eq:two-perm-to-one-perm} as:
\be
\label{eq:nabla-to-faces}
d(\btau, \eta) = \frac 1 {D-1}\Bigl( \nabla(\btau ; \eta) + \sum_{1\le c_1<c_2\le D} \lvert \tau_{c_1}\tau_{c_2}^{-1}\rvert\Bigr) \;. 
\ee
On the other hand, we define:
\be
\nabla^{(2)}(\bsig ; \btau) = \sum_{1\le c_1<c_2\le D} \bigl( \lvert \sigma_{c_1}\tau_{c_1}^{-1}\rvert + \lvert \tau_{c_1}\tau_{c_2}^{-1}\rvert+ \lvert \tau_{c_2}\sigma_{c_2}^{-1}\rvert - \lvert \sigma_{c_1}\sigma_{c_2}^{-1}\rvert \bigr)\ge 0 \; ,
\ee
and with the help of these two relations we express:\footnote{We used $d (\bsig, \btau)  + \frac 1 {D-1} \sum_{1\le c_1<c_2\le D} \lvert \tau_{c_1}\tau_{c_2}^{-1}\rvert  = \frac 1 {D-1}\Bigl(\nabla^{(2)}(\bsig ; \btau) + \sum_{1\le c_1<c_2\le D}  \lvert \sigma_{c_1}\sigma_{c_2}^{-1}\rvert \Bigr)$.}
\be
\nabla^{(2)}(\bsig; \btau) + \nabla(\btau ; \eta) = \nabla(\bsig ; \eta) + (D-1)\bigl[ d(\bsig, \btau) + d(\btau, \eta) - d(\bsig, \eta)\bigr] \;. 
\ee

We conclude that the set $\bfS(\bsig)$ 
is also the set of pairs $(\pi, \btau )$ such that $\pi \ge \Pi_\mathrm{p}(\bsig)$ and:
\be
 \nabla^{(2)}(\bsig; \btau) +\min_{\eta\in H_{\btau,\pi}} \nabla(\btau   ; \eta) =  \min_{\eta_0\in S_n} \nabla(\bsig ; \eta_0). 
\ee
Assuming that $\bsig$ is compatible,  this condition becomes 
$ \nabla^{(2)}(\bsig; \btau) = \min_{\eta\in H_{\btau,\pi}} \nabla(\btau   ; \eta) = 0$.
On the other hand, we also have that 
if $\nabla^{(2)}(\bsig; \btau) =0$, then  $\eta$ renders $\btau$ compatible if and only if it renders $\bsig$ compatible and moreover it is such that $\btau\eta^{-1}\preceq\bsig\eta^{-1}$. 

This allows defining equivalently $\bfS(\bsig)$ as the set of pairs $(\pi, \btau)$ such that $\pi \ge \Pi_\mathrm{p}(\bsig)$ and there exists $\eta\in H_{\btau, \pi}$ which renders $\bsig$ compatible and is such that $\btau\eta^{-1}\preceq\bsig\eta^{-1}$. At this point, in order  to proceed, we would need to determine whether it is possible for a $\pi >\Pi_\mathrm{p}(\btau)$ to contribute to the dominant contribution. We posit that this is not the case. 

\begin{conjecture}
\label{conj:compatible-and-eta}
Consider $\btau\in S_n^D$. Then any $\eta\in S_n$ satisfying $\nabla(\btau, \eta)=0$ is such that $\Pi_\mathrm{p}(\eta)\le\Pi_\mathrm{p}(\btau)$, that is, $\eta\in H_{\btau, \Pi_\mathrm{p}(\btau)}$. 

It is straightforward to prove that if the Gaussian scaling is subadditive then this conjecture holds, that is Conjecture~\ref{conj:subadi} implies the present conjecture.
\end{conjecture}

Assuming this conjecture,\footnote{And assuming that if several permutations $\eta$ render $\bsig$ compatible then the families of invariants such that $\btau\eta^{-1}\preceq \bsig\eta^{-1}$ are disjoint, which is a weak assumption view that one expects usually at most one such $\eta$ to exist.}
then the cumulants of purely connected invariants in Proposition~\ref{prop:limit-of-finite-cumulants-general} become for $\bsig$ compatible: 
\be
\lim_{N\rightarrow \infty} N^{nD-r(\bsig)}\mathcal{K}_{\bsig}[T, \bar T]  = \sum_{\substack{{\eta \in S_{n}}\\{\nabla (\bsig, \eta)=0}}}  \  \sum_{\substack{{\btau \in S^D_{n}}\\{\btau \eta^{-1} \preceq \bsig \eta^{-1}}}}    \varphi _{\Pi_\mathrm{p}(\btau), \btau} (t, \bar t)\,  \mathsf{M}(\bsig\btau^{-1})  \; . 
\ee
This is a sum which can be inverted.  In particular, if there is a unique $\eta$  such that $\nabla(\bsig, \eta)=0$ then, choosing the labeling of $\bsig$ so that this $\eta$ is the identity, we have:
\be
\lim_{N\rightarrow \infty} N^{nD-r(\bsig)}\mathcal{K}_{\bsig}[T, \bar T]  =  \sum_{\btau  \preceq \bsig }   \varphi _{\Pi_\mathrm{p}(\btau), \btau} (t, \bar t)\,  \mathsf{M}(\bsig\btau^{-1}) .
\ee

Using the methods of \cite{Gurau-Schaeffer} one should be able to build infinite families of compatible trace-invariants $\bsig$ of fixed degree $\omega$ and which each have a single $\eta$ for which $\nabla(\bsig, \eta)=0$.  We can conjecturally generalize Thm.~\ref{thm:free-cumulants-melonic} to this subclass of compatible invariants. 

\begin{theorem}
\label{thm:generalization-to-compatible}
If Conj.~\ref{conj:compatible-and-eta} holds, then the finite size  free cumulants of a purely connected compatible invariant $\bsig$ with a unique $\eta\in S_n$ satisfying $\nabla(\bsig, \eta)=0$
scale like $N^{1 - \frac {\omega(\bsig)}{D-1} - nD}$ and:
 \be
 \label{eq:imit-free-finite-compatible}
\kappa_{\bsig}(t, \bar t):= \lim_{N\rightarrow \infty} N^{nD-1+ \frac {\omega(\bsig)}{D-1}}\mathcal{K}_{\bsig}[T, \bar T]  = \sum_{\substack{{\btau \in S^D_{n}}\\{\btau\eta^{-1} \preceq \bsig\eta^{-1}}}}   \varphi _{\Pi_\mathrm{p}(\btau), \btau} (t, \bar t)\,  \mathsf{M}(\bsig\btau^{-1}). 
\ee
Furthermore, Thm.~\ref{thm:free-cumulants-melonic} can be generalized immediately to this case. 
\end{theorem}

\subsection{Mixed perspective on the pure case}

At the end of Sec.~\ref{subsub:cumulants-for-finite-N}, we discussed the difference at the level of finite $N$ free cumulants between pure random tensors $T, \bar T$ and the mixed quantity $T\otimes \bar T$: for any $\btau$,
$\mathbb{E}\bigl[\Tr_\btau(T, \bar T)\bigr]= \mathbb{E}\bigl[\Tr_\btau(T\otimes  \bar T)\bigr]$  so that for $\bsig$ purely connected,   $\mathcal{K}_\bsig[T, \bar T]= \mathcal{K}^\mathrm{m}_\bsig[T\otimes  \bar T]$ \eqref{eq:cum-finN-int-pure-conn}, and if in addition $\bsig$ is melonic, $N^{nD - 1}\mathcal{K}^\mathrm{m}_\bsig[T\otimes  \bar T] $ converges towards $\kappa^\mathrm{m}_\bsig(t\otimes \bar t) = \kappa_\bsig(t, \bar t)$.
 For $\bsig$ connected ($K_\mathrm{m}(\bsig)=1$) but not purely connected ($K_\mathrm{p}(\bsig)>1$), one generally has $\mathcal{K}_\bsig[T, \bar T]\neq \mathcal{K}^\mathrm{m}_\bsig[T\otimes  \bar T]$, as in the mixed case one still has:
\be
 \mathcal{K}^\mathrm{m}_\bsig[T\otimes  \bar T] = \sum_{\btau \in S^D_{n}} \bE\left[\Tr_{\btau} (T,\bar T)\right]\prod_{c=1}^D W^{(N)} ( \sigma_{ c} \tau_{c}^{-1})  \; ,
\ee
while $\mathcal{K}_\bsig[T, \bar T]$ involves an additional sum over partitions \eqref{eq:cum-finN-int-pur}. In fact  (the proof is in Sec.~\ref{sec:proof-facto-mixed-on-pure}): 

\begin{lemma}
\label{lem:facto-mixed-on-pure}Let $D\ge 3$, $\bsig\in S_n^D$ be melonic, connected but not purely connected. Let $\eta$ be the canonical pairing of $\bsig$. Consider $T,\bar T$ satisfying \eqref{eq:scaling-hypothesis-pure}. Then:
$$
\kappa_\bsig^\mathrm{m}(t\otimes \bar t) = \lim_{N\rightarrow \infty} N^{nD - K_\mathrm{p}(\bsig)}\mathcal{K}^\mathrm{m}_\bsig[T\otimes  \bar T] = \kappa_{\Pi_\mathrm{p}(\bsig), \bsig}(t, \bar t)
 $$
 \end{lemma}

Considering now $T_1,  \bar T_1$ and $T_2, \bar T_2$ two independent \textsf{LU}-invariant pure random tensors, in general   $\mathbb{E}\bigl[\Tr_\btau(T_1+T_2, \bar T_1+\bar T_2)\bigr] \neq \mathbb{E}\bigl[\Tr_\btau(T_1\otimes  \bar T_1 + T_2\otimes  \bar T_2)\bigr]$ (see the discussion below \eqref{eq:scaling-hypothesis-pure-multitensors}) and therefore $\mathcal{K}_\bsig[T_1+T_2, \bar T_1+\bar T_2] \neq \mathcal{K}^\mathrm{m}_\bsig[T_1\otimes  \bar T_1 + T_2\otimes  \bar T_2]$, even for purely connected $\bsig$. However, from Prop.~\ref{prop:large-N-linear-gen} and 
Thm.~\ref{thm:finite-free-cumulants}, one has equality for $N$ sufficiently large and purely connected $\bsig$: 
\be
\label{eq:asympt-equality-cumulants-mixed-on-pure}
\mathcal{K}^\mathrm{m}_\bsig[T_1\otimes  \bar T_1 + T_2\otimes  \bar T_2] = \mathcal{K}^\mathrm{m}_\bsig[T_1\otimes  \bar T_1] +   \mathcal{K}^\mathrm{m}_\bsig[T_2\otimes  \bar T_2] = \mathcal{K}_\bsig[T_1,  \bar T_1] +   \mathcal{K}_\bsig[T_2,  \bar T_2] = \mathcal{K}_\bsig[T_1+T_2, \bar T_1+\bar T_2],
\ee
so that their rescaled limits coincide. 

\begin{theorem}
\label{thm:sum-and-mixed-on-pure}
Let $D\ge3$, $\bsig\in S_n^D$ be melonic $\omega(\bsig)=0$, purely connected $K_\mathrm{p}(\bsig)=1$, and labeled such that its canonical pairing is the identity.
Let $T_1,  \bar T_1$ and $T_2, \bar T_2$ be two \emph{independent} \textsf{LU}-invariant pure random tensors satisfying \eqref{eq:scaling-hypothesis-pure-multitensors},
and define the mixed tensor $A=T_1\otimes  \bar T_1 + T_2\otimes  \bar T_2$. Then, 
\begin{align}
\label{eq:first-momcum-in-mixed-on-pure}
 \kappa^\mathrm{m}_{\bsig}(a)  &= \lim_{N\to \infty} N^{nD-1}  \mathcal{K}^{\mathrm{m}}_\bsig[A]
  = \sum_{\substack{{\btau \in S^D_{n}}\\{\btau \preceq \bsig } } }   \varphi ^\mathrm{m}_{\Pi_\mathrm{p}(\btau), \btau} (a)\,  \mathsf{M}(\bsig\btau^{-1}) \;,\\
    \varphi^\mathrm{m}_{\bsig}(a) &= \sum_{\substack{{\btau \in S^D_{n}}\\{\btau \preceq \bsig } } }   \kappa^\mathrm{m}_{\Pi_\mathrm{p}(\btau), \btau}(a) \;,
    \label{eq:second-momcum-in-mixed-on-pure}
\end{align}
where for $\bnu$ purely connected and melonic, $\varphi^\mathrm{m}_\bnu(a)=\lim_{\rightarrow \infty} \Phi^\mathrm{m}_\bnu[A]/N<\infty$. 
In particular, $\kappa^\mathrm{m}_{\bsig}(a) = \kappa_\bsig(t_1+t_2, \bar t_1+\bar t_2)$ and $\varphi^\mathrm{m}_{\bsig}(a) = \varphi_\bsig(t_1+t_2, \bar t_1+\bar t_2)$. As a consequence, 
\be
\label{eq:non-canonical-vanish}
\sum_{{{\vec x, \vec{x'}}\textrm{ s.t. } {\vec{x'}\neq \vec{\bar x}}} } \varphi_{\bsig}(\vec x, \vec{x'})=0, 
\ee
where $\vec x=(x_1,\ldots x_n)$, $x_i\in \{t_1, t_2\}$, $\vec{x'}=(x'_{\bar 1},\ldots x'_{\bar n})$, $x'_{\bar i}\in \{\bar t_1, \bar t_2\}$, such that not $x'_{\bar i} \neq \overline{x_{i}}$ for at least one canonical pair $1\le i \le n$.
\end{theorem}
\begin{proof}
The theorem is proven in Sec.~\ref{sec:sum-and-mixed-on-pure}. \end{proof}

 A consequence of \eqref{eq:non-canonical-vanish} is that if the quantities in the sum are non-negative they must vanish, so that in that case the corresponding $\Phi_\bsig[\vec X, \vec{X'}]$ all have a scale lower than $N^{r(\bsig)}$ in \eqref{eq:scaling-hypothesis-pure-multitensors}. This generalizes a  known result for the Gaussian case, see Footnote \ref{footnote:gaussian-mixed}. The condition $ \varphi_{\bsig}(\vec x, \vec{x'})=0$ for $\vec{x'}\neq \vec{\bar x}$ is one of the conditions for tensor freeness, see Thm.~\ref{thm:equiv-tensor-freeness-cumulants-moments-pure} in Sec.~\ref{sec:tensor-freeness}.

\subsection{Mixed random tensors that scale like a Wishart tensor}
\label{sub:mixed-like-pure-gaussian}

In this section, we will discuss the case when our random tensor is mixed rather than pure. As scaling assumption, we will consider mixed random tensors $A$ that scales like the square Wishart random tensor of Section~\ref{sub:true-Wishart-tensor}, that is we assume that for $\bsig \in S_n^D$:
\be
\label{eq:scaling-wishart-D1}
\lim_{N\rightarrow \infty} \frac 1 { N^{r_W(\bsig)}}\,\Phi^\mathrm{m}_\bsig[A] =\varphi^\mathrm{m}_\bsig(a), \hspace{1.5cm} r_W(\bsig)= n - \min d\bigl( (\bsig  , \mathrm{id}), \eta\bigr),
\ee
where the minimum is taken over $\eta\in S_n$ for which $K_\mathrm{m}(\bsig, \eta)=1$. Note that $d( ( \bsig , \mathrm{id} ) , \eta) = d(\bsig, \eta)  + \lvert \eta \rvert$ and also $r_W(\bsig)=r(\bsig , \mathrm{id})$ with $r$ the Gaussian scaling function in \eqref{eq:def-scling-purely-connected}. 

For the version of the theorem that involves different tensors, one must make the following stronger assumption for any $\vec A=(A_1, \ldots, A_n)$:
\be
\label{eq:scaling-wishart-multitens}
\lim_{N\rightarrow \infty} \frac 1 { N^{r_W(\bsig)}}\,\Phi^\mathrm{m}_\bsig[{\vec A}\,] =\varphi^\mathrm{m}_\bsig(\vec a), 
\ee

This includes the case where there exists a \textsf{LU}-invariant pure random tensor $T,\bar T$ with $D+1$ indices, not necessarily the pure  Gaussian itself but displaying Gaussian scaling \eqref{eq:scaling-hypothesis-pure}, and $A$ is: 
\be
\label{eq:wishart-perspective-on-pure}
A_{i^1 \ldots i^D ; j^1 \ldots j^D} = \sum_{k=1}^N T_{i^1 \ldots i^D \, k}\bar T_{j^1 \ldots j^D \, k}. 
\ee
For this case the results discussed here are to be compared to the pure case with one additional color (see Sec.~\ref{sub:true-Wishart-tensor}), and we review this 
comparison at the end of this subsection.

For the general case of a mixed tensor $A$  with scaling function $r_W$, the thick edges cannot a priori be seen simply as an additional color. In order to understand this, 
let us first study the free moment-cumulant relations for the first order moments $\Phi^\mathrm{m}_\bsig[A]$, corresponding to a melonic connected invariant $\bsig$ such that $(\bsig, \mathrm{id})$ is melonic (see Sec.~\ref{sub:true-Wishart-tensor}).

\begin{theorem}
\label{thm:limit-for-wishart-tensor-first-order}
Let $D\ge 2$, $\bsig\in S_n^D$ be connected and such that $\omega( \bsig , \mathrm{id} )=0$ (which implies that $\bsig$ is melonic) and let $\eta\in S_n$ be the canonical pairing of $(\bsig,\mathrm{id})$. Consider some mixed random tensors $A_1,A_2,\ldots$ that scale as in \eqref{eq:scaling-wishart-D1}, and satisfy \eqref{eq:scaling-wishart-multitens}. Then the limit of the mixed finite  size free cumulant associated to $\bsig$ is:
 \be
\kappa_{\bsig}^\mathrm{m}(\vec a):= \lim_{N\rightarrow \infty} N^{nD-1}\mathcal{K}^\mathrm{m}_{\bsig}[\vec A]  = \sum_{\substack{
{\btau\eta^{-1} \preceq \bsig\eta^{-1}}}}   \varphi^\mathrm{m} _{\Pi(\btau), \btau } (\vec a)\,  \mathsf{M}(\bsig\btau^{-1}),
\ee
and each $\btau$ in the sum is such that  $( \btau,\mathrm{id})$ is melonic with canonical pairing $\eta$. Furthermore:
\begin{itemize}
\item Prop.~\ref{prop:finite-free-prop} allows expressing $\kappa^\mathrm{m}_{\bsig}$ as a rescaled limit of classical cumulants of tensor entries. 
\item If $A_1$ and $A_2$ are independent and scale as in \eqref{eq:scaling-wishart-D1}, then $A=A_1+A_2$ scales as in  \eqref{eq:scaling-wishart-D1}, and $\kappa^\mathrm{m}_\bsig(a)= \kappa^\mathrm{m}_\bsig(a_1)+\kappa^\mathrm{m}_\bsig(a_2)$.
\item The relation can be inverted. With the notations above:
\be
\label{eq:free-com-cum-gen-mixed}
\varphi^\mathrm{m}_{\bsig}(\vec a) = \sum_{\substack{{\btau\eta^{-1}\preceq\bsig\eta^{-1}}}} \kappa^\mathrm{m}_{\Pi\left(\btau\right), \btau }(\vec a). 
\ee
\end{itemize}
\end{theorem}

\begin{proof}
The proof can be found Section~\ref{subsub:free-cumulants-wishart}.
\end{proof}

In the general mixed case one cannot treat the thick edges (the identity) as an additional color because the inverse relation \eqref{eq:free-com-cum-gen-mixed} does not include a sum over permutations $\tau_{D+1}$ such that $\tau_{D+1}\eta^{-1}\preceq \eta^{-1}$, as one would have in the pure case with one additional color. See  Fig.~\ref{fig:ex-preceq-melo-mixed}

\begin{figure}[!h]
\centering
\includegraphics[scale=1]{ex-melo-nonpure.pdf}\hspace{1.5cm}\includegraphics[scale=1]{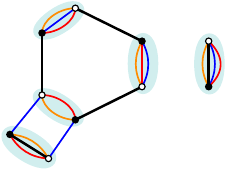}\\
\includegraphics[scale=1]{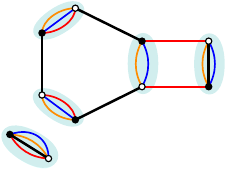}\hspace{1.5cm}\includegraphics[scale=1]{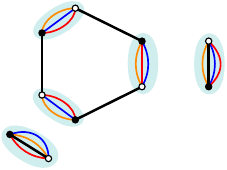}
\caption{A $\bsig$ which is connected but not purely connected and such that $(\bsig, \mathrm{id})$ is melonic (upper left), with canonical pairing given by the blue blobs, and all the  $\btau$ such that $\btau\eta^{-1}\preceq\bsig\eta^{-1}$. The labeling is not indicated, as the same diagrams are obtained regardless of the labeling of $\bsig$. 
}
\label{fig:ex-preceq-melo-mixed}
\end{figure}

The inverse relation \eqref{eq:free-com-cum-gen-mixed} simplifies if the connected components of $\bsig$ are purely connected, as then the canonical pairing on $\bsig$ is $\eta=\mathrm{id}$, to:
\be
\label{eq:free-com-cum-gen-mixed-purelyconn}
\varphi^\mathrm{m}_{\Pi(\bsig), \bsig }(a) = \sum_{\substack{
{ \btau\preceq\bsig}}} \kappa^\mathrm{m}_{\Pi\left(\btau\right), \btau }(a) \;,
\ee
that is one recovers the formula from the pure 
case with $D$ indices with canonical pairing $\eta=\mathrm{id}$  \eqref{eq:inversion-free-cum-melonic}, since for any $\btau \preceq \bsig$, $\Pi(\btau)$ and $\Pi_\mathrm{p}(\btau)$ coincide. Thus for the purely-connected melonic invariants, which are first order in both situations, nothing distinguishes a pure distribution with $D$ indices with Gaussian scaling function $r(\bsig)$ from a mixed distribution with scaling function $r_W(\bsig)=r(\bsig,\mathrm{id})$. The difference between the two comes from the fact that there are more first order invariants in the mixed case, namely those which are connected but not purely connected (Fig.~\ref{fig:ex-invariants-purs-mixed})

Observe that there are invariants for which 
$\btau\eta^{-1}\preceq \bsig\eta^{-1}$ forces $\btau=\bsig$. This is for instance the case for $\bsig$ melonic with all $\sigma_c$ equal. For such invariants the asymptotic moment equals the free cumulant, and it is additive. This phenomenon is specific to the tensor case and does not arise for matrices. 

If $\kappa^\mathrm{m}_\btau(a)=1$ for the first order $\btau$ (connected and $( \btau,\mathrm{id} )$ melonic) and we pick a connected melonic invariant $\bsig\in S_n^D$ with $\eta=\mathrm{id}$ (in which case $\bsig$ is purely connected) then $\varphi^\mathrm{m}_{\bsig}(a)$ is given by the same product of Catalan numbers as in \eqref{eq:asymptotic-moments-cumulants-are-one}. 

\paragraph{Higher orders.}Prop.~\ref{prop:limit-of-finite-cumulants-general} which deals with the higher order in the pure case also generalizes with obvious modifications to the mixed one. In particular, for connected $\bsig$, $\mathcal{K}^\mathrm{m}_\bsig[A]$ scales as:
\be
\label{eq:scling-wishart-finite-cumulants-general}
\mathcal{K}^\mathrm{m}_\bsig[A]_{\lvert_{K_\mathrm{m}(\bsig)=1}} \asymp N^ {r_W(\bsig) - nD} \;, 
\ee
which is consistent with  \eqref{eq:rescaled-higher-moments-and-cumulants} for $D=1$.

\paragraph{The mixed and Wishart point of vue on the pure case.}For the genuine Wishart like case, we have the following.
\begin{proposition}
\label{prop:mixed-persp-on-pure}
Let $A_{i^1\ldots i^D ; j^1 \ldots j^D} = \sum_{k=1}^N T_{i^1\ldots i^D k}\bar T_{j^1 \ldots j^D k}$, $D\ge 2$, for some pure random tensor $T, \bar T$ with $D+1$ indices with Gaussian scaling \eqref{eq:scaling-hypothesis-pure}. Let $\bsig\in S_n^D$  be  such that $( \bsig,\mathrm{id} )$ is melonic  with canonical pairing $\eta$. Then the pure free cumulants of $T,\bar T$ and the mixed free cumulants of $A$ are related by:
\be
\label{eq:mixed-perspective-on-wishart-general}
\kappa_{\Pi_\mathrm{p} ( \bsig,\mathrm{id} ), (\bsig, \mathrm{id} ) }(t,\bar t)= \sum_{ \substack{
{\nu \in S_n} \\
{ \nu\eta^{-1}\preceq \eta^{-1}} } }\kappa^\mathrm{m}_{\Pi\left(\bsig \nu^{-1}\right), \bsig \nu^{-1} }( a)\; \mathsf{M}(\nu)  \; , 
\ee
where $\nu\in S_n$ encodes a change of the thick edges representing the tensors $A$.\footnote{\label{footnote:nu-new-thick}If we change variable to $\bsig'=\bsig\nu^{-1}$ corresponding to a change of labeling of the white vertices, the identity after the change of variables represents the thick edges, but before this change of labeling, the thick edges are represented by the permutation $\nu$.} 
Conversely:
\be
\label{eq:mixed-perspective-on-wishart-general-inverse}
\kappa^\mathrm{m}_{\Pi\left(\bsig\right), \bsig }(a)= \sum_{\substack{{\nu\in S_{n}}\\{ \nu \eta^{-1}\preceq \eta^{-1}}}} \kappa_{\Pi_\mathrm{p} (\bsig, \nu), (\bsig , \nu ) }(t,\bar t) \; .
\ee

If $\bsig$ is purely connected, $\eta=\mathrm{id}$ and 
$\kappa_{( \bsig,\mathrm{id} ) }(t, \bar t)=\kappa^{\mathrm{m}}_{\bsig}(a)$.
\end{proposition}
\begin{proof}
See Section~\ref{sub:proof-mixed-persp-on-pure}.
\end{proof}

\begin{figure}[!h]
\centering
\includegraphics[scale=1]{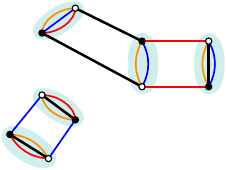}\hspace{1.5cm}\includegraphics[scale=1]{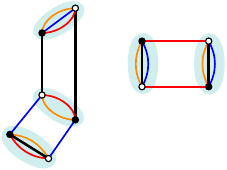}\hspace{1.5cm}
\includegraphics[scale=1]{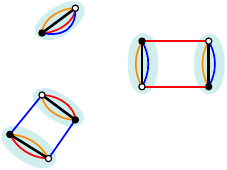}\caption{Examples of contributions to the right-hand side of \eqref{eq:mixed-perspective-on-wishart-general} for $\bsig$ as on the upper left of Fig.~\ref{fig:ex-preceq-melo-mixed}: the new permutation $\nu$ represents a modification of the thick edges (see Footnote \ref{footnote:nu-new-thick}). 
}
\label{fig:ex-preceq-melo-mixed-compar-kappa}
\end{figure}

In particular, if $T$ is a pure Gaussian with $D+1$ indices, $D\ge 2$,  and $A$ is as in \eqref{eq:wishart-perspective-on-pure}, we have from  \eqref{eq:free-cumulants-gaussian} that
 $\kappa_{\Pi_\mathrm{p} (\bsig, \nu), (\bsig , \nu ) }(t,\bar t) =    \delta_{ (\bsig , \nu), (\eta, \ldots, \eta) }$ where $\eta$ is the canonical pairing of $(\bsig, \mathrm{id})$. 
Inserting this into \eqref{eq:mixed-perspective-on-wishart-general-inverse} for $\bsig$  connected with  $( \bsig,\mathrm{id} )$ melonic, we see that the right-hand side is zero unless $\bsig= (\eta, \ldots, \eta)\in S_n^D$ and $\nu=\eta$. For $\bsig$ to be connected, $\eta$ must therefore be a cycle of length $n$. We find for  $\bsig$  connected with  $( \bsig,\mathrm{id} )$ melonic that:
\be
\label{eq:wishart-higher-D-free-cum}
\kappa^\mathrm{m}_{\bsig }(w) = \delta_{\bsig, (\eta, \ldots, \eta) },
\ee
 where $\eta$ is a cycle of length $n$,  and more generally if $\bsig$ is not necessarily connected that $\kappa^\mathrm{m}_{\Pi(\bsig),\bsig }(w) = \delta_{\bsig, (\eta, \ldots, \eta) }$, without condition on $\eta$. This is  different from the  case of a pure  complex Gaussian with two indices (a random matrix) \eqref{eq:free-cumulants-gaussian}, but it directly generalize the result for the $D=1$ Wishart matrix \eqref{eq:wishart-as-free-cumulantvs-pure-D2}.
However, for $D\ge 2$, applying \eqref{eq:free-com-cum-gen-mixed} leads for $\bsig$  connected with  $( \bsig,\mathrm{id} )$ melonic  to:
\be
\label{eq:wishart-higher-D-as-mum}
\varphi^\mathrm{m}_{\bsig }(w) =1,
\ee
as expected, since again $\btau=(\eta, \ldots, \eta)$ is the only non-vanishing term in the sum \eqref{eq:free-com-cum-gen-mixed}.

\newpage
\section{Asymptotic tensor freeness}
\label{sec:tensor-freeness}

\subsection{Matrix freeness}
\label{sub:matrix-freeness}
Let us first review the notion of freeness adapted to unitarily invariant random matrices $A,B\in \mathcal{M}_N(\mathbb{C})$ \cite{Voiculescu, Speicher94}. 
If $M=(M_1, \ldots, M_n)\in \{A,B\}^n$  and defining $\kappa_n(m_1, \ldots, m_n)$ where $m_i\in \{a,b\}$, $A,B$ are asymptotically free if and only if for all $n\ge 2$:
\be
\label{eq:matrix-freeness-cum}
\kappa_n(m_1, \ldots, m_n)=0\;,
\ee 
whenever there exist $1\le i<j\le n$ such that $m_i\neq m_j$. For matrices, one can make sense of  $a,b$ as  being non-commutative random variables, that is, elements of a non-commutative probability space defined as an algebra  $\mathcal{A}$ with a unit element $1$, together with a linear functional $\varphi:\mathcal{A}\rightarrow \mathbb{C}$, that maps the unit element to 1. The family $\varphi(a^n)$, $n\ge 1$ are the moments of $a$.  For a single variable $a$ this relates for $n\ge 1$ to $\varphi_n$ defined before  as $\varphi(a^n)=\varphi_n(a)$ where $a^n\in \mathcal{A}$. 
$\varphi$ corresponds asymptotically to ``$\frac 1 N \mathbb{E}[\Tr(\cdot)]$''.
A random matrix $A$ converges in distribution towards $a$ when $N\rightarrow \infty$ if for all $n\in \mathbb{N}^\star$, $\frac 1 N \mathbb{E}[\Tr(A^n)]\rightarrow \varphi(a^n)$, and the statement above for the \emph{asymptotic freeness} of $A,B$ directly refers to the \emph{freeness} of the non-commutative random variables $a$ and $b$. 

Freeness is equivalent to the vanishing of mixed centered moments (the original formulation \cite{Voiculescu}), that is, for any $n\ge 2$, for any $(m_1, \ldots, m_n)$ with $m_i\in \{a,b\}$ (or some finite fixed set of elements of $\mathcal{A}$), and any $g_i$ in the subalgebra generated by 1 and $m_i$: 
\be
\label{eq:freeness-centered-matrices}
\varphi\bigl(g_1 \cdots g_n  \bigr) =0, 
\ee
whenever $\varphi(g_i)=0$ for all $g_i$, and $g_1\cdots g_n$ are \emph{almost alternating}\footnote{The elements $g_1\cdots g_n$ are \emph{strictly} alternating if  $g_1$ and $g_n$ belong to different subalgebras.}, that is, $g_i$ and $g_{i+1}$ belong to two different subalgebras.  Since the variables can be centered, for two variables $a,b$, this is equivalent to requiring that for all $k\ge 1$ and all $n_1, \ldots, n_k \ge 1$, $m_1\ge 1$, if $k\ge 2$ $m_k\ge 0$, and if $k\ge 3$ $m_2\ldots, m_{k-1}\ge 1$:
\be
\varphi\Bigl(\bigl(a^{n_1} - \varphi(a^{n_1}) 1\bigr) \bigl(b^{m_1} - \varphi(b^{m_1}) 1\bigr) \cdots \bigl(a^{n_k} - \varphi(a^{n_k})1\bigr)  \bigl(b^{m_k} - \varphi(b^{m_k})1\bigr) \Bigr) =0. 
\ee

\subsection{Asymptotic tensor freeness at the level of free cumulants}
\label{sub:tensor-freeness-cumulants}

Having defined free cumulants of pure and mixed random tensors, we define asymptotic (first order) tensor freeness of a collection of random tensors as the vanishing of $\kappa_{\bsig}$ involving different elements\footnote{Called mixed cumulants, not to be confused with  the term ``mixed'' for a tensor, which comes from the quantum information interpretation of \textsf{LU}-invariant distributions as distributions over quantum states.} for any  \emph{first order }$\bsig$. More precisely:

\begin{itemize}
\item {\bf For mixed tensors} $A,B,\ldots$ that scale like the Wishart tensor   \eqref{eq:scaling-wishart-D1} and \eqref{eq:scaling-wishart-multitens}, $D\ge 2$, the first order consists of  the \emph{connected} $\bsig\in S_n^D$, $n\ge 1$, such that $\omega(\bsig,\mathrm{id})=0$. This constrains the thick edges as follows: in the graph $\bsig$, any  cycle of alternated thick edges and canonical pairs that  involves more than one thick edge must be  separating (Fig.~\ref{fig:ex-invariants-purs-mixed}). $A,B,\ldots $ are asymptotically free if for any such $\bsig$ and any $\vec m=(m_1,\ldots, m_n)$ with $m_i\in \{a,b,\ldots\}$:
\be
\kappa^{\mathrm{m}}_{\bsig}(\vec m)=0,
\ee
whenever $m_i\neq m_j$ for some $i\neq j$. 
\item {\bf For pure tensors}, $(T_a, \bar T_a)$, $(T_b , \bar T_b)$ $\ldots$ that scale like a complex Gaussian \eqref{eq:scaling-hypothesis-pure} and \eqref{eq:scaling-hypothesis-pure-multitensors}, $D\ge 3$, the first order consists of  the \emph{purely connected and melonic }$\bsig\in S_n^D$,  $n\ge 1$. We require that $\bsig$ be labeled such that its canonical pairing $\eta$ is the identity and consider for $1\le i \le n$, $\vec X = (X_1, \dots X_n)$ and $\vec{X'} = (X'_{\bar 1},\dots X'_{\bar n})$ with $X_i\in \{T_a, T_b \ldots \}$ and $X'_{\bar s}\in \{\bar T_a, \bar T_b \ldots\}$. Note that we allow $X'_{\bar s}\neq \bar X_s$, that is, one could have $X_s=T_a$ and $X'_{\bar s}=\bar T_b$. In that case,  $(T_a, \bar T_a)$, $(T_b , \bar T_b)$ $\ldots$ are asymptotically free if:
\be
\kappa_{\bsig}(\vec x, \vec{x'})=0,
\ee
whenever $x_i\neq x_j$ or $x_{\bar i}'\neq x_{\bar j}'$, or  $\overline{x_i}\neq x_{\bar j}'$ for some $i, j$. 
\end{itemize}

In the pure case,  if there exists $\eta^{-1}$ such that for each $s$, $X_{\overline{\eta^{-1}( s)}}' = \overline{X_{ s}}$, then by changing the labeling of the white vertices one may replace $\bsig$ by $\bsig \eta$, setting $X_{\bar s}' = \overline{X_{s}}$ with now $\eta$ defining the canonical pairing. After this change of labeling, the black and white vertices can be seen as representing  the inputs and the outputs of the same tensor $X_s\otimes \overline{X_{s}}$,  and are linked by a thick edge encoded by the permutation $\mathrm{id}$. For $\vec X, \vec X'$ satisfying this, the vanishing condition for purely connected $\bsig$ with canonical pairing the identity (i.e.~with $\bar \omega(\bsig;\mathrm{id})=0$)  must be considered in both the pure and mixed cases. On the other hand, the vanishing condition in the pure case for  $\vec X, \vec X'$ which do not satisfy this condition must be considered \emph{only in the pure case}, while  the vanishing condition for $\bsig$ such that $\omega(  \bsig, \mathrm{id} )=0$, connected but not purely, must  be considered \emph{only in the mixed case}.

\subsection{Paired tensors}
\label{sub:paired-tensors}

In order to state the definition of matrix freeness at the level of moments (Sec.~\ref{sub:matrix-freeness}), it was necessary to consider elements in the subalgebras  $\mathcal{A}[a,1]$ generated by $a$ and $1$.  
We let $\mathcal{G}[A]=\{A^n\}_{ n\ge 1}$: the algebra $\mathcal{A}[A,\un]$ generated by $A$ and $\un$ can also be seen as the set of linear combinations of elements of $\{\un\}\cup \mathcal{G}[A]$. In the tensor case,  stating the moments version of tensor freeness will require  extending the notions of tensors and trace-invariants. We will introduce here some sets which will generalize the roles of $\mathcal{G}[A]$, of $\mathcal{A}[A,\un]$, and then of $\mathcal{A}[a,1]$ in Sec.~\ref{sec:limiting-spaces}.

\paragraph{Paired tensors.}A \emph{paired tensor} is a tensor with components $P_{\{i_{c,r} , j_{c,r}\}}\in\mathbb{C}$, where $1\le c\le D$ and $r\in \{0, \ldots, k_c\}$, $1\le i_{c,r} , j_{c,r} \le N$, and where $k_c\ge 0$ and  $\sum_{c=1}^D k_c\ge 1$. There may therefore be no index of color $c$ or there can be several of the same color $c$, in which case the second index~$r$ partitions the inputs and outputs of the same color in pairs. We refer to $r$ as the \emph{shade}: a pair of outputs and  inputs $(i_{c,r} , j_{c,r})$ has color $c$ and shade $r$. The inputs or outputs of the same color~$c$ carrying different shades $r$ are distinguishable. 
Tensors $A$ or $T\otimes \bar T$ considered earlier in the text are paired tensors.
If $\mathcal{D}=\sum_c k_c$, the identity  $\un^{\otimes \mathcal{D}}$ of $\mathcal{M}_N(\mathbb{C})^{\otimes \mathcal{D}}$ can be viewed as a paired tensor: 
$
{\un^{\otimes \mathcal{D}}}_{\{i_{c,r} , j_{c,r}\}} =\prod_{c=1}^D\prod_{r=1}^{k_c} \delta_{i_{c,r} , j_{c,r}}.
$

\paragraph{Trace-invariants of paired tensors.}We generalize the notion of trace-invariants for paired tensors.  Given $n$ labeled paired tensors  $P_1, \ldots, P_n$, if we represent these tensors by thick edges linking a black and a white vertex, with $k_c(i)$ distinguishable half-edges of color $c$ respectively attached to the inputs and the outputs of the tensor $P_i$ for each color $1\le c \le D$, then the trace-invariants of these paired tensors  are encoded   by colored graphs $\mathsf{g}$ with thick edges representing the tensors, and edges of color $c$ representing the index summations. A vertex may now have none or several incident edges for each color $c$. We use the notation $\Tr_{\mathsf{g}}(\vec P)$. See the example in Fig.~\ref{fig:paired-invariant}.

We cannot simply use a single permutation per color  as  before in order to encode labeled invariants, since here a tensor may now have different pairs of  indices of the same color. 
However, in the graph there are cycles which go from the output of a tensor which has subscript $(c,r_1)$, through an edge of color $c$, to the intput of a tensor which has subscript $(c,r_2)$, to the output of the same tensor to which it is paired, through an edge of color $c$, etc. If the paired tensors are labeled from 1 to $n$, we thus obtain  a set of \emph{cycles} $\{\gamma_{c,b}\}$, each of the form $\bigl([k_1, r_1]\cdots   [k_q, r_q]\bigr)$ where $1\le k_1, \ldots, k_q\le n$ are the labels of the thick edges (paired tensors) encountered in the cycle (all distinct), and $r_s$ is the shade of the pair of output and input of color $c$ of the paired tensor number $k_s$ which belongs to the   cycle under consideration. 
 The index $b$ labels the different cycles that have the same color $c$. 
 Starting from the collection of labeled paired tensors, one can reconstruct all the colored edges from this data: an output of color $c$ and shade $r$ of a given tensor $k$ appears in only one cycle $\gamma_{c,b_0}$ of the $\{\gamma_{c,b}\}$, and if $\gamma_{c,b_0}\bigl([k,r])=[k', r']$, an edge of color $c$ is added between the output of color $c$, shade $r$ of $k$ and the input of color $c$, shade $r'$ of $k'$.

\begin{figure}[!h]
\centering
\includegraphics[scale=0.65]{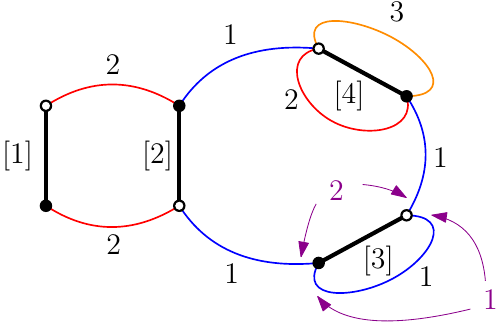}
\caption{Melonic graph of four paired tensors, represented as thick edges with labels between brackets. The color is indicated on the edges, and the shade is indicated in magenta in the only case where it is needed ($k_c>1$).
The labeled invariant is  encoded by the cycles $\gamma_{1,1}=\bigl([2,1] [3,2] [4,1]\bigr)$, $\gamma_{1,2}=\bigl( [3,1]\bigr)$, $\gamma_{2,1}=\bigl([1,1][2,1]\bigr)$, $\gamma_{2,2}=\bigl([4,1]\bigr)$ and $\gamma_{3,1}=\bigl([4,1]\bigr)$. 
}
\label{fig:paired-invariant}
\end{figure}

\paragraph{Melonic invariants of paired tensors.}We can generalize for paired tensors the recursive construction that defines melonic graphs and invariants (Sec.~\ref{sub:Melo}).  For our use here, it is enough to define them with thick edges on the canonical pairs (so in the following, ``melonic graphs \emph{of paired tensors}'' always have thick edges on the canonical pairs).  The definition is the following:
\begin{itemize}
\item A melonic invariant with one paired tensor $P$ is always denoted as $\mathbf{id}_1$ (the number of indices is not made explicit unless there is a possible confusion). It corresponds to a graph with two vertices linked by one thick edge and a number of colored edges. One has:
\be
\label{eq:trace-of-paired-tensor}
\Tr_{\mathbf{id}_1}(P)=\Tr(P) = {\sum}_{\{i_{c,r}\}_{c,r}} P_{\{i_{c,r} , i_{c,r}\}}
\ee
\item A \emph{connected} melonic invariant with more than one paired tensor always has a tensor with $\mathcal{D}\ge 1$ inputs,  of which $\mathcal{D}- 1$ are summed with their paired outputs.  In the colored graph, it corresponds to a thick edge with a number of  edges (possibly zero) with colors in $1, \ldots, D$ linking the same two vertices, with the exception of one pair of an  input and output of the same shade. Removing this tensor and reconnecting the two pending half-edges  as in Fig.~\ref{fig:melon-recursive}, one gets a smaller melonic graph.\footnote{So for a melonic invariant of paired tensors, a paired tensor with a single input and output  (such as the one labeled $[1]$ in Fig.~\ref{fig:paired-invariant}) always corresponds to a canonical pair.} An example is shown in Fig.~\ref{fig:paired-invariant}. 
\end{itemize}

Consider a connected melonic invariant with graph $\mathsf{g}$ as just defined, but for some regular (labeled) tensors $M_1, \ldots M_n$ (so the vertices have exactly one incident edge of color $c$ for each $1\le c \le D$) with $M_i\in\{A, B \ldots\}$ in the mixed case or $M_i=X_i\otimes X_{\bar i}'$ with $X_i\in \{T_a, T_b \ldots\}$ and  $X_{\bar i}'\in \{\bar T_a, \bar T_b \ldots\}$ in the pure case. Then this invariant   is a purely connected melonic invariant $\bsig$ in the usual sense, and the canonical pairing is the identity since by construction, the black and white vertices of  the canonical pairs are linked by the thick edges. 

To generate all first order invariants in the mixed case, one may consider $q$ paired tensors $P^\ell$, $1\le \ell\le q$, of the form: 
\be
\label{eq:paired-tensors-for-general-melo}
P^{\ell}_{\{i_{c,r} , j_{c,r}\}_{\substack{{1\le c \le D}\\{1\le r\le k_\ell}}}}=(M^\ell_1)_{i_{1,1}, \ldots i_{D,1};j_{1,2}, \ldots j_{D,2}}(M^\ell_2)_{i_{1,2}, \ldots i_{D,2};j_{1,3}, \ldots j_{D,3}}\ \cdots\  (M^\ell_{k_\ell})_{i_{1,k_\ell}, \ldots i_{D,k_\ell};j_{1,1}, \ldots j_{D,1}}.
\ee
A melonic invariants of $q$ paired tensors $P^1,\ldots P^q$ of this kind is a trace invariant $\Tr_\bsig[\{M^\ell_i\}]$ for $\bsig \in S_n^D$ with $n=\sum_{\ell=1}^q k_\ell$, connected and such that $(\bsig, \mathrm{id})$ is melonic. The canonical pairing of $(\bsig, \mathrm{id})$ is given by the pairing of inputs and outputs of  \eqref{eq:paired-tensors-for-general-melo}. All the invariants of this kind can be generated this way, since inserting a tensor of the form \eqref{eq:paired-tensors-for-general-melo} reproduces a cycle alternating thick edges and canonical pairs which is separating for $k_\ell>1$.

\paragraph{Paired tensors generated by first-order invariants.}To a tensor $A$ (mixed or pure $A=T_a\otimes  \bar T_a$), we associate a family of paired tensors as follows. Take any first order invariant $\bsig$ (in the mixed case, $\bsig$ is connected and $(\bsig, \mathrm{id})$ is melonic, and in the pure case, with the convention of Sec.~\ref{sub:tensor-freeness-cumulants}, $\bsig$ is purely connected melonic with canonical pairing the identity), and consider the cycles in the graph which alternate edges of color $c$ and canonical pairs.  For each such cycle, choose an edge and split it open. Each edge deletion  removes in the corresponding trace-invariant a summation between the output and the input of two (non-necessarily distinct) tensors, and these indices are now \emph{free indices} in the sense that they are not summed, and are \emph{paired} (the corresponding half-edges are linked by a path of edges of color $c$ for some $c$ and canonical pairs). If $E$ is the set of edges which have been split open (that is, $E$ is a set of pairs of the form $(i, \sigma_c(i))$ for which the summation is not carried in the definition of the trace-invariant \eqref{def:trace-invariants}), we denote by $\Tr_{\bsig_{\setminus E}}(A)$ the resulting ``partial trace-invariant''.  It is a paired tensor  with $k_c$ paired inputs and outputs for each color $c$, and if there are $n$ regular tensors $A$, the number of inputs of $\Tr_{\bsig_{\setminus E}}(A)$  is easily seen to be:
\be
\label{eq:mathcalD-of-generated}
\mathcal{D} =  \sum_{c=1}^D k_c = n(D-1)+K, 
\ee
 where $K$ is the number of connected components. If the canonical pairing of $\bsig$ coincides with the thick edges, the $(D+1)$-colored graph $\bsig_{\setminus E}$ for which the thick edges are represented is connected and it is a tree. The graph is not connected if the canonical pairing of $\bsig$ differs from the thick edges (a situation which occurs only in the mixed case)\footnote{$K$ is one plus the sum over cycles alternating thick edges and canonical pairs of the number of thick edges in the cycle minus one.}, and each connected component is a tree. In the latter case, there are inputs paired with outputs from a different connected component. 
 
 The simplest example is the tensor $A$ itself (in the pure case $A=T_a\otimes  \bar T_a$).  In the mixed case, paired tensors of the form \eqref{eq:paired-tensors-for-general-melo} with all $M_a^\ell$ equal to $A$ are generated  this way. 
See Fig.~\ref{fig:examples-paired-tensors-generated} for more generic examples. 
\begin{figure}[!h]
\centering
\raisebox{1cm}{\includegraphics[scale=1]{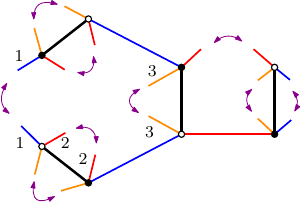}}\hspace{2cm}\includegraphics[scale=1]{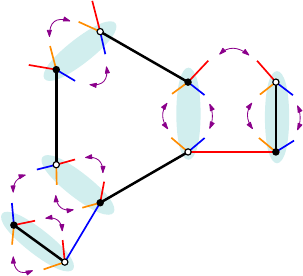}
\caption{Examples of paired tensors generated by the first order invariants of Fig.~\ref{fig:ex-invariants-purs-mixed} (the left one - with nine inputs  - may be considered both for the pure and mixed cases, and the one on the right  - with thirteen inputs - for the mixed case only). The arrows indicate the paired indices. }
\label{fig:examples-paired-tensors-generated}
\end{figure}

We denote by $\mathcal{G}^\mathrm{m}_{D, \mathcal{D}}[A]$ and by  $\mathcal{G}^\mathrm{p}_{D, \mathcal{D}}[T, \bar T]$ the sets of paired tensors with $\mathcal{D}$ inputs generated through this procedure by $A$ for the mixed first order invariants or by $T,\bar T$ (though the choice $A=T\otimes\bar T$)  for the pure first order invariants. These sets generalize the set $\mathcal{G}[A]=\mathcal{G}^\mathrm{m}_{1,1}[A]=\{A^n\}_{ n\ge 1}$ defined above\footnote{Indeed for the mixed $D=1$ case, a first order trace-invariant corresponds graphically to a cycle of length $n\ge 1$ consisting alternatively of edges of color 1 and thick edges, and splitting any edge open in this cycle, one obtains $A^n$, with ${\mathcal{D}}=1$ input.}.
One may also consider first order trace-invariants of some possibly different tensors:  In the mixed case,  $\vec M$ then takes value in a fixed set $S=\{A,B\ldots \}$, while in the pure case, $\vec X$ and $\vec X'$ respectively take value in some sets $\Theta=\{T_a, T_b \ldots\}$ and $\bar \Theta=\{\bar T_a, \bar T_b \ldots\}$.\footnote{In this case, as in Sec.~\ref{sub:tensor-freeness-cumulants}, we keep the thick edges on the canonical pairs. Each thick edge represents $X_s\otimes X'_{\bar s}$, but we allow $X'_{\bar s} \neq \bar X_s$.} We then use the notations   $\mathcal{G}^\mathrm{m}_{D, \mathcal{D}}[S]$ and $\mathcal{G}^\mathrm{p}_{D, \mathcal{D}}[\Theta, \bar \Theta]$ for the resulting sets of paired tensors (for the $D=1$ mixed case, this is the set of matrices $AB^{3}A^2C\cdots$ formed by multiplying words of matrices in the set  $S$). We let $\mathcal{G}^\mathrm{m}_D=\bigcup_{{\mathcal{D}}\ge 1}\mathcal{G}^\mathrm{m}_{D, \mathcal{D}}$ and $\mathcal{G}^\mathrm{p}_D=\bigcup_{{\mathcal{D}}\ge 1}\mathcal{G}^\mathrm{p}_{D, \mathcal{D}}$. We say that an element of  $\mathcal{G}^\mathrm{m}_D[S]$ or $\mathcal{G}^\mathrm{p}_D[\Theta, \bar \Theta]$  is \emph{generated by} $S$ or $(\Theta,\bar \Theta)$.

\subsection{Grouping and ungrouping tensors}
\label{sub:grouping}

\paragraph{Ungrouping tensors.}Consider some paired tensors $H_1, \ldots H_p$ with $H_\ell \in \mathcal{G}^\mathrm{m}_{D}[S_\ell]$, where $S_\ell\subset \{A,B\ldots\}$, or $H_\ell \in \mathcal{G}^\mathrm{p}_{D}[\Theta_\ell, \bar \Theta_\ell]$, where $\Theta_\ell\subset \{T_a,T_b\ldots\}$ and  $\bar \Theta_\ell\subset \{\bar T_a,\bar T_b\ldots\}$ and a colored graph $\mathsf{g}$ of $\vec H$. Assume that $\vec H$ involves in total $n$ regular tensors $M_1, \ldots M_n$, $M_i\in \{A,B\ldots\}$ or $M_i=X_i\otimes X_{\bar i}'$ with $X_i\in \{T_a, T_b \ldots \}$ and $X_{\bar i}'\in \{\bar T_a, \bar T_b\ldots \}$.  One can encode the index summations represented by both the edges of $\mathsf{g}$ as well as the internal edges of all the  $H_\ell$  using a $D$-tuple of permutations $\bsig\in S_n^D$, and then:  
\be
\label{eq:paired-invariant-vs-usual-invariant}
\Tr_{\mathsf{g}}(\vec H) = \Tr_\bsig(\vec M),
\ee
which in the pure case where $M_i=X_i\otimes X_{\bar i}'$ can also be written with the dedicated notation $\Tr_{\mathsf{g}}(\vec H) = \Tr_\bsig(\vec M)=\Tr_\bsig(\vec X, \vec X')$. See the examples of Fig.~\ref{fig:paired-invariant-ungrouping}.
\begin{lemma}
\label{lem:melopaired-vs-first-order}
With these notations, $\mathsf{g}$ is a melonic graph of paired tensors if and only if $\bsig$ is first order, i.e.~$(\bsig, \mathrm{id})$ is  connected and melonic in the mixed case, or $\bsig$ is purely connected  with canonical pairing the identity in the pure case.
\end{lemma}
\proof It is equivalent to show that every cycle that alternates canonical pairs and at least two edges of color $c$ for any $1\le c \le D$ is separating (where in the case of paired tensors, canonical pairs coincide with the thick edges and the inputs and outputs are of the same shade). This holds because of the fact that the paired inputs and outputs of an element $H_\ell$ of $\mathcal{G}^\mathrm{m}_{D}[S_\ell]$ or $\mathcal{G}^\mathrm{p}_{D}[\Theta_\ell, \bar \Theta_\ell]$ are linked inside $H_\ell$ by a sequence of edges of color $c$ and canonical pairs. 
 \qed
 
 \

\begin{figure}[h!]
\centering
\includegraphics[scale=0.9]{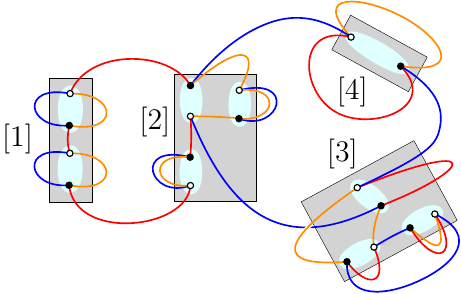}\hspace{1.5cm}\includegraphics[scale=0.9]{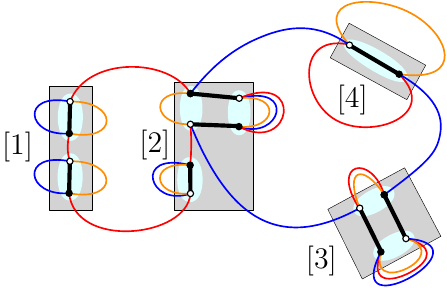}
\caption{Two examples of melonic graphs of paired tensors $\mathsf{g}$, together with some explicit choices of paired tensors $H_1, \ldots H_4$ generated by $T, \bar T$ on the left (the white and black vertices) and by   $A$  on the right (the thick edges). The grey box labeled $[i]$ represents the paired tensor $H_i$, and the blue blobs represent the canonical pairs. The edges which are not entirely included inside the grey boxes are edges of $\mathsf{g}$ and the ones entirely included in the grey boxes are edges of the $H_i$.  Considering the edges of  $\mathsf{g}$ as well as those internal to the paired tensors, one gets a first order invariant  of the regular tensors $T, \bar T$ 
in the pure case and 
 $A$
 in the mixed case. 
}
\label{fig:paired-invariant-ungrouping}
\end{figure}

Assuming $\mathsf{g}$ to be melonic and if now $E$ is a set of edges of  ${\mathsf{g}}$ with colors in  $\{1, \ldots D\}$ and \emph{with one edge per cycle that alternates edges of color $c$ and paired inputs and outputs}, we use the notation $\Tr_{\mathsf{g}_{\setminus E}}(\vec H)$ for the partial trace-invariant, which has $\lvert E \rvert$ inputs. From \eqref{eq:paired-invariant-vs-usual-invariant} and noticing that by construction there is still  in  $\bsig$ one edge of $E$ per cycle alternating non-thick edges and canonical pairs: 
\be
\Tr_{\mathsf{g}_{\setminus E}}(\vec H)=\Tr_{\bsig_{\setminus E}}(\vec M),
\ee 
which in the mixed case is an element of $\mathcal{G}^\mathrm{m}_{D}[S]$ with $S=\cup_\ell S_\ell$, and in the pure case is an element of $\mathcal{G}^\mathrm{p}_{D}[\Theta, \bar \Theta]$, with $\Theta=\cup_\ell \Theta_\ell$ and $\bar \Theta=\cup_\ell \bar \Theta_\ell$ and $\Tr_{\bsig_{\setminus E}}(\vec M)=\Tr_{\bsig_{\setminus E}}(\vec X, \vec X')$.
In both cases $\mathsf{g}_{\setminus E}$ is a tree, and $\Tr_{\mathsf{g}_{\setminus E}}(\vec H)$ is a paired tensor,
and the result of applying to it $\Tr_{\mathbf{id}_1}=\Tr$ defined in  \eqref{eq:trace-of-paired-tensor} is:
\be
\Tr \bigl(\Tr_{\mathsf{g}_{\setminus E}}(\vec H)\bigr) = \Tr \bigl(\Tr_{\bsig_{\setminus E}}(\vec M)\bigr) = \Tr_{\mathsf{g}}(\vec H) = \Tr_\bsig(\vec M). 
\ee

\paragraph{Grouping tensors.}
\emph{If now we remove the condition on $E$} so that it can be \emph{any set}  of edges of color $1, \ldots D$ of $\mathsf{g}$ (a melonic graph of paired tensors), then the colored graph $
\mathsf{g}_{\setminus E}$ might no longer be connected (Fig.~\ref{fig:paired-invariant-grouping}). Letting $\vec H_{\lvert_\jmath}$ be the restriction of $\vec H$ to the set $R_\jmath\subset \{1, \dots p\}$ of the labels of the  thick edges that belong to the connected component number $\jmath$ of $\mathsf{g}_{\setminus E}$, the partial trace of $\vec H_{\lvert_\jmath}$ associated to the connected component number $\jmath$ is a paired tensor $P_\jmath$. Starting from a half-edge in $\mathsf{g}_{\setminus E}$ (created by splitting an edge open), and alternatively following paired tensors and pairs of inputs and outputs of the same shade, one must eventually arrive at another half-edge. Connecting these two half-edges to form an edge (the resulting edge might or might not be an edge of $\mathsf{g}$), and repeating the operation for all the half-edges of $\mathsf{g}_{\setminus E}$, one obtains a graph $\mathsf{h}$ encoding an invariant of the paired tensors $\vec H$ whose connected components $\mathsf{h}_\jmath$ involve the same $H_\ell$, $\ell\in R_j$ as the  $P_\jmath$. To recover $P_\jmath$ from $\mathsf{h}_\jmath$, one must remove some edges, at most one edge is removed for each cycle alternating paired indices and colored edges.
\begin{figure}[!h]
\centering
\includegraphics[scale=0.55]{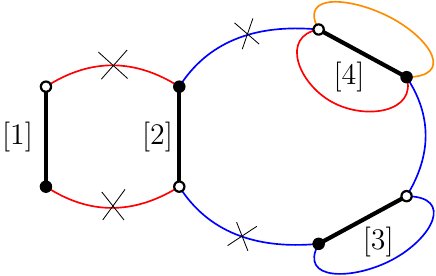}\hspace{1.6cm}\includegraphics[scale=0.55]{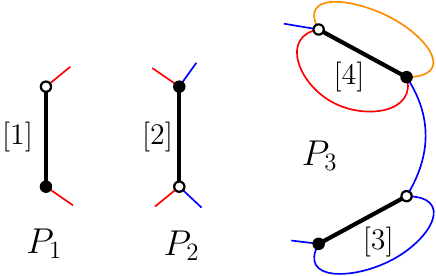}\hspace{1.6cm}\includegraphics[scale=0.55]{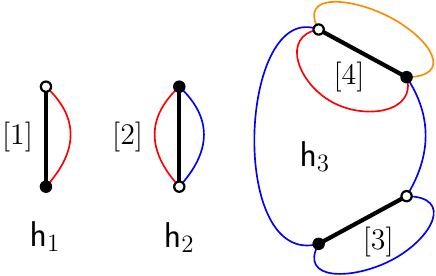}
\caption{
Left: the example  of Fig.~\ref{fig:paired-invariant} and a set of edges $E$ (crossed). Center: splitting these edges open, one gets a collection of paired tensors $P_\jmath$. Right: the melonic graphs $\mathsf{h}_\jmath$. 
}
\label{fig:paired-invariant-grouping}
\end{figure}

In addition, since $\mathsf{g}$ is a melonic graph of paired tensors that are elements of $\mathcal{G}^\mathrm{m}_{D}[S_\ell]$ or  $\mathcal{G}^\mathrm{p}_{D}[\Theta_\ell, \bar \Theta_\ell]$, there is a $\bsig\in S_n^D$ 
 first order (Lemma~\ref{eq:paired-invariant-vs-usual-invariant}), and the relevant canonical pairing $\eta$,
and a $\vec M$, such that $\Tr_{\mathsf{g}}(\vec H) = \Tr_\bsig(\vec M)$ (in the pure case, $M_i=X_i\otimes X_i'$ and one can also write $\Tr_{\mathsf{g}}(\vec H) = \Tr_\bsig(\vec M)=\Tr_\bsig(\vec X, \vec X')$). There is also a $\btau\in S_n^D$  with connected components $\btau_\jmath$ satisfying $\Tr_{\mathsf{h}_\jmath}(\vec H_\jmath) = \Tr_{\btau_\jmath}(\vec M_\jmath)$ with $\vec M_\jmath=\vec M_{\lvert_{R_\jmath}}$ (so in the pure case, $\Tr_{\mathsf{h}_\jmath}(\vec H_\jmath) = \Tr_{\btau_\jmath}(\vec X_\jmath, \vec X_\jmath')$),  and one has\footnote{Indeed the cycles that alternate edges of color $c$ and paired indices in $\mathsf{h}$ or $\mathsf{g}$ correspond to cycles that alternate edges of color $c$ and canonical pairs in $\btau$ or $\bsig$. The operation of splitting open the edges of $E$ and reconnecting them to the first half-edge that follows, indeed builds a non-crossing permutation for each such cycle (it's the smallest that leave the edges not in $E$ unchanged). } $\btau\eta^{-1} \preceq \bsig\eta^{-1}$, and therefore in the mixed case $(\btau, \mathrm{id})$ is melonic with canonical pairing $\eta$   (Thm.~\ref{thm:limit-for-wishart-tensor-first-order}), and in the pure case $\btau$ is melonic with canonical pairing the identity  (Thm.~\ref{thm:free-cumulants-melonic}). In both cases, the connected components $\{\btau_\jmath\}$  of $\btau$ are first order invariants, so from Lemma~\ref{eq:paired-invariant-vs-usual-invariant} \emph{the $\{\mathsf{h}_\jmath\}$ are connected melonic graphs of the paired tensors $\vec H_{\lvert_\jmath}$.} If in addition $E$ is such that $P_\jmath$ is obtained from   $\btau_\jmath$ (which is first order) by removing one edge  per cycle alternating edges of color $1\le c \le D$ and canonical pairs, then by definition of the generated sets:
\be
\label{eq:Pj-is-generated}
P_\jmath \in \mathcal{G}^\mathrm{m}_{D}[S_\jmath] \qquad \textrm{(mixed case)},\qquad \textrm{and}\qquad  P_\jmath \in \mathcal{G}^\mathrm{p}_{D}[\Theta_\jmath, \bar \Theta_\jmath] \qquad \textrm{(pure case)}
\ee
where $S_\jmath=\cup_{\ell \in R_j} S_\ell$, $\Theta_\jmath=\cup_{\ell \in R_j} \Theta_\ell$   and $\bar \Theta_\jmath=\cup_{\ell \in R_j} \bar \Theta_\ell$. Furthermore, the trace-invariants corresponding to the $\mathsf{h}_\jmath$ satisfy:
\be
\label{eq:trace-of-one-paired-vs-trace-invariant}
 \Tr_{\mathsf{h}_\jmath}(\vec H_\jmath) =  \Tr_{\mathbf{id}_1} (P_\jmath)=  \Tr(P_\jmath). 
\ee

One can also see $\Tr_\mathsf{g}(\vec H)$ itself as a melonic invariant $\mathsf{k}$ of the paired tensors $\vec P$. The edges of the set $E$ considered above encode summations between the indices of the $P_\jmath$. Letting $\mathsf{k}$ be the colored graph  with one thick edge per paired tensor $P_\jmath$ and whose colored edges are  the edges of  $E$, we have:
\be
\label{eq:trace-of-n-paired-vs-trace-invariant}
\Tr_\mathsf{g}(\vec H) = \Tr_{\mathsf{k}}(\vec P). 
\ee

In particular,  starting from $\bsig\in S_n^D$ purely connected with $\bar \omega(\bsig; \mathrm{id})=0$ and \emph{regular} tensors $\vec Q\in \{A, B \ldots\}$ (or the pure equivalent $Q_i=X_i\otimes X_{\bar i}'$ with $X_i\in \{T_a, T_b \ldots \}$ and $X_{\bar i}'\in \{\bar T_a, \bar T_b \ldots \}$) instead of \emph{paired} tensors $\vec H$, choosing a set of edges $E$ and taking the $P_\jmath$ to be the connected components of $\Tr_{\bsig_{\setminus E}}(\vec Q)$ (which rewrites in the pure case as $\Tr_{\bsig_{\setminus E}}(\vec X, \vec X')$), 
one gets some $\bnu_\jmath$ instead of $\mathsf{h}_\jmath$, and the relations \eqref{eq:trace-of-one-paired-vs-trace-invariant} and  \eqref{eq:trace-of-n-paired-vs-trace-invariant} become $\Tr_{\bnu_\jmath}(\vec Q)= \Tr(P_\jmath)$ and $\Tr_\mathsf{\bsig}(\vec Q) = \Tr_{\mathsf{k}}(\vec P)$.

\paragraph{Alternating and almost alternating invariants.}Given a trace-invariant of paired tensors $H_1, \ldots H_p$ encoded by its colored graph $\mathsf{g}$,  with $H_\ell \in \mathcal{G}^\mathrm{m}_{D}[Q_\ell]$ with $Q_\ell\in \{A,B\ldots\}$,  or $H_\ell \in \mathcal{G}^\mathrm{p}_{D}[X_\ell, \bar X_\ell]$, where $(X_\ell, \bar X_\ell) \in \{(T_a, \bar T_a),(T_b, \bar T_b)\ldots\}$ (that is, here, if $X_\ell=T_a$, one must have $\bar X_\ell=\bar T_a$, etc),
we will say that $(\mathsf{g}, \vec H)$ is \emph{alternating} or \emph{strictly alternating}  if it has more than one paired tensor and all the  edges of color $1\le c \le D$ of $\mathsf{g}$  that link some $H_\ell$, $H_{\ell'}$ with $\ell\neq \ell'$ are such that $H_\ell$ and $H_{\ell'}$  are generated by different tensors $Q_\ell\neq Q_{\ell'}$ or $(X_\ell, \bar X_\ell)\neq(X_{\ell'}, \bar X_{\ell'})$. Note that  any number of edges may  link the inputs and outputs of a given $H_\ell$.

We say that $(\mathsf{g}, \vec H)$  is \emph{almost alternating}, if at least one edge of some  color $1\le c \le D$ links some paired tensors $H_\ell$, $H_{\ell'}$ with $\ell\neq \ell'$ and generated by different tensors $Q_\ell\neq Q_{\ell'}$, and \emph{at most one} edge  links two paired tensors $H_\ell$, $H_{\ell'}$  with $\ell\neq \ell'$ and generated by the same $Q_\ell=Q_{\ell'}$. %
\begin{figure}[!h]
\centering
\includegraphics[scale=0.5]{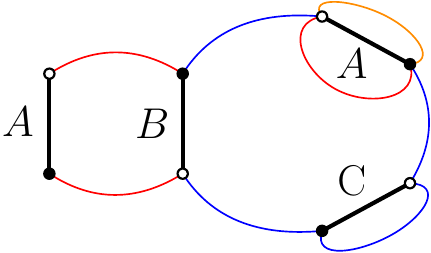}\hspace{2cm}\includegraphics[scale=0.5]{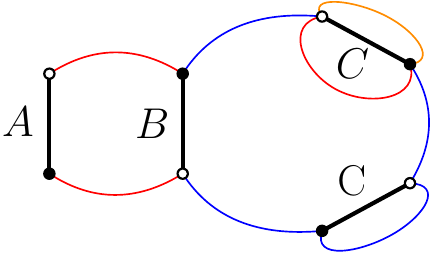}\caption{
Examples of alternating (left) and almost alternating (right) melonic $(\mathsf{g}, \vec H)$, where the letters indicate the generators for the paired tensors represented by the thick edges. In the pure case, the generators $A,B,C$ are replaced by $(T_a, \bar T_a),(T_b, \bar T_b),(T_c, \bar T_c)$.
}
\label{fig:paired-invariant-alternating}
\end{figure}

With the same notations, if $(\mathsf{g}, \vec H)$ is not alternating, we may form groups of tensors to form coarser invariants which are alternating, as we now detail.
In $\mathsf{g}$, let $E^{\neq} = E^{\neq}(\mathsf{g}, \vec H)$ be the set of edges of color $c\in\{1, \ldots D\}$ that link paired tensors $H_\ell, H_{\ell'}$ \emph{generated by different} tensors $Q_\ell\neq Q_{\ell'}$\footnote{For the example on the right of Fig.~\ref{fig:paired-invariant-alternating}, $E^{\neq} $ is the set of crossed edges on the left of Fig.~\ref{fig:paired-invariant-grouping}. }. We choose $E^\mathrm{int}$ a set of edges of $\mathsf{g}_{\setminus E^{\neq}}$, such that $E^\mathrm{int}$ contains one edge per cycle of $\mathsf{g}_{\setminus E^{\neq}}$ that alternates edges of color $c$ and paired inputs and outputs, and let $E^\mathrm{tot}= E^{\neq}\cup E^\mathrm{int}$. 
We let $P^\neq_\jmath=P^\neq_\jmath (\mathsf{g}, \vec H, E^\mathrm{int})$ and $\mathsf{h}^\neq_\jmath=\mathsf{h}^\neq_\jmath(\mathsf{g}, \vec H, E^\mathrm{int})$ be the $P_\jmath$ and $\mathsf{h}_\jmath$ of $\mathsf{g}_{\setminus E^\mathrm{tot}}$, constructed as above.  From the construction,  all paired tensors $H_\ell$ that belong to the same $P^\neq_\jmath$ are generated by the \emph{same} tensor $Q_\ell=Q(\jmath)$ in  $\{A, B \ldots\}$, so that from \eqref{eq:Pj-is-generated},  in the mixed case, $P^\neq_\jmath\in  \mathcal{G}^\mathrm{m}_{D}[Q(\jmath)]$, and in the pure case, $Q(\jmath)=X(\jmath)\otimes \bar X(\jmath)$ and $P^\neq_\jmath \in  \mathcal{G}^\mathrm{p}_{D}[X(\jmath), \bar X(\jmath)]$. The colored graph $\mathsf k$ built above is in the present situation denoted by $\mathsf{k}^\neq=\mathsf{k}^\neq(\mathsf{g}, \vec H, E^\mathrm{int})$. From the construction, every edge $e$ of $\mathsf{k}^\neq$ either links paired tensors $P^\neq_\jmath, P^\neq_{\jmath'}$ with $\jmath\neq \jmath'$ that are generated by different tensors $Q(\jmath)\neq Q(\jmath')$, or $e$ links and input and output of one  paired tensor $P^\neq_\jmath$: $(\mathsf{k}^\neq, \vec P^\neq)$ is strictly alternating.

\subsection{Asymptotic moments of random paired tensors}
\label{sub:asympt-moments-of-paired}

\paragraph{Asymptotics.} 
Consider any paired-tensors $H_1, \ldots H_q$,  $H_\ell \in \mathcal{G}^\mathrm{m}_{D}[S_\ell]$, where $S_\ell\subset \{A,B\ldots\}$, or $H_\ell \in \mathcal{G}^\mathrm{p}_{D}[\Theta_\ell, \bar \Theta_\ell]$, where $\Theta_\ell\subset \{T_a,T_b\ldots\}$ and  $\bar \Theta_\ell\subset \{\bar T_a,\bar T_b\ldots\}$,   and a trace-invariant of these $q$ paired tensors with colored graph $\mathsf{g}$, that involves a total of  $n$  regular tensors $M_1, \ldots M_n$, $M_i\in \{A,B\ldots\}$ or $M_i=X_i\otimes X_{\bar i}'$ with  with $X_i\in \{T_a, T_b \ldots \}$ and $X_{\bar i}'\in \{\bar T_a, \bar T_b \ldots \}$. Then from \eqref{eq:paired-invariant-vs-usual-invariant}, there is a $\btau\in S_n^D$ such that $\Tr_{\mathsf{g}}(\vec H)=\Tr_\btau(\vec M)$. 
Furthermore, $\mathsf{g}$ is a connected melonic graph of  paired tensors  if and only if $\btau$ is  
first order (Lemma~\ref{eq:paired-invariant-vs-usual-invariant}). 
Assume that the regular random tensors $\{A,B\ldots\}$  or $\{T_a, \bar T_a, T_b, \bar T_b \ldots\}$ scale as assumed in the previous sections respectively in the mixed \eqref{eq:scaling-wishart-D1}
and \eqref{eq:scaling-wishart-multitens} and pure \eqref{eq:scaling-hypothesis-pure} and \eqref{eq:scaling-hypothesis-pure-multitensors} cases.
In both cases,  the expectations of trace-invariants for
\emph{first order} $\btau$ are
of order $N$, and we may set:
\be
\label{eq:asymptotic-moments-paired-tensors}
\phi_{\mathsf{g}}\bigl(\vec h\bigr)  : =\lim_{N\rightarrow \infty}  \frac 1 {N} 
 \mathbb{E}\bigl[\Tr_{\btau}(\vec M)\bigr]. 
\ee
In the mixed and pure case,  this respectively corresponds to the dedicated notations: 
\be
\label{eq:asympt-of-melo-paired-vs-first-order-regular}
\phi_{\mathsf{g}}(\vec h)=\varphi_{\btau}^\mathrm{m}(\vec m)\qquad \textrm{ or }\qquad  \phi_{\mathsf{g}}(\vec h)=\varphi_{\btau}(\vec x, \vec{x'}).
\ee
We say that $(\mathsf{g}, \vec h)$ is \emph{almost or strictly  alternating} if $(\mathsf{g}, \vec H)$ is.

Considering a set $E$ of edges of $\mathsf{g}$ satisfying the condition 
above \eqref{eq:Pj-is-generated}
and splitting open the edges of $E$ defines as in \eqref{eq:Pj-is-generated} some $P_\jmath\in \mathcal{G}^\mathrm{m}_{D}[S_\jmath]$ or $\mathcal{G}^\mathrm{p}_{D}[\Theta_\jmath, \bar \Theta_\jmath]$, and with the notations of Sec.~\ref{sub:grouping},  some $\vec H_\jmath$, $\btau_\jmath$, and $\mathsf{k}$. The $\btau_\jmath$ being 
first order, we may define as in \eqref{eq:asymptotic-moments-paired-tensors}  
$\phi_{\mathsf{h}_\jmath}(\vec h_\jmath)  = \varphi_{\btau_\jmath}^\mathrm{m}(\vec m_\jmath)$, or in the pure case  $\phi_{\mathsf{h}_\jmath}(\vec h_\jmath)  = \varphi_{\btau_\jmath}(\vec x_\jmath, \vec {{x}_{\jmath}'})$, so that \eqref{eq:trace-of-one-paired-vs-trace-invariant} becomes: 
\be
\label{eq-two-relations-between-asymptotic-paired-and-unpaired}
 \phi_{\mathsf{h}_\jmath}(\vec h_\jmath) =  \phi_{\mathbf{id}_1} (p_\jmath)=  \phi(p_\jmath).
 \ee
Asymptotically,  \eqref{eq:trace-of-n-paired-vs-trace-invariant} implies:
\be
\label{eq:asymptotic-relation-grouping}
\phi_\mathsf{g}(\vec h) = \phi_{\mathsf{k}}(\vec p). 
\ee
In particular,  starting from $\bsig\in S_n^D$ purely connected with $\bar \omega(\bsig; \mathrm{id})=0$ and $\vec Q\in \{A, B \ldots\}$ (or $Q_i=X_i\otimes X_{\bar i}'$ with $X_i\in \{T_a, T_b \ldots \}$ and $X_{\bar i}'\in \{\bar T_a, \bar T_b \ldots \}$)   instead of $\vec H$, with $\bnu_\jmath$ instead of $\mathsf{h}_\jmath$: $\varphi_{\bnu_\jmath}^\mathrm{m}(\vec q)= \phi(p_\jmath)$ and $\varphi^\mathrm{m}_{\bsig}(\vec q) = \phi_{\mathsf{k}}(\vec p)$ (or pure equivalent $\varphi_{\bnu_\jmath}(\vec x, \vec{x'})= \phi(p_\jmath)$ and $\varphi_{\bsig}(\vec x, \vec{x'}) = \phi_{\mathsf{k}}(\vec p)$).

\paragraph{Multilinearity.} {We now consider the asymptotics of $\Tr_{\mathsf{g}}(\vec H'')$, where $H''_\ell= H_\ell + t_\ell H'_\ell$ with $t_\ell\in \mathbb{C}$ and $H_\ell,H'_\ell\in \mathcal{G}^\mathrm{m}_{D}[S]$ with $S=\{A,B\ldots\}$, or $H_\ell \in \mathcal{G}^\mathrm{p}_{D}[\Theta, \bar \Theta]$, where $\Theta= \{T_a,T_b\ldots\}$ and  $\bar \Theta= \{\bar T_a,\bar T_b\ldots\}$,  $\mathsf{g}$ a connected melonic graph of these paired tensors. By multilinearity:
\be 
\label{eq:develop-multilinear}
\Tr_{\mathsf{g}}(\vec H'') = \sum_{k=0}^p \ \sum_{1\le \ell(1) < \ell(2) <\ldots < \ell(k)\le p}\Bigl(\prod_{j=1}^k t_{\ell(j)} \Bigr)\;\Tr_{\mathsf{g}}\Bigl(\vec H''_{[\ell(1), \ldots, \ell(k)]}\Bigr), 
\ee
where setting $I=\{\ell(1), \ldots, \ell(k)\}$, the $\ell$th component of $\vec H''_{[I]}$ is $H_\ell$ if $\ell\notin I$ and is $H_\ell'$ otherwise. Due to \eqref{eq:paired-invariant-vs-usual-invariant}, \eqref{eq:scaling-hypothesis-pure-multitensors} and \eqref{eq:scaling-wishart-multitens}, 
$\Tr_{\mathsf{g}}(\vec H''_{[I]})$ is of order $N$ and the  limit of  the rescaled expectation is: 
\be 
\label{eq:varphi-multilinear}
\phi_{\mathsf{g}}(\vec h'')  = \sum_{k=0}^p \ \sum_{1\le \ell(1) < \ell(2) <\ldots < \ell(k)\le p}\Bigl(\prod_{j=1}^k t_{\ell(j)} \Bigr)\;\phi_{\mathsf{g}}\Bigl(\vec h''_{[I)]}\Bigr), \qquad h''_\ell= h_\ell + t_\ell h'_\ell\;,
\ee
where  the $\ell$th component of $\vec h''_{[I]}$ is $h_\ell$ if $\ell\notin I$ and  $h_\ell'$ otherwise:   $\phi_{\mathsf{g}}$ is multilinear by construction.

\subsection{Asymptotically centering paired tensors}
\label{sub:centering-paired}
For any $\mathcal{D}\in \mathbb{N}^\star$, we let:
\be
\un_{\mathcal{D}} = \frac{\un^{\otimes \mathcal{D}}}{N^{\mathcal{D}-1}}. 
\ee

If $H\in \mathcal{G}^\mathrm{m}_{D, \mathcal{D}}[S]$ with $S= \{A,B\ldots\}$ or $H\in \mathcal{G}^\mathrm{p}_{D, \mathcal{D}}[\Theta, \bar \Theta]$, where $\Theta= \{T_a,T_b\ldots\}$ and  $\bar \Theta= \{\bar T_a,\bar T_b\ldots\}$, in order to  asymptotically center $H$, we may define at finite $N$:
\be
\label{eq:Gell}
G = H -  \phi(h) \un_\mathcal{D}, 
\ee 
where  $\phi(h)=\phi_{\mathbf{id}_1}(h)$, so if $H$  involves a total of  $n$  regular tensors $M_1, \ldots M_n$, $M_i\in \{A,B\ldots\}$ or $M_i=X_i\otimes X_{\bar i}'$ with  $X_i\in \{T_a, T_b \ldots \}$ and $X_{\bar i}'\in \{\bar T_a, \bar T_b \ldots \}$, there exists a first order $\bsig\in S_n^D$  such that $ \phi(h)=\varphi_\bsig^\mathrm{m}(\vec m)$ in the mixed case and  $ \phi(h)=\varphi_\bsig(\vec x, \vec{x'})$ in the pure case.

We need to justify that the rescaling \eqref{eq:Gell} is the appropriate one for considering quantities
like $\phi_\mathsf{g}$, that it will be asymptotically multilinear, and that $G$ thus defined will be  asymptotically centered: $\phi(g)=0$. 
The reason why the identity is rescaled by $N^{\mathcal{D}-1}$ is because for any $\bsig\in S_n^{D}$: 
\be
\label{eq:scaling-of-id-general-invariant}
\Tr_\bsig(\un_D) =N^{n-d(\bsig, \mathrm{id})} = N^{1-(D-1)(K_\mathrm{p}(\bsig)-1) - \bar \omega(\bsig;\mathrm{id})}, 
\ee
so if $\bsig$ is purely connected and  $\bar \omega(\bsig; \mathrm{id})=0$ one has  $\Tr_\bsig(\un_D)=N$.

\begin{lemma}
\label{lem:preimage-of-identity}
Any element of $\mathcal{G}^\mathrm{m}_{D, \mathcal{D}}[\un_D]$  \emph{which originates from a purely connected $\bsig$} is equal to~$\un_{\mathcal{D}}$. Therefore for any $\mathcal{D}$ there exists $\bsig\in S_n^D$ with $n=(\mathcal{D}-1)/(n-1)$, $\bsig$ purely connected with $\bar \omega(\bsig;\mathrm{id})=0$,  and a set $E$ consisting of $\mathcal{D}$ edges, one per cycle alternating 
edges of color $c$ and thick edges, such that:
\be
\un_{\mathcal{D}}= \Tr_{\bsig_{\setminus_E}}\left(\un_D \right)
\ee
\end{lemma}
\proof With the notations of the lemma, every path of $\bsig$ consisting alternatively of thick edges representing  and edges of color $c$ and whose extremities are half-edges represents a Kronnecker delta in the corresponding indices. The cycles of this kind contribute as factors of $N$, and the overall exponent of $N$ is $n-d(\bsig, \mathrm{id}) - \mathcal{D} = 1 - \mathcal{D}$. \qed

\

If  $\mathsf{g}$ is a  connected melonic graph  of $\un_{\mathcal{D}_1}, \ldots \un_{\mathcal{D}_p}$, then from Lemma~\ref{lem:preimage-of-identity}, \eqref{eq:paired-invariant-vs-usual-invariant} and \eqref{eq:scaling-of-id-general-invariant}:
\be
\label{eq:paired-invariant-of-identities}
\Tr_{ \mathsf{g} } \left(\un_{\mathcal{D}_1}, \ldots, \un_{\mathcal{D}_p}\right) = N, 
\ee}
\emph{leading to the same global rescaling $N$ for melonic invariants as in} \eqref{eq:asymptotic-moments-paired-tensors}. By linearity of  $\Tr=\Tr_{\mathbf{id}_1}$, we do find that $G$ introduced in \eqref{eq:Gell} satisfies $\phi(g)=\phi_{\mathbf{id}_1}(g)=0$.

\

Lemma~\ref{lem:preimage-of-identity} also implies that the following sets both contain $\un_{ \mathcal{D}}$:
\begin{itemize}
\item The set $\mathcal{G}^\mathrm{p}_{D, \mathcal{D}}[\Theta, \bar \Theta;\un_D]$  of paired tensors with $\mathcal{D}$ inputs obtained from pure first order invariants  (Sec.~\ref{sub:paired-tensors}, see also Fig.~\ref{fig:examples-paired-tensors-generated}), where any thick edge may represent $\un_D$, or any of the elements of the form $X_i\otimes X_{\bar i}'$ with  $X_i\in \Theta$ and $X_{\bar i}'\in \bar \Theta$. 
\item   The set $\mathcal{G}^\mathrm{m}_{D, \mathcal{D}}[S;\un_D]$ of paired tensors with $\mathcal{D}$ inputs obtained from mixed first order invariants, where any thick edge may represent any element of $S$, but only thick edge that coincide with a canonical pair may be replaced with $\un_D$.
\end{itemize}
Respectively let $\mathcal{A}^\mathrm{m}_{D,\mathcal{D}}[S;\un_D]$ and $\mathcal{A}^\mathrm{p}_{D,\mathcal{D}}[\Theta, \bar \Theta; \un_D]$ be the sets of complex linear combinations of elements of $\mathcal{G}^\mathrm{m}_{D, \mathcal{D}}[S;\un_D]$ and $\mathcal{G}^\mathrm{p}_{D, \mathcal{D}}[\theta, \bar \theta;\un_D]$.   \eqref{eq:Gell} is an element of $\mathcal{A}^\mathrm{m}_{D,\mathcal{D}}[S;\un_D]$ in the mixed case and of $\mathcal{A}^\mathrm{p}_{D,\mathcal{D}}[\Theta, \bar \Theta; \un_D]$ in the pure case. We let $\mathcal{A}^\mathrm{m}_D=\cup_{\mathcal{D}\ge 1}\mathcal{A}^\mathrm{m}_{D,\mathcal{D}}$ and $\mathcal{A}^\mathrm{p}_D=\cup_{\mathcal{D}\ge 1}\mathcal{A}^\mathrm{p}_{D,\mathcal{D}}$. For $D=1$ and $S=\{A\}$, $\mathcal{A}^\mathrm{m}_1[A;\un]=\mathcal{A}^\mathrm{m}_{1,1}[A; \un]$ is the algebra $\mathcal{A}[A,\un]$ generated by $\un$ and $A$.
 
\

Considering $\vec H'$ where $H'_\ell=H_\ell + \alpha_\ell \un_{\mathcal{D}_\ell}$ where  $H_\ell \in \mathcal{G}^\mathrm{m}_{D, \mathcal{D}_\ell}  [S_\ell]$,  $S_\ell\subset \{A,B\ldots\}$, or $H_\ell \in \mathcal{G}^\mathrm{p}_{D, \mathcal{D}_\ell}   [\Theta_\ell, \bar \Theta_\ell]$, where $\Theta_\ell\subset \{T_a,T_b\ldots\}$ and  $\bar \Theta_\ell\subset \{\bar T_a,\bar T_b\ldots\}$, in order to make sense of  $\phi_{\mathsf{g}}\bigl(\vec{h}'\bigr)$, where $\mathsf{g}$ is a connected melonic graph of $p$ paired tensors, we write by multilinearity:
\be 
\label{eq:develop-centered-grouped}
\Tr_\mathsf{g}(\vec H') = \sum_{k=0}^p \ \sum_{1\le \ell(1) < \ell(2) <\ldots < \ell(k)\le p}\Bigl(\prod_{j=1}^k \alpha_{\ell(j)} \Bigr)\;\Tr_\mathsf{g} \Bigl(\vec H_{[\ell(1), \ldots, \ell(k)]}\Bigr), 
\ee
where letting $I=\{\ell(1), \ldots, \ell(k)\}$, the component of index $\ell$ of the vector $\vec H_{[I]}$ is $H_\ell$ if $\ell\notin I$ and is $\un_{\mathcal{D}_\ell}$ otherwise. We therefore need to study the asymptotics of the right-hand-side.

\begin{figure}[!h]
\centering
\includegraphics[scale=1.1]{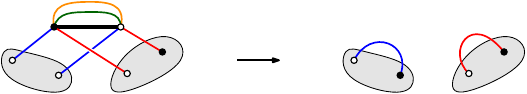}
\caption{Removal of a thick edge ($D=4$, $l=2$) from a purely connected melonic graph~$\bsig$ with $\bar \omega(\bsig; \mathrm{id})=0$. The grey blobs represent parts of the graph which are disconnected from one another due to the properties of $\bsig$, so that there are two connected components on the right. }
\label{fig:thick-edge-removal}
\end{figure}

If $\bsig\in S_n^D$ is purely connected and $\bar \omega(\bsig; \mathrm{id})=0$, choosing a thick edge labeled $i$, we may remove this thick edge and the two vertices it links from the graph, as well as the $D-l$ colored edges that link these two vertices, and then reconnect the $l$ pending half-edges of color $1,\ldots D$ as in Fig.~\ref{fig:thick-edge-removal}. Permutation-wise, this corresponds to removing $i$ from all cycles of the permutations $\sigma_c$, leading to a new $D$-tuple $\bsig_{\setminus i}$. This invariant is still melonic with $\bar \omega(\bsig; \mathrm{id})=0$, but it now has $l$ (purely) connected components.  One has for any $\vec M$ with\footnote{The result does not depend on which $M_i$ is set to $\un_D$.} $M_n=\un_D$:
\be
\Tr_{\bsig}\bigl(\vec M\bigr) =  \frac 1 {N^{l-1}} \Tr_{\bsig_{\setminus n}}\bigl(\vec M_{\setminus n})\bigr), 
\ee
where $l-1=D-1 + l-D$. In the pure case, $\vec M$ has components $M_i=X_i\otimes X_{\bar i}'$.  Since both sides are first order, the expectation of the left-hand side is of order $N$ and that of the trace-invariant on the right-hand side is of order $l$, so that in the end we find (here $\Pi(\bsig_{\setminus n})=\Pi_\mathrm{p}(\bsig_{\setminus n})$):
$$
\phi_{\bsig}\bigl(m_1, \ldots, m_{n-1},1_D\bigr):=  \lim_{N} \mathbb{E}\biggl[\frac {\Tr_{\bsig}\bigl(\vec M\bigr)}{N} \biggr]=  \lim_{N}   \mathbb{E}\biggl[\frac{\Tr_{\bsig_{\setminus n}}\bigl(\vec M_{\setminus n})\bigr)}{{N^{l}} }\biggr]= \varphi^\mathrm{m}_{\Pi(\bsig_{\setminus n}), \bsig_{\setminus n}}(m_1, \ldots, m_{n-1}), 
$$
where for the pure case the right hand side can be replaced by $\varphi_{\Pi_\mathrm{p}(\bsig_{\setminus n}), \bsig_{\setminus n}}(\vec x_{\setminus n}, \vec{x'}_{\setminus n})$. 
One may iterate this process for all the thick edges whose labels are in $I=\{n-k+1,\ldots, n\}$, and as the procedure does not depend on the ordering of the thick edge removals, one may denote the resulting invariant by $\bsig_{\setminus I}$. We may then define $\phi_{\bsig}\bigl(m_1, \ldots, m_{n-k},1_D, \ldots, 1_D\bigr)$ in the same manner, and notice that for any $k<n$: 
$$
\phi_{\bsig}\bigl(m_1, \ldots, m_{n-k},1_D, \ldots, 1_D\bigr) = \varphi^\mathrm{m}_{\Pi(\bsig_{\setminus I}), \bsig_{\setminus I}}(m_1, \ldots, m_{n-k}),
$$
or $\varphi_{\Pi_\mathrm{p}(\bsig_{\setminus I}), \bsig_{\setminus I}}(\vec x_{\setminus I}, \vec{x'}_{\setminus I})$ for the pure case. For $k=n$, one has from \eqref{eq:scaling-of-id-general-invariant} that $\phi_{\bsig}(1_D)=1$.

\emph{If now  $\mathsf{g}$ is a more general connected melonic graph of $p$ paired tensors}, and with the notations of \eqref{eq:develop-centered-grouped}, combining \eqref{eq:paired-invariant-vs-usual-invariant} and Lemma~\ref{lem:melopaired-vs-first-order} for  the $H_\ell$ for  $\ell\notin I$ as well as Lemma~\ref{lem:preimage-of-identity} for the $\ell\in I$, we see that there exists a first order $\bsig\in S_n^D$ and $\vec M=(M_1, \ldots, M_n)$ such that $\Tr_\mathsf{g} (\vec H_{[I]}) = \Tr_\bsig(\vec M),$ where $I=\{\ell(1), \ldots, \ell(k)\}$.
Each $M_i$ comes from one of the components of $\vec H_{[I]}$, so it takes value in the set $\{\un_D, A,B\ldots\}$ (in the pure case, each $M_i$ different from $\un_D$ writes as  $M_i=X_i\otimes X_{\bar i}'$). We therefore define 
$\phi_\mathsf{g} (\vec h_{[I]}) =
\lim_{N\rightarrow \infty} \mathbb{E}[
\Tr_\bsig(\vec M )]/N
$, where the component of index $\ell$ of  $\vec h_{[I]}$ is $h_\ell$ if $\ell\notin\{\ell(1), \ldots, \ell(k)\}$ and is $1_{\mathcal{D}_\ell}$ otherwise. If $k=p$, we define $\phi_\mathsf{g}(1_{\mathcal{D}_1}, \ldots, 1_{\mathcal{D}_p})$ as the limit of  \eqref{eq:paired-invariant-of-identities} rescaled by $1/N$, which is always equal to 1 regardless of  the $\mathcal{D}_\ell$:
\be
\label{eq:varphi-of-one}
\phi_\mathsf{g}(1_{\mathcal{D}_1}, \ldots, 1_{\mathcal{D}_p})=1. 
\ee

 We define as above $\mathsf{g}_{\setminus I}$ as the (non-necessarily connected) melonic graph of paired tensors obtained by removing the thick edges with indices in $I=\{\ell(1), \ldots, \ell(k)\}$ from the graph in the same way as in Fig.~\ref{fig:thick-edge-removal} (the procedure can be adapted here since the inputs and outputs of the same color  are paired). We also denote by $\phi_{\Pi(\mathsf{g}_{\setminus I}), \mathsf{g}_{\setminus I}}(\vec h_{\setminus I})$ the product of the $\phi_{\mathsf{g}_j}$ over the connected components $\mathsf{g}_j$ of $\mathsf{g}_{\setminus I}$, where $\vec h_{\setminus I}$ gathers the $h_\ell$ for $\ell \notin I=\{\ell(1), \ldots, \ell(k)\}$. One has:
\be
\label{eq:varphi-of-one-and-contraction}
\phi_\mathsf{g} \bigl(\vec h_{[I]}\bigr) = \phi_{\Pi(\mathsf{g}_{\setminus I}), \mathsf{g}_{\setminus I}}(\vec h_{\setminus I}). 
\ee
Since none of the arguments onn the right-hand side is 1, it can be expressed in terms of asymptotic moments of  $\vec m$ or $\vec x, \vec{x'}$. Finally coming back to \eqref{eq:develop-centered-grouped}, we see that all the terms in the sum are of order $N$, so we may rescale and take the limit of the expectation, to obtain: 
 \be 
\phi_\mathsf{g}(\vec h') = \sum_{k=0}^p \ \sum_{1\le \ell(1) < \ell(2) <\ldots < \ell(k)\le p}\Bigl(\prod_{j=1}^k \alpha_{\ell(j)} \Bigr)\;\phi_\mathsf{g} \bigl(\vec h_{[I]}\bigr) , 
\qquad h'_\ell = h_\ell + \alpha_\ell 1_{\mathcal{D}_\ell}
\ee
The notation  $h'_\ell = h_\ell + \alpha_\ell 1_{\mathcal{D}_\ell}$ is consistent with multilinearity, since  the component of index $\ell$ of  $\vec h_{[I]}$ is $h_\ell$ if $\ell\notin\{\ell(1), \ldots, \ell(k)\}$ and $1_{\mathcal{D}_\ell}$ otherwise.

\subsection{Free cumulants of paired tensors}
\label{sec:free-cumulants-paired}

We are now able to introduce free cumulants associated to connected melonic graphs of paired tensors. In Sec.~\ref{sub:paired-tensors}, we detailed how trace-invariants of paired tensors can be encoded by graphs  $\mathsf{g}$  or cycles $\{\gamma_{c,b}\}$, each of the form $\bigl((k_1, r_1)\cdots   (k_{q_{c,b}}, r_{q_{c,b}})\bigr)$ where $1\le k_1, \ldots k_q\le n$ are the labels of the thick edges encountered in the cycle, and $r_s$ is the shade of the pair of output and input of color $c$ of the paired tensor number $k_s$ which is in the cycle.

Consider any paired-tensors $H_1, \ldots H_q$,  $H_\ell \in \mathcal{G}^\mathrm{m}_{D}[S_\ell]$, where $S_\ell\subset \{A,B\ldots\}$, or $H_\ell \in \mathcal{G}^\mathrm{p}_{D}[\Theta_\ell, \bar \Theta_\ell]$, where $\Theta_\ell\subset \{T_a,T_b\ldots\}$ and  $\bar \Theta_\ell\subset \{\bar T_a,\bar T_b\ldots\}$,   and a trace-invariant of these $q$ paired tensors with colored graph $\mathsf{g}$, encoded by some such cycles $\{\gamma_{c,b}\}$.   We may independently consider some non-crossing \emph{permutations} $\{\tau_{c,b}\}$ with $\tau_{c,b}\preceq \gamma_{c,b}$ in $S_{\mathrm{NC}}(\gamma_{c,b})$, which is isomorphic to $\mathrm{NC}(q_{c,b})$, where $q_{c,b}$ is the number of elements in the cycle $\gamma_{c,b}$. 
We use the notation $\{ \mathsf{h}\mid \mathsf{h} \preceq \mathsf{g} \}$ to refer to the set of  $\mathsf{h}$ built from all the $\tau_{c,b}\preceq \gamma_{c,b}$, and the notation $\mathsf{M}(\mathsf{h}, \mathsf{g})=\prod_{c,b} \mathsf{M}(\tau_{c,b}\gamma_{c,b}^{-1})$.

We may therefore define free cumulants for connected melonic graphs $\mathsf{g}$ of  $q$ paired tensors for any $\vec h =(h_1, \ldots h_q)$ as:
\be
\label{eq:paired-free-cumulant-moment}
\varkappa_{\mathsf{g}} (\vec h)= \sum_{\mathsf{h} \preceq \mathsf{g}} \phi_{\Pi(\mathsf{h}), \mathsf{h}}(\vec h) \; \mathsf{M}(\mathsf{h}, \mathsf{g}),
\ee
where as above, $\phi_{\Pi(\mathsf{h}), \mathsf{h}}(\vec h)$ is the product of the $\phi_{\mathsf{h}_j}$ over the connected components $\mathsf{h}_j$ of $\mathsf{h}$. By multilinearity of the $\phi$, the $\varkappa$ are multilinear.

To understand this formula better, one can use \eqref{eq:asympt-of-melo-paired-vs-first-order-regular} and Lemma~\ref{lem:melopaired-vs-first-order}: we may gather the colored edges of $\mathsf{g}$ together with the colored edges of each $H_\ell$, to form a colored graph of $M_1, \ldots M_n$,   namely, a first order invariant $\bsig\in S_n^D$, such that $\Tr_\mathsf{g}(\vec H)=\Tr_\bsig(\vec M)$. In the mixed case, $M_i\in \{A, B \ldots\}$, while in the pure case, $M_i=X_i\otimes X_{\bar i}'$ with $X_i\in \{ T_a,  T_b \ldots\}$ and $X_{\bar i}'\in  \{\bar T_a, \bar T_b \ldots\}$ (so $\Tr_\bsig(\vec M)=\Tr_\bsig(\vec X, \vec X')$).
\emph{To each $\mathsf{h}\preceq \mathsf g$ corresponds a $\brho\in S_n^D$ which satisfies $\brho \eta^{-1}\preceq \bsig\eta^{-1}$, where $\eta\in S_n$ is the canonical pairing of $\bsig$ (pure) or $(\bsig, \mathrm{id})$ (mixed), but not any such $\brho$ may occur in the sum, as the edges internal to the $H_\ell$ must be the same in $\brho$ and $\bsig$}, as will be detailed in the proof of Prop.~\ref{prop:cumul-of-paired-vs-non-paired}, see \eqref{eq:kappa-paired-vs-not-constrained}. One has:
\be
\label{eq:varphi-paired-vs-not-for-h}
\phi_{\Pi(\mathsf{h}), \mathsf{h}}(\vec h) =  \varphi^\mathrm{m}_{\Pi(\brho), \brho}(\vec m)  \quad \textrm{(mixed)},\qquad \textrm{and}\qquad  \phi_{\Pi(\mathsf{h}), \mathsf{h}}(\vec h) =  \varphi_{\Pi_\mathrm{p}(\brho), \brho}(\vec x, \vec{x'}) \quad \textrm{(pure)},
 \ee 
where  in the pure case $\Pi(\brho)=\Pi_\mathrm{p}(\brho)$ due to the fact that  the connected components of $\brho$ are purely connected (since in that case the canonical pairing is the identity). 

If $\mathsf{g}$ is not necessarily connected, we set $\varkappa_{\Pi(\mathsf g), \mathsf{g}} (\vec h)$ to be the  product of the $\varkappa_{\mathsf{g}_i}$ over the connected components $\mathsf{g}_i$. We can invert the formula in the lattice $\bigtimes_{c, b} S_{\mathrm{NC}}(\gamma_{c,b})\cong \bigtimes_{c, b} \mathrm{NC}(q_{c,b})$, to obtain: 
\be
\label{eq:paired-free-moment-cumulant}
\phi_{\mathsf{g}} (\vec h)= \sum_{\mathsf{h} \preceq \mathsf{g}} \varkappa_{\Pi(\mathsf{h}), \mathsf{h}}(\vec h). 
\ee

The relation between $\varkappa$ and $\kappa$ is more complicated than the relation between $\phi$ and $\varphi$  \eqref{eq:varphi-paired-vs-not-for-h}: see \eqref{eq:kappa-paired-vs-not-for-h} later in this section.
With the same notations as above, the following proposition, analogous to Prop.~11.15 of \cite{NicaSpeicher}. is proved in Sec.~\ref{sub:proof-of-sec-cum-paired}.
\begin{proposition}
\label{prop:cumul-of-1}
Let $\mathsf{g}$ be a connected melonic graph of  $q\ge 2$ paired tensors $H_1, \ldots H_q$, $\mathcal{D}_\ell$ be the number of inputs of $H_\ell$,  then one has 
$
\varkappa_{\mathsf{g}} (\vec h) =0
$
if one of the $h_\ell$ is $1_{\mathcal{D}_\ell}$ (i.e.~if $H_\ell=\un_{\mathcal{D}_\ell}$). 
\end{proposition}

This of course applies as well to  $\kappa_\bsig(\vec x, \vec{x'})$ and  $\kappa^\mathrm{m}_\bsig(\vec m)$. The second result we will need relates free cumulants of paired tensors $\vec H$ and those obtained by grouping them to obtain bigger paired tensors $\vec P$ as in Sec.~\ref{sub:grouping}.
See the proof in Sec.~\ref{sub:proof-of-sec-cum-paired}. It generalizes Thm.~11.12 of \cite{NicaSpeicher}.

\begin{proposition}
\label{prop:cumul-of-paired-vs-non-paired}
Let $\mathsf{g}$ be a connected melonic graph of  $q$ paired tensors $H_1, \ldots H_q$, $H_\ell \in \mathcal{G}^\mathrm{m}_{D}[S_\ell]$, where $S_\ell\subset \{A,B\ldots\}$, or $H_\ell \in \mathcal{G}^\mathrm{p}_{D}[\Theta_\ell, \bar \Theta_\ell]$, where $\Theta_\ell\subset \{T_a,T_b\ldots\}$ and  $\bar \Theta_\ell\subset \{\bar T_a,\bar T_b\ldots\}$,  $E$ a set of edges of $\mathsf{g}$, the $P_\jmath$ obtained as connected components of $\Tr_{\mathsf{g}_{\setminus E}}(\vec H)$, 
the $\mathsf{h}_\jmath$ such that 
$\Tr_{\mathsf{h}_\jmath}(\vec H_\jmath) =  \Tr(P_\jmath)$ \eqref{eq:trace-of-one-paired-vs-trace-invariant}, and which constitute the connected components of $\mathsf{h}$, which is a melonic invariant of $\vec H$, and finally, $\mathsf{k}$ such that  $\Tr_\mathsf{g}(\vec H) = \Tr_{\mathsf{k}}(\vec P)$  \eqref{eq:trace-of-n-paired-vs-trace-invariant}. Then:
$$
\varkappa_{\mathsf{k}} (\vec p) = \sum_{\substack{{\mathsf{h}' \preceq \mathsf{g}}\\{\Pi(\mathsf{h'})\vee\Pi(\mathsf{h}) = 1_q}}} \varkappa_{\Pi(\mathsf{h}'),\mathsf{h}'} (\vec h) ,
$$
where  $\Pi(\mathsf{g})$,  $\Pi(\mathsf{h})$ are the partitions of $\{1, \ldots, q\}$ induced by the connected components of   $\mathsf{g}$, $\mathsf{h}$. 
\end{proposition}

We may apply the proposition to a first order $\bsig\in S_n^D$ instead of the more general $\mathsf{g}$: in the mixed case, if $\eta\in S_n$ is the  canonical pairing  of $(\bsig, \mathrm{id})$, $\mathsf{k}$ is a connected melonic graph of  $q$ paired tensors $P_1, \ldots P_q$ generated by $A,B\ldots$,   $\vec M$ is such that $\Tr_\mathsf{k}(\vec P)=\Tr_\bsig(\vec M)$ \eqref{eq:paired-invariant-vs-usual-invariant}, $\btau_\ell$ with support $R_\ell\subset \{1,\ldots n\}$ is such that $\Tr(P_\ell) = \Tr_{\btau_\ell}(\vec M_{\lvert B_\ell})$,  and $\btau\in S_n^D$ has connected components $\{\btau_\ell\}$:
\be
\label{eq:kappa-paired-vs-not-for-h}
\varkappa_{\mathsf{k}} (\vec p) = \sum_{\substack{{\brho\eta^{-1}\preceq \bsig\eta^{-1}}\\{\Pi(\brho\eta^{-1})\vee\Pi(\btau\eta^{-1}) = 1_n}}} \kappa^\mathrm{m}_{\Pi(\brho),\brho} (\vec m). 
\ee

In the pure case,  the $P_1, \ldots P_q$ are generated by  $T_a, T_b \ldots$ and $\bar T_a, \bar T_b \ldots$,  one has $M_i=X_i\otimes X_{\bar i}'$ and  $\Tr_\mathsf{k}(\vec P)=\Tr_\bsig(\vec X, \vec X')$, $\btau_\ell$ is such that $\Tr(P_\ell) =\Tr_{\btau_\ell}\bigl(\overrightarrow{X\otimes X'}_{\lvert B_\ell}\bigr)$, and:
\be
\label{eq:kappa-paired-vs-not-for-h-pure}
\varkappa_{\mathsf{k}} (\vec p) = \sum_{\substack{{\brho\preceq \bsig}\\{\Pi(\brho)\vee\Pi(\btau) = 1_n}}} \kappa_{\Pi_\mathrm{p}(\brho),\brho} (\vec x, \vec{x'}).
\ee

 The proof of Prop.~\ref{prop:cumul-of-paired-vs-non-paired} in Sec.~\ref{sub:proof-of-sec-cum-paired} relies on the useful relation \eqref{eq:restricted-sum-group-freecum}, which here simplifies to:
\begin{align}
\label{eq:kappa-paired-vs-not-constrained}
\varkappa_{\mathsf{k}} (\vec p) &= \sum_{\substack{{\brho\in S_n^D}\\{\btau\eta^{-1}\preceq\brho\eta^{-1}\preceq \bsig\eta^{-1}}}} \varphi^\mathrm{m}_{\Pi(\brho),\brho} (\vec m)& \textrm{(mixed)}\\[+0.1cm]
\varkappa_{\mathsf{k}} (\vec p) &= \sum_{\substack{{\brho\in S_n^D}\\{\btau\preceq\brho\preceq \bsig}}} \varphi_{\Pi(\brho),\brho} (\vec x, \vec{x'})&\textrm{(pure)}
\end{align}

Consider now the  situation where, in the mixed case, $\mathsf{k}$ is a connected melonic graph of  $P_1, \ldots P_q$ with $P_\ell$ of the form \eqref{eq:paired-tensors-for-general-melo}, and $\vec M$ such that $\Tr_\mathsf{k}(\vec P)=\Tr_\bsig(\vec M)$, so that $\bsig\in S_n^D$ is connected and $(\bsig, \mathrm{id})$ is melonic. \emph{Then $\btau$ coincides with $(\eta, \eta, \ldots, \eta)$}, so that in this particular situation:
\be
\label{eq:kappa-paired-vs-not-for-p-particular-melo}
\varkappa_{\mathsf{k}} (\vec p)  = \kappa^\mathrm{m}_{\bsig} (\vec m). 
\ee

\subsection{Asymptotic tensor freeness at the level of moments}
\label{sub:different-formulations-of-asymptotic-tensor-freeness}

\subsubsection{Mixed case}

We consider some mixed random tensors $A,B\ldots$ with $D\ge 2$ inputs which scale as \eqref{eq:scaling-wishart-D1}, and we consider the corresponding asymptotic moments and first order free cumulants $\varphi_\bsig(\vec m)$, $\kappa_\bsig(\vec m)$, where $\vec m \in \{a,b,\ldots\}^n$, as well as the asymptotic moments $\phi$ and cumulants $\varkappa$ associated to connected melonic invariants of paired tensors. The following is proved in Sec.~\ref{sub:proof-equiv-tensor-freeness-cumulants-moments-mixed}.

\begin{theorem}[Asymptotic mixed tensor freeness]
\label{thm:equiv-tensor-freeness-cumulants-moments-mixed}
The following statements are equivalent:
\begin{enumerate}
\item For any $n\ge 2$, any $\bsig\in S_n^D$ connected and with $\omega(\bsig, \mathrm{id})=0$, and any $\vec m=(m_1, \ldots, m_n)\in \{a,b,\ldots\}^n$, $\kappa^\mathrm{m}_\bsig(\vec m)=0$ whenever there exists  $1\le i<j\le n$ such that $m_i\neq m_j$.  
\item The two following conditions are satisfied:
\begin{enumerate}[label=$-$]
\item For any $n\ge 2$, any $\bsig\in S_n^D$ connected such that $(\bsig, \mathrm{id})$ melonic with canonical pairing $\eta\neq \mathrm{id}$, and any $\vec m=(m_1, \ldots, m_n)\in \{a,b,\ldots\}^n$, $\kappa^\mathrm{m}_\bsig(\vec m)=0$ whenever there exists $i\in \{1,\ldots n\}$ such that $m_i\neq m_{\eta(i)}$. 
\item For any $q\ge 2$, any paired tensors $H_1, \ldots H_q$ such that $\forall\; 1\le \ell \le q$, $H_\ell\in \mathcal{G}^\mathrm{m}_{D}[Q_\ell]$ with $Q_\ell\in \{A,B\ldots\}$, and any connected melonic graph $\mathsf{g}$  of the paired tensors $\vec H$, $\varkappa_\mathsf{g}(\vec h)=0$ whenever there exists $1\le \ell<\ell'\le q$ such that $Q_\ell\neq Q_{\ell'}$. 
\end{enumerate}
\item The two following conditions are satisfied:
\begin{enumerate}[label=$-$]
\item For any $n\ge 2$, any $\bsig\in S_n^D$ connected and such that $(\bsig, \mathrm{id})$ is melonic with canonical pairing $\eta\neq \mathrm{id}$, and any $\vec m=(m_1, \ldots, m_n)\in \{a,b,\ldots\}^n$, $\varphi^\mathrm{m}_\bsig(\vec m)=0$ whenever there exists  $i\in \{1,\ldots n\}$ such that $m_i\neq m_{\eta(i)}$. 
\item For any $q\ge 2$, any $\mathcal{D}_1, \ldots \mathcal{D}_q\ge 1$, any paired tensors $H_1, \ldots H_q$ such that for $1\le \ell \le q$, $H_\ell\in \mathcal{G}^\mathrm{m}_{D, \mathcal{D}_\ell} [Q_\ell]$ with $Q_\ell\in \{A,B\ldots\}$, and any connected melonic graph $\mathsf{g}$  of $H_1, \ldots H_q$,  letting for each $\ell$:  
$h'_\ell= h_\ell - \phi(h_\ell) 1_{\mathcal{D}_\ell}$, then $\phi_\mathsf{g}(\vec h')=0$ whenever $(\mathsf{g}, \vec h)$ is almost alternating. 
\end{enumerate}
\end{enumerate}
\end{theorem}

Notice the dissymmetry: the first statements of point 2 and 3 do not require the notion of paired tensors nor introducing $ \mathcal{G}^\mathrm{m}_{D}[Q_\ell]$. These notions are needed for the second statements of 2 and 3, because of the necessity  to center the $h_\ell$. For  point 3, for $A,B, C$ asymptotically free,  the example on the right of Fig.~\ref{fig:paired-invariant-ungrouping} with $H_1, \ldots H_4$ respectively   generated by $A$, $B$, $C$ and $C$  is almost alternating (Fig.~\ref{fig:paired-invariant-alternating}) and therefore vanishes when the $H_i$ are asymptotically centered.  

For points 2 and 3, one can replace $\mathcal{G}^\mathrm{m}_{D, \mathcal{D}_\ell} [Q_\ell]$ by $\mathcal{A}^\mathrm{m}_{D, \mathcal{D}_\ell} [Q_\ell; \un_D]$ etc, see Rk.~\ref{remark:replace-by-algebr} and Rk.~\ref{remark:replace-by-algebr-second-part}.

\subsubsection{Pure case}
\label{subsub:tens-freeness-pure}
We now consider some pure random tensors $T_a, \bar T_a$, $T_b, \bar T_b\ldots$ for $D\ge 3$, which scale as \eqref{eq:scaling-hypothesis-pure} and \eqref{eq:scaling-hypothesis-pure-multitensors}, and we consider the corresponding asymptotic moments and first order free cumulants $\varphi_\bsig(\vec x, \vec{x'})$, $\kappa_\bsig(\vec x, \vec{x'})$, where for each $i$, $x_i \in \{t_a, t_b  \ldots\}$ and $x_{\bar i}'\in \{ \bar t_a, \bar t_b \ldots\}$, as well as the $\phi$ and $\varkappa$.

\begin{theorem}[Asymptotic pure tensor freeness]
\label{thm:equiv-tensor-freeness-cumulants-moments-pure}
The following  statements are equivalent:
\begin{enumerate}
\item For any $n\ge 2$, any $\bsig\in S_n^D$ purely connected and melonic with canonical pairing the identity, and any $\vec x\in \{t_a,t_b\ldots\}^n$ and $\vec {x'}\in \{\bar t_a,\bar t_b\ldots\}^n$, $\kappa_\bsig(\vec x, \vec{x'})=0$ whenever  $x_i\neq x_j$ or $x_{\bar i}'\neq x_{\bar j}'$, or  $\overline{x_i}\neq x_{\bar j}'$ for some $i, j$.  
\item The two following conditions are satisfied:
\begin{enumerate}[label=$-$]
\item For any $n\ge 2$, any $\bsig\in S_n^D$ purely connected and melonic with canonical pairing the identity, and any $\vec x\in \{t_a,t_b\ldots\}^n$ and $\vec{x'}\in \{\bar t_a,\bar t_b\ldots\}^n$, $\kappa_\bsig(\vec x, \vec{x'})=0$ whenever there exists $i\in \{1,\ldots n\}$ such that $\overline{x_i}\neq x_{\bar i}'$. 
\item For any $q\ge 2$, any paired tensors $H_1, \ldots H_q$ such that $\forall\; 1\le \ell \le q$, $H_\ell\in \mathcal{G}^\mathrm{p}_{D}[X_\ell, \bar X_\ell]$ where $(X_\ell, \bar X_\ell) \in \{(T_a, \bar T_a),(T_b, \bar T_b)\ldots\}$, and any connected melonic graph $\mathsf{g}$  of the paired tensors $\vec H$, $\varkappa_\mathsf{g}(\vec h)=0$ whenever $\exists$ $1\le \ell<\ell'\le q$ such that $(X_\ell, \bar X_\ell)\neq(X_{\ell'}, \bar X_{\ell'})$. 
\end{enumerate}
\item The two following conditions are satisfied:
\begin{enumerate}[label=$-$]
\item For any $n\ge 2$, any $\bsig\in S_n^D$ purely connected and melonic with canonical pairing the identity, and any $\vec x\in \{t_a,t_b\ldots\}^n$ and $\vec{x'}\in \{\bar t_a,\bar t_b\ldots\}^n$, $\varphi_\bsig(\vec x, \vec{x'})=0$ whenever there exists $i\in \{1,\ldots n\}$ such that $\overline{x_i}\neq x_{\vec i}'$. 
\item For any $q\ge 2$, any $\mathcal{D}_1, \ldots \mathcal{D}_q\ge 1$, any $H_1, \ldots H_q$ such that $\forall$ $1\le \ell \le q$, $H_\ell\in\mathcal{G}^\mathrm{p}_{D, \mathcal{D}_\ell}   [X_\ell, \bar X_\ell]$ where $(X_\ell, \bar X_\ell) \in \{(T_a, \bar T_a),(T_b, \bar T_b)\ldots\}$, and any connected melonic graph $\mathsf{g}$  of $H_1, \ldots H_q$,  letting for each $\ell$:  
$h'_\ell= h_\ell - \phi(h_\ell) 1_{\mathcal{D}_\ell}$, then $\phi_\mathsf{g}(\vec h')=0$ whenever $(\mathsf{g}, \vec h)$ is almost alternating. 
\end{enumerate}
\end{enumerate}
\end{theorem}

The proof can be found in Sec.~\ref{sec:equiv-tensor-freeness-cumulants-moments-pure}. For points 2 and 3, one can replace $H_\ell\in \mathcal{G}^\mathrm{p}_{D, \mathcal{D}_\ell} [X_\ell, \bar X_\ell]$ by $H_\ell\in \mathcal{A}^\mathrm{m}_{D, \mathcal{D}_\ell} [X_\ell, \bar X_\ell; \un_D]$, see Rk.~\ref{remark:replace-by-algebr} and Rk.~\ref{remark:replace-by-algebr-second-part} and Thm.~\ref{thm:equiv-tensor-freeness-cumulants-moments-pure-algebr}.

For the second statement of point 3, for  $T_a, \bar T_a$, $T_b, \bar T_b$, $T_c, \bar T_c$ asymptotically free, the example on the left of Fig.~\ref{fig:paired-invariant-ungrouping}, with $H_1$, $H_2$, $H_3$, $H_4$   respectively generated by  $T_a, \bar T_a$, $T_b, \bar T_b$, $T_c, \bar T_c$ and $T_c, \bar T_c$, vanishes when the $H_i$ are asymptotically  centered.

\newpage
\section{Limiting spaces and tensor freeness}
\label{sec:limiting-spaces}

Here we provide an answer to the question: can we give a more precise meaning to the letters $a,b,h, t, \bar t$ used in the asymptotic moments and free cumulants $\varphi^\mathrm{m}_\bsig(\vec a)$, $\varphi_\bsig(t, \bar t)$, $\phi_\mathsf{g}(\vec h)$, etc, in the same was as for matrices they can be seen as non-commutative random variables (Sec.~\ref{sub:matrix-freeness}). 
Our approach is to introduce spaces whose elements are constructed algebraically from a set of generators in the same way as the paired tensors that are centered in the formulation of asymptotic tensor freeness (Sec.~\ref{sub:different-formulations-of-asymptotic-tensor-freeness}), so that asymptotic tensor freeness of the generators $A,B,C\ldots$ can be reformulated as tensor freeness of the algebraic spaces generated by $a,b,c\ldots$ The trace of these elements - which define the moments - correspond to melonic graphs in the generators, and therefore the random tensors considered converge in distribution to such elements.

\subsection{Tensorial probability spaces}
\label{sub:general-construction-algebra}

\paragraph{General idea.}Elements $x,y$ of a non-commutative algebra can be multiplied on the left or right: $xy$ and $yx$  are a priori different. In this sense they have one input (right) and one output (left), and if the input of $x$ is multiplied with the output of $y$, then $x$ has no more available input, and $xy$ still has an input (that of $y$, by associativity) and an output (that of $x$).

 The main objects in the present case are elements $x,y$, which can be multiplied  on the left and right in a number of ways. They have a number of inputs and outputs of different kinds, and $x$ and $y$ can be multiplied on the left or right if they have an input and output of the same kind. After multiplication, it may still be possible to multiply $x$ on the right in $xy$ by some $z$, by using other inputs of $x$, different from the one  used to multiply $x$ with $y$. 
 
 Starting from a set of generators, we will build by multiplication some elements with more inputs and outputs. The space is organized in sets of  elements with the same number of  inputs and outputs of the same kind. Some special elements generalize the role of the identity. A unital linear functional plays the same role as for non-commutative probability spaces.

 \paragraph{Elements.}For $D\ge1$ and $\vec k=(k_1, \ldots k_D)$ where $k_c\in \mathbb{N}$ and $\sum_c k_c >0$, we consider complex vector spaces $\mathcal{A}_{D, \vec k}$ (for the addition) of elements which have $k_c$ inputs of color $c$ on the right $k_c$ outputs of color $c$ on the left,  which are distinguishable and paired: if $x\in\mathcal{A}_{D, \vec k}$  and $k_c\ge 1$, we respectively label the different inputs and outputs of color $c$ of $x$ from $1$ to $k_c$. Inputs and outputs with the same label have the same shade. We define the graded space $\mathcal{A}_{D} = \bigcup_{\vec k}\mathcal{A}_{D, \vec k}$.

  \paragraph{Multiplication.}If $c\in\{1,\ldots D\}$, $x\in \mathcal{A}_{D, \vec k^{(1)}}$ and $y\in \mathcal{A}_{D, \vec k ^{(2)}}$, and $s^{(1)},s^{(2)}\in \mathbb{N}^\star$ with  $k_c^{(1)}\ge s^{(1)}$ and $k_c^{(2)}\ge s^{(2)}$, we may consider the product: 
 \be
 \label{eq:multiplication-xy}
 x\cdot_{(c\;;\;s^{(1)},s^{(2)})}y\  \in\  \mathcal{A}_{D, \vec k} \quad \textrm{ where }\ \left\{ \begin{array}{l}
    k_c=k_c^{(1)} + k_c^{(2)} - 1     \\
   k_{c'}=k_{c'}^{(1)} + k^{(2)}_{c'} \qquad  \textrm{ if }   c'\neq c
    \end{array}
    \right. ,
 \ee
obtained by multiplying the input of $x$ of color $c$ and shade  $s^{(1)}$,  with the output of $y$ of color $c$ and shade $s^{(2)}$. Inputs and outputs of different colors cannot be multiplied, and any inputs and outputs of the same color may be multiplied. If $k_c>0$ in \eqref{eq:multiplication-xy}, the result $x\cdot_{(c\;;\;s^{(1)},s^{(2)})}y$ itself has distinguishable inputs and outputs of color $c$,  which are labeled from 1 to $k_c$. The output $(c,s^{(1)})$ of $x$ and the input $(c,s^{(2)})$ of $y$ are paired in $x\cdot_{(c\;;\;s^{(1)},s^{(2)})}y$, and the input and outputs of $x,y$ that are not involved in the multiplication remain paired after the multiplication. A difference with non-commutative algebras is that if $x\cdot_{(c\;;\;s^{(1)},s^{(2)})}y$ is multiplied on the right by another element $z\in \mathcal{A}_{D, \vec k^{(3)}}$, then the input of $x\cdot_{(c\;;\;s^{(1)},s^{(2)})}y$ multiplied with the  output of $z$ may correspond (by associativity) either to an input of $y$ \emph{or to an input of} $x$, and similarly for multiplication from the  left.

  \paragraph{Graphical representation.}An expression representing the multiplication of a number $n$ of elements $x,y,z\ldots$ gets rapidly messy when $n$ grows whereas it is easily represented graphically. We use the same representation as for paired tensors: an element $x\in \mathcal{A}_{D, \vec k}$ is represented by a thick edge linking a white vertex (left) and a black vertex (right). We represent the inputs and outputs by some half-edges, respectively connected to the black vertex and the white vertices, and which are labeled by a pair $(c,s)$, where $1\le c\le D$  is the color and $1\le s\le k_c$ is the shade (if $k_c=0$ no half-edge is represented).  For each color, around the white and black vertices incident to a given thick edge, the shade is represented as growing from top to bottom (Fig.~\ref{fig:multiplication-pure}). Connecting the half-edge representing the input $(c,s^{(1)})$ of $x$ with the half-edge representing the output $(c,s^{(2)})$ of $y$ to form an edge labeled $(c\;;s^{(1)}, s^{(2)})$ represents the multiplication  $x\cdot_{(c\;;\;s^{(1)},s^{(2)})}y$, see Fig.~\ref{fig:multiplication-pure}.

\begin{figure}[!h]
\centering
\includegraphics[scale=1.3]{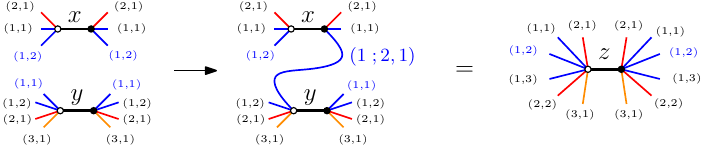}
\caption{Here $x\in \mathcal{A}_{3, \vec k^{(1)}}$ and $y\in \mathcal{A}_{3, \vec k ^{(2)}}$ with $\vec k^{(1)} =(2,1,0)$ and $\vec k^{(2)} =(2,1,1)$. Performing the multiplication  $x\cdot_{(1\;;\;1,2)}y$, we obtain a new element $z\in \mathcal{A}_{3, \vec k}$ with $\vec k=(3,2,1)$.  The labels of the edges and half-edges involved in the multiplication are in blue. For a given color, the shade grows from top to bottom around each vertex. The half-edges are relabeled on the right, but the vertical ordering of the half-edges is coherent for the three  graphs (so for instance the half-edges $(1,2)$ of $y$ are relabeled $(1,3)$ in $z$). }
\label{fig:multiplication-pure}
\end{figure}

The result of multiplying a number of elements is a \emph{tree}, see Fig.~\ref{fig:multiplication-tree-pure}. In such a tree $\mathcal{T}$ that represents an element $z$, consider an half-edge  $e$ corresponding to an output $(c,s_1)$ of an element $x_1$ involved in the multiplication.  It is the extremity of a path in $\mathcal{T}$ which follows the thick-edge $x_1$ attached to $e$, then either an half-edge $(c,s_1)$, or an edge labeled $(c\;;s_1,s_2)$, then a thick edge $x_2$, and so on, until finally an half-edge labeled $(c,s_r)$  is met for some $r\ge 1$. If now instead of seing these two half-edges as an output of $x_1$ and an input of $x_r$ respectively, we see them as an input and output of $z$, then they are paired and will carry the same shade.  In Fig.~\ref{fig:multiplication-tree-pure}, the labels are those of the half-edges of $z$ and not of the elements $x,y,w,u,v$ (the shades of these elements is encoded in the top to bottom ordering of half-edges around the thick edges).

\begin{figure}[!h]
\centering
\includegraphics[scale=1.3]{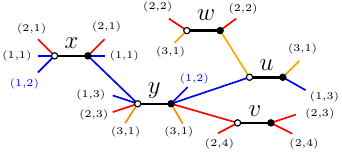}
\caption{Example of multiplication of some elements $x,y,w,u,v$.  We keep the convention of Fig.~\ref{fig:multiplication-pure}: for each thick edge, the shade of half-edges of the same color grows from top to bottom, allowing the identification of the shade locally for each thick edge. The labels represented are those of the element $z$ resulting from the multiplication (and not those of the elements $x,y,w,u,v$ themselves): the indices which are paired are the extremities of paths that alternate colored edges and paired inputs and outputs.   
}
\label{fig:multiplication-tree-pure}
\end{figure}

Associativity in this context is the fact that the result of multiplying several elements does not depend on the order in which the elements are multiplied, as long as it leads to the same tree (Fig.~\ref{fig:multiplication-tree-pure}). 

  \paragraph{Generated spaces.}Consider a set $\mathsf{s}=\{a,b,c\ldots\}$ of elements (called generators), which may be infinite (see the mixed case in Sec.~\ref{sub:generators-mixed}). We let $\mathcal{G}_{D, \vec k}[\mathsf{s}]$ be the set of  elements constructed by multiplication of a finite number of elements of $\mathsf{s}$ and which have $k_c$ inputs of color $c$.   We then let $\mathcal{A}_{D, \vec k}[\mathsf{s}]$ be the set of linear combinations of elements of $\mathcal{G}_{D, \vec k}[\mathsf{s}]$ with complex coefficients, and $\mathcal{A}_{D} [\mathsf{s}]= \bigcup_{\vec k}\mathcal{A}_{D, \vec k} [\mathsf{s}] $ (many $\vec k$ correspond to an empty set). The generators will differ in the pure (Sec.~\ref{sub:generators-pure}) and mixed (Sec.~\ref{sub:generators-mixed}) cases.

        \paragraph{The case $D=1$.}For $D=1$,  $\mathcal{A}_{1,1}$ (i.e.~the vector $\vec k$ only has a single component which is equal to one) is a non-commutative algebra. The graph representing the multiplication $xyz\ldots$   of a number of some elements $x,y,z\cdots$  is a line starting from an half-edge of color 1 representing the output of $x$, a white vertex, a thick edge, a black vertex representing the input of $x$, an edge of color 1 representing the multiplication of $x$ and $y$, a white vertex representing the output of $y$, and so on. Then $\mathcal{G}_{1, 1}[\mathsf{s}]$ corresponds to the set of words in the generators, and $\mathcal{A}_{1, 1} [\mathsf{s}]$ is the algebra generated by $\mathsf{s}$.

        \paragraph{Identities.} For any $D\ge1$ and any $\vec k$, one can define a particular element $1_{\vec k}\in \mathcal{A}_{D, \vec k}$ such that for any   $x\in \mathcal{A}_{D, \vec k'}$:
      \be
      \label{eq:mult-id-alg}
      1_{\vec k} \cdot_{(c_1\;;\;s_1,s'_1)}x = x \cdot_{(c_2\;;\;s_2,s'_2)}1_{\vec k}  
      \ee
      gives the same result regardless of the values of $c_1, s_1, s'_1, c_2, s_2,s'_2$,   whenever the products are well defined\footnote{\label{footnote:shadewise-tensor-product}The result differs from $x$ (and it is not  an element of  $\mathcal{A}_{D, \vec k'}$), but  it coincides with $x\otimes 1_{\vec k''}$ for the appropriate notion of shadewise tensor product, where $k_{\tilde c}'' = k_{\tilde c}' $ for $\tilde c \neq c$ and $k_{c}'' = k_{c}' -1$. Similarly,  one may then argue that for any $\vec k$, $1_{\vec k}= 1^{\otimes \sum_c k_c}$, where $1$ is the usual identity of an algebra.}.
 Any multiplication of a number of elements $1_{\vec k^{(1)}}, \ldots 1_{\vec k^{(n)}}$ gives the same elements $1_{\vec k}$ in the destination space. Choosing $ \mathsf{s}=\{1_{\vec k_0}\}$, then for any $\vec k$ which corresponds to a tree representing a multiplication of copies of $1_{\vec k_0}$ (Fig.~\ref{fig:multiplication-tree-pure}), $\mathcal{G}_{D, \vec k} [1_{\vec k_0}]=\{1_{\vec k}\}$ and  $\mathcal{A}_{D, \vec k} [1_D]=\mathbb{C}1_{\vec k}$,  and so on.

   \paragraph{Trace.}The trace $\phi$ is a linear functional from $\mathcal{A}_{D}$  to  $\mathbb{C}$, which satisfies the following property for $x\in \mathcal{A}_{D, \vec k^{(1)}}$ and $y\in \mathcal{A}_{D, \vec k ^{(2)}}$:
          \be
          \label{eq:trace-property}
          \phi\Bigl( x\cdot_{(c\;;\;s^{(1)},s^{(2)})}y\Bigr) =  \phi\Bigl( y\cdot_{(c\;;\;s^{(2)},s^{(1)})}x\Bigr) , 
          \ee
whenever the parameters $\vec k^{(1)}, \vec k^{(2)}, c, s^{(1)},s^{(2)}$ are such that one of the products is well defined. 

The trace acts as follows with respect to the elements $1_{\vec k}$, 
         \begin{align}
         \label{eq:identity-alg}
         \nonumber
          \phi\Bigl( 1_{\vec k} \cdot_{(c^L\;;\;l,l')}x_1 \cdot_{(c_1^L\;;\;l_1,l'_1)}x_2\cdots \cdot_{(c_p^L\;;\;l_p,l'_p)}x_p  \Bigr) & =  \phi\Bigl( x_1 \cdot_{(c_1^R\;;\;r_1,r'_1)}x_2 \cdots \cdot_{(c_p^R\;;\;l_p,l'_p)} x_p\cdot_{(c^R\;;\;r,r')}  1_{\vec k}  \Bigr)\\ 
          &=\phi(x_1)\phi(x_2)\cdots \phi(x_p), 
          \end{align}
whenever the values of the colors and shades involved in the multiplication are such that all elements $x_1,\ldots x_p$ are multiplied with $1_{\vec k}$. Graphically, it means that the thick edge $e$ corresponding to $1_{\vec k}$ is connected to the rest of the graph only through its inputs (left of \eqref{eq:identity-alg}) or only through its outputs (upper right in \eqref{eq:identity-alg}), and every $x_i$ is connected to  $e$ via an edge, see Fig.~\ref{fig:property-identity}.
   \begin{figure}[!h]
\centering
\includegraphics[scale=1.2]{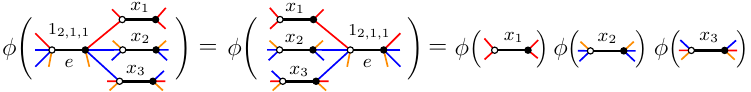}
\caption{Graphical representation of the property \eqref{eq:identity-alg}.
}
\label{fig:property-identity}
\end{figure}       

In particular, for any   $x\in \mathcal{A}_{D, \vec k'}$:
         \be
          \phi\Bigl( 1_{\vec k} \cdot_{(c_1\;;\;s_1,s'_1)}x\Bigr) =  \phi\Bigl( x \cdot_{(c_2\;;\;s_2,s'_2)}1_{\vec k}  \Bigr)=\phi(x) , 
          \ee
          regardless of the values of $c_1, s_1, s'_1, c_2, s_2,s'_2$,   whenever it is well defined, and for any $\vec k$: 
          \be
          \label{eq:phi-of-identity}
          \phi( 1_{\vec k}) =  1.
          \ee

    \paragraph{The trace and melonic graphs.}
For non-commutative probability spaces, \eqref{eq:trace-property} corresponds to the property $\phi(xy)=\phi(yx)$: 
\begin{enumerate}
\item For any word of elements, one can take the rightmost one and multiply it on the left instead of the right and this will give the same result: the trace is cyclic. 
\item Equivalently, one can see $\phi$ as adding the missing multiplication between the input and output of the word (the missing edge in the line graph representing the word), and $\phi$ applied to  any word obtained by removing an edge from this graph gives the same  value. 
\end{enumerate}

These two points of view generalize in the present case as follows: consider a tree $\mathcal{T}$ as in Fig.~\ref{fig:multiplication-tree-pure}, which represents an element $z$ obtained multiplying $q$ other elements $h_1, \ldots h_q$. 
\begin{enumerate}
\item  Take any element $x$ formed multiplying some of the $h_i$ and which in $z$ has a single multiplied  input  (in $\mathcal{T}$, the corresponding subtree is incident to a single edge $e$ labeled $(c;s_1,s_2)$, the rest being half-edges).
Consider the half-edge which is paired in $z$ with the output $(c;s_1)$ of $x$: it corresponds to the input $(c,s_3)$ of some element $y$ involved in the multiplication. The element $z'$ obtained by splitting $e$ open an reconnecting $x$ to the graph by joining its output $(c;s_1)$ with the input $(c,s_3)$ of  $y$ is such that $\phi(z)=\phi(z')$. See Fig.~\ref{fig:local-operations-same-phi}.

\begin{figure}[!h]
\centering
\includegraphics[scale=1.15]{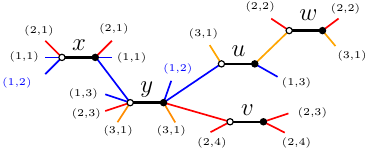}\hspace{0.5cm}\includegraphics[scale=1.15]{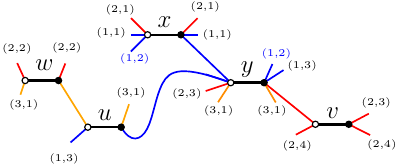}
\caption{   On the left, an element $z'$ obtained from $z$ of Fig.~\ref{fig:multiplication-tree-pure} performing the operation \eqref{eq:trace-property} for $w$ versus the rest, and on the right $z''$ obtained from $z$ through this operation, this time for the subtree containing $u,w$ versus the subtree containing $x, y, v$. From \eqref{eq:trace-property}, $\phi(z)=\phi(z')=\phi(z'')$.}
\label{fig:local-operations-same-phi}
\end{figure}

\item Equivalently, one can see $\phi$ as adding the missing multiplications  between the paired inputs and outputs of $z$. Graphically,  each pair of half-edges (the extremities of paths that alternate colored edges and paired inputs and outputs) is replaced by an edge. The result does not have any pending half-edges and \emph{it is a melonic graph $\mathsf{g}$ of paired tensors} (Sec.~\ref{sub:paired-tensors}), see Fig.~\ref{fig:melonic-same-phi}. 
 One may therefore express   $\phi(z)$ as a quantity of $h_1, \ldots h_q$, labeled by $\mathsf{g}$: 
\be
\label{eq:phi-an-melo-paired}
\phi(z) = \phi_\mathsf{g}(\vec h).
\ee 
 Reciprocally, starting from a melonic graph $\mathsf{g}$ of elements $h_1, \ldots h_q$ and removing a set of edges $E$ containing precisely one edge per cycle alternating colored edges and paired inputs and outputs, one obtains a tree $\mathsf{g}_{\setminus E}$ corresponding to an element $z_{\mathsf{g}_{\setminus E}}$ obtained by multiplication of $h_1, \ldots h_q$. 
 Doing this for all  sets $E$ satisfying this condition, one generates all the different trees obtained by doing the procedure of point 1 above in all possible ways. Applying $\phi$ to any of these elements $z_{\mathsf{g}_{\setminus E}}$ gives the same value $\phi(z_{\mathsf{g}_{\setminus E}}) = \phi_\mathsf{g}(\vec h)$. The element $z_{\mathsf{g}_{\setminus E}}$ plays a role analogous to ``the asymptotics of $\Tr_{\mathsf{g}_{\setminus E}}(\vec H)$'', see below.

\begin{figure}[!h]
\centering
\includegraphics[scale=1.3]{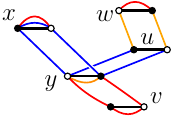}
\caption{  The three elements $z$ of Fig.~\ref{fig:multiplication-tree-pure} and $z'$, $z''$ of Fig.~\ref{fig:local-operations-same-phi} give the same melonic graph $\mathsf{g}$ of $x$, $y$, $w$, $u$ and $v$ when the missing edges are completed.  }
\label{fig:melonic-same-phi}
\end{figure}
\end{enumerate}

     \paragraph{Properties for melonic graphs.}Due to the properties of the trace $\phi$ and definition of multiplication between elements, the properties found asymptotically for the moments of paired tensors in      Sec.~\ref{sub:asympt-moments-of-paired}, Sec.~\ref{sub:centering-paired} and Sec.~\ref{sec:free-cumulants-paired} are satisfied here. The quantities $\phi(z) = \phi_\mathsf{g}(\vec h)$ of \eqref{eq:phi-an-melo-paired} are linear in each $h_i$ because $z$ is, and $\phi$ is linear. The property below \eqref{eq:mult-id-alg} should be compared with Lemma~\ref{lem:preimage-of-identity}. \eqref{eq:varphi-of-one} should then be compared here with \eqref{eq:phi-of-identity}. 
\eqref{eq:varphi-of-one-and-contraction} is satisfied here due to \eqref{eq:trace-property} and \eqref{eq:identity-alg}. 
\eqref{eq:asymptotic-relation-grouping} is also satisfied here due to the definition of multiplication of elements.  Free cumulants $\varkappa_\mathsf{g}(\vec h)$ associated to the quantities $\phi_\mathsf{g}(\vec h)$ can be defined through the formula \eqref{eq:paired-free-cumulant-moment},  and the inverse formula \eqref{eq:paired-free-moment-cumulant} holds. 
Prop.~\ref{prop:cumul-of-1} and Prop.~\ref{prop:cumul-of-paired-vs-non-paired} are also satisfied for the quantities of this section.

\paragraph{Tensorial probability space.}The generalization of a non-commutative probability space - called tensorial probability space - is then defined as the data $(\mathcal{A}_{D}, \phi)$ of a space and trace as defined above, where if $\mathcal{A}_{D}= \bigcup_{\vec k}\mathcal{A}_{D, \vec k} $, each non-empty $\mathcal{A}_{D, \vec k}$ is required  to contain $1_{\vec k}$.  Elements of $\mathcal{A}_{D}$ generalize non-commutative random variables, and the moments of $x\in \mathcal{A}_{D}$  are the $\phi_\mathsf{g}(x)$ for all melonic\footnote{In the sense of melonic graphs of paired tensors, Sec.~\ref{sub:paired-tensors}. } graphs $\mathsf{g}$ in copies of $x$. From \eqref{eq:phi-an-melo-paired} they correspond to the $\phi(z)$ for $z\in \mathcal{G}_{D} [x]$: $\phi(z) = \phi_\mathsf{g}(x)$, but  there are in general different elements $z\in\mathcal{G}_{D} [x]$ which correspond to the same moment $\phi_\mathsf{g}(x)$. Similarly, the joint moments of a set of random variables $\mathsf{s}\subset \mathcal{A}_{D}$ are given by melonic graphs in copies of the random variables in $\mathsf{s}$, namely the $\phi_\mathsf{g}(\vec h)$  for all $q\ge 1$, all $\vec h \in \mathsf{s}^q$, and all melonic graphs $\mathsf{g}$ of $\vec h$ (non-labeled). They correspond to the $\phi(z)$ for $z\in \mathcal{G}_{D} [\mathsf{s}]$, again with repetition\footnote{For non-commutative probability spaces, this repetition occurs when considering more than one random variable: different words in the random variables $h_1h_2h_3\ldots h_r$ or $h_rh_1h_2\ldots h_{r-1}$ or  $h_{r-1}h_rh_1\ldots h_{r-2}$ etc  correspond with this convention to the same moment. This convention is different from that of \cite{NicaSpeicher}, where joint moments are defined as the $\phi(z)$ for words $z$ in the random variables.}.

\ 

As for non-commutative random variables, we define convergence in distribution as follows: let $(\mathcal{A}^{(N)}_{D}, \phi^{(N)})$ for $N\in \mathbb{N}$ and $(\mathcal{A}_{D}, \phi)$ be as introduced above.
Consider $x\in \mathcal{A}_{D}$ and for each $N$, $x_N\in \mathcal{A}^{(N)}_{D}$. Then $x_N$ is said to converge in distribution to $x$ (written  $x_N\substack{{\mathrm{dist}}\\{\longrightarrow}\\{\ }} x$)
if all moments converge: 
\be
\lim_{N\rightarrow \infty} \phi_\mathsf{g}(x_N) = \phi_\mathsf{g}(x), 
\ee
for any melonic graph $\mathsf{g}$ of any number of copies of $x$. In the same way, $x_N^{(1)}, \ldots x_N^{(p)}\in \mathcal{A}^{(N)}_{D}$ converge in distribution to  $x^{(1)}, \ldots x^{(p)}\in \mathcal{A}_{D}$ if all the joint moments of the former converge towards the joint moments of the latter, that is:
\be
\lim_{N\rightarrow \infty} \phi_\mathsf{g}(\vec h_N) = \phi_\mathsf{g}(\vec h), 
\ee
for all $q\ge 1$,  all $\vec h_N\in \mathsf{s}_N^q$ where  $\mathsf{s}_N=\{x_N^{(1)}, \ldots x_N^{(p)}\}$, and all melonic graph $\mathsf{g}$ of $\vec h_N$. Furthermore, in the equation above, $\vec h\in \mathsf{s}^q$ where $\mathsf{s}=\{x^{(1)}, \ldots x^{(p)}\}$ is such that if $\vec h_N=(x_N^{(i_1)}, \ldots, x_N^{(i_q)})$, where $i_1, \ldots i_q\in \{1,\ldots p\}$, then $\vec h=(x^{(i_1)}, \ldots, x^{(i_q)})$.

\subsection{Generators in the pure case}
\label{sub:generators-pure}

We now specify the generators in the pure case. We assume the existence of some sets of elements $\theta=\{t_a, t_b  \ldots\}$ and $\bar \theta=\{\bar t_a, \bar t_b \ldots\}$. Elements of $\theta$ (resp.~$\bar \theta$) can only be multiplied on the left (resp.~right): they have $D$ distinguishable outputs  (resp.~inputs) labeled by their color $1,\ldots D$. We define the set of  tensor products $\mathsf{s}[\theta,\bar \theta]=\{x \otimes x'\}_{x\in \theta, x' \in \bar \theta}$, whose elements  have a pair of inputs and outputs of color $c$ for each $1\le c\le D$. We represent graphically an element $x \otimes x'$ of $\mathsf{s}[\theta,\bar \theta]$ as a thick edge with $x$ labeling the white vertex and $x'$ labeling the black vertex. 

We consider the set  $\mathcal{G}^\mathrm{p}_{D, \vec k}[\theta, \bar \theta]=\mathcal{G}_{D, \vec k}\bigl[\mathsf{s}[\theta,\bar \theta]\bigr]$ generated using this set. Its elements correspond to trees as in Fig.~\ref{fig:multiplication-tree-generators-pure}, but for which the black and white vertices have exactly one incident edge or half-edge of each color $1\le c \le D$ (so the shade can be omitted): see Fig.~\ref{fig:multiplication-tree-generators-pure}. We also let $\mathcal{A}^\mathrm{p}_{D, \vec k}[\theta,\bar \theta]=\mathcal{A}_{D, \vec k}\bigl[\mathsf{s}[\theta,\bar \theta]\bigr]$ and $\mathcal{A}^\mathrm{p}_{D} [\theta,\bar \theta]=\mathcal{A}_{D} \bigl[\mathsf{s}[\theta,\bar \theta]\bigr]$.

\begin{figure}[!h]
\centering
\includegraphics[scale=1.3]{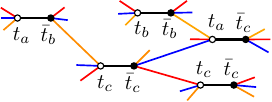}
\caption{Example of element of $\mathcal{G}^\mathrm{p}_{3, \vec k}[\theta, \bar \theta]$, with  $\vec k=(4,4,3)$,  $\theta=\{t_a, t_b,t_c\}$, and $\bar \theta=\{\bar t_a, \bar t_b,\bar t_c\}$.  }
\label{fig:multiplication-tree-generators-pure}
\end{figure}
   
An element of $\mathcal{A}^\mathrm{p}_{D, \vec k}[\theta,\bar \theta]$ built from $n$ generators has $k_c\ge 1$ for $1\le c \le D$, and a total of $\mathcal{D} =  \sum_{c=1}^D k_c = n(D-1)+1$ inputs and the same number of outputs (see \eqref{eq:mathcalD-of-generated}). All the elements in  $\mathcal{G}^\mathrm{p}_{D, \vec k}[\theta, \bar \theta]$ and in $\mathcal{A}^\mathrm{p}_{D, \vec k}[\theta, \bar \theta]$  for $\vec k$ satisfying ${\sum k_c =\mathcal{D} }$ with $\mathcal{D}$ fixed  involve the same number $n=(\mathcal{D}-1)/(D-1)$ of generators. In particular, letting $ \vec h =(1,1,1,\ldots 1)$, one has $\mathcal{G}^\mathrm{p}_{D, \vec h}[\theta, \bar \theta] = \mathsf{s}[\theta, \bar \theta]$.

      \paragraph{Generators of identities.} Let $1_D$ be the particular element $1_{1,\ldots, 1}$ satisfying the properties listed in the previous section (so $1_D$ has a single input and output of color $c$ for each $1\le c \le D)$. Choosing as only generator $ \mathsf{s}=\{1_D\}$, then for any $\vec k$ which correspond to a tree representing a multiplication of pure generators (Fig.~\ref{fig:multiplication-tree-generators-pure}), one has $\mathcal{G}_{D, \vec k} [1_D]=\{1_{\vec k}\}$, $\mathcal{A}_{D, \vec k} [1_D]=\mathbb{C}1_{\vec k}$,  etc. 
      
     Now considering the generating set $\mathsf{s}=\mathsf{s}[\theta, \bar \theta]\cup \{1_D\}$, each $\mathcal{G}^\mathrm{p}_{D, \vec k}[\theta, \bar \theta]$ is in particular supplemented with the element $1_{\vec k}$ (but also all elements generated by some copies of $1_D$ and some elements of $ \mathsf{s}[\theta, \bar \theta]$).  We let  $\mathcal{G}^\mathrm{p}_{D, \vec k}[\theta, \bar \theta;1_D]=\mathcal{G}_{D, \vec k}[\mathsf{s}]$ and similarly for  $\mathcal{A}^\mathrm{p}_{D, \vec k}[\theta,\bar \theta;1_D]$ and $\mathcal{A}^\mathrm{p}_{D} [\theta,\bar \theta;1_D]$.

        \begin{remark}
        \label{rk:color-0-pure}In this construction the $t$ and $\bar t$ are not the generators themselves, as they do not belong to the generated space: they can be seen as ``pre-generators''. One could instead add a color $0$ which would play a special role, and consider directly the space generated by multiplication of the $t$ and $\bar t$ themselves,  but then the trace would only be defined on the subspace for which each cycle alternating colored edges and thick edges has the same number of black and white vertices. The advantage of our approach is that one can treat the identity on the same level as the generators, and the constructions for the mixed and pure cases only differ by the choice of generators. 
   \end{remark}

   \paragraph{The trace and melonic graphs.}With this choice of generating set, the points 1 and 2 of the homonymous paragraph of Sec.~\ref{sub:general-construction-algebra} are modified as follows:
   \begin{enumerate}
\item The shade is now trivial and one can simply follow the paths which alternate thick edges and edges of color $c$.      
\item Here, replacing the paired half-edges by full edges, one obtains a melonic graph $\bsig$ in the usual sense\footnote{That is, a purely connected melonic graph as in Sec.~\ref{sub:Melo} whose canonical pairing is the identity, and not a more general melonic graph of paired tensors as in Sec.~\ref{sub:paired-tensors}. }, whose canonical pairing is given by the thick edges, and one may express  $\phi(z)$ in terms of the $\vec x$, $\vec {x'}$ which compose the generators $x_1\otimes x_1', \ldots, x_n\otimes x_n'$ multiplied to form $z$, as:
\be
\phi(z)  =\varphi_\bsig(\vec x, \vec {x'}).
\ee
 Reciprocally, starting from a melonic graph  $\bsig$  with canonical pairing given by the thick edges, and removing a set of edges $E$ containing precisely one edge per cycle alternating colored edges and canonical pairs, one obtains a tree $\bsig_{\setminus E}$ of the kind of Fig.~\ref{fig:multiplication-tree-generators-pure}. 
\end{enumerate}

\paragraph{Pure random tensors and convergence.}Consider the trace $\phi^{(N)} = \frac 1 N \mathbb{E}[\Tr(\cdot)]$, where $\Tr=\Tr_{\mathbf{id}_1}$ has been defined in \eqref{eq:trace-of-paired-tensor}. A pure tensor $(T,\bar T)$, together with the rescaled identity $\un_D=\frac {\un^{\otimes D}} {N^{D-1}}$, generates a tensorial probability space: 
\be
\Bigl(\mathcal{A}^\mathrm{p}_{D} \bigl[T,\bar T; \un_D\bigr], \phi^{(N)}\Bigr).
\ee 
The moments of the generator $T\otimes \bar T$ correspond to the $\phi^{(N)}(z)$ for $z\in\mathcal{G}^\mathrm{p}_{D} [T,\bar T]$, that is, the  $\frac 1 N  \Phi_\bsig(T,\bar T) = \frac 1 N \mathbb{E}[ \Tr_\bsig(T, \bar T)]$ for $\bsig$ purely connected and melonic, and the family of moments converges\footnote{Both for finite $N$ and at the limit, only one representative per equivalence classe under $\sim_\mathsf{p}$ must be considered in the family of moments. }: 
\be
\lim_{N\rightarrow \infty} \frac 1 N \mathbb{E}\bigl[ \Tr_\bsig(T, \bar T) \bigr]= \varphi_\bsig(t, \bar t), 
\ee
for all $n\ge 1$ and all $\bsig\in S_n^D$ purely connected and melonic. The $\varphi_\bsig(t, \bar t)$ can be seen as the moments of the generator $t\otimes \bar t$ of a tensorial probability space $(\mathcal{A}^\mathrm{p}_{D} [t,\bar t; 1_D], \phi)$: for $z\in\mathcal{G}^\mathrm{p}_{D} [t,\bar t]$, we set $\phi(z)=\varphi_\bsig(t, \bar t)$ if joining the paired half-edges of the tree corresponding to $z$, one obtains the graph $\bsig$. We conclude that $(T, \bar T)$ converges in distribution to $(t, \bar t)$ (in the sense that $T\otimes \bar T$ converges in distribution to $t\otimes \bar t$). This generalizes for a collection of pure random tensors and their joint convergence to a collection of pregenerators of a tensorial probability space. As an example, a complex pure Gaussian tensor converges in distribution in this sense, to a $(t, \bar t)$ with moments $\varphi_\bsig(t, \bar t)=1$  (Thm.~\ref{thm:melonic-gaussian}).

\paragraph{Tensor freeness in the pure case.}As specified in Sec.~\ref{sub:general-construction-algebra}, given some sets $\theta= \{t_a, t_b  \ldots\}$ and $\bar \theta=\{ \bar t_a, \bar t_b \ldots\}$ and the tensorial probability space $(\mathcal{A}^\mathrm{p}_{D}[\theta,\bar \theta;1_D], \phi)$ that they generate together with $1_D$, all the quantities $\varphi_\bsig$, $\phi$, $\phi_\mathsf{g}$,  $\kappa_\bsig$, $\varkappa_\mathsf{g}$ obtained considering tensor asymptotics of pure tensors $T_a, \bar T_a$, $T_b, \bar T_b\ldots$ are now understood in terms of traces of elements of $(\mathcal{A}^\mathrm{p}_{D}[\theta,\bar \theta;1_D], \phi)$, and they satisfy the same properties as those derived in Sec.~\ref{sub:asympt-moments-of-paired}, Sec.~\ref{sub:centering-paired} and Sec.~\ref{sec:free-cumulants-paired}.  As a consequence, Thm.~\ref{thm:equiv-tensor-freeness-cumulants-moments-pure} can be reformulated as an equivalence of conditions regarding the quantities $\varphi_\bsig$, $\phi$, $\phi_\mathsf{g}$,  $\kappa_\bsig$, $\varkappa_\mathsf{g}$ defined  for elements of $(\mathcal{A}^\mathrm{p}_{D}[\theta,\bar \theta;1_D], \phi)$. Asymptotic tensor freeness of pure tensors $T_a, \bar T_a$, $T_b, \bar T_b\ldots$ can then be understood as tensor freeness of the limiting random variables (or the subspaces they generate). 

 We do not state again the equivalence for the case where  there exists $1\le i\le n$ such that $\overline{x_i}\neq x_{\bar i}'$, as for this case things remain unchanged.

\begin{theorem}[Pure tensor freeness]
\label{thm:equiv-tensor-freeness-cumulants-moments-pure-algebr}
The following  statements are equivalent:
\begin{enumerate}
\item For any $n\ge 2$, any $\bsig\in S_n^D$ purely connected and melonic with canonical pairing the identity, and any $\vec x\in \{t_a,t_b\ldots\}^n$ and $\vec {x'}\in \{\bar t_a,\bar t_b\ldots\}^n$, $\kappa_\bsig(\vec x, \vec{x'})=0$ whenever  $x_i\neq x_j$ or $x_{\bar i}'\neq x_{\bar j}'$, or  $\overline{x_i}\neq x_{\bar j}'$ for some $i\neq j$.  
\item  For any $q\ge 2$, any $h_1, \ldots h_q$ such that $\forall\; 1\le \ell \le q$, $h_\ell\in \mathcal{A}^\mathrm{p}_{D}   [x_\ell, \bar x_\ell; 1_D]$ where $(x_\ell, \bar x_\ell) \in \{(t_a, \bar t_a),(t_b, \bar t_b)\ldots\}$, and any connected melonic graph $\mathsf{g}$  of $\vec h$, $\varkappa_\mathsf{g}(\vec h)=0$ whenever there exists $1\le \ell<\ell'\le q$ such that $(x_\ell, \bar x_\ell)\neq(x_{\ell'}, \bar x_{\ell'})$. 
\item 
For any $q\ge 2$, any $h_1, \ldots h_q$ such that for $1\le \ell \le q$, $h_\ell\in \mathcal{A}^\mathrm{p}_{D}   [x_\ell, \bar x_\ell; 1_D]$ where $(x_\ell, \bar x_\ell) \in \{(t_a, \bar t_a),(t_b, \bar t_b)\ldots\}$, and any connected melonic graph $\mathsf{g}$  of $h_1, \ldots h_q$,  $\phi_\mathsf{g}(\vec h)=0$ whenever $(\mathsf{g}, \vec h)$ is almost alternating \emph{and} for all $\ell$, $\phi(h_\ell)=0$. 
\end{enumerate}
\end{theorem}

In analogy with usual freeness, by definition, the  $\mathcal{A}^\mathrm{p}_{D}   [t_a, \bar t_a; 1_D]$, $\mathcal{A}^\mathrm{p}_{D}   [t_b, \bar t_b; 1_D] \ldots$ are said to be tensorially free if  2 and 3 are satisfied. The equivalence between 2 and 3 of this theorem should then be compared to Thm.~11.16 of \cite{NicaSpeicher}, and that between  1 and 2 with Thm.~11.20 of \cite{NicaSpeicher}. 
      
The proof of Thm.~\ref{thm:equiv-tensor-freeness-cumulants-moments-pure-algebr} is similar to that of Thm.~\ref{thm:equiv-tensor-freeness-cumulants-moments-pure}, with the modifications pointed out in  Rk.~\ref{remark:replace-by-algebr} and Rk.~\ref{remark:replace-by-algebr-second-part}.  There is a formulation analogous to Thm.~\ref{thm:equiv-tensor-freeness-cumulants-moments-pure}, which involves the $\mathcal{G}^\mathrm{p}_{D}   [x_\ell, \bar x_\ell]$ instead of the $\mathcal{A}^\mathrm{p}_{D}   [x_\ell, \bar x_\ell; 1_D]$, but we do not state it here.

\subsection{Generators in the mixed case}
\label{sub:generators-mixed}

We assume the existence of some set $\mathsf{s}_0=\{a_1,a_2,\ldots, a_p\}$ of elements (which we may also call ``pregenerators'') with  $D$ distinguishable outputs  (resp.~inputs) labeled by their color $1,\ldots D$. We define the set of ordered  tensor products $\mathsf{s}[\mathsf{s}_0]=\{a_{i_1} \otimes a_{i_2} \otimes \cdots a_{i_r}\}_{r\ge 1, i_s\in \{1,\ldots p\}}$, whose elements  have by convention  a pairing of the inputs and outputs of color $c$ for each $1\le c\le D$, defined in the same way as \eqref{eq:paired-tensors-for-general-melo}: the output of $a_{i_1} \otimes a_{i_2} \otimes \cdots a_{i_r}$ labeled $(c,s)$ corresponds to the output of color $c$ of $a_{i_s}$, and the paired input labeled  $(c,s)$ corresponds to the input of color $c$ of $a_{i_{s-1}}$. 
 We represent graphically such an element $a_{i_1} \otimes a_{i_2} \otimes \cdots a_{i_r}$ of $\mathsf{s}[a_1,a_2,\ldots, a_p]$ as a cycle consisting alternatively of thick edges labeled cyclically from 1 to $r$ and blue blobs (Fig.~\ref{fig:generators-mixed-graphical})  and with $a_{i_k}$ labeling the thick edge number $k$, for $1\le k \le r$. 
 
\begin{figure}[!h]
\centering
\includegraphics[scale=1.2]{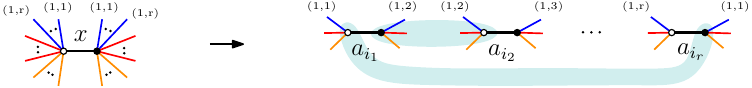}
\caption{Graphical representation of the generators in the mixed case.}
\label{fig:generators-mixed-graphical}
\end{figure}

We let $\mathcal{G}^\mathrm{m}_{D, \vec k}[\mathsf{s}_0]=\mathcal{G}_{D, \vec k}\bigl[\mathsf{s}[\mathsf{s}_0]\bigr]$ generated from the set $\mathsf{s}[\mathsf{s}_0]$. Its elements correspond to trees as in Fig.~\ref{fig:multiplication-tree-pure}, but for which the black and white vertices incident to a given thick edge have the same number of  incident edges or half-edges of each color $1\le c \le D$: see Fig.~\ref{fig:multiplication-tree-generators-mixed}. Each thick edge may be replaced by a cycle as in Fig.~\ref{fig:generators-mixed-graphical}. We define $\mathcal{A}^\mathrm{m}_{D, \vec k}[\mathsf{s}_0]$ and $\mathcal{A}^\mathrm{m}_{D} [\mathsf{s}_0]$ in  the same way.

\begin{figure}[!h]
\centering
\includegraphics[scale=1.2]{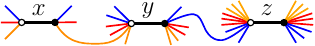}\\[+2ex]
\includegraphics[scale=1.2]{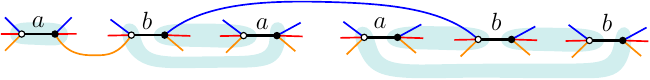}
\caption{Multiplication of generators in the mixed case: Example of element of $\mathcal{G}^\mathrm{m}_{3, \vec k}[\mathsf{s}_0]$, with  $\vec k=(5,6,5)$ and  obtained multiplying some generators $x,y,z$ for $\mathsf{s}_0=\{a,b\}$. The generators are $x=a$, $y=b\otimes a$ and $z=a\otimes b \otimes b$. For the top diagram,  around every vertex, the shade of each color grows from top to bottom. The pregenerators $a,b$ are shown explicitly on the bottom diagram, using the representation of Fig.~\ref{fig:generators-mixed-graphical}.  }
\label{fig:multiplication-tree-generators-mixed}
\end{figure}

      \paragraph{Identities.} We define $\mathcal{G}^\mathrm{m}_{D, \vec k}[\mathsf{s}_0;1_D]=\mathcal{G}_{D, \vec k}\bigl[\mathsf{s}[\mathsf{s}_0]\cup \{1_D\}\bigr]$ and similarly for  $\mathcal{A}^\mathrm{m}_{D, \vec k}[\mathsf{s}_0;1_D]$ and $\mathcal{A}^\mathrm{m}_{D} [\mathsf{s}_0;1_D]$.       Note that for every $r\ge 1$, $\mathcal{G}_{D}[1_D]$ contains $1_{r, r, \ldots r}$, which has the same number of inputs of each color as  an element  $a_{i_1} \otimes a_{i_2} \otimes \cdots a_{i_r}\in \mathsf{s}[\mathsf{s}_0]$. Note also that due to the convention for the pairing of indices, an element of $\mathsf{s}[1_D]$ formed of $r$ copies of $1_D$ differs from $1_{r,r,\ldots r}$.

        \begin{remark}In this construction, the set of generators is $\mathsf{s}[\mathsf{s}_0]$ and not $\mathsf{s}_0$ (which is only a strict subset of $\mathsf{s}[\mathsf{s}_0]$). It is not entirely satisfying to use an infinite family of generators for a finite set of pregenerators (leading also to issues for instance when considering moments, see below). As for the pure case, one could instead add a color $0$ which plays the role of the blue blobs. This allows generating the elements in $\mathsf{s}[\mathsf{s}_0]$ by multiplication of elements of $\mathsf{s}_0$ themselves, using the new color 0,  but then the same issue as for the pure case (Rk.~\ref{rk:color-0-pure}) occurs regarding the trace, and there is an additional issue due to the identities. 
   \end{remark}

   \paragraph{The trace and melonic graphs.}With this choice of generating set, point 2 of the homonymous paragraph of Sec.~\ref{sub:general-construction-algebra} is modified as follows:
   \begin{enumerate}
     \setcounter{enumi}{1}  
\item Here, if one replaces paired half-edges by full edges and replaces the thick edges by the cycles alternating thick edges and blue blobs (Figs.~\ref{fig:generators-mixed-graphical} and~\ref{fig:multiplication-tree-generators-mixed}), one obtains\footnote{Of course a labeling has to be chosen, but the quantities are ultimately considered up to relabeling.} a connected  graph $\bsig$ with  $\omega(\bsig,\mathrm{id})=0$, and the canonical pairing of $(\bsig,\mathrm{id})$ is given by the \emph{blue blobs}, see Fig.~\ref{fig:melonic-trace-mixed}. One may then express  $\phi(z)$ in terms of the $n\ge q$ elements of $\mathsf{s}_0=\{a_1, \ldots, a_p\}$ which compose the $q$ generators involved in the multiplication leading to $z$. There is a $\vec m\in \mathsf{s}_0^n$ such that:
\be
\phi(z)  =\varphi^\mathrm{m}_\bsig(\vec m).
\ee
 Reciprocally, starting from $\bsig$ connected with $(\bsig,\mathrm{id})$ melonic with  canonical pairing represented by blue blobs and removing a set of edges $E$ containing precisely one edge per cycle alternating colored edges and canonical pairs, one obtains a tree $\bsig_{\setminus E}$ of the kind of Fig.~\ref{fig:multiplication-tree-generators-mixed}. 
\end{enumerate}

\begin{figure}[!h]
\centering
\includegraphics[scale=1.2]{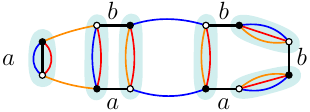}
\caption{In the mixed case, when the graphical representation of the generators is made explicit (bottom of Fig.~\ref{fig:multiplication-tree-generators-mixed}), the trace can be seen as connecting the paired half-edges  (extremities of paths alternating canonical pairs and edges of a given color)  to obtain a connected graph $\bsig$ such that $(\bsig,\mathrm{id})$ melonic, with canonical pairing given by the blue blobs. }
\label{fig:melonic-trace-mixed}
\end{figure}

\paragraph{Mixed random tensors and convergence.}Consider the trace $\phi^{(N)} = \frac 1 N \mathbb{E}[\Tr(\cdot)]$, where $\Tr=\Tr_{\mathbf{id}_1}$ has been defined in \eqref{eq:trace-of-paired-tensor}.  A mixed tensor $A$, together with the rescaled identity $\un_D=\frac {\un^{\otimes D}} {N^{D-1}}$, generates a tensorial probability space: 
\be
\Bigl(\mathcal{A}^\mathrm{m}_{D} \bigl[A; \un_D\bigr], \phi^{(N)}\Bigr).
\ee 
In the mixed case, we modify slightly the notion of moments: we define the moments of $A$ \emph{as a pregenerator} to be the family of joint moments of the elements of $\mathsf{s}[A]$, whose elements are of the form \eqref{eq:paired-tensors-for-general-melo}. They consist in the  $\phi^{(N)}(z)$ for $z\in\mathcal{G}^\mathrm{m}_{D} [A] = \mathcal{G}_{D} \bigl[ \mathsf{s}[A]\bigr]$, that is, the  $\frac 1 N \Phi^\mathrm{m}_\bsig(A) = \frac 1 N \mathbb{E}[ \Tr_\bsig(A)]$ for $\bsig$ connected satisfying $\omega(\bsig,\mathrm{id})=0$, and the family of joint moments converges\footnote{Both for finite $N$ and at the limit, only one representative per equivalence classe under $\sim_\mathsf{m}$ must be considered in the family of moments. }: 
\be
\lim_{N\rightarrow \infty} \frac 1 N \mathbb{E}\bigl[ \Tr_\bsig(A) \bigr]= \varphi^\mathrm{m}_\bsig(a), 
\ee
for all $n\ge 1$ and all $\bsig\in S_n^D$ connected satisfying $\omega(\bsig,\mathrm{id})=0$. The $\varphi^\mathrm{m}_\bsig(a)$ can be seen as the moments of the pregenerator $a$ of a tensorial probability space $(\mathcal{A}^\mathrm{m}_{D} [a; 1_D], \phi)$: for $z\in\mathcal{G}^\mathrm{m}_{D} [a]$, we set $\phi(z)=\varphi^\mathrm{m}_\bsig(a)$ if joining the paired half-edges of the tree corresponding to $z$ and making the cyclic structure of the generators explicit, one obtains the graph $\bsig$. We conclude that $\mathsf{s}[A]$ converges in distribution to $\mathsf{s}[a]$, which we call convergence in distribution of $A$ to $a$, \emph{as pregenerators}. This generalizes for a collection of mixed random tensors and their joint convergence to a collection of pregenerators of a tensorial probability space. As an example, the mixed Wishart tensor $W$ of Sec.~\ref{sub:true-Wishart-tensor} converges in distribution in this sense, to a pregenerator $w$ with moments $\varphi^\mathrm{m}_\bsig(w)=1$  \eqref{eq:wishart-higher-D-as-mum}.

\paragraph{Another point of view on non-commutative probability spaces.}For $D=1$, the trace of an element $z\in\mathcal{G}^\mathrm{m}_{1} [a]$ is of the form  $\phi(z)  =\varphi^\mathrm{m}_\sigma(a)$, where $\sigma$ is a connected cycle alternating $n$ thick edges representing copies of $a$ and $n$ edges of color 1. 
If $z$ is obtained by multiplying only copies of the generator $a$, $z\in \mathcal{G}_{1} [a]$, the blue blobs coincide with the thick edges:
\be
\label{eq:another-perspective-ncps}
\phi(a^n)=\varphi^\mathrm{m}_\gamma(a),
\ee 
where $\gamma=(12\ldots n)\in S_n$. They correspond to the moments of $a$ as an element of $(\mathcal{A}^\mathrm{m}_{1} [a; 1], \phi)$. On the other hand, if $z$ involves other generators $a^{\otimes r}$, one still gets a connected cycle alternating thick edges and edges of color 1, but with blue blobs representing a non-crossing pairing of the black and white vertices. Since one retains only one representative per graph (i.e.~per equivalence class under $\sim_\mathrm{m}$), any multiplication of generators involving a total of $n$ copies of $a$ corresponds to the same moment $\phi(a^n)=\varphi^\mathrm{m}_\gamma(a)$ \eqref{eq:another-perspective-ncps}.  This additional degeneracy is not present for $D>1$ as then the  canonical pairing of $(\bsig, \mathrm{id})$ is unique and  introducing tensor products of copies of $a$ is required because the canonical pairs can differ from the thick edges. Therefore, the moments of $a$ (both as a generator and as a pregenerator) in the tensorial probability space $(\mathcal{A}^\mathrm{m}_{1} [a; 1], \phi)$ are the same as those of the  random variable $a$ in the non-commutative probability space $(\mathcal{A} [a, 1], \phi)$.

\paragraph{Tensor freeness in the mixed case.}As specified in Sec.~\ref{sub:general-construction-algebra}, given some set $\mathsf{s}_0= \{a, b  \ldots\}$ and the tensorial probability space $(\mathcal{A}^\mathrm{m}_{D}[\mathsf{s}_0;1_D], \phi)$, all the quantities $\varphi^\mathrm{m}_\bsig$, $\phi$, $\phi_\mathsf{g}$,  $\kappa^\mathrm{m}_\bsig$, $\varkappa_\mathsf{g}$ obtained considering tensor asymptotics of mixed tensors $A,B\ldots$ are now understood in terms of traces of elements of $(\mathcal{A}^\mathrm{m}_{D}[\mathsf{s}_0;1_D], \phi)$, and they satisfy the same properties as those derived in Sec.~\ref{sub:asympt-moments-of-paired}, Sec.~\ref{sub:centering-paired} and Sec.~\ref{sec:free-cumulants-paired}.  Thm.~\ref{thm:equiv-tensor-freeness-cumulants-moments-mixed} can therefore be reformulated as an equivalence of conditions regarding the quantities $\varphi^\mathrm{m}_\bsig$, $\phi$, $\phi_\mathsf{g}$,  $\kappa^\mathrm{m}_\bsig$, $\varkappa_\mathsf{g}$ for elements of $(\mathcal{A}^\mathrm{m}_{D}[\mathsf{s}_0;1_D], \phi)$. Asymptotic tensor freeness of mixed tensors $A,B\ldots$ can then be understood as tensor freeness of the limiting random variables (or the subspaces they generate). 

 We do not state again the equivalence for the case where  there exists $1\le i\le n$ such that $m_i\neq m_{\eta(i)}$, $\eta$ being the canonical pairing of $(\bsig,\mathrm{id})$ in the first statement, as for this case things remain unchanged.

\begin{theorem}[Mixed tensor freeness]
\label{thm:equiv-tensor-freeness-cumulants-moments-mixed-algebr}
The following  statements are equivalent:
\begin{enumerate}
\item For any $n\ge 2$, any $\bsig\in S_n^D$ connected and with $\omega(\bsig, \mathrm{id})=0$, and any $\vec m=(m_1, \ldots, m_n)\in \{a,b\ldots\}^n$ such that for all $i$, $m_i=m_{\eta(i)}$, where $\eta$ is the canonical pairing of $(\bsig,\mathrm{id})$, $\kappa^\mathrm{m}_\bsig(\vec m)=0$ whenever there exists  $1\le i<j\le n$ such that $m_i\neq m_j$.  
\item  For any $q\ge 2$, any $h_1, \ldots h_q$ such that $\forall\; 1\le \ell \le q$, $h_\ell\in \mathcal{A}^\mathrm{m}_{D}[q_\ell; 1_D]$ with $q_\ell\in \{a,b\ldots\}$, and any connected melonic graph $\mathsf{g}$  of $\vec h$, $\varkappa_\mathsf{g}(\vec h)=0$ whenever $\exists$ $\ell<\ell'$ such that $q_\ell\neq q_{\ell'}$. 
\item 
For any $q\ge 2$,  any $h_1, \ldots h_q$ such that for $1\le \ell \le q$, $h_\ell\in \mathcal{A}^\mathrm{m}_{D}   [q_\ell;1_D]$ with $q_\ell\in \{a,b\ldots\}$, and any connected melonic graph $\mathsf{g}$  of $\vec h$,   $\phi_\mathsf{g}(\vec h)=0$ whenever $(\mathsf{g}, \vec h)$ is almost alternating \emph{and} for every $\ell$, $\phi(h_\ell)=0$. 
\end{enumerate}
\end{theorem}

See the comments below Thm.~\ref{thm:equiv-tensor-freeness-cumulants-moments-pure-algebr}.

\subsubsection{Composition and graph operad}

One may alternatively describe the limiting spaces using composition of diagrams instead of multiplication. The properties of composition at the limit are the same as those for the finite $N$ quantities $\Tr_\bsig$, $\Tr_\mathsf{g}$,  $\Tr_{\bsig_{\setminus E}}$, $\Tr_\mathsf{g_{\setminus E}}$, described in Sec.~\ref{sub:asympt-moments-of-paired} and Sec.~\ref{sub:centering-paired}. The construction should be analogous to the operad of graph operations described by Male in \cite{Male}, with differences such as: only the thick edges are replaces by elements from the space, the colored edges represent multiplication and instead of one input and output (a source and a sink in \cite{Male}), there are now a certain number of inputs and outputs of the same shade. We do not expect these differences to affect the  mathematical properties described in \cite{Male}.  
\newpage

\newpage

\appendix

\addtocontents{toc}{\protect\setcounter{tocdepth}{1}}

\section{Proofs of Sec.~\ref{sec:char-lu-inv-tensor-distributions}}
\label{app:proofs-of-section-Invariance}

We gather here the proofs of statements in 
Sec.~\ref{sec:char-lu-inv-tensor-distributions}.

\subsection{Proofs of Sec.~\ref{sec:basis}}
\label{app:proofs-of-section-Invariance-itself}

\paragraph{Proof of Thm.~\ref{thm:independence}:}
We start by proving the linear independence in the mixed case. We consider $n_\mathrm{max}$ and the family $\Tr_{[\bsig]_{\mathrm{m}} }$ of trace-invariants for $ [\bsig]_{\mathrm{m}}  \in S_n^D/{\sim_\mathrm{m}}$ 
with $n\le n_{\rm max}$, a set of corresponding complex numbers $\lambda_{ [\bsig]_{\mathrm{m}}  }$, and the linear combination:
\[\sum_{[\bsig]_{\mathrm{m}} \in \bigcup_{n\le N}S_n^D/{\sim_\mathrm{m}}}\lambda_{[\bsig]_{\mathrm{m}}} 
\Tr_{[\bsig]_{\mathrm{m}}} (A)=0 \; , \qquad \forall A\in M_N(\mathbb{C})^{\otimes D}\;.
\]
But then trivially, if $A$ is a complex random tensor with finite moments, this implies:
\be
\label{eq:kernel}
 \sum_{[\bsig]_{\mathrm{m} } \in \bigcup_{n\le N} S_n^D/{\sim_\mathrm{m}} }\lambda_{ [\bsig]_{\mathrm{m} } } 
 \; \mathbb{E} \bigl[ \Tr_{[\bsig]_{\mathrm{m}}} (A) \; \Tr_{[\btau]_{\mathrm{m}}} (\bar A)\bigr]  = 0  \; ,
\ee
and in order to conclude it is enough to exhibit a particular random tensor such that the covariance matrix:
\[
G= \Bigl\{ \mathbb{E} \bigl[ \Tr_{[\bsig]_{\mathrm{m}}} (A) \; \Tr_{[\btau]_{\mathrm{m}}} (\bar A)\bigr] \Bigr\}_{[\bsig]_{\mathrm{m} },  [\btau]_{\mathrm{m} } \in \bigcup_{n\le N} S_n^D/{\sim_\mathrm{m}} 
}\,, 
\]
is invertible for $N$ large enough.
Let us  take $A\in M_N(\mathbb{C})^{\otimes D}$ distributed according to the Ginibre ensemble, that is the $N^{2D}$ components $A_{i^1 \ldots i^D; j^1 \ldots j^D}$ are independent complex Gaussians with covariance $\mathbb{E}[\bar A_{\vec{i} ; \vec j}A_{\vec k ; \vec l}]=\prod_{c=1}^D \delta_{i^c, k^c}\delta_{j^c, l^c} / N^D$. The only non zero elements of the covariance matrix $G$ are obtained from the Wick theorem as:
\be
\mathbb{E} \bigl[ \Tr_{[\bsig]_{\mathrm{m}}} (A) \; \Tr_{[\btau]_{\mathrm{m}}} (\bar A)\bigr] = \sum_{\eta\,\in S_n} N^{\sum_{c=1}^D \#( \sigma_c \eta \tau_c^{-1}\eta^{-1}) - nD} \;,
\ee
where the permutation $\eta$ encodes the Wick pairings between $A$s and $\bar A$s. In terms of the distance function $d_\mathrm{m}( [ \bsig ]_{\mathrm{m}}, [ \btau]_{\mathrm{m} } )$  \eqref{eq:def-distance-orbits-mixte}. It becomes:
\[
G_{[\bsig]_{\mathrm{m}}, [\btau]_{\mathrm{m}}} = \mathbb{E} \bigl[ \Tr_{[\bsig]_{\mathrm{m}}} (A) \; \Tr_{[\btau]_{\mathrm{m}}} (\bar A)\bigr]   = C_{[\bsig]_{\mathrm{m}}, [\btau]_{\mathrm{ m }} } N^{-d_\mathrm{m}( [ \bsig ]_{\mathrm{m}}, [ \btau]_{\mathrm{m} } )}(1+O(N^{-1}))  \; , 
\]
where $C_{[\bsig]_{\mathrm{m}}, [\btau]_{\mathrm{ m }} }>0$ is the number of $\eta\in S_n$ for which the distance is attained.  Since $d_\mathrm{m}([\bsig]_{\mathrm{m}}, [ \bsig]_{\mathrm{m}})=0$, the diagonal terms are strictly positive $ C_{[\bsig]_{\mathrm{m}}, [\bsig]_{\mathrm{ m }} } >0$ and of order $1$, while the off-diagonal terms are all suppressed in $1/N$, as $d_\mathrm{m}( [ \bsig ]_{\mathrm{m}}, [ \btau]_{\mathrm{m} } ) \ge 1$ for any two different classes.  The following lemma completes the proof, as for $N$ larger enough, $G$ satisfies the inequality in the lemma.

\begin{lemma}
\label{lem:positive-mat}
Let $A=\{a_{ij}\}$ be a $N\times N$ Hermitian matrix with positive diagonal entries $a_{ii}>0$. If all off-diagonal entries satisfy $|a_{ij}|\le \sqrt{a_{ii}a_{jj}}/N$, then $A$ is positive definite and invertible
\end{lemma}

\proof Without loss of generality, we may assume that $A$ has diagonal entries equal to one (otherwise we multiply $A$ on the left and  right by the matrix $\{a_{ii}^{-1/2} \delta_{ij}\}$) and off-diagonal entries bounded as $|a_{ij}| \le 1/N, \, i\neq j$. We split $A=\un + \tilde A$, where $\tilde A$ has null diagonal and operator-norm less than $1$, as it is bounded by the Frobenius norm,
$\|\tilde A\|_{\rm op}^2 \le \Tr( \tilde A^2) < 1$. Therefore,  $\un + \tilde A $ is positive-definite and  $(\un + \tilde A )^{-1} = \sum_{k\ge 0} A^k$, where the series converges in operator norm.
\qed 

\ 

The proof is similar in the pure case, choosing $T_1$ and $T_2$ to be independent, distributed  complex Gaussians with covariances 
$\mathbb{E}[\bar T_{\vec{i}}\ T_{\vec k }]=\prod_{c=1}^D \delta_{i^c, k^c} / N^{D/2}$ 
and evaluating the covariance matrix: 
\be
H_{[\bsig]_{\mathrm{p}}, [\btau]_{\mathrm{p}}} = \mathbb E_{T_1, T_2}\bigl[\Tr_{[\bsig]_{\mathrm{p}}} ( T_2,  T_1)  \Tr_{[\btau]_{\mathrm{p}}} (\bar T_1, \bar T_2)  \bigr] = D_{[\bsig]_{\mathrm{p}}, [\btau]_{\mathrm{p}}} N^{- d_\mathrm{p}([\bsig]_{\mathrm{p}}, [\btau]_{\mathrm{p}})}(1+O(N^{-1})) , 
\ee
where $D_{[\bsig]_{\mathrm{p}}, [\btau]_{\mathrm{p}}}>0$ is the number of $\eta, \nu\in S_n$ for which the distance is attained. 
\qed

\subsection{Proof of Thm.~\ref{thm:finite-free-cumulants}}
\label{app:proofs-of-finite-free}

In the mixed case, due to the $\mathsf{LU}$-invariance of $B$, we can write $Z_A(B) $ in terms of the tensor HCIZ integral \cite{CGL, CGL2}. Denoting
with $U=U^1\otimes \cdots \otimes U^D$ and $dU = dU^1\dots dU^D$, where $dU^c$ denotes the Haar measure on $U(N)$, from the moment-cumulant relation \eqref{eq:classical-cumulant} we have:
\[
 \frac{\partial^{n}}{\partial z^n} \log \bE_A \bE_{U}[e^{z\Tr(B^T U A U^\dagger)}] \bigr \rvert_{z=0}   =  \sum_{\pi \in \mathcal{P}(n)}\lambda_\pi \, \prod_{G\in \pi}  \bE_A \left[\int dU  \; \left[  \sum_{\vec i, \vec j, \vec a, \vec b} B_{\vec i ; \vec j} U_{\vec i, \vec a} A_{\vec a;\vec b}  \overline{U_{\vec j, \vec b }} \right]^{|G|} \right] \; .
\]
The integral over the unitary group is computed from  \eqref{eq:WeinDef}:
\[
 \begin{split}
  & \frac{\partial^{n}}{\partial z^n}\log Z_A(zB) \bigr \rvert_{z=0}    \crcr
  & =  \sum_{\pi \in \mathcal{P}(n)}\lambda_\pi \prod_{G\in \pi}
  \sum_{\bsig_{|_G},\btau_{|_G} \in S^D_{|G|}} 
\sum_{i,j,a,b}\big(     \prod_{c=1}^D W^{(N)}(\sigma_{c|_G} \tau_{c|_G}^{-1} )  \; \delta_{i^c_s j^c_{ \sigma_{c|_G} (s) } } \delta_{a^c_s b^c_{\tau_{c|_G} (s)}}
 \big)  \bE    \Big[ \prod_{s\ge 1}^{|G| } A_{\vec a_s;\vec b_s}B_{\vec i_s;\vec j_s } \Big]  \crcr
 & = \sum_{\pi \in \mathcal{P}(n)}\lambda_\pi \prod_{G\in \pi}  \sum_{\bsig_{|_G},\btau_{|_G} \in S^D_{|G|}}
  \big(\prod_{c=1}^D W^{(N)} (\sigma_{ c|_G} \tau_{c|_G}^{-1}) \big) \; 
  \Tr_{\bsig_{|_G}} (B) \;   \bE \left[\Tr_{\btau_{|_G} } (A)\right]\; .
 \end{split}
\]
As $ \Tr_{\bsig}(B) $ factors over the mixed connected  components in $\Pi(\bsig)$ and the summand is invariant under relabelings, we can 
exchange the sum over partitions with the ones over permutations to obtain:
\[
\frac{\partial^{n}}{\partial z^n}\log Z_A(zB) \bigr \rvert_{z=0} = \sum_{\bsig, \btau \in S^D_{n}}  \Tr_{\bsig}(B) \sum_{\substack{{\pi \in \mathcal{P}(n)}\\{\pi \ge \Pi( \bsig, \btau)}}} \lambda_\pi  \prod_{G\in \pi}   \bE \left[\Tr_{\btau_{|_G} } (A)\right]\; \prod_{c=1}^D W^{(N)} (\sigma_{ c|_G} \tau_{c|_G}^{-1}) \; .
\]

In the pure case, similar manipulations lead to the similar formulae, but with the ``mixed'' replaced by ``pure''. In detail, we start from:
\[
\begin{split}
& \partial^n_{z}\partial^n_{\bar z}
 \log \bE_{T ,\bar T}[e^{z J\cdot T  + \bar z \bar J \cdot \bar T }]\bigr\rvert_{z=\bar z=0} \crcr
 & \qquad = 
 \sum_{\Pi \in \mathcal{P}(n,\bar n) }
\lambda_{\Pi}\prod_{G = B\cup \bar B\in \Pi} 
  \bE_{T, \bar T} \left[ \int dU \;
  \left[  \sum_{\vec i \vec a} J_{\vec i} U_{\vec i, \vec a} T_{\vec a} \right]^{|B| }
\left[  \sum_{ \vec j \vec b} \bar J_{\vec j} \overline{ U_{\vec j, \vec b} }   \bar T_{\vec b} \right]^{|\bar B|}   
  \right] \;,
\end{split}
\]
and we integrate over the unitary group.
In this context, it is convenient to regard the permutations in  \eqref{eq:WeinDef} as bijective mappings from the white elements to the black ones, that is $s\to \overline{\sigma_c}(s)$. Denoting $S_{B,\bar B }$ the set of bijections from the elements of $B$ to the ones of $\bar B$, we have:
\[
 \begin{split}
 & \partial^n_{z}\partial^n_{\bar z}
 \log \bE_{T ,\bar T}[e^{z J\cdot T  + \bar z \bar J \cdot \bar T }]\bigr\rvert_{z=\bar z=0}
 =  \sum_{\Pi \in \mathcal{P}(n,\bar n)} \lambda_\Pi \prod_{G=B\cup \bar B\in \Pi}
 \;  \sum_{\bsig_{|_B},\btau_{|_B} \in S^D_{B, \bar B}}  \crcr
& \qquad \qquad \qquad \sum_{i,j,a,b}
\big(     \prod_{c=1}^D W^{(N)}(\sigma_{c|_B} \tau_{c|_B}^{-1} )  \; \delta_{i^c_s  j^c_{ \overline{ \sigma_{c|_B} (s) } } } \delta_{a^c_s b^c_{ \overline{ \tau_{c|_B} (s)} } }
 \big)  \bE    \Big[ \prod_{s\ge 1}^{|B| } 
  T_{\vec a_s} \bar T_{\vec b_{\bar s}}  J_{\vec i_s} \bar J_{\vec j_{\bar s}} \Big]  \crcr
& = \sum_{\Pi \in \mathcal{P}(n,\bar n)} \lambda_\Pi \prod_{G=B\cup \bar B\in \Pi}
 \;  \sum_{\bsig_{|_B},\btau_{|_B} \in S^D_{B,\bar B}} \big(     \prod_{c=1}^D W^{(N)}(\sigma_{c|_B} \tau_{c|_B}^{-1} ) \;  \Tr_{\bsig_{|_B}}(J,\bar J) \; \bE \left[  \Tr_{\btau_{|_B}}(T,\bar T)  \right]\; ,
 \end{split}
\]
and we observe that $ \sigma_{c|_B} \tau_{c|_B}^{-1}$ is a permutation of the black elements having the same cycle-type as the permutation $ \tau_{c|_B}^{-1} \sigma_{c|_B}$ of the white ones. The invariants now factor over the pure connected components, and when commuting the sum over $\Pi$ with the ones over $\bsig$ and $\btau$, the compatibility condition that $\bsig_{|_B}, \btau_{|_B} $ correspond to the parts of $\Pi$ implies that $\Pi$ must be coarser that the partitions in pure connected components $\Pi_{\mathrm{p}}(\bsig)$ 
and $\Pi_{\mathrm{p}}(\btau)$, leading to:
\[
\begin{split}
& \partial^n_{z}\partial^n_{\bar z}
 \log \bE_{T ,\bar T}[e^{z J\cdot T  + \bar z \bar J \cdot \bar T }]\bigr\rvert_{z=\bar z=0} \crcr
 &\qquad  = \sum_{\bsig,\btau\in S_{n,\bar n}^D}
 \Tr_{\bsig}(J,\bar J) 
 \sum_{\substack{{\Pi \in \mathcal{P}(n,\bar n)}\\{\Pi \ge \Pi_\mathrm{p}( \bsig) \vee \Pi_{\mathrm{p}}( \btau) } }} \lambda_\Pi
 \prod_{G=B\cup \bar B\in \Pi} 
\bE \left[\Tr_{\btau_{|_B}} (T,\bar T)\right]\; \prod_{c=1}^D W^{(N)} (\sigma_{ c|_B} \tau_{c|_B}^{-1}) \;.
\end{split}
\]
In order to invert the relations we start by averaging the expectations over the unitary group:
\be
\label{eq:unit-inv-lu-inv-0}
\bE \left[\Tr_{\bsig}(A)\right] = \bE_A \left[ \int dU \, \Tr_{\bsig}(UAU^\dagger)\right] \;, 
\qquad \bE \left[\Tr_{\bsig}(T,\bar T)\right] = \bE_{T,\bar T} \left[ \int dU \, \Tr_{\bsig} (UT, \bar T U^\dagger)\right] \;,  
\ee
which is, using the definition \eqref{def:trace-invariants} of the trace-invariants and  \eqref{eq:WeinDef}:
\begin{align}
\label{eq:unit-inv-lu-inv}
 \bE \left[\Tr_{\bsig}(A)\right] 
= & \sum_{\btau, \bnu \in S^D_n}  
 \sum_{ i ,  j , a ,  b}  \bE \biggl[\prod_{s=1}^n A_{\vec  a_s ;  
\vec b_s }\biggr] 
 \prod_{c,s}\delta_{i_s^{c}, j_{\sigma_c(s)}^{c}} 
\,  \delta_{i_s^{c}, j_{\tau_c(s)}^{c}} \, \delta_{a_s^{c}, b_{\nu_c(s)}^{c}} 
W^{(N)} ( \tau_{c} \nu_{c}^{-1})  \\
 \bE \left[\Tr_{\bsig}(T,\bar T)\right] 
= & \sum_{\btau, \bnu \in S^D_{n,\bar n}}  
 \sum_{  i , j , a , b}  \bE \biggl[\prod_{s=1}^n T_{\vec  a_s } \bar T_{\vec b_s }\biggr] 
 \prod_{c,s}\delta_{i_s^{c}, j_{\sigma_c(s)}^{c}} 
\,  \delta_{i_s^{c}, j_{\tau_c(s)}^{c}} \, \delta_{a_s^{c}, b_{\nu_c(s)}^{c}} 
W^{(N)} ( \tau_{c} \nu_{c}^{-1})  
\; ,\nonumber
\end{align}
where again, in the pure case, the sum should be considered as a sum over bijections between the white and the black vertices. Recalling that 
$ \sum_{c=1}^D\#(\sigma_c\tau_c^{-1}) =  nD - d(\bsig,\btau)$, this is:
\[
 \bE \left[\Tr_{\bsig}(A)\right]  =  \sum_{\btau \in S^D_n} N^{     nD - d(\bsig,\btau)    } \mathcal{G}_{\btau} [A] \;,  \qquad 
 \mathcal{G}_{\btau} [A]  = \sum_{\bnu \in S^D_n}\mathbb{E} \left[\Tr_{\bnu}(A)\right]\,  \prod_{c=1}^D W^{(N)} ( \nu_{c} \tau_{c}^{-1}) \;, 
\]
and similarly for the pure case.
We define the multiplicative extension in the mixed case:
\be
\begin{split}
\label{eq:int-proof-thm-discr}
\prod_{G\in \pi} \bE\left[\Tr_{\bsig_{\lvert_G}}(A)\right] & =  \sum_{\substack{{\btau \in S^D_n}\\{\Pi(\btau)\le \pi}}} N^{   nD - d(\bsig,\btau)    }  \mathcal G _{\pi, \btau} [A] \;,
\crcr
\mathcal G_{\pi, \btau}[A] &= \sum_{\substack{{\bnu \in S^D_{n}}\\{\pi \ge \Pi(\bnu)}}}
\prod_{G\in \pi} \bE\left[\Tr_{\bnu_{\lvert_G}}(A)\right]  \, \prod_{c=1}^D W^{(N)} ( \nu_{c|_G}\tau_{c|_G}^{-1}) \; ,
\end{split}
\ee
while in the pure case we have $\Pi\in  \mathcal{P}(n,\bar n)$ and $\Pi \ge \Pi_{\mathrm{p}}(\bsig)$. Comapring with \eqref{eq:cum-finN-int} leads to: 
\be
\label{eq:int-proof-thm-discr2}
\mathcal{K}^{\mathrm{m}}_\bsig[A] =\sum_{\substack{{\pi \in \mathcal{P}(n)}\\{\pi \ge \Pi(\bsig)}}} \lambda_\pi \  \mathcal G_{\pi, \bsig }  [A] \;, \qquad  \mathcal G_{\pi, \bsig}[A] =\sum_{\substack{{\pi' \in \mathcal{P}(n)}\\{\pi \ge \pi' \ge \Pi(\bsig)}}} \mathcal{K}^{\mathrm{m}}_{\pi',\bsig}[A] \; ,
\ee
where the last equality follows by M\"oebius inversion. The pure case is similar.
\qed

\subsection{Proof of  Prop.~\ref{prop:finite-free-prop}}
\label{app:proofs-of-finite-free-prop}

In the mixed case, using the $\mathsf{LU}$ invariance and following the steps \eqref{eq:unit-inv-lu-inv-0},  \eqref{eq:unit-inv-lu-inv} and \eqref{eq:int-proof-thm-discr} but without carrying the summation over $i$s and $j$s, one computes for $\pi\in\mathcal{P}(n)$: 
$$
\prod_{G\in \pi}\mathbb{E}\left[\prod_{s\in G} A_{i_s^1,\ldots i_s^D\;;\;j_s^1,\ldots j_s^D}\right]=\sum_{\btau\in S_n^D\textrm{ s.t. } \Pi(\btau)\le \pi} \left(\prod_{c=1}^D\prod_{s=1}^n \delta_{i_s^c, j^c_{\tau_c(s)}}\right) \mathcal{G}_{\pi, \btau}[A]\;, 
$$
whch leads to:
\be
 k_n\Bigl(\bigl\{ A_{i_s^1,\ldots i_s^D\;;\;j_s^1,\ldots j_s^D}  \bigr\}_{1\le s \le n}\Bigr)=\sum_{\btau\in S_n^D} \left(\prod_{c=1}^D\prod_{s=1}^n \delta_{i_s^c, j^c_{\tau_c(s)}}\right) \mathcal{K}_\btau^\mathrm{m}[ A] \;.
\ee
We choose $\sigma_c\in S_n$ and impose  $i_s^c=i_c(\sigma_c(s))$ and $j_{s'}^c=i_c(s')$ for all $1\le s \le n$: the non-vanishing terms in the sum are such that $i_s^c=i_c(\sigma_c(s))=j^c_{\tau_c(s)}=i_c(\tau_c(s))$, and since the $i_c(s')$ are all distinct, this imposes $\tau_c=\sigma_c$. Therefore:
$$
k_n\Bigl(\bigl\{A_{i_1(\sigma_1(s)), \ldots, i_D(\sigma_D(s))\; ;\;   i_1(s), \ldots, i_D(s)  } \bigr\}_{1\le s \le n}\Bigr)=\mathcal{K}_\bsig^\mathrm{m}[ A]\;.
$$
The proof is similar in the pure case.\qed
\section{Proofs of Sec.~\ref{sec:free-cum}}

\subsection{Proof of Thm.~\ref{thm:limit-of-finite-cumulants}}
\label{sec:proof-free-cumulants-pure-melonic}

\paragraph{General asymptotic.}
The finite $N$ free cumulants in Thm.~\ref{thm:limit-of-finite-cumulants} simplify
for a purely connected $\bsig$, as the sum over partitions reduces to one term corresponding to the one set partition. Expressing it in terms of classical cumulants we have:
\be
\label{eq:cum-finN-pure-conn-phi}
\mathcal{K}_{\bsig}[T, \bar T] _{\bigl\lvert_{K_\mathrm{p}(\bsig)=1}} = \sum_{\btau \in S^D_{n,\bar n}}  \sum_{ \pi \ge \Pi_\mathrm{p}(\btau)}  \Phi _{\pi, \btau} [T, \bar T]  \prod_{c=1}^D W^{(N)} ( \sigma_{ c} \tau_{c}^{-1}) \;,
\ee
where all the partitions encountered in this proof are bipartite, and all the permutations should be seen as mapping white to black vertices.

The Gaussian scaling hypothesis \eqref{eq:scaling-hypothesis-pure} yields:
\be
\lim_{N\rightarrow \infty} \frac 1 {N^{n- \sum_{B\in \pi} \min d(\btau_{\lvert_B}, \eta_B)}}\, 
\Phi _{\pi, \btau } [T, \bar T] = \varphi _{\pi, \btau} (t, \bar t)  \; ,  
\ee
where each minimum is taken over the $\eta_B\in S_{\lvert B \vert}$ with  $K_\mathrm{p}(\btau_{\lvert_B}, \eta_B) = 1$. We let: 
\[
H_{\btau,\pi}=\bigl\{\eta \in S_n \mid \Pi_\mathrm{p}(\eta) \le \pi \quad \mathrm{and}\quad \forall B\cup \bar B\in \pi,\ K_\mathrm{p}(\btau_{\lvert_B}, \eta_{\lvert_B})=1\bigr\} \;,
\]
the sets of permutations over which the minimum in the Gaussian scaling is taken, that is:
\be
\sum_{B\in \pi}\  \ \min_{\eta_B\in S_{\lvert B \rvert},\, K_\mathrm{p}(\btau_{\lvert_B}, \eta_B)=1} d(\btau_{\lvert_B}  , \eta_B ) = \min_{\eta \in H_{\btau,\pi}}\sum_{B\in \pi} d(\btau_{\lvert_B}  , \eta_{\lvert_B} ) =  \min_{\eta \in H_{\btau,\pi}} d(\btau  ; \eta )  \; .
\ee
Taking into account that asymptotically the Weingarten functions behave as:
\[
\prod_{c=1}^D W^{(N)} ( \sigma_c \tau_c^{-1}) =  N^{- n D - d(\bsig,\btau) }  \mathsf{M}(\bsig\btau^{-1}) (1+O(N^{-2})) \; ,
\]
we obtain the $N\rightarrow \infty$ asymptotic behaviour of the finite $N$ free cumulant:
\begin{align}
\label{eq:cum-finN-pure-conn-phi-first-asympt}
\mathcal{K}_{\bsig}[T, \bar T]  = \frac1{N^{nD}}\sum_{\btau \in S^D_{n}}  \sum_{ \pi\ge \Pi_\mathrm{p}(\btau)}  \left [ \varphi _{\pi, \btau} (t, \bar t)\,  \mathsf{M}(\bsig\btau^{-1})  N^{r(\bsig) - \Delta_\pi(\bsig; \btau)} +o\left(N^{r(\bsig) - \Delta_\pi(\bsig; \btau)}\right) \right]\;, 
\end{align}
where $\Delta_\pi(\bsig;\btau)=d(\bsig, \btau) + \min_{\eta \in H_{\btau,\pi}} d(\btau  , \eta ) - \min_{\eta_0\in S_n} d(\bsig, \eta_0)$. Now, for any $\eta\in S_n$, one has the $D$-fold triangular inequality:
\be
\label{eq:Dfold-triangular}
d(\bsig, \btau) + d(\btau  , \eta ) \ge  d(\bsig, \eta), 
\ee
with equality if and only if $\btau\eta^{-1}\preceq\bsig\eta^{-1}$. It follows that  
$\Delta_\pi(\bsig;\btau)\ge 0$ and at 
large $N$, the dominant contribution is given by $\Delta_\pi(\bsig;\btau) =  0$, and only the couples $(\pi,\btau)$ belonging to the set:
 \be
\bfS(\bsig) = \Bigl\{(\pi, \btau) \ \bigl\lvert \  \pi \ge \Pi_\mathrm{p}(\btau)\ \  \mathrm{and}\ \ d(\bsig, \btau) + \min_{\eta \in H_{\btau,\pi}} d(\btau  , \eta ) =\min_{\eta_0\in S_n} d(\bsig, \eta_0)\Bigr\} \;,
\ee
contribute to the cumulant. Note that this set is nonempty, $(\Pi_\mathrm{p}(\bsig), \bsig )\in \bfS(\bsig)$, as for $K_\mathrm{p}(\bsig)=1$ we have $H_{\bsig, \Pi_p(\bsig)}=S_n$.
 
\paragraph{Expressions in terms of degrees.}
In order to use our results on the degree of graphs, it is convenient to parametrize the set $\bfS(\bsig)$ using the degree $\bar \omega$ (Thm.~\ref{thm:second-degree}). For $\btau\in S_n^D$, $\pi \ge \Pi_\mathrm{p}(\bsig)$ and $\eta\in H_{\btau, \pi}$, as the restriction
of $\bsig, \eta$ to each block of $\pi$ is connected, we have $ \#(\pi) = K_p(\bsig,\eta)$ 
and:
\be
\label{eq:bar-omega-annex}
\bar \omega(\bsig;\eta) = D \#(\pi) - (D-1) K_\mathrm{p}(\bsig) - n  +  d(\bsig, \eta) \;. 
\ee
At the same time, we also have the following.
\begin{lemma}
\label{lem:introduce-C}
For any $\bsig, \btau \in S_n^D$ we have:
\be
\label{eq:C}
\mathcal{C}(\bsig ; \btau) = K_\mathrm{p}(\bsig, \btau) - K_\mathrm{p}(\btau) + d(\bsig, \btau) \ge 0 \; ,  \qquad \mathcal{C}(\bsig ; \bsig)=0 \;.
\ee
\end{lemma}
\proof Consider the $2D$ colored graph $(\bsig,\btau)$. Deleting an edge corresponding to $\sigma_c$ in each cycle of $\sigma_c\tau_c^{-1}$ does not disconnect the graph.
Deleting the remaining $d(\bsig,\btau)$ edges in $\bsig$ leads to the graph $\btau$, and the number of connected components cannot increase by more than $1$ for each of these edges.
\qed

\

The $D$-fold triangular inequality \eqref{eq:Dfold-triangular} is expressed equivalently for any $\bsig, \btau \in S_n^D$, $K_\mathrm{p}(\bsig)=1$, $\pi\ge \Pi_\mathrm{p}(\btau)$ and $\eta\in H_{\btau, \pi}$ as:
\be
\label{eq:D-fold-triangular-degree-purelyconnected}
\mathcal{C}(\bsig ; \btau) + D\bigl(K_\mathrm{p}(\btau) - \#(\pi) \bigr) +  \bar \omega (\btau  ; \eta )  \ge  \bar \omega (\bsig  ; \eta ) \;,  
\ee
hence the condition defining the set $\bfS(\bsig)$ becomes in terms of the degree:
\be
\label{eq:lo-generic-purely-conn-degree}
\mathcal{C}(\bsig ; \btau) + D\bigl(K_\mathrm{p}(\btau) - \#(\pi) \bigr) +  \min_{\eta \in H_{\btau,\pi}} \bar \omega (\btau  ; \eta )  =  \min_{\eta_0 \in S_n} \bar \omega (\bsig  ; \eta_0 ).
\ee

 \paragraph{The melonic case.} 
 If $\bsig$ is purely connected and melonic, $\omega(\bsig)=0$ and $K_\mathrm{p}(\bsig)=1$, then from Thm.~\ref{thm:second-degree}, for $D\ge 3$ there exists a unique $\eta_0 \in S_n$ such that $\bar \omega (\bsig  ; \eta_0 )=0$, and \eqref{eq:lo-generic-purely-conn-degree} simplifies to:
\be
\mathcal{C}(\bsig ; \btau) + D\bigl(K_\mathrm{p}(\btau) - \#(\pi) \bigr) +  \min_{\eta \in H_{\btau,\pi}} \bar \omega (\btau  ; \eta )  =  0 \;,
\ee
and since all the three terms are non-negative, this implies $\mathcal{C}(\bsig ; \btau)=0$,  $\pi=\Pi_\mathrm{p}(\btau)$, and   $\min_{\eta \in H_{\btau,\Pi_\mathrm{p}(\btau)}} \bar \omega (\btau  ; \eta )  =  0$. The last condition tells us that $\btau$ must be melonic, but not necessarily purely connected.

From \eqref{eq:D-fold-triangular-degree-purelyconnected}, for any $\btau$ with $\mathcal{C}(\bsig ; \btau)=0$ we have that for any $\eta \in H_{\btau,\Pi_\mathrm{p}(\btau)}$:
\be
 \bar \omega (\btau  ; \eta )  \ge  \bar \omega (\bsig  ; \eta ), 
\ee
with equality if and only if $\btau\eta^{-1}\preceq\bsig\eta^{-1}$. Let $\eta\in  H_{\btau,\Pi_\mathrm{p}(\btau)}$ be such that $\bar \omega (\btau  ; \eta )  =  0$: then $\bar \omega (\bsig  ; \eta )=0$, and as $\bsig$ is purely connected, $\eta$ is the canonical pairing of $\bsig$, and $\btau \eta^{-1}\preceq \bsig \eta^{-1}$. 

Reciprocally, if $\bsig$ is melonic purely connected and $\eta$ is its canonical pairing, and if $\btau$ is such that $\btau  \eta^{-1}\preceq \bsig \eta^{-1}$, then for $\pi=\Pi_\mathrm{p}(\btau)\vee\Pi_\mathrm{p}( \eta)$ one has $\eta \in H_{\btau, \pi}$ and from \eqref{eq:D-fold-triangular-degree-purelyconnected}:  
\be
\mathcal{C}(\bsig ; \btau) + D\bigl(K_\mathrm{p}(\btau) - \#(\pi) \bigr) +  \bar \omega (\btau  ;  \eta )  =  \bar \omega (\bsig  ;  \eta ) = 0 \;.  
\ee
It follows that $\pi=\Pi_\mathrm{p}(\btau)$, that is, $\Pi_\mathrm{p}(\btau)\ge\Pi_\mathrm{p}(\eta)$ or said otherwise, $\eta \in H_{\btau, \Pi_\mathrm{p}(\btau)}$, and on the other hand $\mathcal{C}(\bsig ; \btau)=\bar \omega (\btau  ; \eta )=0$, hence $\min_{\eta \in H_{\btau,\Pi_\mathrm{p}(\btau)}} \bar \omega (\btau  ; \eta )  =  0$ and in particular $\btau$ is melonic.

\

If $\omega(\bsig)=0$ and $K_\mathrm{p}(\bsig)=1$,  and if $\eta$ corresponds to the canonical pairing of $\bsig$, then the asymptotics of \eqref{eq:cum-finN-pure-conn-phi-first-asympt} is:
 \be
\kappa_{\bsig}(t, \bar t):= \lim_{N\rightarrow \infty} N^{nD-r(\bsig)}\mathcal{K}_{\bsig}[T, \bar T]  = \sum_{\substack{{\btau \in S^D_{n}}\\{\btau \eta^{-1}\preceq \bsig \eta^{-1} } } }   \varphi _{\Pi_\mathrm{p}(\btau), \btau } (t, \bar t)\,  \mathsf{M}(\bsig\btau^{-1}) \; ,  
\ee
where, as $\bsig$ is melonic $r(\bsig)=1$ (see Thm.~\ref{thm:melonic-gaussian}).

\subsection{Proof of Thm.~\ref{thm:free-cumulants-melonic}}
\label{sec:proof-free-cumulants-pure-melonic-multinv}

 The first assertion has been proven in Thm.~\ref{thm:limit-of-finite-cumulants}. The asymptotic additivity derives from \eqref{eq:additivity-pure} as well as the  fact that $(T_1+T_2, \bar T_1+\bar T_2)$ scales as \eqref{eq:scaling-hypothesis-pure}, which is itself a consequence of the multilinearity of the classical cumulants and of \eqref{eq:scaling-hypothesis-pure-multitensors}.

 \paragraph{Multiplicativity.} \eqref{eq:multiplicativty-melonic-free-cumulants} is the analogue of Prop.~10.21 in \cite{NicaSpeicher}: the mapping  $\bsig, \btau\in S_n \rightarrow \mathsf{M}(\bsig\btau^{-1})$ is multiplicative because $\sigma, \tau\in S_n \rightarrow \mathsf{M}(\sigma\tau^{-1})$ is. For $\bsig\in S_n^D$ with pure connected components $\bsig_1^\mathrm{p}, \ldots, \bsig_q^\mathrm{p}$, since $\Pi_\mathrm{p}(\eta)\le\Pi_\mathrm{p}(\bsig)$, letting $G_i = B_i\cup \bar B_i$ be the support of $\bsig_i^\mathrm{p}$:
\[
\kappa_{\bsig}  = \sum_{\substack{{\btau \in S^D_{n}}\\{\btau_{|_{B_i}} \preceq \bsig_{|_{B_i}} }}}   \varphi _{\Pi_\mathrm{p}(\btau), \btau} \,  \mathsf{M}(\bsig\btau^{-1}) 
 =  \sum_{\{\btau_{|_{B_i} } \preceq \bsig_i^\mathrm{p} \} }  \prod_{i=1}^q  \varphi_{ \Pi_\mathrm{p}(\btau_{|_{B_i} }), \btau_{|_{B_i}} } \,  \mathsf{M}(\bsig_i^\mathrm{p}\btau_{|_{B_i}}^{-1})
 = \prod_{i=1}^q \kappa_{\bsig_i^\mathrm{p}} \;.
\]

 \paragraph{Inversion.} If all the tensors are identical, the inversion is an example of M\"{o}bius inversion. The direct product of $D$ copies of the lattice of non crossing partitions $\mathrm{NC}(n)$ is the lattice $\left(\mathrm{NC}(n)\right)^{\times D}$ whose M\"{o}bius function is $ \mathsf{M}(\bnu)= \prod_{c=1}^D \mathsf{M}(\nu_c)$ (Prop.~10.14 of \cite{NicaSpeicher}), and we recall that $\mathrm{NC}(n)$ is isomorphic to the lattice of non crossing permutations on $\mu$, $S_{\mathrm{NC}}(\mu)$, for any cycle $\mu$ with $n$ elements. 
 
 Here for each $c$ we fix some $\mu_c$ with $n$ elements for which $\sigma_c\preceq \mu_c$, and we work in the lattice product $\bigtimes_{c=1}^D S_{\mathrm{NC}}(\mu_c)$, which contains $\bsig$. Any interval $[\btau,\bsig]$ where $\btau \preceq \bsig$ is itself a lattice (Remark 9.26 in \cite{NicaSpeicher}). In particular, letting $\bsig$ be melonic and with canonical pairing the identity, and $\mathbf{id}{_n}\in S_n^D$ be the $D$-tuple of permutations $(\mathrm{id}, \ldots, \mathrm{id})$, that is, the melonic invariant consisting in $n$ disjoint two-vertex graphs with canonical pairing the identity, the M\"{o}bius inversion (Proposition 10.6 in \cite{NicaSpeicher}) ensures that \eqref{eq:def-mult-ext-free-cumulants-melonic} can be inverted on the interval $[\mathbf{id}{_n}, \bsig]$, hence for any $\bnu \in [\mathbf{id}{_n}, \bsig]$:
\be
\kappa_{\Pi_\mathrm{p}(\bnu), \bnu} (\vec x, \vec{x'})= \sum_{\btau \preceq \bnu} \varphi_{\Pi_\mathrm{p}(\btau), \btau}(\vec x, \vec{x'}) \,\mathsf{M}(\bnu\btau^{-1}) \;, \qquad 
\varphi_{\Pi_\mathrm{p}(\bnu), \bnu} (\vec x, \vec{x'})= \sum_{\btau \preceq \bnu} \kappa_{\Pi_\mathrm{p}(\btau), \btau}(\vec x, \vec{x'}),  
\ee
and in particular this holds for $\bnu=\bsig$.

\subsection{Proof of Lemma.~\ref{lem:facto-mixed-on-pure} }
\label{sec:proof-facto-mixed-on-pure}
The only difference with the proof of Thm.~\ref{thm:limit-of-finite-cumulants} is that the exponent of $N$ for the terms in he sum is now expressed as 
 $$-nD+ K_\mathrm{p}(\bsig) -\bigl[\mathcal{C}(\bsig ; \btau) + D\bigl(K_\mathrm{p}(\btau) - \#(\pi) \bigr) +  \min_{\eta \in H_{\btau,\pi}} \bar \omega (\btau  ; \eta ) + (K_\mathrm{p}(\bsig) - K_\mathrm{p}(\bsig, \btau))\bigr],$$ where the quantity between bracket is still non-negative. For it to vanish, one has the additional condition  $K_\mathrm{p}(\bsig) = K_\mathrm{p}(\bsig, \btau)$, that is, $\Pi_\mathrm{p}(\btau)\le \Pi_\mathrm{p}(\bsig)$. The $D$-fold triangular inequality \eqref{eq:Dfold-triangular} is in this case seen to be equivalent to: 
 $$
\mathcal{C}(\bsig ; \btau) + D\bigl(K_\mathrm{p}(\btau) - \#(\pi) \bigr) + \bigl(K_\mathrm{p}(\bsig) -K_\mathrm{p}(\bsig, \btau) \bigr) + \bar \omega (\btau  ; \eta )  \ge  \bar \omega (\bsig  ; \eta ) + D\bigl(K_\mathrm{p}(\bsig) -K_\mathrm{p}(\bsig, \eta) \bigr).  
$$
One  then proceedes as in Sec.~\ref{sec:proof-free-cumulants-pure-melonic} to show that the set of $\btau$ such that  $\mathcal{C}(\bsig ; \btau) + D\bigl(K_\mathrm{p}(\btau) - \#(\pi) \bigr) +  \min_{\eta \in H_{\btau,\pi}} \bar \omega (\btau  ; \eta ) + (K_\mathrm{p}(\bsig) - K_\mathrm{p}(\bsig, \btau))=0$ coincides with the set of $\btau$ such that $\btau\eta^{-1}\preceq\bsig\eta^{-1}$, where $\eta$ is the canonical pairing of $\bsig$ (and $\btau$). One obtains the same limit as in Thm.~\ref{thm:limit-of-finite-cumulants}, and it factorizes over the pure connected components.

\subsection{Proof of Thm.~\ref{thm:sum-and-mixed-on-pure} }
\label{sec:sum-and-mixed-on-pure}

Even though $A=T_1\otimes \bar T_1 + T_2\otimes \bar T_2$ is mixed, we  express $\mathcal{K}^\mathrm{m}_{\bsig}[A] $ in terms of  pure quantities by multilinearity:
$$
 \mathcal{K}^\mathrm{m}_\bsig[A] = \sum_{\btau \in S^D_{n}}\ \sum_{\{A_i\in \{T_1\otimes \bar T_1, T_2\otimes \bar T_2\}\}}\ \bE\left[\Tr_{\btau} (A_1, \ldots, A_n)\right]\prod_{c=1}^D W^{(N)} ( \sigma_{ c} \tau_{c}^{-1})  \; .
$$
For each assignment of the $A_i$s, $\Tr_{\btau} (A_1, \ldots, A_n)$ is the same as a pure invariant $\Tr_{\btau} (\vec X, \vec {\bar X})$, where $\vec X$ (resp.~$\vec {\bar X}$) assigns some tensors $T_1$ or $T_2$ to the white (resp.~black) vertices of $\btau$, with the condition that the components $\bar X_i$ of $\vec {\bar X}$ and the components $X_i$ of $\vec X$ satisfy $\bar X_i = \overline{(X_i)}$.
We may therefore use the classical moment cumulant formula over the \emph{pure} connected components of $\btau$:
$$
 \mathcal{K}^\mathrm{m}_\bsig[A] = \sum_{\btau \in S^D_{n}}\ \sum_{\{\vec X\}}\sum_{\pi\ge \Pi_\mathrm{p}(\btau)} \Phi_{\pi, \btau} \bigl[\vec X, \vec{\bar X}\bigr]\prod_{c=1}^D W^{(N)} ( \sigma_{ c} \tau_{c}^{-1})  \; .
$$ 
Since by assumption the tensors satisfy \eqref{eq:scaling-hypothesis-pure-multitensors}, the proof of Thm.~\ref{thm:limit-of-finite-cumulants}  goes through and one obtains: 
 \be
 \label{eq:intermediate-in-proof-mixed-on-pure}
\kappa_{\bsig}^\mathrm{m}(a) = \lim_{N\rightarrow \infty} N^{nD-1} \mathcal{K}^\mathrm{m}_\bsig[A]  = \sum_{\{\vec x\}}\sum_{\substack{{\btau \in S^D_{n}}\\{\btau \preceq \bsig}}}   \varphi _{\Pi_\mathrm{p}(\btau), \btau} (\vec x, \vec{\bar x})\,  \mathsf{M}(\bsig\btau^{-1}) \;.
\ee
In particular, since the labeling is such that the canonical pairing is the identity, for all the terms in the sum the connected and purely connected components of $\btau$ coincide. For any $\bnu$ satisfying this property, one has: 
\be
\Phi_\bnu^\mathrm{m}[A] = \sum_{\{A_i\in \{T_1\otimes \bar T_1, T_2\otimes \bar T_2\}\}} \Phi_\bnu^\mathrm{m}[A_1, \ldots, A_n] = \sum_{\{\vec X\}} \Phi_\bnu[\vec X, \vec{\bar X}] = N^{r(\bnu)}\sum_{\vec x} \varphi_\bnu(\vec x, \vec{\bar x}) + o(N^{r(\bnu)}), 
\ee
where the assumption is necessary for the second equality and he third uses \eqref{eq:scaling-hypothesis-pure-multitensors}. For such  $\bnu$:
\be
\label{eq:varphim-mixedonpure}
\varphi^\mathrm{m}_{\bnu} (a)=\lim_{N\rightarrow \infty} \frac1{N^{r(\bnu)}} \Phi_\bnu^\mathrm{m}[A]  = \sum_{\vec x} \varphi_\bnu(\vec x, \vec{\bar x}), 
\ee
so that \eqref{eq:first-momcum-in-mixed-on-pure} is proven. The inversion \eqref{eq:second-momcum-in-mixed-on-pure} has already been proven (Sec.~\ref{sec:proof-free-cumulants-pure-melonic-multinv}). The fact that $\kappa^\mathrm{m}_{\bsig}(a) = \kappa_\bsig(t_1+t_2, \bar t_1+\bar t_2)$ comes from \eqref{eq:asympt-equality-cumulants-mixed-on-pure}, and $\varphi^\mathrm{m}_{\bsig}(a) = \varphi_\bsig(t_1+t_2, \bar t_1+\bar t_2)$ from the inverse relations. However: 
\be
\label{eq:varphip-mixedonpure}
\varphi_{\bnu} (t_1+t_2, \bar t_1+\bar t_2)=\lim_{N\rightarrow \infty} \frac1{N^{r(\bnu)}} \Phi_\bnu[T_1+T_2, \bar T_1+\bar T_2]  = \sum_{\vec x, \vec{x'}} \varphi_\bnu(\vec x, \vec{\bar x}), 
\ee
where $\vec x=(x_1,\ldots x_n)$, $x_i\in \{t_1, t_2\}$, $\vec{x'}=(x'_{\bar 1},\ldots x'_{\bar n})$, $x'_{\bar i}\in \{\bar t_1, \bar t_2\}$. 
For \eqref{eq:varphim-mixedonpure} and \eqref{eq:varphip-mixedonpure} to be equal,  \eqref{eq:non-canonical-vanish} must hold.

\subsection{Proof of Thm.~\ref{thm:limit-for-wishart-tensor-first-order}.}
\label{subsub:free-cumulants-wishart}

The proof follows the one in the pure case presented in
Section~\ref{sec:proof-free-cumulants-pure-melonic} by replacing $\btau$ with
$(\btau, \mathrm{id})$ and observing that the pure connected components of $(\btau, \mathrm{id})$ are one to one with the mixed connected components of $\btau$, thus:
$$
\mathcal{K}^\mathrm{m}_{\bsig}[A]  = \frac1{N^{nD}}\sum_{\btau \in S^D_{n}}  \sum_{ \pi\ge \Pi(\btau)}\Bigl(  \varphi^\mathrm{m} _{\pi, \btau} (a)\,  \mathsf{M}(\bsig\btau^{-1}) N^{ n  - d(\bsig, \btau) - \min_{\eta \in H_{ (\btau, \mathrm{id} ) ,\pi}} d( ( \btau,\mathrm{id} ) , \eta ) }+\ldots\Bigr), 
$$
where for $\pi \ge \Pi(\btau)$, 
$H_{( \btau,\mathrm{id} ),\pi}=\bigl\{\eta \in S_n \mid \Pi(\eta) \le \pi \quad \mathrm{and}\quad \forall B\in \pi,\ K_\mathrm{m}(\btau_{\lvert_B}, \eta_{\lvert_B})=1\bigr\}$. 

\paragraph{General asymptotics.}The $D$-fold triangular inequality \eqref{eq:Dfold-triangular} is now replaced with: 
\be
\label{eq:new-triangular-wishart}
d(\bsig, \btau) + d(( \btau,\mathrm{id} ) , \eta ) \ge  d\bigl(( \bsig,\mathrm{id} ) , \eta\bigr), 
\ee
with equality if and only if $( \btau,\mathrm{id} ) \eta^{-1}\preceq 
(\bsig,\mathrm{id} )\eta^{-1}$, that is $\btau \eta^{-1}\preceq \bsig\eta^{-1}$, and the dominant contributions satisfy:
 \be
\label{eq:general-formula-asymptotics-Wishart}
d(\bsig, \btau) + \min_{\eta \in H_{( \btau,\mathrm{id} ),\pi}} d\bigl( ( \btau,\mathrm{id} ), \eta \bigr) =\min_{\eta_0\in S_n} d\bigl(( \bsig,\mathrm{id} ), \eta_0\bigr) \; , 
\ee
which is saturated at least for $\btau=\bsig$, hence for a connected $\bsig$, $\mathcal{K}^\mathrm{m}_{\bsig}[A]$ scales as $N^{r_W(\bsig)-nD}$. 

\paragraph{First order.} Using the degree, $\bar \omega(( \btau,\mathrm{id} ), \eta)$,  \eqref{eq:bar-omega-annex} gives rise 
to a term $K_\mathrm{m}(\btau)$ instead of $K_\mathrm{p}(\btau)$, and applying Lemma~\ref{lem:introduce-C} to $( \bsig,\mathrm{id} )$ and $( \btau,\mathrm{id} )$ we get:
 \be
 \label{eq:introduce-tildeC}
\mathcal{C}_\mathrm{m}(\bsig, \btau) = d(\bsig, \btau) -K_\mathrm{m}(\btau) + K_\mathrm{m}(\bsig, \btau)\ge 0 \; , 
\ee
where for $\bsig$ connected $K_\mathrm{m}(\bsig, \btau)=1$, and \eqref{eq:new-triangular-wishart} is equivalent to: 
\be
\mathcal{C}_\mathrm{m}(\bsig ; \btau) + (D+1)\bigl(K_\mathrm{m}(\btau) - \#(\pi) \bigr) +  \bar \omega (( \btau,\mathrm{id} ) ; \eta )  \ge  \bar \omega (( \bsig,\mathrm{id} ) ; \eta ).  
\ee

The rest follows as in Sec.~\ref{sec:proof-free-cumulants-pure-melonic}, identifying for connected $\bsig$ such that $( \bsig,\mathrm{id} )$ is melonic, the set of $(\pi, \btau)$ satisfying \eqref{eq:general-formula-asymptotics-Wishart} as $\pi=\Pi(\btau)$ and $\btau \eta^{-1}\preceq \bsig\eta^{-1}$, where $\eta$ is the canonical pairing of $( \bsig,\mathrm{id} )$. The additivity follows from the additivity of the finite $N$ free cumulants.

\paragraph{Inversion.} The inversion is slightly more involved. We start from the multiplicative extension to any non-necessarily connected $\bsig$ with $(\bsig,\mathrm{id} )$ melonic with canonical pairing $\eta$:
\be
\kappa^\mathrm{m}_{\Pi(\bsig), \bsig} (\vec a)= \sum_{\btau\eta^{-1}\preceq \bsig\eta^{-1}} \varphi^\mathrm{m}_{\Pi(\btau), \btau}(\vec a) \mathsf{M}(\bsig\btau^{-1}) = \sum_{\substack{{\btau \in S^D_{n}}\\{( \btau,\mathrm{id} )\eta^{-1} \preceq ( \bsig,\mathrm{id} )\eta^{-1}}}}   \varphi^\mathrm{m} _{\Pi(\btau), \btau} (\vec a)\,  \mathsf{M}(\bsig\btau^{-1}),  
\ee
where $\Pi(\btau)=\Pi_\mathrm{p}( \btau,\mathrm{id} )$. Changing variables to  $\bsig'=\bsig\eta^{-1}$ and $\btau'=\btau\eta^{-1}$, we rewrite this as:
\be
\label{eq:partial-in-proof-inv-mixed}
\kappa^\mathrm{m}_{\Pi_\mathrm{p}\left( \bsig' ,\eta^{-1}\right), \bsig'\eta }(\vec a_\eta) = \sum_{\substack{{\btau' \in S^D_{n}}\\{ ( \btau', \eta^{-1} )\preceq\; ( \bsig',  \eta^{-1} ) } } }     
\varphi ^\mathrm{m}_{\Pi_\mathrm{p}(\btau', \eta^{-1}), \btau'\eta } (\vec a_\eta)\,  \mathsf{M}(\bsig'\btau'^{-1}),
\ee
where the notation $\vec a_\eta$ indicates that after the change of variable, the variable $a_s$ is now associated to the thick edge going from the white vertex $s$ to the black vertex 
$\overline{\eta^{-1}(s)}$. After this change of variable, which corresponds to a relabeling of the white  vertices to set the canonical pairing to the identity, one has $\bar \omega( (\bsig',\eta^{-1} ) ;\mathrm{id})=0$, 
$\Pi_\mathrm{p}(\bsig,\mathrm{id} )=\Pi_\mathrm{p}(\bsig'\eta, \mathrm{id})= \Pi_\mathrm{p}(\bsig', \eta^{-1})$,  and similarly for $\btau$. 
For $\tilde \bsig=( \bsig', \sigma'_{D+1} )\in S_n^{D+1}$ we denote:
\be
g(\tilde \bsig) = \kappa^\mathrm{m}_{\Pi_\mathrm{p}\left(\bsig' , \sigma_{D+1}'\right), \bsig' \sigma_{D+1}'^{-1} }(\vec a_{\sigma_{D+1}'})
\; ,\qquad  f(\tilde \bsig)= \varphi^\mathrm{m}_{\Pi_\mathrm{p}\left(\bsig' , \sigma_{D+1}'\right),\bsig' \sigma_{D+1}'^{-1}  }(\vec a_{\sigma_{D+1}'}).  
\ee 

We fix a connected $\bnu'$ with $\bar \omega( ( \bnu', \eta^{-1} ); \mathrm{id})=0$ and  consider the relations above for all the $\bsig'$ such that $\bsig'\preceq \bnu'$: the $\bsig'$ are not necessarily connected and satisfy $\bar \omega( ( \bsig', \eta^{-1}); \mathrm{id})=0$. We rewrite \eqref{eq:partial-in-proof-inv-mixed} for any $\tilde \bsig= ( \bsig' , \sigma_{D+1}' )\in \mathcal{L}=[ ( \mathbf{id}_n , \eta^{-1} ) , 
( \bnu' ,  \eta^{-1} ) ]$  as:
\be
g(\tilde \bsig) = \sum_{\substack{{\tilde \btau \in \mathcal{L}}\\[+0.5ex]{ \tilde \btau\preceq \tilde \bsig}}}  f(\tilde \btau)\,  \mathsf{M}(\tilde\bsig\tilde \btau^{-1}) \; ,
\ee
which holds as in the lattice interval $\mathcal L$  we have $\sigma_{D+1}' = \eta^{-1}$ and $\tau_{D+1}'=\eta^{-1}$.
Note that $\mathcal{L}$ is a sublattice of $\bigtimes_{c=1}^{D} S_{\mathrm{NC}}( \nu'_c)
\times S_{\mathrm{NC}}( \eta^{-1} ) $, and we can invert this formula (Proposition 10.6 in \cite{NicaSpeicher}) to obtain for all
 $\tilde \bsig= ( \bsig' , \sigma_{D+1}' )\in \mathcal{L}=[ ( \mathbf{id}_n , \eta^{-1} ) , 
( \bnu' ,  \eta^{-1} ) ]$:
\be
 f(\tilde \bsig) = \sum_{\substack{{\tilde \btau \in \mathcal{L}}\\[+0.5ex]{ \tilde \btau\preceq \tilde \bsig}}}  g(\tilde \btau) \; , 
\ee
and again $ (\btau', \tau'_{D+1} ) \in \mathcal{L}$ forces $\tau_{D+1}'=\eta^{-1}$. Retracing our steps we conclude.

\subsection{Proof of Prop.~\ref{prop:mixed-persp-on-pure}}
  \label{sub:proof-mixed-persp-on-pure}

Starting from the free moment-cumulant formula for the pure case:
\be
\kappa_{\Pi_\mathrm{p} ( \bsig,\mathrm{id} ),  ( \bsig, \mathrm{id} ) }(t,\bar t)=  \sum_{\substack{{ ( \btau , \tau_{D+1} ) \in S^{D+1}_{n}}\\{( \btau, \tau_{D+1} )\eta^{-1} \preceq\; ( \bsig,\mathrm{id} )\eta^{-1}}}}     \varphi_{\Pi_\mathrm{p}( \btau, \tau_{D+1}), (\btau, \tau_{D+1}) } (t, \bar t) \,  \mathsf{M}(\bsig\btau^{-1})\mathsf{M}(\tau_{D+1}),  
\ee
we use the invariance $ \varphi_{\Pi_\mathrm{p}(\btau, \tau_{D+1}), ( \btau, \tau_{D+1} ) }(t, \bar t) = \varphi_{\Pi_\mathrm{p}( \btau\tau_{D+1}^{-1},  \mathrm{id}), ( \btau\tau_{D+1}^{-1} ,  \mathrm{id} ) }(t, \bar  t) $ and observe that:
\be
\varphi_{\Pi_\mathrm{p}(\btau\tau_{D+1}^{-1}, \mathrm{id}), (\btau\tau_{D+1}^{-1} ,  \mathrm{id}) }(t, \bar t)  = \varphi^\mathrm{m}_{\Pi(\btau\tau_{D+1}^{-1}),\btau\tau_{D+1}^{-1}}(a).
\ee
The invariant $\btau'=\btau\tau_{D+1}^{-1}$ is such that $ ( \btau' , \mathrm{id})$ is melonic, with canonical pairing $\eta\tau_{D+1}^{-1}$, so that:
\be
\varphi^\mathrm{m}_{\Pi(\btau'),\btau'}(a) =\sum_{\substack{{ \bnu' \tau_{D+1}\eta^{-1}\preceq\btau'\tau_{D+1}\eta^{-1}}}} \kappa^\mathrm{m}_{\Pi\left(\bnu'\right), \bnu' }( a). 
\ee
Changing back to $\btau' \tau_{D+1} =\btau$ and $\bnu' \tau_{D+1} =\bnu$ and putting the two relations together, we have:
\[
\kappa_{\Pi_\mathrm{p} ( \bsig,\mathrm{id} ), ( \bsig, \mathrm{id}) }(t,\bar t)=  \sum_{\substack{{ ( \btau , \tau_{D+1} ) \in S^{D+1}_{n}}\\{(\btau , \tau_{D+1})\eta^{-1} \preceq\; ( \bsig,\mathrm{id} )\eta^{-1}}}}      \mathsf{M}(\bsig\btau^{-1})\mathsf{M}(\tau_{D+1}) \sum_{\substack{{ \bnu\eta^{-1}\preceq\btau\eta^{-1}}}} \kappa^\mathrm{m}_{\Pi\left(\bnu \tau_{D+1}^{-1}\right), \bnu \tau_{D+1}^{-1} }( a) \; . 
\]
One has to be careful that for $\btau'$ and $\bnu'$, the identity is the permutation encoding the thick edges and the canonical pairing is $\eta\tau_{D+1}^{-1}$, but after the change of labeling the permutation encoding the thick edges is $\tau_{D+1}$, and the canonical pairing is $\eta$.
As $\kappa^{\mathrm{m}}_{\Pi(\bnu \tau_{D+1}^{-1}), \bnu \tau_{D+1}^{-1} }( a)$, factors over the connected components of $\Pi(\bnu \tau_{D+1}^{-1})=\Pi_\mathrm{p}(\bnu \tau_{D+1}^{-1} , \mathrm{id})=\Pi_\mathrm{p}(\bnu, \tau_{D+1})$, exchanging the sums over $\btau$ and $\bnu$, we get:
\[
\begin{split}
\kappa_{\Pi_\mathrm{p}( \bsig,\mathrm{id} ), ( \bsig, \mathrm{id} ) }(t,\bar t)& =  \sum_{\tau_{D+1}\eta^{-1}\preceq \eta^{-1}}\mathsf{M}(\tau_{D+1})\sum_{\bnu\eta^{-1}\preceq\bsig\eta^{-1}} \kappa^{\mathrm{m}}_{\Pi\left(\bnu \tau_{D+1}^{-1}\right), \bnu \tau_{D+1}^{-1} }( a) \crcr
& \qquad \times 
\sum_{\substack{{\btau \in S^{D}_{n}}\\{\bnu\eta^{-1}\preceq\btau\eta^{-1} \preceq \bsig\eta^{-1}}}}      \mathsf{M}(\bsig\btau^{-1}) \; ,
\end{split}
\]
and the last sum is $\delta(\bnu,\bsig)$, so that setting $\nu=\tau_{D+1}$, we get:
\be
\kappa_{\Pi_\mathrm{p} ( \bsig,\mathrm{id} ), (\bsig, \mathrm{id} ) }(t,\bar t)= \sum_{ \substack{
{\nu \in S_n} \\
{ \nu\eta^{-1}\preceq \eta^{-1}} } }\kappa^\mathrm{m}_{\Pi\left(\bsig \nu^{-1}\right), \bsig \nu^{-1} }( a)\; \mathsf{M}(\nu)  \; . 
\ee

Reciprocally, starting from the mixed moments in terms of the pure ones which we express in terms of pure cumulants, and substituting into the mixed cumulant moment formula, we get: 
\be
\kappa^\mathrm{m}_{\Pi\left(\bsig\right), \bsig }(a)= \sum_{\substack{{\btau \in S^D_{n}}\\{\btau\eta^{-1} \preceq \bsig\eta^{-1}}}}   \,  \mathsf{M}(\bsig\btau^{-1})\sum_{\substack{{ ( \bnu , \nu_{D+1} ) \in S^{D+1}_{n}}\\{(\bnu ,  \nu_{D+1}) \eta^{-1}\preceq ( \btau,\mathrm{id} )\eta^{-1}}}} \kappa_{\Pi_\mathrm{p}(\bnu,  \nu_{D+1} ), (\bnu,  \nu_{D+1} )}(t,\bar t) \;,
\ee
and exchanging the summations and summing the M\"oebius functions we find:
\be
\kappa^{\mathrm{m}}_{\Pi\left(\bsig\right), \bsig }(a)= \sum_{\substack{{\nu\in S_{n}}\\{ \nu \eta^{-1}\preceq \eta^{-1}}}} \kappa_{\Pi_\mathrm{p} (\bsig ,  \nu ), ( \bsig,  \nu ) }(t,\bar t) \;.
\ee

\section{Proofs of Sec.~\ref{sec:tensor-freeness}}

\subsection{Proofs of Sec.~\ref{sec:free-cumulants-paired}}
\label{sub:proof-of-sec-cum-paired}

\paragraph{Proof of Prop.~\ref{prop:cumul-of-1}.}We directly generalize the proof by induction of \cite{NicaSpeicher}. If $q=2$, the connected melonic graphs of two paired tensors are such that the two thick edges form a square with e.g. the edges of color 1, while each one of the other colored edges links the same two vertices as one of the thick edges. For such a graph, one has if one of the tensors is the identity on the appropriate space:
$\varkappa_{\mathsf{g}} (h_1, 1_{\mathbf{id}_1})= \phi_{\mathsf{g}} (h_1, 1_{\mathcal{D}_2}) -  \phi_{\mathbf{id}_1} (h_1) \phi_{\mathbf{id}_1} (1_{\mathcal{D}_2}) $, 
since $\mathsf{M}((1)(2))=1$ and $\mathsf{M}((12))=-1$. Using \eqref{eq:varphi-of-one} as well as \eqref{eq:varphi-of-one-and-contraction} and since after removing any of the thick edges (Fig.~\ref{fig:thick-edge-removal}), the invariant is $\mathbf{id}_1$, we have $\varkappa_{\mathsf{g}} (h_1, 1_{\mathcal{D}_2})= \phi_{\mathbf{id}_1} (h_1) -  \phi_{\mathbf{id}_1} (h_1)=0. $

We now assume $q>2$. Up to a relabeling, we may assume that $h_q=1_{\mathcal{D}_q}$.  From \eqref{eq:paired-free-moment-cumulant}:
\be
\phi_{\mathsf{g}}(h_1, \ldots, h_{q-1}, 1_{\mathcal{D}_q}) =  \varkappa_{\mathsf{g}} (h_1, \ldots, h_{q-1}, 1_{\mathcal{D}_q})  + \sum_{\substack{{\mathsf{h} \preceq \mathsf{g}}\;\mid\;{\mathsf{h} \neq \mathsf{g}}}} \varkappa_{\Pi(\mathsf{h}), \mathsf{h}}(\vec h)
\ee
From \eqref{eq:varphi-of-one-and-contraction}, the left-hand-side is $\phi_{\Pi(\mathsf{g}_{\setminus q})\mathsf{g}_{\setminus q}}(h_1, \ldots, h_{q-1})$. For the rightmost term, the condition $\mathsf{h} \neq \mathsf{g}$ forces $\mathsf{h}$ to have more than one connected component, and  from the induction hypothesis, $h_q$ must be alone in a component $\mathbf{id}_1$, whose contribution is $\varkappa_{\mathbf{id}_1}(1_{\mathcal{D}_q})=\phi_{\mathbf{id}_1}(1_{\mathcal{D}_q})=1$. This means that this sum can be rewritten as a sum over $\mathsf{h'} \preceq \mathsf{g}_{\setminus q}$ of  $\varkappa_{\Pi(\mathsf{h'}), \mathsf{h'}}(\vec p)$, that is, $\phi_{\Pi(\mathsf{g}_{\setminus q})\mathsf{g}_{\setminus q}}$ (since $\mathsf{g}_{\setminus q}$ might not be connected, we obtain the multiplicative extension of \eqref{eq:paired-free-moment-cumulant}). To summarize:
$$
\phi_{\Pi(\mathsf{g}_{\setminus q})\mathsf{g}_{\setminus q}}(h_1, \ldots, h_{q-1}) = \varkappa_{\mathsf{g}} (h_1, \ldots, h_{q-1}, 1_{\mathcal{D}_q})  + \phi_{\Pi(\mathsf{g}_{\setminus q})\mathsf{g}_{\setminus q}}(h_1, \ldots, h_{q-1}), 
$$
so that $\varkappa_{\mathsf{g}} (h_1, \ldots, h_{q-1}, 1)=0$.  \qed

\paragraph{Proof of Prop.~\ref{prop:cumul-of-paired-vs-non-paired}.}Let us recall the notations of the proposition:   $\mathsf{g}$ is a connected melonic graph of  $q$ paired tensors $H_1, \ldots H_q$, $H_\ell \in  \mathcal{G}^\mathrm{m}_{D}[S_\ell]$, where $S_\ell\subset \{A,B\ldots\}$, or $H_\ell \in \mathcal{G}^\mathrm{p}_{D}[\Theta_\ell, \bar \Theta_\ell]$, where $\Theta_\ell\subset \{T_a,T_b\ldots\}$ and  $\bar \Theta_\ell\subset \{\bar T_a,\bar T_b\ldots\}$,  $E$ is a set of edges of $\mathsf{g}$, the $P_\jmath$  are obtained as connected components of $\Tr_{\mathsf{g}_{\setminus E}}(\vec H)$, 
the $\mathsf{h}_\jmath$ with support $R_\jmath\subset\{1, \ldots q\}$ such that 
$\Tr_{\mathsf{h}_\jmath}(\vec H_\jmath) =  \Tr(P_\jmath)$ \eqref{eq:trace-of-one-paired-vs-trace-invariant}, and which constitute the connected components of $\mathsf{h}$, which is a melonic invariant of $\vec H$, and finally, $\mathsf{k}$ is such that  $\Tr_\mathsf{g}(\vec H) = \Tr_{\mathsf{k}}(\vec P)$  \eqref{eq:trace-of-n-paired-vs-trace-invariant}. 
One develops using \eqref{eq:paired-free-cumulant-moment}:
$$
\varkappa_{\mathsf{k}} (\vec p)= \sum_{\mathsf{k}' \preceq \mathsf{k}} \phi_{\Pi(\mathsf{k'}), \mathsf{k'}}(\vec p) \; \mathsf{M}(\mathsf{k'}, \mathsf{k}).
$$
For each $\mathsf{k'}$, from \eqref{eq:paired-invariant-vs-usual-invariant} we may gather the colored edges of $\mathsf{k'}$ together with the colored edges that link the $H_\ell$ inside each  $P_\jmath$, to form a colored graph $\mathsf{h'}$ of $H_1, \ldots H_q$. From \eqref{eq:asymptotic-relation-grouping}, one has $ \phi_{\Pi(\mathsf{k'}), \mathsf{k'}}(\vec p) =  \phi_{\Pi(\mathsf{h'}), \mathsf{h'}}(\vec h)$.

The summation over $\mathsf{k'} \preceq \mathsf{k}$ can be exchanged for a summation over  $\mathsf{h'} \preceq \mathsf{g}$, but with the condition that the edges of color in $1, \ldots D$ that are  \emph{internal} to the $P_\jmath$ are the same in $\mathsf{h'}$ and $\mathsf{g}$, or said otherwise, $\mathsf{h} \preceq \mathsf{h'}$, where we have recalled the definition of $\mathsf{h}$ at the beginning of the proof. This is because $\mathsf{h}$  is the smallest of the $\mathsf{h'} \preceq \mathsf{g}$ leaving these edges unchanged.
We will have shown that: 
\be
\label{eq:restricted-sum-group-freecum}
\varkappa_{\mathsf{k}} (\vec p)= \sum_{\mathsf{h} \preceq \mathsf{h'} \preceq \mathsf{g}} \phi_{\Pi(\mathsf{h'}), \mathsf{h'}}(\vec h) \; \mathsf{M}(\mathsf{h'},\mathsf{g}), 
\ee
if we justify that $\mathsf{M}(\mathsf{h'},\mathsf{g})=\mathsf{M}(\mathsf{k'},\mathsf{k})$. We let $\{\gamma_{c,b}\}$ be the cycles encoding $\mathsf{k}$, $\{\eta_{c,b}\}$ be the non-crossing permutations of the $\{\gamma_{c,b}\}$  whose cycles encode $\mathsf{k'}$, and similarly, 
$\{\epsilon_{c,b}\}$ the cycles encoding $\mathsf{g}$, and $\{\zeta_{c,b}\}$ the non-crossing permutations of the $\{\epsilon_{c,b}\}$  whose cycles encode $\mathsf{h'}$. We need to show that 
\be
\label{eq:moebius-is-moebius}
\prod_{c,b} \mathsf{M}(\eta_{c,b}\gamma_{c,b}^{-1})=\prod_{c,b} \mathsf{M}(\zeta_{c,b}\epsilon_{c,b}^{-1}). 
\ee

Each cycle $\epsilon_{c,b}$ entirely included in one of the $P_\jmath$ is left unchanged: $\epsilon_{c,b}=\zeta_{c,b}$,  contributing with a factor 1  to the right-hand side, so for these cycles there is nothing to prove. 
We now consider a cycle of $\gamma_{c,b}$ which is completed into a cycle of $\epsilon_{c,b}$ by the portions that are internal to some of the $\{P_\jmath\}$ and alternate edges of color $c$ and thick edges.
For each thick edge labeled $i$ for which the edge of color $c$ of $\mathsf{g}$ which is connected to its  white vertex (the input of the corresponding paired tensor) is internal to one of the  $\{P_\jmath\}$, since this edge of color $c$ is in both $\mathsf{g}$ and $\mathsf{h'}$, applying $\epsilon_{c,b}^{-1}$ and then $\zeta_{c,b}$, one returns to $i$, so the contribution to $\zeta_{c,b}\epsilon_{c,b}^{-1}$ is the singlet $(i)$.
  For the thick edges for which the edge of color $c$ connected to its  white vertex is an edge of $\mathsf{k}$ (not internal to one of the $\{P_\jmath\}$), then the cycles $\{\zeta_{c,b}\epsilon_{c,b}^{-1}\}_b$ coincide exactly with  the cycles of $\{\eta_{c,b}\gamma_{c,b}^{-1}\}_b$. Since the singlets contribute as factors 1 to the  $\mathsf{M}$ (see the definition \eqref{eq:Moebius-on-NC}), \eqref{eq:moebius-is-moebius} indeed holds.

If we also denote $\{\theta_{c,b}\}$ the non-crossing permutations of the $\{\epsilon_{c,b}\}$  whose cycles encode $\mathsf{h}$, the condition in the sum \eqref{eq:restricted-sum-group-freecum} can be expressed as summations for all  $c,b$ over $\zeta_{c,b}$ satisfying $\theta_{c,b}\preceq \zeta_{c,b}\preceq\epsilon_{c,b}$. Recalling  that each $\epsilon_{c,b}$ is a single cycle, the direct product $\bigtimes_{c,b} S_\mathrm{NC}(\epsilon_{c,b})$ is a lattice. We therefore want to apply the partial inversion formula which can be found in e.g.~Prop.~10.11 of \cite{NicaSpeicher}. Indeed, for any $\{\vartheta_{c,b}\}_{c,b}$ in $\bigtimes_{c,b} S_\mathrm{NC}(\epsilon_{c,b})$, $\mathsf{h''}$ encoded by the cycles of $\vartheta_{c,b}$ is a (not necessarily connected) melonic graph of paired tensors $\vec h$  and so taking the product of \eqref{eq:paired-free-moment-cumulant} for the connected components of $\mathsf{h''}$, one has that  $ \phi_{\Pi(\mathsf{h''}), \mathsf{h''}}(\vec h)= \sum_{\mathsf{h'}\preceq \mathsf{h''}}\varkappa_{\Pi(\mathsf{h'}), \mathsf{h'}}(\vec h)$, where by definition the sum is to be understood as a sum over  $\zeta_{c,b}$ whose cycles encode $\mathsf{h'}$. 
One can therefore apply the aforementioned proposition, obtaining:
\be
 \sum_{\mathsf{h} \preceq \mathsf{h'} \preceq \mathsf{g}} \phi_{\Pi(\mathsf{h'}), \mathsf{h'}}(\vec h) \; \mathsf{M}(\mathsf{h'},\mathsf{g}) =  \sum_{\substack{{\mathsf{h}' \preceq \mathsf{g}}\\{\Pi(\mathsf{h'})\vee\Pi(\mathsf{h}) = 1_q}}} \varkappa_{\Pi(\mathsf{h}'),\mathsf{h}'} (\vec h) .
\ee
This concludes the proof. \qed

\subsection{Proofs of Sec.~\ref{sub:different-formulations-of-asymptotic-tensor-freeness}
}

\subsubsection{Proof of Thm.~\ref{thm:equiv-tensor-freeness-cumulants-moments-mixed}
}
\label{sub:proof-equiv-tensor-freeness-cumulants-moments-mixed}

We first show that points 1 and 2 are equivalent:
\begin{lemma}
\label{lem:mixed-algebra-formulation-cumulants}
The two following conditions are equivalent:
\begin{enumerate}[label=(\alph{enumi})]
 \item For any $n\ge 2$, any $\bsig\in S_n^D$ connected and with $\omega(\bsig, \mathrm{id})=0$, and any $\vec m=(m_1, \ldots, m_n)\in \{a,b\ldots\}^n$, $\kappa^\mathrm{m}_\bsig(\vec m)=0$ whenever there exists  $1\le i<j\le n$ such that $m_i\neq m_j$.  
 
\item The two following conditions are satisfied:\begin{itemize}

\item For any $n\ge 2$, any $\bsig\in S_n^D$ connected such that $(\bsig, \mathrm{id})$ melonic with canonical pairing $\eta\neq \mathrm{id}$, and any $\vec m=(m_1, \ldots, m_n)\in \{a,b\ldots\}^n$, $\kappa^\mathrm{m}_\bsig(\vec m)=0$ whenever there exists $i\in \{1,\ldots n\}$ such that $m_i\neq m_{\eta(i)}$. 

\item For any $q\ge 2$, any paired tensors $H_1, \ldots H_q$ such that $\forall\; 1\le \ell \le q$, $H_\ell\in  \mathcal{G}^\mathrm{m}_{D}[Q_\ell]$ with $Q_\ell\in \{A,B\ldots\}$, and any connected melonic graph $\mathsf{g}$  of $H_1, \ldots H_q$, $\varkappa_\mathsf{g}(\vec h)=0$ whenever there exists $1\le \ell<\ell'\le q$ such that $Q_\ell\neq Q_{\ell'}$. 
\end{itemize}
\end{enumerate}
\end{lemma}
\proof - {\it From (b) to (a):}
By the first statement of \emph{(b)}, we only need to prove  \emph{(a)} for $\vec m$ such that for every $1\le i\le n$, $m_i=m_{\eta(i)}$. Consider $\vec m$ satisfying this property, and $\bsig\in S_n^D$ connected and with $\omega(\bsig, \mathrm{id})=0$. In the colored graph of $\bsig$, all the cycles alternating canonical pairs and thick edges correspond to the same element of $\{a,b\ldots\}$. Label each one of these cycles with $\ell\in\{1, \ldots q\}$ and for each cycle form the paired tensor  \eqref{eq:paired-tensors-for-general-melo} with $k_\ell$ the number of thick edges in the cycle number $\ell$, and $M_1^\ell=\ldots=M_{k_\ell}^\ell$ all equal to the element of $\{A,B\ldots\}$ represented by each thick edge in the cycle. Applying \emph{(b)} to the connected melonic graph $\mathsf{k}$  obtained grouping the tensors to form $H_1, \ldots H_q$, we have by \eqref{eq:kappa-paired-vs-not-for-p-particular-melo}:
$$
0=\varkappa_{\mathsf{k}} (\vec h)  = \kappa_{\bsig} ^\mathrm{m}(\vec m). 
$$

\noindent{\it From (a) to (b):} Again, we only need to prove that the property \emph{(a)} for $\vec m$ satisfying $m_i=m_{\eta(i)}$ for every $1\le i\le n$ implies the second statement of   \emph{(b)}.  Assuming   \emph{(a)}  and with $\mathsf{g}$, $\vec H$, $\vec Q$ as in the statement \emph{(b)}, we prove that $\varkappa_\mathsf{g}(\vec h) =0$. There exist $\bsig$  connected and with $\omega(\bsig, \mathrm{id})=0$ and  $\vec m$ with at least $i,j\in \{1,\ldots n\}$ such that $m_i\neq m_j$ and such that $\Tr_\mathsf{g}(\vec H)=\Tr_\bsig(\vec m)$. By \eqref{eq:kappa-paired-vs-not-for-h}:
\be
\label{eq:atobinproof}
\varkappa_\mathsf{g}(\vec h) = \sum_{\substack{{\brho\eta^{-1}\preceq \bsig\eta^{-1}}\\{\Pi(\brho\eta^{-1})\vee\Pi(\btau\eta^{-1}) = 1_n}}} \kappa^\mathrm{m}_{\Pi(\brho),\brho} (\vec m), 
\ee
where $\btau$ is defined above \eqref{eq:kappa-paired-vs-not-for-h}. We know that there exist $i,j\in \{1, \ldots n\}$ such that $m_i\neq m_j$. By construction, $\btau$ has no connected component  containing some $i,j$ for which $m_i\neq m_j$. In order for the condition $\Pi(\brho\eta^{-1})\vee\Pi(\btau\eta^{-1}) = 1_n$ to be satisfied,  $\brho$ must therefore have a connected component containing some $i,j$  for which $m_i\neq m_j$. Therefore, by \emph{(a)},  every $\brho$ in the sum \eqref{eq:atobinproof} must be such that  $\kappa^\mathrm{m}_{\Pi(\brho),\brho} (\vec m)=0$, and so $\varkappa_\mathsf{g}(\vec h) =0$. 
\qed

\

\begin{remark}
\label{remark:replace-by-algebr}
In Lemma~\ref{lem:mixed-algebra-formulation-cumulants}, one may replace $H_\ell\in  \mathcal{G}^\mathrm{m}_{D}[Q_\ell]$ by $H_\ell\in  \mathcal{G}^\mathrm{m}_{D}[Q_\ell, \un_D]$ in the second statement of  \emph{(b)}: then in order for the condition $\Pi(\brho\eta^{-1})\vee\Pi(\btau\eta^{-1}) = 1_n$ to be satisfied, $\brho$ must have a connected component containing either $1_D$, or some $i,j$  for which $m_i\neq m_j$, and the same conclusions hold. By multilinearity, one may also replace $\mathcal{G}^\mathrm{m}_{D}[Q_\ell]$ by $ \mathcal{A}^\mathrm{m}_{D}[Q_\ell]$ or $\mathcal{A}^\mathrm{m}_{D}[Q_\ell; \un_D]$.
\end{remark}

We then show the equivalence between point 2 and point  3. We start with the easiest part. 
\begin{lemma}
\label{lem:proof-of-freeness-easy-part}
The two following assertions are equivalent: 
\begin{enumerate}[label=\emph{(a-\normalfont{\roman*})}]
\item For any $n\ge 2$,  $\bsig\in S_n^D$ connected and with $(\bsig, \mathrm{id})$ melonic with canonical pairing $\eta\neq \mathrm{id}$, and any $\vec m=(m_1, \ldots, m_n)\in \{a,b\ldots\}^n$, $\kappa^\mathrm{m}_\bsig(\vec m)=0$ whenever there exists  $i\in \{1,\ldots n\}$ such that $m_i\neq m_{\eta(i)}$. 
\item For any $n\ge 2$,  $\bsig\in S_n^D$ connected and $(\bsig, \mathrm{id})$ melonic with canonical pairing $\eta\neq \mathrm{id}$, and any $\vec m=(m_1, \ldots, m_n)\in \{a,b\ldots\}^n$, $\varphi^\mathrm{m}_\bsig(\vec m)=0$ whenever there exists $i\in \{1,\ldots n\}$ such that $m_i\neq m_{\eta(i)}$. 
\end{enumerate}
\end{lemma}

\proof We fix $\bsig\in S_n^D$ , $\eta\in S_n$,  $\vec m=(m_1, \ldots, m_n)$ as in the lemma and $i$ such that $m_i\neq m_{\eta(i)}$. 

\

\noindent{\it From \emph{(a-ii)} to \emph{(a-i)}:}   From Thm.~\ref{thm:limit-for-wishart-tensor-first-order}:
\be
\kappa_{\bsig}^\mathrm{m}(\vec m)= \sum_{\substack{{\btau\eta^{-1} \preceq \bsig\eta^{-1}}}}   \varphi^\mathrm{m} _{\Pi(\btau), \btau} (\vec m)\,  \mathsf{M}(\bsig\btau^{-1}),
\ee
 Since  $m_i\neq m_{\eta(i)}$, then the cycle in the graph which alternates thick edges and canonical pairs and  contains the thick edge labeled $i$ is a separating cycle.  Since any $\btau$ satisfying $\btau\eta^{-1} \preceq \bsig\eta^{-1}$ is such that $(\btau, \mathrm{id})$ has the same canonical pairing $\eta$, and since the permutation defining the thick edges of $\btau$ is still the identity, there is for every $\btau$ in the sum a connected component $\btau_j$ with support $B_j$ with the same separating cycle, which satisfies  $\omega(\btau_j, \mathrm{id}_{\lvert B_j})=0$, has canonical pairing $\eta_{\lvert B_j}$, and contains the same thick edge $i$ for which $m_i\neq m_{\eta_{\lvert B_j}(i)}$, for which $\varphi^\mathrm{m}_{\btau_j}(\vec m_{\lvert B_j})=0$. From (a-ii) every term in the sum vanishes, and so does $\kappa_{\bsig}^\mathrm{m}(\vec m)$. 
 
 \

\noindent{\it From \emph{(a-i)} to \emph{(a-ii)}:} Reciprocally, one has from Thm.~\ref{thm:limit-for-wishart-tensor-first-order}: 
\be
\varphi_{\bsig}^\mathrm{m}(\vec m) = \sum_{\substack{{\btau\eta^{-1}\preceq\bsig\eta^{-1}}}} \kappa^\mathrm{m}_{\Pi\left(\btau\right), \btau}(\vec m). 
\ee
The discussion is exactly the same as above, and assuming (a-i) we find that $\varphi^\mathrm{m}_{\bsig}(\vec m) =0$ for any $\bsig$ connected with $(\bsig, \mathrm{id})$ melonic, and therefore the same holds for $\varphi^\mathrm{m}_{\bsig}(\vec m)$.  \qed

\ 

 We recall that for $\mathsf{g}$ a melonic graph of paired tensors $H_1, \ldots H_p$ with $H_\ell \in  \mathcal{G}^\mathrm{m}_{D}[R_\ell]$ with $R_\ell\in \{A,B\ldots\}$,  $(\mathsf{g}, \vec H)$ is \emph{almost alternating} if it has at least one edge linking paired tensors $H_\ell$, $H_{\ell'}$ generated by different tensors $R_\ell\neq R_{\ell'}$, and \emph{at most one} edge of $\mathsf{g}$ which links two paired tensors $H_\ell$, $H_{\ell'}$ generated by the same $R_\ell=R_{\ell'}$. It is strictly alternating if different any two paired tensors linked by an edge are generated by different elements of $\{A,B\ldots\}$. In order to prove the equivalence of the remaining statements, we will need the following key lemma.

\begin{lemma}
\label{lem:almost-alternating-subgraphs}
Let $\mathsf{g}$ be a connected melonic graph of the paired tensors $H_1, \ldots H_q$, $H_\ell\in \mathcal{G}^\mathrm{m}_D[Q_\ell]$ with $Q_\ell\in \{A,B\ldots\}$, or $H_\ell\in \mathcal{G}^\mathrm{p}_D[X_\ell, \bar X_\ell]$ with   $ (X_\ell, \bar X_\ell) \in \{(T_a, \bar T_a),(T_b, \bar T_b)\ldots\}$, such that $(\mathsf{g}, \vec H)$ is almost alternating, and consider $\mathsf{h}\preceq\mathsf{g}$. Then either $(\mathsf{h}, \vec H)$ has an almost alternating connected component, or it has a connected component consisting of a single paired tensor. 
\end{lemma}
\proof We consider a system of cycles $\{\gamma_{c,b}\}$ encoding  $\mathsf{g}$ and 
the non-crossing permutations $\{\tau_{c,b}\}$ with $\tau_{c,b}\preceq \gamma_{c,b}$, whose cycles encode $\mathsf{h}$.

The statement of the lemma holds if $\mathsf{h}=\mathsf{g}$, because $(\mathsf{g}, \vec H)$ is almost alternating. Otherwise, $\tau_{c,b}\neq \gamma_{c,b}$ for some $c,b$. We can go from  $\{\gamma_{c,b}\}$ to $\{\tau_{c,b}\}$ by a sequence $\{\rho_{c,b}{(i)}\}$, $i=1, \ldots, k$,  where $\rho_{c,b}{(1)}=\gamma_{c,b}$, $\rho_{c,b}{(k)}=\tau_{c,b}$, and there is a single $c,b$ for which $\rho_{c,b}{(i)}\neq \rho_{c,b}{(i+1)}$ (all the others are equal), and for this $c,b$ one has: $\rho_{c,b}{(i+1)}\preceq \rho_{c,b}{(i)}$ and $\#( \rho_{c,b}{(i+1)}) = \#( \rho_{c,b}{(i)}) + 1$. At this step, a pair of edges of color $c$ \emph{that both belong to the same cycle} $\rho_{c,b}{(i)}$ has been \emph{flipped}: they remain connected to the same white vertices, but the black vertices to which they are connected have been exchanged, dividing the cycle of $\rho_{c,b}{(i)}$ to which they belonged in two. Note that this corresponds to a multiplication of $\rho_{c,b}{(i)}$ by an appropriate transposition. 

We want to show that at each step $i$ from 2 to $k$, the following property is true:  \emph{there exist \emph{two} connected components of $\{\rho_{c,b}^{(i)}\}$ whose colored edges are all the same as in  $\mathsf{g}$,  with the exception of a single edge.} We call this property $P(i)$. We will update a pair $(\mathcal{C}_1(i), \mathcal{C}_2(i))$ of connected components of this kind (which ``validate'' $P(i)$).

$P(2)$ holds because a single pair of edges has been flipped to go from $\{\gamma_{c,b}\}$ to $\{\rho_{c,b}^{(2)}\}$, creating two connected components  $(\mathcal{C}_1(2), \mathcal{C}_2(2))$ which validate $P(2)$. 

At the next step, the pair of edges to be flipped belong to the same connected component of  $\{\rho_{c,b}^{(2)}\}$, say, $\mathcal{C}_1(2)$. Flipping these edges creates two connected components: one of them has two edges that were not in $\mathsf{g}$, while the other, which we call $\mathcal{C}_1(3)$, has a single edge that is not in $\mathsf{g}$. Setting $\mathcal{C}_2(3)=\mathcal{C}_2(2)$, we have a pair which validates $P(3)$. 
At every step that follows, either the two edges are flipped in a connected component which is neither $\mathcal{C}_1(i)$ nor  $\mathcal{C}_2(i))$, in which case $\mathcal{C}_1(i+1):=\mathcal{C}_1(i)$ and  $\mathcal{C}_2(i+1):=\mathcal{C}_2(i)$, or it is performed in say $\mathcal{C}_1(i)$, in which case as for step 2 to 3 a component $\mathcal{C}_1(i+1)$ is created with a single edge not in $\mathsf{g}$, and we update   $\mathcal{C}_2(i+1):=\mathcal{C}_2(i)$. This proves that $P(i)$ holds at every step.

\

Having proven that $P(k)$ is true, we consider the following situations. 
\begin{enumerate}[label=$-$] 
\item If there is no edge of $\mathcal{C}_1(k)$ or $\mathcal{C}_2(k)$ that is also in $\mathsf{g}$, or if all the edges of $\mathcal{C}_1(k)$ or $\mathcal{C}_2(k)$ that are also in $\mathsf{g}$ link the inputs and outputs of a paired tensor, then this connected component of~$\mathsf{h}$ consists of a single paired tensor with a number of edges linking its inputs and outputs, that is, a melonic graph of a single paired tensor, and the statement of the lemma holds. 
\end{enumerate}
Otherwise both $\mathcal{C}_1(k)$ and $\mathcal{C}_2(k)$ contain edges that also belong to $\mathsf{g}$ or  link paired tensors $H_\ell$, $H_{\ell'}$ with $\ell\neq \ell'$. 
Since $\mathsf{g}$ is almost alternating,  it  has at most one edge $e$ which links two paired tensors generated by the same $Q_\ell\in \{A,B\ldots\}$ (or the pure equivalent, by $ (X_\ell, \bar X_\ell) \in \{(T_a, \bar T_a),(T_b, \bar T_b)\ldots\}$). 
\begin{enumerate}[label=$-$] 
\item  If $\mathsf{g}$ is \emph{strictly} alternating, or if $e$ exists but has been flipped at some step and is therefore not as in $\mathsf{g}$, or if $e$ remains in  $\mathsf{h}$ as it was in $\mathsf{g}$  but does not belong to $\mathcal{C}_1(k)$ or $\mathcal{C}_2(k)$, then  there is at least one edge in  $\mathcal{C}_1(k)$ and one in $\mathcal{C}_2(k)$ that are also in $\mathsf{g}$ and the edges of $\mathcal{C}_1(k)$ and $\mathcal{C}_2(k)$ that are also in $\mathsf{g}$   still all link in $\mathsf{h}$ paired tensors generated by different elements of $\{A,B\ldots\}$.
 There is only one edge of $\mathcal{C}_1(k)$ which is not in $\mathsf{g}$, so it might link two paired tensors $H_\ell$, $H_{\ell'}$  of $\mathcal{C}_1(k)$ with $\ell\neq \ell'$ and which are generated by the same element of $\{A,B\ldots\}$, but no other edge of $\mathcal{C}_1(k)$ can have this property, so $\mathcal{C}_1(k)$ is almost alternating (and $\mathcal{C}_2(k)$ as well). 
 \end{enumerate}
Otherwise, $e$ exists in $\mathsf{g}$ and remains in  $\mathsf{h}$, and it belongs to say $\mathcal{C}_1(k)$. 
\begin{enumerate}[label=$-$] 
\item  We then consider $\mathcal{C}_2(k)$ (this is why we needed  \emph{two} connected components  to validate $P(i)$): it has at most one edge linking paired tensors generated by the same element of $\{A,B\ldots\}$, and it has at least one edge linking  paired tensors generated by different elements of $\{A,B\ldots\}$: it is almost alternating. \qed
 \end{enumerate}

\begin{lemma}
\label{lem:hard-part-of-proof-freeness-mixed}
The two following assertions are equivalent: 
\begin{enumerate}[label=\emph{(b-\normalfont{\roman*})}]
\item For any $q\ge 2$, any paired tensors $H_1, \ldots H_q$ such that $\forall 1\le \ell \le q$, $H_\ell\in  \mathcal{G}^\mathrm{m}_{D}[Q_\ell]$ with $Q_\ell\in \{A,B\ldots\}$, and any connected melonic graph $\mathsf{g}$  of $H_1, \ldots H_q$, $\varkappa_\mathsf{g}(\vec h)=0$ whenever there exists $1\le \ell<\ell'\le q$ such that $Q_\ell\neq Q_{\ell'}$. 
\item For any $q\ge 2$, any $\mathcal{D}_1, \ldots \mathcal{D}_q\ge 1$, any paired tensors $H_1, \ldots H_q$ such that for $1\le \ell \le q$, $H_\ell\in \mathcal{G}^\mathrm{m}_{D, \mathcal{D}_\ell} [Q_\ell]$ with $Q_\ell\in \{A,B\ldots\}$, and any connected melonic graph $\mathsf{g}$  of $H_1, \ldots H_q$,  letting for each $\ell$:  $h'_\ell= h_\ell - \phi(h_\ell) 1_{\mathcal{D}_\ell}$, then $\phi_\mathsf{g}(\vec h')=0$ whenever $(\mathsf{g}, \vec h)$ is almost alternating. 
\end{enumerate}
\end{lemma}

\proof We generalize the proof of Thm.~11.16 of \cite{NicaSpeicher}. We fix some paired tensors $H_1, \ldots H_q$, $H_\ell\in \mathcal{G}^\mathrm{m}_{D, \mathcal{D}_\ell} [Q_\ell]$ with $Q_\ell\in \{A,B\ldots\}$,   and a connected melonic graph $\mathsf{g}$  of these paired tensors.

\

\noindent{\it From \emph{(b-i)} to \emph{(b-ii)}:} We assume that (b-i) holds and  that  $(\mathsf{g}, \vec h)$ is almost alternating. Recalling the free moment-cumulant relation for melonic invariants of paired tensors  \eqref{eq:paired-free-moment-cumulant}: 
\be
\label{eq:mom-cum-centered}
\phi_{\mathsf{g}} (\vec h')= \sum_{\mathsf{h} \preceq \mathsf{g}} \varkappa_{\Pi(\mathsf{h}), \mathsf{h}}(\vec h'). 
\ee
From Lemma~\ref{lem:almost-alternating-subgraphs}, since $(\mathsf{g}, \vec h)$ is almost alternating, every $\mathsf{h} \preceq \mathsf{g}$ is such that:\\[+1ex]
$-$ Either  $(\mathsf{h}, \vec h)$ has an almost alternating connected component with support $B\subset\{1,\ldots, q\}$, which in particular has at least two paired tensors $H_\ell$, $H_{\ell'}$  with $\ell\neq \ell'\in B$ and generated by different elements $Q_\ell\neq Q_{\ell'}$. We focus on the contribution to the right-hand-side of \eqref{eq:mom-cum-centered} for such a $\mathsf{h}$, and using the multilinearity of $\varkappa$, we expand all the entries of the factor $\varkappa_{\mathsf{h}_{\lvert_B}}(\vec h'_{\lvert_B})$. From Prop.~\ref{prop:cumul-of-1}, the terms involving $1_D$ vanish,  and by (b-i), the term involving $\vec H_{\lvert B}$ also vanishes. \\[+1ex]
$-$ Or  $(\mathsf{h}, \vec h)$ has a connected component which is a melonic graph of a single paired tensor, contributing to a factor $\varkappa_{\mathbf{id}_1}(h'_\ell)=\phi_{\mathbf{id}_1}(h'_\ell)=\phi(h'_\ell)=0$ by linearity of $\phi$. 

\

\noindent{\it From \emph{(b-ii)} to \emph{(b-i)}:} We assume that (b-ii) holds  and that there exists $1\le \ell<\ell'\le q$ such that $Q_\ell\neq Q_{\ell'}$. We recall the free cumulant-moment relation for melonic invariants of paired tensors \eqref{eq:paired-free-cumulant-moment}:
\be
\label{eq:recall-paired-free-cumom}
\varkappa_{\mathsf{g}} (\vec h)= \sum_{\mathsf{h} \preceq \mathsf{g}} \phi_{\Pi(\mathsf{h}), \mathsf{h}}(\vec h) \; \mathsf{M}(\mathsf{h}, \mathsf{g}). 
\ee

In order to prove (b-i), we proceed by induction on the number  $q$ of paired tensors $H_1, \ldots, H_q$ considered. If $q=2$, as in the proof of Prop.~\ref{prop:cumul-of-1} in Sec.~\ref{sub:proof-of-sec-cum-paired}, the connected melonic graphs of two paired tensors are such that the two thick edges form a square with e.g.~the edges of color 1, while each one of the other colored edges links the same two vertices as one of the thick edges.
 From \eqref{eq:recall-paired-free-cumom} and Prop.~\ref{prop:cumul-of-1}, letting $h_\ell' =  h_\ell -\phi(h_\ell)1_{\mathcal{D}_\ell}$, we have for any such invariant
$\varkappa_{\mathsf{g}} (h_1, h_2)=\varkappa_{\mathsf{g}} (h'_1, h'_2)= \phi_{\mathsf{g}} (h'_1, h'_2)$, 
due to the fact that $\mathsf{M}((1)(2))=1$ and $\mathsf{M}((12))=-1$. Since by assumption $Q_1\neq Q_2$, $(\mathsf{g}, \vec h)$ is alternating,  so by (b-ii) one has $\phi_{\mathsf{g}} (h'_1, h'_2)=0$.

We now assume that $q>2$, that (b-ii) holds, and that (b-i) holds for any $q'<q$, and we prove that (b-i) holds for $q$ paired tensors. In order to apply (b-ii), we need to group the paired tensors $H_1, \ldots, H_q$ so that the invariant of paired tensors for the different groups is \emph{alternating}. This is described at the end of Sec.~\ref{sub:grouping}: removing  $E^\mathrm{tot}=E^{\neq}\cup E^{\mathrm{int}}$, where $E^{\neq}$ is the set of edges linking paired tensors generated by different  $Q_\ell$ and $E^{\mathrm{int}}$  contains one edge per cycle of $\mathsf{g}_{\setminus E^{\neq}}$ that alternates edges of color $c$ and paired inputs and outputs, we consider the paired tensors  $P^\neq_\jmath$ corresponding to the connected components of $\mathsf{g}_{\setminus E^{\mathrm{tot}}}$  and the colored graph $\mathsf{k}^\neq$ whose colored edges are the edges of  $E^{\mathrm{tot}}$ and whose thick edges represent the paired tensors $\vec P^\neq$. The invariant encoded by $\mathsf{k}^\neq$ is melonic for the paired tensors $\vec P^\neq$, and $(\mathsf{k}^\neq, \vec F^\neq)$ is alternating. Applying \eqref{eq:recall-paired-free-cumom} to $\varkappa_{\mathsf{k}^\neq} (\vec f^\neq)=0$, centering the $\vec f^\neq$ and repeating the argument of the first part of this proof (from (b-i) to (b-ii)) with the application of (b-ii) instead of (b-i), one finds that $\varkappa_{\mathsf{k}^\neq} (\vec f^\neq)=0$.

We can apply Prop.~\ref{prop:cumul-of-paired-vs-non-paired} with $\mathsf{h}$ replaced by $\mathsf{h}^\neq$ (notation introduced in Sec.~\ref{sub:grouping}):
\begin{align}
0=\varkappa_{\mathsf{k}^\neq} (\vec f^\neq)  = \sum_{\substack{{\mathsf{h}' \preceq \mathsf{g}}\\{\Pi(\mathsf{h'})\vee\Pi(\mathsf{h}^\neq) = 1_q}}} \varkappa_{\Pi(\mathsf{h}'),\mathsf{h}'} (\vec h)  =\; \varkappa_{\mathsf{g}}(\vec h)\  +  \sum_{\substack{{\mathsf{h}' \prec \mathsf{g}}\\{\Pi(\mathsf{h'})\vee\Pi(\mathsf{h}^\neq) = 1_q}}} \varkappa_{\Pi(\mathsf{h}'),\mathsf{h}'} (\vec h),
\end{align}
where in the rightmost sum, $\preceq$ has been replaced by $\prec$ to indicate that the case $\mathsf{h}' = \mathsf{g}$ is excluded. As in the proof of Lemma~\ref{lem:mixed-algebra-formulation-cumulants}, since by assumption there exist  $1\le \ell<\ell'\le q$ such that $Q_\ell\neq Q_{\ell'}$, and since by definition of $\mathsf{h}^\neq$, the connected components of $\mathsf{h}^\neq$ only involve $H_\ell, H_{\ell'}$ generated by the same element $Q_\ell =Q_{\ell'}$, the condition $\Pi(\mathsf{h'})\vee\Pi(\mathsf{h}^\neq) = 1_q$ imposes  every $\mathsf{h'}$ in the rightmost sum to have at least one connected component involving both a $H_\ell$ and a  $H_{\ell'}$ with $Q_\ell\neq Q_{\ell'}$. Since $\mathsf{h}' \prec \mathsf{g}$ we can apply the induction hypothesis to $\mathsf{h'}$, which is also a melonic invariant of paired tensors, thus concluding that the rightmost term is zero. We have therefore shown that $\varkappa_{\mathsf{g}}(\vec h)=0$, which concludes the proof. \qed

\begin{remark}
\label{remark:replace-by-algebr-second-part}
In Lemma~\ref{lem:hard-part-of-proof-freeness-mixed}, one may replace in (b-i) $\mathcal{G}^\mathrm{m}_{D}[Q_\ell]$ by $\mathcal{A}^\mathrm{m}_{D}[Q_\ell; \un_D]$ and in (b-ii) $H_\ell\in  \mathcal{G}^\mathrm{m}_{D, \mathcal{D}_\ell}[Q_\ell]$ and $\phi_\mathsf{g}(\vec h')=0$ by  $H_\ell\in  \mathcal{A}^\mathrm{m}_{D}[Q_\ell; \un_D]$ satisfying $\phi(h_\ell)=0$ and $\phi_\mathsf{g}(\vec h)=0$. The proof is \emph{mutatis mutandis} the same. 
\end{remark}

\subsubsection{Proof of Thm.~\ref{thm:equiv-tensor-freeness-cumulants-moments-pure}}
\label{sec:equiv-tensor-freeness-cumulants-moments-pure}
Since the proof is very close to the mixed case, we only specify what differs. We start with the equivalence between the first and second point (analogous to Lemma~\ref{lem:mixed-algebra-formulation-cumulants}).

\ 

In the pure case, the fact that the second point implies the first point is trivial. We can assume that for all $i$, $\overline{x_i}=x_{\bar i}'$. In the second statement, one can just choose each $H_\ell$ to be $Q_\ell=X_\ell\otimes \bar X_\ell$ (and not a more general element of $\mathcal{G}^\mathrm{m}_{D}[Q_\ell]$). Then a melonic graph $\mathsf{g}$ of $\vec H$ is a purely connected and melonic $\bsig$ with canonical pairing the identity and $1\le \ell\le q$ identifies with $1\le i \le n$. Since $Q_\ell = X_\ell\otimes \bar X_\ell$ with $(X_\ell, \bar X_\ell) \in \{(T_a, \bar T_a),(T_b, \bar T_b)\ldots\}$, one  has $Q_i\neq Q_{j}$ for some $i\neq j$  if and only if $x_i\neq x_j$  (which  is equivalent to  $x_{\bar i}'\neq x_{\bar j}'$ and $\overline{ x_i} \neq x_{\bar j}'$ due to the fact that $\overline{x_i}=x_{\bar i}'$).

Conversely, the fact that the first point implies the second point is proved exactly in the same manner as for the mixed case: since it is obvious if $\overline{x_i}\neq x_{\bar i}'$ for some $i$, one only needs to prove it under the assumptions  that $\overline{x_i}=x_{\bar i}'$ for all $i$, and that there exists $i,j$ such that $x_i\neq x_j$. In the pure case, $\bsig$ first order is purely connected melonic and with canonical pairing the identity. One has to consider \eqref{eq:kappa-paired-vs-not-for-h-pure} instead of \eqref{eq:atobinproof}. It is still true that  $\overline{x_i}=x_{\bar i}'$ for all $i$ both in $\btau$ and $\brho$, that $\btau$ has no connected component involving some $x_i\neq x_j$,  and that if $\brho$ was to also satisfy the same property, then it would not be possible to satisfy $\Pi(\brho)\vee\Pi(\btau)=1_n$, so all $\brho$ in the sum \eqref{eq:kappa-paired-vs-not-for-h-pure} have a connected component involving $x_i\neq x_j$, and $\varkappa_\mathsf{g}(\vec h)=0$.

\ 

We pursue with the equivalence between the first statement of the second point and the first statement of the third point (analogous to Lemma~\ref{lem:proof-of-freeness-easy-part}). Here nothing changes with the mixed case, due to the fact that if there exists $i$ such that $\overline{x_i}\neq x_{\bar i}'$, then it is also true for one connected component of every $\btau$ such that  $\btau\preceq \bsig$ (here $\eta=\mathrm{id}$). 

\ 

Remains the equivalence between the second statements of the second and third points of the theorem (analogous to Lemma~\ref{lem:hard-part-of-proof-freeness-mixed}). For this part, the proof is as for the mixed case. \qed


\begin{thebibliography}{99}




  \bibitem{NicaSpeicher}
  A.~Nica and R.~Speicher, 
  ``Lectures on the combinatorics of free probability," 
  volume 13. Cambridge University Press, 2006.
  
  
  \bibitem{Collins03} B.~Collins, ``Moments and Cumulants of Polynomial random variables on unitary groups, the Itzykson Zuber integral and free probability'', Int. Math. Res. Not. {\bf 2003}(17), 953-982 (2003).

\bibitem{CMSS}
B.~Collins, J.A.~Mingo, P.~\'Sniady, R.~Speicher,
``Second Order Freeness and Fluctuations of Random Matrices, III. Higher order freeness and free cumulants'',
Documenta Mathematica {\bf12} (2007) 1-70.







\bibitem{1Nexpansion1}
R.~Gurau, ``The $1/N$ expansion of colored tensor models'',  Annales Henri Poincar\'e 12:829-847, 2011. 

\bibitem{1Nexpansion2}
R.~Gurau, V.~Rivasseau, ``The $1/N$ expansion of colored tensor models in arbitrary dimension'', Europhys. Lett. 95:50004, 2011.

\bibitem{1Nexpansion3}
R.~Gurau, ``The Complete $1/N$ Expansion of Colored Tensor Models in Arbitrary Dimension'', Ann. Henri Poincar\'e 13, 399–423 (2012). 

\bibitem{critical}
V.~Bonzom, R.~Gurau, A.~Riello, V.~Rivasseau, ``Critical behavior of colored tensor models in the large $N$ limit'', Nucl. Phys. B853 (2011) 174-195. 







\bibitem{Gurau-universality}
R.~Gurau,``Universality for random tensors'',  Annales de l'I.H.P. Probabilités et statistiques 50.4 (2014): 1474-1525.

\bibitem{uncoloring}
V.~Bonzom, R.~Gurau, and V.~Rivasseau, ``Random tensor models in the large $N$ limit: Uncoloring the colored tensor models'', Physical Review D, 85:084037, 2012.

\bibitem{bonzom-SD}
V.~Bonzom, ``Revisiting random tensor models at large $N$ via the Schwinger-Dyson equations'', J. High Energ. Phys. 2013, 160 (2013). 

\bibitem{Gurau-book}
R,~Gurau, ``Random Tensors'', Oxford University Press, 2016.

\bibitem{enhanced-1}
V.~Bonzom, T.~Delepouve, V.~Rivasseau, ``Enhancing non-melonic triangulations: A tensor model mixing melonic and planar maps'', Nuclear Physics B, 895:161–191, 2015. 

\bibitem{Walsh-maps}
V.~Bonzom, L.~Lionni, and V.~Rivasseau, ``Colored triangulations of arbitrary dimensions are stuffed Walsh maps'', Electronic Journal of Combinatorics, 24(1):\#P1.56, 2017.

\bibitem{multicritical}
L.~Lionni and J.~Th\"urigen, ``Multi-critical behaviour of 4-dimensional tensor models up to order 6'', Nucl. Phys. B 941 (2019) 600-635.  

\bibitem{bonzom-review}
V.~Bonzom, ``Large $N$ Limits in Tensor Models: Towards More Universality Classes of Colored Triangulations in Dimension $d\ge2$'', SIGMA 12 (2016), 073. 


\bibitem{Lionni-thesis}
L.~Lionni, ``Colored discrete spaces: higher dimensional combinatorial maps and quantum gravity'', PhD thesis, Universit\'e Paris-Sud 2017, Springer thesis 2018, 2017.


\bibitem{Gurau:2012ix}
R.~Gurau,
``The Schwinger Dyson equations and the algebra of constraints of random tensor models at all orders,''
Nucl. Phys. B \textbf{865}, 133-147 (2012).

\bibitem{Gurau:2013pca}
R.~Gurau,
``The 1/N Expansion of Tensor Models Beyond Perturbation Theory,''
Commun. Math. Phys. \textbf{330}, 973-1019 (2014).






\bibitem{Tensors1}
A. Auffinger, G.B. Arous, and J. Cerny, ``Random Matrices and Complexity of Spin Glasses’’, Comm. Pure Appl. Math., 66, 165-201 (2013). 


\bibitem{Tensors2}
A.~Montanari and E.~Richard, ``A statistical model for tensor PCA’’,  Advances in Neural Information Processing
Systems (NIPS). Montréal, Canada, 2014.

\bibitem{Tensors3}
N. Sasakura, “Real tensor eigenvalue/vector distributions of the Gaussian tensor model via a four-fermi theory,” PTEP 2023, no.1, 013A02 (2023). 

\bibitem{Tensors4}
Razvan Gurau, ``On the generalization of the Wigner semi-circle law to real symmetric tensors’’,  arXiv:2004.02660
(2020).

\bibitem{Tensors5}
J. H.~de M. Goulart, R.~Couillet, and P.~Comon, ``A Random
Matrix Perspective on Random Tensors’’,.  Journal of Ma-
chine Learning Research (2022).


\bibitem{Tensors6}
N.~Cheng, C.~Lancien, G.~Penington, M.~Walter, and F.~Witteveen. “Random Tensor Networks with Non-trivial Links’’, Annales Henri Poincaré 25.4
(2024), pp. 2107–2212.

\bibitem{Tensors7}
B.~Au and J.~Garza-Vargas,``Spectral asymptotics for contracted tensor ensembles’’, Electron. J. Probab., 28:Paper No. 113, 32, 2023. 

\bibitem{Tensors8}
S.~Majumder, N.~Sasakura, ``Three cases of complex eigenvalue/vector distributions of symmetric order-three random tensors’’, PTEP 2024, no. 9, 093A01 (2024).

\bibitem{Tensors9}
S.~Dartois and B.~McKenna, ``Injective norm of real and complex random tensors I: From spin glasses to geometric entanglement’’,  arXiv preprint
arXiv:2404.03627 (2024).

\bibitem{Tensors10}
R.~Bonnin, ``Universality of the Wigner-Gurau limit
for random tensors’’,  arXiv:2404.14144
(2024).

\bibitem{Tensors11}
D.~Kunisky, C.~Moore, and A.~Wein, ``Tensor cumulants
for statistical inference on invariant distributions’’,  in 2024 IEEE 65th Annual Symposium on Foundations of Computer Science (FOCS), Chicago, IL, USA, 2024, pp. 1007-1026.

\bibitem{Tensors12}
R.~Bonnin and C.~Bordenave, ``Freeness for tensors’’, arXiv:2407.18881.




  \bibitem{IonReview}
 B.~Collins, I.~Nechita, 
  ``Random matrix techniques in quantum information theory,"
  Journal of Mathematical Physics {\bf 57}, 015215 (2016).

    \bibitem{DNL}
S.~Dartois, L.~Lionni, I.~Nechita, ``On the joint distribution of the marginals of multipartite random quantum states'', Random Matrices: Theory and Applications, {\bf 9}(03):2050010, 2020.



\bibitem{CGL} 
B.~Collins, R.~Gurau, L.~Lionni, 
``The tensor Harish-Chandra--Itzykson--Zuber integral I: Weingarten calculus and a generalization of monotone Hurwitz numbers'', accepted for publication in J. Eur. Math. Soc. (JEMS).

\bibitem{CGL2} 
B.~Collins, R.~Gurau, L.~Lionni, 
``The tensor Harish-Chandra--Itzykson--Zuber integral II: detecting entanglement in large quantum systems'',
Commun. Math. Phys. 401, 669–716 (2023).




\bibitem{LU-1}
R.~Horodecki, P.~Horodecki, M.~Horodecki, and K.~Horodecki,
``Quantum entanglement'',
Rev. Mod. Phys. 81, 865. 

\bibitem{LU-2}
B.~Kraus, ``Local Unitary Equivalence of Multipartite Pure States'',
Phys. Rev. Lett.  {\bf104}, 020504.  

\bibitem{LU-3}
B.~Kraus.
``Local unitary equivalence and entanglement of multipartite pure states'', 
Phys. Rev. A {\bf 82}, 032121. 

\bibitem{LU-4}
M.~Walter, D.~Gross, and J.~Eisert, 
``Multipartite Entanglement'',
 In Quantum Information (eds D. Bruss and G. Leuchs), 2016.
 


 
 
 \bibitem{random-state1}
P.~Hayden, D.~Leung, P.~W.~Shor and A.~Winter, ``Randomizing quantum states: constructions and applications," Comm.~Math.~Phys.~250 (2004), 371--391.

\bibitem{random-state2}
P.~Hayden, D.W.~Leung, and A.~Winter, ``Aspects of generic entanglement'', Communications in Mathematical Physics, 265(1):95–117, 2006.



\bibitem{random-state4}
B.~Collins, I.~Nechita, K.~Zyczkowski, ``Random graph states, maximal flow and Fuss-Catalan distributions'', J. Phys. A: Math. Theor. 43 (2010), no. 27, 275303.

\bibitem{random-state5}
G.~Aubrun, ``Partial transposition of random states and non-centered semicircular distributions'', Random Matrices Theor. Appl. 1 1250001.



\bibitem{random-state6}
G.~Aubrun, I.~Nechita, ``Realigning quantum states'', Journal of Mathematical Physics, 53(10):102210, 2012.

\bibitem{random-state3}
B.~Collins, I.~Nechita, K.~Zyczkowski, ``Area law for random graph states'', J. Phys. A Math. Theor. 46(30), 305302 (2013).

\bibitem{random-state7}
M.~Christandl, B.~Doran, S.~Kousidis, and M.~Walter, ``Eigenvalue distributions of reduced density matrices'', Communications in Mathematical Physics, 332(1):1–52, 2014.

\bibitem{random-state8}
G.~Aubrun, S.J. Szarek, D.~Ye, ``Entanglement thresholds for random induced states'', Comm. Pure Appl. Math., 67 (1) (2014) 129-171.


\bibitem{funwithreplicas}
G.~Penington, M.~Walter, F.~Witteveen
``Fun with replicas: tripartitions in tensor networks and gravity'',
 J. High Energ. Phys. 2023, 8 (2023).

 \bibitem{random-state-52}
 T.~Banica, I.~Nechita, ``Asymptotic eigenvalue distributions of block-transposed Wishart matrices'', 	J. Theoret. Probab. 26 (2013), 855-869.
 
\bibitem{MingoPopa1}
J.A.~Mingo, M.~Popa, ``Freeness and the partial transposes of Wishart random matrices'', Can. J. Math. 71 (2019) 659-681.

\bibitem{MingoPopa2}
J.A.~Mingo, M.~Popa, ``On partial transposes of unitarily invariant random matrices'', Int. J. Mathematics 35, no. 11, 2450043 (2024).


 
 
 
 
 
 
 
  \bibitem{poly-LU-inQI}
M.S.~Leifer, N.~Linden, and A.~Winter,
``Measuring polynomial invariants of multiparty quantum states'',
Phys. Rev. A 69, 052304.


\bibitem{Heroetal}
M.~W. Hero, J.~F. Willenbring, L.~Kelly Williams,
``The measurement of quantum entanglement and enumeration of graph coverings'',  arXiv:0911.0222 [math.RT].


\bibitem{Vrana1}
P.~Vrana, ``An algebraically independent generating set of the algebra of local unitary invariants'',  	arXiv:1102.2861 [quant-ph].

\bibitem{Vrana2}
P.~Vrana, ``The algebra of local unitary invariants of identical particles'',  	arXiv:1107.2438 [quant-ph].


	\bibitem{BGR1}
	J.~Ben Geloun, S.~Ramgoolam, 
	Counting Tensor Model Observables and Branched Covers of the 2-Sphere
	Ann. Inst. H. Poincare Comb. Phys. Interact. 1 (2014) 1, 77-138.
	
	
\bibitem{BGR2}
J.~Ben Geloun and S.~Ramgoolam,
``Tensor Models, Kronecker coefficients and Permutation Centralizer Algebras,''
 	JHEP 11 (2017) 092.

	
	\bibitem{TuckNort}
J.~Turner, J.~Morton,
``A Complete Set of Invariants for LU-Equivalence of Density Operators'',
SIGMA 13 (2017), 028. 

	\bibitem{BGR3}
J.~Ben Geloun and S.~Ramgoolam, ``All-orders asymptotics of tensor model observables from symmetries of restricted partitions'', J. Phys. A: Math. Theor. 55 435203, 2022.








\bibitem{MS04}
J.~Mingo and R.~Speicher,
``Second Order Freeness and Fluctuations of Random Matrices: I. Gaussian and Wishart matrices and Cyclic Fock spaces''. J. Funct. Anal., 235, 2006, pp. 226-270.

 
\bibitem{Voiculescu}
D.~Voiculescu,
 ``A strengthened asymptotic freeness result for random matrices with applications to free entropy'', 
International Mathematics Research Notices, Volume 1998, Issue 1, 1998, Pages 41–63. 

\bibitem{Biane}
P.~Biane, ``Some properties of crossings and partitions'', Discrete Mathematics 175 (1997), 41–53.


  
  \bibitem{Speicher94}
  R.~Speicher, ``Multiplicative functions on the lattice of non-crossing partitions and free convolution'', Mathematische Annalen 298 (1994), 611–628.









\bibitem{Weingarten}
  D.~Weingarten,
  ``Asymptotic behavior of group integrals in the limit of infinite rank,''
Journal of Mathematical Physics {\bf 19}, 999 (1978).


   \bibitem{ColSni}
  B.~Collins, P.~\'Sniady, 
  ``Integration with Respect to the Haar Measure on Unitary, Orthogonal and Symplectic Group,'' Commun. Math. Phys. 264, 773-795 (2006).
  


\bibitem{LL-higher}
L.~Lionni,
``From higher order free cumulants to non-separable hypermaps'',
[arXiv:2212.14885]. 
  

  



\bibitem{Deksen}
H.~Deksen   ``Polynomial bounds for rings of invariants'', proceedings of the AMS 129 - 4, 955–963
S 0002-9939(00)05698-7. 


\bibitem{ONproduct}
R.~C.~Avohou, J.~Ben Geloun, N.~Dub,
``On the counting of $O(N)$ tensor invariants'',
Adv. Theor. Math. Phys. 24 (2020) 4, 821-878. 

\bibitem{finite-free-cum}
O.~Arizmendi, D.~Perales, ``Cumulants for finite free convolution'', Journal of Combinatorial Theory, Series A 155, (2018), 244-266.


\bibitem{Gurau-Schaeffer}
R.~Gurau, G.~Schaeffer, ``Regular colored graphs of positive degree'', Ann. Inst. Henri Poincaré Comb. Phys. Interact. 3 (2016), no. 3, pp. 257–320. 

\bibitem{SYK}
V.~Bonzom, L.~Lionni, and A.~Tanasa, ``Diagrammatics of a colored SYK model and of an SYK-like tensor model, leading and next-to-leading orders'', J. Math. Phys., 58:052301, 2017.

\bibitem{SYK2}
E.~Fusy, L.~Lionni, and A.~Tanasa, ``Combinatorial study of graphs arising from the Sachdev-Ye-Kitaev model'', European Journal of Combinatorics, Volume 86, May 2020, 103066, 2018.




 
\bibitem{octa}
V.~Bonzom and L.~Lionni, ``Counting gluings of octahedra'', The Electronic Journal of Combinatorics, 24(3):P3.36, 2017.


\bibitem{Bonzom-balls}
V.~Bonzom, ``Maximizing the number of edges in three-dimensional colored triangulations whose
building blocks are balls'', 2018.  


 
 \bibitem{ledoux}
M.~Ledoux, ``The Concentration of Measure Phenomenon,'' Mathematical Surveys and Monographs, American Mathematical Society, 2001.


 
\bibitem{Male}
C.~Male, ``Traffic distributions and independence: permutation invariant random matrices and the three notions of
independence'', arXiv:1111.4662.



\end{thebibliography}
\end{document}